\documentclass[fleqn,a4paper,12pt]{article}

\usepackage[UKenglish]{babel} % Note the order (last is default)
\usepackage{graphicx,lscape}
\usepackage{enumerate}
\usepackage[figuresright]{rotating}
\usepackage[utf8]{inputenc}
\usepackage{private}
\usepackage{epstopdf}
\usepackage{hyperref}
\usepackage{bm}
\usepackage{url}
\usepackage{setspace}
\usepackage{float}
\usepackage{amsmath}
\usepackage{apacite}
\usepackage{natbib}
\usepackage{dcolumn}
\usepackage{booktabs}
\usepackage{tabularx}
\usepackage[labelfont=bf]{caption}
\newcolumntype{.}{D{.}{.}{-1}}
\usepackage[flushleft]{threeparttable}
\usepackage{color,psfrag}
\usepackage{soul}
\usepackage{nccmath}
\usepackage{subcaption}
\usepackage{longtable}

\usepackage{chngcntr}
\usepackage{apptools}
\usepackage{soul}
\usepackage[table]{xcolor}

\AtAppendix{\counterwithin{theorem}{section}}
\AtAppendix{\counterwithin{lemma}{section}}

\input{ee.sty}
\usepackage[margin=1in]{geometry}

\newtheorem{theorem}{Theorem}
\newtheorem{assumption}{Assumption}
\newtheorem{proposition}{Proposition}
\newtheorem{lemma}{Lemma}
\newtheorem{corollary}{Corollary}

\newtheorem{example}{Example}

\newcommand{\Var}{\operatorname{Var}}

\long\def\symbolfootnote[#1]#2{\begingroup%
\def\thefootnote{\fnsymbol{footnote}}\footnote[#1]{#2}\endgroup}

\graphicspath{{Figures/}}

\begin{document}
\begin{titlepage}

\title{Jackknife Instrumental Variable Inference\thanks{\scriptsize Federico Crudu's research benefited from the financial support provided by the University of Siena via the F-CUR grant 2274-2022-CF-PSR2021-FCUR 001 and by the National Recovery and Resilience Plan (NRRP), Mission 4, Component 2, Investment 1.1, call for tender No. 104 published on 2.2.2022 by the Italian Ministry of University and Research (MUR), funded by the European Union--NextGenerationEU--Project Title Stem in Higher Education \& Women INequalitieS [SHE WINS], CUP I53D23004810006, Grant Assignment Decree No. 1060 adopted on 07/17/2023 by the Italian Ministry of University and Research (MUR). Giovanni Mellace has benefited from financial support from the Independent Research Fund Denmark via the grant 
903800031B. This research has been conducted using the UK Biobank Resource under application number 183120. This work uses data provided by patients and collected by the NHS as part of their care and support. We are grateful to all UK Biobank participants who contributed their data to make this research possible. Individual-level data were accessed by Giovanni Mellace and Magnus Jørck-Thomsen at the University of Southern Denmark. The authors are grateful for helpful comments and suggestions from Jean-Marie Dufour, as well as from participants at the 2024 International Association for Applied Econometrics Conference, the Aarhus Workshop in Econometrics VII, the 10th edition of the Annual Scientific Conference of Romanian Academic Economists from Abroad in Cluj-Napoca,  the Rome Workshop on Econometrics: Frontiers in Causal Inference at EIEF, the Siena Workshop on Econometric Theory and Applications, the Tenth Italian Congress of Econometrics and Empirical Economics, and the Workshop on Causal Inference + Machine Learning at the University of Groningen, and from seminar participants at CEBA St. Petersburg State University, the Gran Sasso Science Institute, the University of Bristol, the University of Genova, and the University of Southern Denmark. We also thank Magnus Jørck-Thomsen for excellent research assistance.
}
}
%\title{\href{https://www.youtube.com/watch?v=yhS9LnDoo_w&ab_channel=ChemicalBrothersVEVO}{The Test}}
\vspace{1cm}
\author{Federico Crudu\footnote{\scriptsize Department of Economics and Statistics, Piazza San Francesco 7/8, 53100 Siena, Italy, federico.crudu@unisi.it}  \\
{\small Universit\`a di Siena and CRENoS}
\\
\and
Giovanni Mellace\footnote{\scriptsize Department of Economics, Campusvej 55, 5230 Odense M, Denmark, giome@sam.sdu.dk} \\
{\small University of Southern Denmark}
\\
\and
Zsolt S\'andor\footnote{\scriptsize Department of Business Sciences, Piata Libertătii 1,
530104 Miercurea Ciuc, Romania, sandorzsolt@uni.sapientia.ro} \\
{\small Sapientia Hungarian University of Transylvania}
  \\}
\vspace{1cm}
\date{\today}

\maketitle
%\vspace{10cm}

\thispagestyle{empty}
%\end{titlepage}
%%%%%%%%%%%%%%%%%%%%%%%%%%%%%%%%%%%%%%%%%%%%%%
%\doublespacing
\begin{abstract}
\noindent
This paper introduces a class of jackknife-based test statistics for linear regression models with endogeneity and heteroskedasticity in the presence of many potentially weak instrumental variables. The tests may be used when considering hypotheses on the full parameter vector or hypotheses defined as linear restrictions. We show that in the limit and under the null the proposed statistics are distributed as a combination of chi squares but by modifying the objective function we derive more familiar chi square limits. An extensive simulation study shows the competitive finite sample properties of the proposed tests in particular against Anderson-Rubin-type of statistics. Finally, we provide an empirical illustration that applies the proposed tests to study the effect of alcohol consumption on body mass index using genetic variants as instrumental variables using the UK Biobank.

\end{abstract}

\noindent
{\small
}

\end{titlepage}
\newpage
\doublespacing
%%%%%%%%%%%%%%%%%%%%%%%%%%%%%%%%%%%%%%%%%%%%%%%%%%%%%
\section{Introduction}

As emphasized by \cite{MS22}, many contemporary empirical applications of instrumental variables (IV) estimation involve a large number of instruments, often with limited identifying content per instrument. Beyond the settings discussed by these authors, similar many-instrument environments arise in shift-share (Bartik) designs that exploit finely disaggregated shocks \citep{BorusyakHullJaravel22,GoldsmithPinkhamSorkinSwift20}, in Mendelian randomization studies that use large sets of genetic variants as instruments \citep{Davies15,Sanderson22,PLB24}, and in recent machine learning–assisted IV approaches that construct instruments from high-dimensional predictors \citep{BelloniChernozhukovHansen12,ChernozhukovChetverikovDemirer18}. In such settings, instrument proliferation combined with weak first-stage relationships can undermine conventional IV inference, creating a demand for robust procedures that address not only estimation but also hypothesis testing. We address this need by developing a unified trinity of Wald, Lagrange multiplier, and distance-type tests for linear IV models with many potentially weak instruments.

The construction and asymptotic analysis of estimators and test statistics often rely on the definition of a given objective function. In this context, a typical analytical strategy consists of using some stochastic expansion argument and the application of a suitable central limit theorem (CLT) \citep[see, e.g.,][]{NM94}.
 Classical likelihood theory \citep{Engle84} and inference based on the generalised method of moments \citep[henceforth GMM; see][]{NW87, NW09} exploit this strategy to derive the well-known {\it trinity} of tests, Wald ($W$), Lagrange multiplier ($LM$) and likelihood ratio ($LR$) in maximum likelihood, and the corresponding distance statistic ($D$) in GMM. In line with this idea, we develop an analogous trinity of tests for linear IV models with many weak instruments. 

 In particular, this paper proposes a unified approach to inference for linear regression models with multiple endogenous and exogenous variables, heteroskedastic disturbances and many (potentially weak) instrumental variables (IVs). A jackknife device is typically applied to remove the noise component due to the presence of heteroskedasticity while, ideally, retaining the whole signal (matrix) of the instruments  \citep[see, among others,][]{HNWCS12,CHNSW14,BC15,CMS21,MS22,Yap23,MO22}.\footnote{See \cite{PH77} and \cite{AIK99} for earlier applications of the jackknife to IV models and \cite{BD99} and \cite{DM06} for extensive simulation studies on the properties of jackknife IV estimators in finite samples.} 
 In our setup, we allow the number of instruments, say, $k$ to grow proportionally with the sample size, say, $n$ and $k\le n$. This is a standard assumption in this literature. There are, however, some studies that allow the number of instruments to be larger than the sample size \citep[see, for example,][]{HK14,DKM24}. Furthermore, while the number of instruments is allowed to diverge, we impose a restriction on the minimal amount of signal the instruments should collectively carry. As shown in \cite{MS22}, this is a necessary condition to obtain consistent IV estimators and tests.\footnote{The sufficiency part was proven in \cite{CS05}.}
 Our test statistics are derived from two types of objective functions.
 The first type is defined by quadratic forms, such as those underlying the jackknife IV estimators JIVE1 and JIVE2 \citep[][]{AIK99, CMS21,MS22}.
 The second type is defined by ratios of quadratic forms, such as the objective function for the limited information maximum likelihood (LIML) estimator \citep[see, among others,][]{AR49, Bekker94, HHN08} and the more general jackknife-based alternatives heteroskedasticity-robust LIML (HLIM) and symmetric JIVE (SJIVE) of \cite{HNWCS12} and \cite{BC15}, respectively.\footnote{The LIML case is the object of a companion paper.} We show that the tests directly related to objective functions are asymptotically distributed as a weighted average of chi-squares, sometimes referred to as chi-bar-square distribution \citep{Vuong89, Hansen21}. The weights of this distribution can be estimated and p values can be computed \citep{Farebrother84}. However, by appropriately modifying the objective function, it is possible to obtain more conventional chi-square distributions in the limit. 
 
Via extensive Monte Carlo experiments, we investigate the tests' finite-sample properties in terms of size and power. Although the simulations do not display a clear ranking of the tests, we can still identify specific scenarios where certain statistics outperform the others. For example, when instruments are weak, some $LM$ statistics suffer from significant loss of power against distant alternatives, whereas $W$ may have low power near the null. By contrast, $D$ tends to have good size and power properties across a range of designs. Furthermore, compared to the Anderson-Rubin ($AR$) tests of \cite{CMS21} and \cite{MS22}, $D$ and $W$ exhibit higher power in our experiments. In the JIVE case, these statistics are relatively easy to compute (compared to SJIVE and HLIM) and, in some cases in our simulations, they outperform $AR$ both in terms of size and power. Moreover, our tests are computationally less burdensome in our implementations then their $AR$ counterparts.

We further illustrate the practical relevance of our methodology through an empirical application that studies the causal effect of alcohol consumption on body mass index using genetic variants as instrumental variables in the UK Biobank. This setting naturally involves a large number of instruments, a feature that raises concerns about instrument strength and the reliability of conventional IV inference. In this context, the proposed statistics can be implemented straightforwardly and allow us to conduct inference without relying on standard approximations that may perform poorly when instruments are numerous or weak. The empirical results therefore complement the simulation findings by showing that the proposed procedures can be effectively applied in a realistic setting with many potentially weak instruments.

The main contribution of this paper is fourfold: first, we introduce a trinity of tests for linear IV models robust to many weak instruments; second, we establish the limiting distribution of these tests under the null; third, we compare their finite-sample performance and highlighting scenarios of dominance; four, we illustrate the applicability of the proposed tests in a realistic setting with many potentially weak instruments through an empirical study that uses genetic variants to estimate the effect of alcohol consumption on body mass index.
 Alongside the main contributions, we provide two auxiliary results that could be of independent interest to the readership: first, we prove a general consistency result for constrained and unconstrained jackknife-based estimators; second, we introduce a cross-fit variance estimator for the $AR$ test of \cite{CMS21} along the lines of that proposed by \cite{MS22}.\footnote{See Proposition \ref{Pcons} in Appendix \ref{Auxres} and equation \eqref{Phicf} in Section \ref{implem}, respectively.}

 The results in this paper are largely new, particularly for tests based on ratios of quadratic forms and the jackknife approach of \cite{BC15}. They also generalise earlier results, such as the $LM$ test of \cite{HHN08} and Bekker's LIML statistic \citep{Kleibergen02}, which are developed under the assumption of homoskedasticity. There are, however, some interesting overlaps. In particular, the $LM$ test built on the JIVE2 objective function coincides with that of \cite{MO22}. Moreover, the Wald test based on the JIVE2 objective function is related to the results in \cite{Yap23}, while the Wald test based on the HLIM objective function was first introduced by \cite{HNWCS12}. This unified treatment of $W$, $LM$, and $D$ tests (the latter being novel) under many weak instruments is, to our knowledge, new. Finally, the modification applied to the objective functions to obtain a standard chi-square limit is reminiscent of that used in \cite{Kleibergen02} to obtain the $K$ statistic \citep[see also][for a many instrument version of the $K$ statistic]{BK03}. Also in our case, the degrees of freedom of the chi-square distribution do not depend on the number of instruments.

%\textcolor{red}{\bf add one sentence about the connection of our results with those in the early literature of the 90/00 and Kean and Neal. Cite also Bound et al. (1995) and Nelson and Startz (1990)?}
Our analysis also relates to the weak instrument literature of the 1990s and early 2000s, which showed that standard Wald-type procedures can break down under weak identification and thereby spurred the development of robust alternatives \citep[among others]{NS90,BJB95,Dufour97,SS97,WZ98,Kleibergen02}. The present paper complements this literature, and recent contributions such as \cite{KN23}, by showing how analogous concerns can be addressed in linear IV models with many potentially weak instruments through a unified testing framework.

 As noted, several studies have considered $W$ and $LM$ statistics, but, to the best of our knowledge, the $D$ statistic in this paper, both in the chi-bar-square version and in the chi-square version, is new. This fills a gap in the weak-instrument literature, which so far has mostly developed $AR$-, $W$- and $LM$-type tests but lacked an analogue of the $LR$ test. The $D$ statistic we propose provides such an analogue, with promising size and power properties when compared to its competitors.

% %In this context, a relevant aspect is that of JIVE objective functions and for the chi-bar-square limit the $D$ statistic is numerically equal to the $W$ statistic \textcolor{red}{\bf it should be $D=W=LM$ but only $D$ and $W$ produce equal p values, this should be due to the fact that the nuisance parameters of the chi-bar-square distribution are estimated differently for $D$/$W$ and $LM$, I have to double check that to be sure}. For the chi-square limit, this is not the case, but, from our simulations we see that $D$ and $LM$ tend to be numerically very close. 

 There are clear similarities between the conditional $LR$ ($CLR$) statistic introduced by \cite{Moreira03} and its more recent versions designed to allow for instrument numerosity \citep[see][]{AMO22,LWZ22}, in that both $D$ and $CLR$ are based on a discrepancy between objective functions evaluated at different values. However, there is also a difference, namely that $D$ does not rely on a conditional argument. 

 Moreover, to the best of our knowledge, no previous work in this setting explicitly considers general linear restrictions (including the important case of testing a subset of parameters). This may be an important aspect, as in a conventional linear regression model one may partial out variables, via the Frisch-Waugh-Lovell (FWL) theorem, and conduct inference on the remaining regressors. However, as pointed out by \cite{CSW23} and \cite{MS24}, the FWL theorem prevents the jackknife device from removing the noise component \citep[see also][for a similar result in the context of least squares estimation and time series data]{MS25}. Introducing an inferential method for linear restrictions avoids possible issues connected to the interplay between the jackknife approach and the FWL theorem.

The remainder of the paper is organised as follows. Section \ref{model} introduces the model and the test statistics for simple hypotheses and linear restrictions and shows the asymptotic behaviour of the tests. Section \ref{chi2section} describes how to modify our statistics to have a chi square limiting distribution. Section \ref{numerical} studies the finite sample properties of the tests in a Monte Carlo experiment and in an empirical application using genetic instruments. Finally, Section \ref{conclusion} offers some conclusions. Proofs and auxiliary results are deferred to the Appendix.

 Throughout the paper we use the following notational conventions: unless differently stated $a$, $\va$ and $\mA$ denote a scalar, a vector and a matrix, respectively; $a_i$ is the $i$-th element of vector $\va$ and $A_{ij}$ is the $(i,j)$-th element of matrix $\mA$. $\va'$ and $\mA'$ are the transposes of $\va$ and $\mA$. For a square matrix $\mA$, $\lambda(\mA)$ denotes the vector of eigenvalues of matrix $\mA$, while  $\lambda_{\min}(\mA)$ and $\lambda_{\max}(\mA)$ are the smallest and largest eigenvalues of $\mA$ respectively. Moreover, $\rk(\mA)$ denotes the rank of matrix $\mA$. Finally, $\mI$ is an identity matrix and when we need to be more specific, we include a subindex. Thus, $\mI_n$ denotes a $n\times n$ identity matrix; similarly, $\viota$ and $\viota_n$ denote vectors of ones. The symbol $*$ is used for elementwise or Hadamard multiplication and $^{(n)}$ denotes elementwise power. Hence,
\begin{align*}
\underbrace{\va*\va*\dots*\va}_{n\; times}=\va^{(n)}\;\;
\mbox{and}\;\;
\underbrace{\mA*\mA*\dots*\mA}_{n\; times}=\mA^{(n)}.
\end{align*} 
Analogously, $\mA^{(-1)}$ is used for the elementwise inverse of $\mA$, where the generic $(i,j)$ entry of $\mA^{(-1)}$ is equal to $A_{ij}^{-1}$.

%%%%%%%%%%%%%%%%%%%%%%%%%%%%%%%%%%%%%%%%%%%%%%%%%%%%%%

\section{Model and test statistics}\label{model}
Let us consider the model
\begin{align*}
\vy&=\mX\vbeta+\vvarepsilon\\
\mX&=\mZ\mPi+\mV
\end{align*}
where $\mX$ is a $n\times g$ matrix containing potentially both endogenous and exogenous variables and $\mZ$ is a $n\times k$ (nonstochastic) matrix of instruments and $\E[\mX]=\mZ\mPi$, where the components of $\mPi$ are allowed to vary with the sample size $n$. Such assumptions are made for convenience and may be generalised.\footnote{
We may, for example, consider $\mZ$ to be stochastic and in this case $\E[\mX]$ should be interpreted as a conditional expectation with respect to $\mZ$. The linearity of $\E[\mX]$ may also be relaxed as suggested in, e.g.,  \cite{Bekker94} and \cite{CHNSW14}.}
%See e.g. \cite{Bekker94}, \cite{HNWCS12}, \cite{CHNSW14}. 
The rows of the disturbance couple $(\vvarepsilon,\mV)$, say $(\varepsilon_i,\mV_i')$ $i=1,\dots,n$, are independent with zero mean and covariance matrices
\begin{align}\label{Sigmai}
\mSigma_i=
\begin{pmatrix}
\sigma_i^2 & \vsigma_{i12}\\
\vsigma_{i21} & \mSigma_{i22}
\end{pmatrix}.
\end{align}
%while the covariance matrix of the rows $(y_i,\mX_i')$ are 
% \begin{align}\label{Omegai}
%\mOmega_i=
%\begin{pmatrix}
%1 & \vbeta'\\
%\vzeros & \mI_g
%\end{pmatrix}
%\mSigma_i
%\begin{pmatrix}
%1 & \vzeros\\
%\vbeta & \mI_g
%\end{pmatrix}.
%\end{align}

\subsection{Test for the whole parameter vector}\label{tests}
In this section we are interested in test statistics for the null hypothesis
\begin{align}\label{simplenull}
H_0:\vbeta=\vbeta_0
\end{align}
based on a certain objective function $Q_n(\vbeta)$ that is assumed to produce consistent estimators in the many instruments sense and to be robust to the presence of heteroskedasticity. We consider two types of objective functions. The first type is based on the ratio of two quadratic forms as in \cite{HNWCS12} and \cite{BC15}; specifically,
\begin{align}\label{objfun}
Q_n(\vbeta)=\frac{(\vy-\mX\vbeta)'\mC(\vy-\mX\vbeta)}{\frac{1}{\tr(\mB)}(\vy-\mX\vbeta)'\mB(\vy-\mX\vbeta)}
\end{align}
where, for the SJIVE estimator of \cite{BC15}, we have
\begin{align}\label{jackmat}
\mC&=\mP+
\left(\mP\mD(\mI-\mD)^{-1}\mP-\frac{1}{2}\mP\mD(\mI-\mD)^{-1}-\frac{1}{2}\mD(\mI-\mD)^{-1}\mP\right)-\mB\\\nonumber
\mB&=(\mI-\mP)\mD(\mI-\mD)^{-1}(\mI-\mP)
\end{align}
where $\mD=\diag(\mP)$ is a diagonal matrix that contains the diagonal elements of $\mP=\mZ(\mZ'\mZ)^{-1}\mZ'$. We note that in this case $\tr(\mB)=\tr(\mP)=k$. For the HLIM estimator of \cite{HNWCS12} we specify
\begin{align}\label{jackmat2}
\mC=\mP- \mD, \;\; \mB=\mI.
\end{align}
Such estimators mimic the ratio of quadratic forms structure typical of LIML. Another possibility would be to define an objective function that only considers the numerator of equation \eqref{objfun}:
\begin{align}\label{objfun2}
Q_n(\vbeta)=(\vy-\mX\vbeta)'\mC(\vy-\mX\vbeta).
\end{align}
The resulting estimator is the JIVE1 if $\mC$ is chosen as in equation \eqref{jackmat} or the JIVE2 if $\mC$ corresponds to the specification in equation \eqref{jackmat2}.\footnote{Both specifications of $\mC$ remove the source of bias that would make an estimator inconsistent in the many instruments sense when heteroskedasticity is present. However, they treat the information contained in the IVs differently. In fact, for $\mC$ defined as in  equation \eqref{jackmat} $\E[\mX'\mC\mX]=\mPi'\mZ'\mZ\mPi$, while the specification in equation \eqref{jackmat2} produces $\E[\mX'\mC\mX]=\mPi'\mZ'\mZ\mPi-\mPi'\mZ'\mD\mZ\mPi$. See \cite{BC15} for a discussion.} \footnote{We label the JIVE case with $\mC$ as in equation \eqref{jackmat} as JIVE1. However, it is important to notice that this is a symmetric version of the JIVE1 proposed in \cite{AIK99} and it mimics the structure of the SJIVE objective function in \cite{BC15}. To the best of our knowledge, the corresponding estimator was first introduced in \cite{CMS21}.}
Whether we use the objective functions in equation \eqref{objfun} or equation \eqref{objfun2}, the estimator
\begin{align}\label{estim}
\widehat\vbeta=\arg\min_{\vbeta}Q_n(\vbeta)
\end{align}
is consistent for $\vbeta$ under $H_0$.  

To test the null hypothesis in equation \eqref{simplenull} we consider the following statistics
\begin{align}\label{D}
D &= -r_{\min}\frac{\sigma^2}{k}\left(Q_n(\widehat\vbeta)-Q_n(\vbeta_0)\right)\\\label{LM}
LM &= r_{\min} \vtau_n'\mH^{-1}\vtau_n\\\label{W}
W &= r_{\min} \vpsi_n'\mH^{-1}\vpsi_n
\end{align}
where $\mH=\mPi'\mZ'\mC\mZ\mPi$, $r_{\min}=\lambda_{\min}(\mH)$ is the smallest eigenvalue of $\mH$ and $\sigma^2$ is the top left entry of matrix $\mSigma=\lim\frac{1}{\tr(\mB)}\sum_{i=1}^nB_{ii}\mSigma_i$ as $n\to\infty$.\footnote{We note that in the SJIVE case $\mH=\mPi'\mZ'\mZ\mPi$ while in the HLIM case $\mH=\mPi'\mZ'\mZ\mPi-\mPi'\mZ'\mD\mZ\mPi$. Notice also that $\mSigma$ depends on the choice of $\mB$. To avoid notation clutter, we omit this dependency.} Furthermore, 
%\begin{align}\label{tau}
%\vtau_n=\textcolor{red}{+}\frac{1}{\sqrt k} \mX'\widehat\mC(\vbeta_0)\vvarepsilon, \text{ \ where}
%\end{align}
%\begin{align}\label{C-hat}
%\widehat\mC(\vbeta)=\mC-\lambda(\vbeta)\mB\text{ \ with } \lambda(\vbeta)=\frac{1}{\tr(\mB)}Q_n(\vbeta) \text{ \ and}
%\end{align}
%\begin{align}\label{psi}
% \vpsi_n=\frac{1}{\sqrt k} \mH(\widehat\vbeta-\vbeta_0).
%\end{align}

\begin{align}
\vtau_{n} &  = \frac{1}{\sqrt{k}}\mX^{\prime}\widehat{\mC}%
(\vbeta_{0})\vvarepsilon,\text{ \ where}\label{tau}\\
\widehat{\mC}(\vbeta) &  =\mC-\lambda(\vbeta)\mB\text{ \ with }\lambda
(\vbeta)=\frac{1}{\tr(\mB)}Q_{n}(\vbeta)\text{ \ and}\label{C-hat}\\
\vpsi_{n} &  =\frac{1}{\sqrt{k}}\mH(\widehat{\vbeta}-\vbeta_{0}).\label{psi}%
\end{align}

When $Q_n(\vbeta)$ is the objective function of a JIVE estimator (equation \eqref{objfun2}), $\sigma^2=1$ and $\lambda(\vbeta)=0$.
The limiting distribution of the tests in equations \eqref{D} to \eqref{W} is derived in Theorem \ref{basicresult} below under a set of assumptions. 

The assumptions we use match those in \cite{BC15} but different versions can be found in the literature. \cite{CMS21} provide a discussion of some of the assumptions; below we discuss the other assumptions. In what follows it is understood that $c_u$ is a generic positive constant that may take different values in different instances.
\begin{assumption}\label{ass1}
The generic diagonal element $P_{ii}$ of the projection matrix $\mP$ satisfies $\max_i{P_{ii}}\leq 1-1/c_u$, with $1<c_u<\infty$. In addition, $k\to\infty$ as $n\to\infty$. 
\end{assumption}
\begin{assumption}\label{ass2}
$\E[\varepsilon^4_i]\le c_u$ and $\E\left[ \left\Vert \mV_{i}\right\Vert ^{4}\right] \leq c_{u}$ with $0<c_u<\infty$, for any $i$. 
\end{assumption}
Let $r_{\max}=\lambda_{\max}(\mH)$ denote the largest eigenvalue of $\mH$. %Moreover, let us define a constant $\kappa$ such that $0\le\kappa<\infty$.
\begin{assumption}\label{ass5}
$\sqrt{k}/r_{\min}\to0$ and $r_{\max}/k$ is bounded when $n\to\infty$. %$k/r_{\max}\to\infty$, $\sqrt{k}/r_{\min}\to 0$ \textcolor{red}{\bf there is a different set of assumptions in the published version of the BC paper which may be more general for the estimation part: }.
\end{assumption}
%\noindent Assumption \ref{ass5} formalizes the fact that we are dealing with potentially weak instruments. 

\begin{assumption}\label{ass3}
The covariance matrices of the disturbances %are bounded, $0\le\mSigma_i\le c_u\mI_{1+g}$ and 
satisfy $\frac{1}{\tr(\mB)}\sum_{i=1}^nB_{ii}\mSigma_i\to\mSigma$ as $n\to\infty$ where $\mSigma=\left(
\begin{array}
[c]{cc}%
\sigma^{2} & \vsigma_{12}\\
\vsigma_{21} & \mSigma_{22}%
\end{array}
\right)  $ is positive definite.%\footnote{\textcolor{blue}{This assumption does not allow for the possibility of including exogenous regressors in the model due to the fact that it requires $\mSigma_{22}$ to be positive definite, which excludes that some components of the $\mV_{i}$'s are zero. However, it is straightforward to modify Assumption \ref{ass3} for the case when some components of the $\mV_{i}$'s are zero, so we do not provide it here.}}
%\textcolor{red}{NOT NEEDED: can be decomposed as $\mSigma_i$ in equation \eqref{Sigmai} and $\sigma_i^2>\underline{\sigma}^2>0$. }
%We denote $\sigma^2$ as the top left entry of $\mSigma$.
\end{assumption}

\begin{assumption}\label{ass5n} 
$\frac{1}{k^{2}}\sum_{i=1}^{n}\left\Vert \mPi^{\prime}\mZ^{\prime
}\mC_{i}\right\Vert ^{4}\rightarrow0$ holds when $n\rightarrow\infty$, where
$\mC_{i}$ is column $i$ of matrix $\mC$.
\end{assumption}

\begin{assumption}
\label{ass6} The matrix sequence $\frac{1}{k}(\mF+\mG)$ is convergent when $n\rightarrow\infty$, where
\begin{align}
\mF  &  =\mPi^{\prime}\mZ^{\prime}\mC\mD_{\sigma^{2}}\mC\mZ\mPi,\label{G}\\
\mG  &  =\sum_{i\neq j}C_{ij}^{2}\left(  \sigma_{j}^{2}\mSigma_{i22}^{\ast
}+\vsigma_{i21}^{\ast}\vsigma_{j12}^{\ast}\right)  ,
\end{align}
$\mD_{\sigma^{2}}$ a diagonal matrix containing $\sigma_{i}^{2}$ in its diagonal
entries and $\mSigma_{i}^{\ast}=\Var\left[  \left(  \varepsilon_{i},\mV_{i}^{\prime}
-\frac{\varepsilon_{i}\vsigma_{12}}{\vsigma^{2}}\right)^{\prime}  \right]  $, and $\lim\frac{1}{k}(\mF+\mG)=\mPhi$.
\end{assumption}

%\rule{\linewidth}{0.4pt}
%\textcolor{red}{\bf between the lines is a comment on the assumptions} 
Some comments are in order. Assumption \ref{ass1} and Assumption \ref{ass2} are standard conditions commonly found in the literature \citep[see, e.g., ][]{HNWCS12,BC15,CMS21}: the former is a technical condition on the behaviour of the diagonal elements of the projection matrix $\mP$ and the latter is a regularity condition on the boundedness of the second moments of the disturbances. Assumption \ref{ass5} formalises the fact that we are dealing with potentially weak instruments. 
Assumption \ref{ass3} does not allow for the possibility of including exogenous regressors in the model due to the fact that it requires $\mSigma_{22}$ to be positive definite, which excludes that some components of the $\mV_{i}$'s are zero. However, it is straightforward to modify this assumption for the case when some components of the $\mV_{i}$'s are zero, so we do not provide it here.
\noindent Assumption \ref{ass5n} corresponds to Assumption 5 from \cite{HNWCS12}. 
%Although a direct comparison is not easy, we can see that
%\begin{align*}
%\sum_{i=1}^{n}\left\Vert \mPi^{\prime}\mZ^{\prime}\mC_{i}\right\Vert ^{2}%
%&=\sum_{i=1}^{n}\mC_{i}^{\prime}\mZ\mPi\mPi^{\prime}\mZ^{\prime}\mC_{i}=\sum_{i=1}%
%^{n}\tr\left(  \mPi^{\prime}\mZ^{\prime}\mC_{i}\mC_{i}^{\prime}\mZ\mPi\right)\\ & =\tr\left(
%\mPi^{\prime}\mZ^{\prime}\mC^{2}\mZ\mPi\right)  \leq c_{u}\tr\left(  \mPi^{\prime
%}\mZ^{\prime}\mZ\mPi\right)
%\end{align*}
%\citep[see, e.g., (A.6) from the Supplement to ][]{CMS21}. In Assumption 5 from
%\cite{HNWCS12} $\vz_{i}$ corresponds to our $\mPi^{\prime}\mZ_{i}$, so $\tr\left(
%\mPi^{\prime}\mZ^{\prime}\mZ\mPi\right)  =\sum_{i=1}^{n}\left\Vert \vz_{i}\right\Vert
%^{2}$, and therefore,
%$\sum_{i=1}^{n}\left\Vert \mPi^{\prime}\mZ^{\prime}\mC_{i}\right\Vert ^{2}\leq
%c_{u}\sum_{i=1}^{n}\left\Vert \vz_{i}\right\Vert ^{2}$. 
Notice that
{\footnotesize
\begin{align*}
\frac{1}{k^{2}}\sum_{i=1}^{n}\left\Vert \mPi^{\prime}\mZ^{\prime}\mC_{i}\right\Vert
^{4}=\frac{1}{k^{2}}\sum_{i=1}^{n}\left(  \mC_{i}^{\prime}\mZ\mPi\mPi^{\prime
}\mZ^{\prime}\mC_{i}\right)  ^{2}=\frac{1}{k^{2}}\sum_{i=1}^{n}\tr\left(
 \mC_{i}^{\prime}\mZ\mPi\mPi^{\prime}\mZ^{\prime}\mC_{i} \mC_{i}^{\prime}\mZ\mPi\mPi^{\prime
}\mZ^{\prime}\mC_{i}\right).
\end{align*}
}%
Since $\sum_{i=1}^{n}\mC_{i}\mC_{i}^{\prime}=\mC^{2}$ and $\mC_{i}\mC_{i}^{\prime}$ is
positive semidefinite, $\mC_{i}\mC_{i}^{\prime}\leq \mC^{2}$ for any $i$.
This and $\mZ^{\prime}\mC^{2}\mZ\leq c_{u}\mZ^{\prime}\mZ$ \citep[see (A.6) from the Supplement to ][]{CMS21} imply that
\begin{align*}
\frac{1}{k^{2}}\sum_{i=1}^{n}\left\Vert \mPi^{\prime}\mZ^{\prime}\mC_{i}\right\Vert
^{4}&\leq
\frac{c_{u}}{k^{2}}\tr\left(\mZ\mPi\mPi^{\prime
}\mZ^{\prime}\mZ\mPi\mPi^{\prime}\mZ^{\prime}\right),
\end{align*}
so we obtain that $\frac{1}{k^{2}}\sum
_{i=1}^{n}\left\Vert \mPi^{\prime}\mZ^{\prime}\mC_{i}\right\Vert ^{4}\leq
\frac{c_{u}}{k^{2}}\tr\left(  \mH^{2}\right)  $. This goes to $0$ if $r_{\max
}/k\rightarrow0$, so a primitive condition for Assumption \ref{ass5n} is that
the instruments are weak. Another primitive condition, which is also valid for
strong instruments, is $\mC_{i}\mC_{i}^{\prime}\leq k^{-\alpha}\mC^{2}$ for any $i$
and for some $\alpha>0$. This means that there is no observation $i$ for which
$\mC_{i}\mC_{i}^{\prime}$ is excessively large. In this case $\frac{1}{k^{2}}%
\sum_{i=1}^{n}\left\Vert \mPi^{\prime}\mZ^{\prime}\mC_{i}\right\Vert ^{4}\leq
\frac{k^{-\alpha}}{k^{2}}\tr\left(  \mH^{2}\right)  $, which goes to $0$ by Assumption \ref{ass5}.
Finally, Assumption \ref{ass6} corresponds to Assumption 6 from \cite{HNWCS12}.

%\rule{\linewidth}{0.4pt}

The next result shows that the test statistics specified above are distributed asymptotically as a weighted sum of chi squares.
\begin{theorem}
\label{basicresult} Let $T\in\{D,LM,W\}$. If Assumptions \ref{ass1} to
\ref{ass6} are satisfied and in addition $r_{\min}\mH^{-1}$  is convergent, then under the null hypothesis $T\rightarrow
_{d}\vzeta^{\prime}\mXi\vzeta\sim\bar{\chi}^{2}(\vvarphi)$ for
$\vzeta\sim\mathcal{N}(\vzeros,\mPhi)$, where $\mPhi$ is given in Assumption \ref{ass6}, $\mXi=\lim r_{\min}\mH^{-1}$ and $\vvarphi$ is the vector of
eigenvalues of $\mXi\mPhi$.
\end{theorem}

The statistics $D$, $LM$, $W$ are not feasible as they are based on the
unknown quantities $\mH$ and $\sigma^{2}$. By replacing these by their feasible
counterparts%
\begin{equation}
\widehat{\mH}=\mX^{\prime}\widehat{\mC}(\widehat{\vbeta})\mX\text{ \ \ and
\ \ }\widehat{\sigma}^{2}=\dfrac{(\vy-\mX\widehat{\vbeta})^{\prime}%
\mB(\vy-\mX\widehat{\vbeta})}{\tr\left(  \mB\right)  },\label{Hh_s2}%
\end{equation}
respectively, we obtain the feasible statistics%
\begin{align*}
\widehat{D} &  =-\widehat{r}_{\min}\dfrac{\widehat{\sigma}^{2}}{k}\left(
Q_{n}(\widehat{\vbeta})-Q_{n}\left(  \vbeta_{0}\right)  \right)  ,\\
\widehat{LM} &  =\widehat{r}_{\min}\vtau_{n}^{\prime}\widehat{\mH}^{-1}\vtau
_{n},\\
\widehat{W} &  =\widehat{r}_{\min}\widehat{\vpsi}_{n}^{\prime}\widehat{\mH}%
^{-1}\widehat{\vpsi}_{n},
\end{align*}
where $\widehat{r}_{\min}=\lambda_{\min}(\widehat{\mH})$ and $\widehat{\vpsi}%
_{n}=\dfrac{1}{\sqrt{k}}\widehat{\mH}(\widehat{\vbeta}-\vbeta_{0})$.

\begin{corollary}
\label{feasibleresult}Let $\widehat{T}\in\{\widehat{D},\widehat{LM}%
,\widehat{W}\}$. If the assumptions of Theorem \ref{basicresult} hold, then
under the null hypothesis $\widehat{T}\rightarrow_{d}\vzeta^{\prime}\mXi
\vzeta\sim\bar{\chi}^{2}(\boldsymbol{\varphi})$ for $\vzeta\sim\mathcal{N}%
(\vzeros,\mPhi)$, where $\mXi$, $\boldsymbol{\varphi}$, and $\mPhi$ are as in Theorem
\ref{basicresult}.
\end{corollary}

It is interesting to notice that test statistics based on quadratic forms such as Wald and $LM$ can be easily modified to obtain a more conventional chi square limiting distribution \citep[see also][]{MO22}. This seems not to be the case for $D$ statistics, unless we take into account the simple case with one endogenous variable ($g=1$). The following example illustrates that case.
\begin{example}[A model with one endogenous variable]
Let us now consider the case of the simple linear regression model with $g=1$:
\begin{align*}
\vy&=\vx\beta+\vvarepsilon\\
\vx&=\mZ\vpi+\vv.
\end{align*}
The model is simple but it is commonly encountered in empirical applications \citep[see the review of the literature in][]{ASS19}. The simplification allows to write $r_{\min}=r_{\max}=r$. In this case the asymptotic distribution of the test statistics in equations (\ref{D}) to (\ref{W}) is $\chi^2_1$, as stated in Corollary \ref{basicresult3} below. %\textcolor{red}{Federico: should we adapt the assumptions? Notation should also be adapted, e.g., define $\varphi$.}
\begin{corollary}\label{basicresult3}
Let $T\in \{D, LM, W\}$. In the single endogenous regressor case when $g=1$, under the assumptions of Theorem \ref{basicresult} $T \to_d \zeta^2$ for $\zeta\sim\mathcal N(0,\varphi^2)$. Furthermore, $\frac{T}{\varphi^2} \to_d \chi^2_1$.
%If assumptions \ref{ass1} to \ref{ass2} and $\frac{r}{\sqrt k}\to \infty$ hold, then $T \to_d \zeta^2$ for $\zeta\sim\mathcal N(0,\varphi^2)$. Furthermore, $\frac{T}{\varphi^2} \to_d \chi^2_1$. 
\end{corollary}
Notice that $\varphi^2$ is the scalar version of the variance covariance matrix $\mPhi$ defined in Assumption \ref{ass6}
\end{example}

\subsection{General linear restrictions}
%\textcolor{red}{\bf remove from here\dots}Le us consider the general case where the null hypothesis is expressed in terms of a potentially non linear function of $\vbeta$. Hence,
%\begin{align}\label{nlrestrictions}
%H_0:\va(\vbeta_0)=\vzeros
%\end{align}
%where $\va:\mathbb R^g\to\mathbb R^p$ with $g\ge p$.
%%The corresponding restricted estimator is
%%\begin{align}\label{estrestricted}
%%\widetilde\vbeta=\arg\min_{\vbeta} Q(\vbeta) \mbox{  subject to  } \va(\vbeta)=\vzeros.
%%\end{align}
%The constrained estimator $\widetilde\vbeta$ is given by the optimization of the Lagrangian
%\begin{align}\label{estrestricted}
%\mathcal L=  Q_n(\vbeta)+ \vgamma' \va(\vbeta)
%\end{align}
%The consistency of $\widetilde\vbeta$ can be derived adapting the arguments in Theorem 9.1 in \cite{NM94} assuming some form of asymptotic locality for the alternative hypothesis.\footnote{Typically we say that the alternative hypothesis is asymptotically local if $\va(\vbeta_0)=\frac{\vdelta}{\sqrt n}$. In our case the $\sqrt n$ rate should be replaced with an appropriate rate.} \textcolor{red}{\bf \dots to here}
Let us assume that the null hypothesis is defined as a set of linear equality
restrictions
\begin{equation}
H_{0}:\mA\vbeta=\va \label{lrestrictions}%
\end{equation}
with alternative hypothesis $H_{1}:\mA\vbeta\neq \va$, where $\mA$ is a $p\times g$
matrix with $g\geq p$ and $\rk(\mA)=p$. The corresponding Lagrangian is
\[
\mathcal{L}(\vbeta,\vgamma)=Q_{n}(\vbeta)+2\vgamma^{\prime}\left(  \mA\vbeta
-\va\right)  ,
\]
where $Q_{n}(\vbeta)$ is defined as in equation \eqref{objfun}. The first
order conditions (FOCs) are (see also footnote on page \pageref{ftn})
\begin{align}
\frac{\partial\mathcal{L}}{\partial\vbeta}  &  =\frac{-2\mX^{\prime}%
\widehat{\mC}(\vbeta)(\vy-\mX\vbeta)}{\frac{1}{\tr(\mB)}(\vy-\mX\vbeta)\mB(\vy-\mX\vbeta)}%
+2\mA^{\prime}\vgamma=\dfrac{-2}{\widehat{\sigma}^{2}({\vbeta})}\mX^{\prime}\widehat{\mC}%
({\vbeta})(\vy-\mX{\vbeta})+2\mA^{\prime}\vgamma,\label{dLb}\\
\frac{\partial\mathcal{L}}{\partial\vgamma}  &  =\mA\vbeta-\va, \label{dLg}%
\end{align}
where recall $\widehat{\mC}(\vbeta)=\mC-\lambda(\vbeta)\mB$ and $\lambda(\vbeta)=\frac{1}%
{k}Q_{n}(\vbeta)$. Let us denote by $\widetilde{\vbeta}$ and
$\widetilde{\vgamma}$ the values of $\vbeta$ and $\vgamma$ that solve the FOCs
(i.e., first order conditions). The test statistics for the null in equation
\eqref{lrestrictions} are
\begin{align}
D_{a}  &  =-r_{\min}\frac{\sigma^{2}}{k}\left(  Q_{n}(\widehat{\vbeta}%
)-Q_{n}(\widetilde{\vbeta})\right)  ,\label{Da}\\
LM_{a}  &  =r_{\min}\vxi_{n}^{\prime}\mH^{-1}\vxi_{n},\label{LMa}\\
W_{1a} &  =\frac{r_{\min}}{k}(\mA\widehat{\vbeta}-\va)^{\prime}(\mA\mH^{-1}\mA^{\prime})^{-1}(\mA\widehat{\vbeta}-\va),\label{W1a}\\
W_{2a} &  =r_{\min}\bm\vartheta_{n}^{\prime}\mH^{-1}\bm\vartheta_{n},\label{W2a}%
\end{align}
where
\begin{equation}
\vxi_{n}=\frac{1}{\sqrt{k}}\mX^{\prime}\widehat{\mC}(\widetilde{\vbeta
})(\vy-\mX\widetilde{\vbeta}) \label{ksi}%
\end{equation}
and
\begin{equation}
\bm\vartheta_{n}=\frac{1}{\sqrt{k}}\mH(\widehat{\vbeta}-\widetilde{\vbeta}). \label{theta}%
\end{equation}

%\st{ and $\boldsymbol\vartheta_n=-\frac{1}{\sqrt k}\mH(\widetilde\vbeta-\widehat\vbeta)$}

The next result shows that these test statistics are distributed asymptotically as a weighted sum of chi squares.

\begin{theorem}
\label{basicresulta}
%\textcolor{red}{\bf a consistency argument for $\widetilde\vbeta$ is missing, alternatively we can assume $\widetilde\vbeta$  is some consistent restricted estimator}
Let $T\in\{D_{a},LM_{a},W_{1a},W_{2a}\}$. If Assumptions \ref{ass1} to \ref{ass6} are
satisfied and in addition $r_{\min}\mH^{-1}$  and $(\mA\mH^{-1}\mA^{\prime})^{-1}\mA\mH^{-1}$ are convergent as $n\rightarrow\infty$, 
%with $\mP_{\mH^{-1/2}\mA'}=\mH^{-1/2}\mA'(\mA\mH^{-1}\mA')^{-1}\mA\mH^{-1/2}$ and $\frac{1}{k^{2}}\sum_{i=1}^{n}\left\Vert \mPi^{\prime}\mZ^{\prime}\mC_{i}\right\Vert ^{4}\rightarrow0$ hold,where $\mC_{i}$ is column $i$ of matrix $\mC$,
then under the null hypothesis \eqref{lrestrictions} $T\rightarrow_{d}\vzeta^{\prime}\mXi_{a}\vzeta\sim
\bar{\chi}^{2}(\boldsymbol{\varphi})$ for $\vzeta\sim\mathcal{N}(\vzeros,\mPhi)$ and
$\boldsymbol{\varphi}=\lambda(\mXi_{a}\mPhi)$, where 
\begin{equation}
\mXi_{a}=\lim r_{\min}\mH^{-1}\mA^{\prime}(\mA\mH^{-1}\mA^{\prime})^{-1}\mA\mH^{-1}%
\;as\;n\rightarrow\infty\label{lim}%
\end{equation} 
and $\mPhi$ is given in Assumption \ref{ass6}.
%is defined as in Theorem \ref{basicresult}.
\end{theorem}

\section{Chi square approximations}\label{chi2section}
%\textcolor{red}{\bf fix this one} The test statistics based on based on quadratic or bilinear forms can be modified to accept a standard chi square limit. This can be easily see from the limit result $T\to_d\vzeta'\mW\vzeta=\vz\mPhi^{1/2}\mW\mPhi^{1/2}\vz$, $T\in\{D,LM,W,LM_a,W_a\}$ and $\vz\sim\mathcal N(\vzeros,\mI)$. 
Under conditions presented in this section, it is possible to show that the proposed test statistics follow a standard chi square limit instead of a weighted sum of chi squares. In general, it is unclear which approximation provides better finite sample behaviour, yet the former is typically more commonly used by practitioners than the latter.\footnote{Section \ref{numerical} explores the finite sample properties of test statistics under either asymptotic distribution.} 
By using appropriate modifications, we can derive a set of test statistics that have a standard chi square limit. In some cases such modifications are relatively straightforward, as in the $LM$ case of \cite{MO22} and, in general, for tests defined by quadratic forms. When dealing with distance statistics, the adjustments are not so obvious.

Let us suppose we want to define a test statistic for the simple null
hypothesis \eqref{simplenull}. In order to find a chi square limit for the
distance statistic we modify the objective function in the following way
\begin{equation}
Q_{n}^{\ast}(\vbeta)=\frac{(\vy-\mX\vbeta)^{\prime}\mJ(\vy-\mX\vbeta)}{\frac{1}%
{\tr(\mB)}(\vy-\mX\vbeta)^{\prime}\mB(\vy-\mX\vbeta)},\label{objfun*}%
\end{equation}
where $\mJ=\mC\mX\mPhi^{-1}\mX^{\prime}\mC$. Let $T^{\ast}$ denote a generic test
statistic, either distance, Lagrange multiplier or Wald, that is obtained from
equation \eqref{objfun*} and with $\widehat{\vbeta}$ the estimator obtained
from the optimisation of the objective function in equation \eqref{objfun}.
Specifically, $T^{\ast}$ is one of the test statistics%
\begin{align}
D^{\ast} &  =-\frac{\sigma^{2}}{k}\left(  Q_{n}^{\ast}(\widehat{\vbeta}%
)-Q_{n}^{\ast}(\vbeta_{0})+2\frac{\sqrt{k}}{\sigma^{2}}\vpsi_{n}^{\prime}%
\mPhi^{-1}\mX^{\prime}\mC(\vy-\mX\widehat{\vbeta})\right)  ,\label{D*}\\
LM^{\ast} &  =\vtau_{n}^{\prime}\mPhi^{-1}\vtau_{n},\label{LM*}\\
W^{\ast} &  =\vpsi_{n}^{\prime}\mPhi^{-1}\vpsi_{n}.\label{W*}%
\end{align}
The following theorem shows that under the null $T^{\ast}$ converges in
distribution to a chi square.

\begin{theorem}
\label{simplechi2} If Assumptions \ref{ass1} to
\ref{ass6} are satisfied, then under the null hypothesis \eqref{simplenull} the
statistic $T^{\ast}\in\{D^{\ast},LM^{\ast},W^{\ast}\}$ is asymptotically chi
square distributed with $g$ degrees of freedom.
\end{theorem}

If we are testing linear restrictions as defined in (\ref{lrestrictions}), we
modify the objective function to%
\begin{equation}
Q_{an}^{\ast}(\vbeta)=\frac{(\vy-\mX\vbeta)^{\prime}\mJ_{a}(\vy-\mX\vbeta)}{\frac{1}%
{\tr(\mB)}(\vy-\mX\vbeta)^{\prime}\mB(\vy-\mX\vbeta)},\label{objfun*a}%
\end{equation}
where $\mJ_{a}=\mC\mX\mGamma_{n}^{\prime}\left(  \mGamma_{n}\mPhi\mGamma_{n}^{\prime
}\right)  ^{+}\mGamma_{n}\mX^{\prime}\mC$ with $\mGamma_{n}=\mA^{\prime}%
(\mA\mH^{-1}\mA^{\prime})^{-1}\mA\mH^{-1}$
%\textcolor{red}{the subscript $n$ may be misleading because the other matrices are not indexed by $n$; maybe we could omit the subscript and denote the limit by $\overline{\mGamma}$ or $\mGamma_{0}$ } \textcolor{blue}{we could do that, I'm fine with either choice}
and $\left(  \mGamma_{n}\mPhi\mGamma_{n}^{\prime}\right)  ^{+}$ is a generalised
inverse of $\mGamma_{n}\mPhi\mGamma_{n}^{\prime}$ that satisfies $\mGamma_{n}%
\mPhi\mGamma_{n}^{\prime}\left(  \mGamma_{n}\mPhi\mGamma_{n}^{\prime}\right)
^{+}\mGamma_{n}\mPhi\mGamma_{n}^{\prime}=\mGamma_{n}\mPhi\mGamma_{n}^{\prime}$ and
is bounded as $n\rightarrow\infty$.\footnote{Note that since $\mGamma_{n}$ is
idempotent, it is not of full rank, and therefore, $\mGamma_{n}\mPhi\mGamma
_{n}^{\prime}$ is not of full rank either.} 
The test statistics for the null hypothesis in equation \eqref{lrestrictions} are 
%\textcolor{blue}{the use of $\vpsi_{n}   =\frac{1}{\sqrt{k}}\mH(\widehat{\vbeta}-\vbeta_{0})$ in equation 72 could be problematic because the the null is not expressed in terms of $\vbeta_0$, we can use $\bm\vartheta_{n}=\mGamma_{n}\vtau_{n}+o_{p}\left(  1\right)$ (later on page 28) and the relationship between $\vtau_n$ and $\vpsi_n$}
\begin{align}
%D_{a}^{\ast} &  =\dfrac{\sigma^{2}}{k}\left(  Q_{an}^{\ast}(\widetilde{\vbeta})-Q_{an}^{\ast}(\widehat{\vbeta})-\frac{2\sqrt{k}}{\sigma^{2}}\vpsi_{n}^{\prime}\mGamma_{n}^{\prime}\left(  \mGamma_{n}\mPhi\mGamma_{n}^{\prime}\right)^{+}\mGamma_{n}\mX^{\prime}\mC\left(  \vy-\mX\widehat{\vbeta}\right)  \right)\label{D*a}\\
D_{1a}^{\ast} &  =\dfrac{\sigma^{2}}{k}\left(  Q_{an}^{\ast}(\widetilde{\vbeta})-Q_{an}^{\ast}(\widehat{\vbeta})-\frac{2\sqrt{k}}{\sigma^{2}}\bm{\vartheta}_{n}^{\prime}\left(  \mGamma_{n}\mPhi\mGamma_{n}^{\prime}\right)^{+}\mGamma_{n}\mX^{\prime}\mC\left(  \vy-\mX\widehat{\vbeta}\right)  \right)\label{D1*a}\\
D_{2a}^{\ast} &  =\dfrac{\sigma^{2}}{k}\left(  Q_{an}^{\ast}(\widetilde{\vbeta})-Q_{an}^{\ast}(\widehat{\vbeta})-\frac{2\sqrt{k}}{\sigma^{2}}\vxi_{n}^{\prime}\left(  \mGamma_{n}\mPhi\mGamma_{n}^{\prime}\right)^{+}\mGamma_{n}\mX^{\prime}\mC\left(  \vy-\mX\widehat{\vbeta}\right)  \right)\label{D2*a}\\
LM_{a}^{\ast} &  =\vxi_{n}^{\prime}\left(  \mGamma_{n}\mPhi\mGamma_{n}^{\prime
}\right)  ^{+}\vxi_{n}\label{LM*a}\\
W_{1a}^{\ast} &  =\frac{1}{k}(\mA\widehat{\vbeta}-\va)^{\prime}(\mA\mH^{-1}\mPhi
\mH^{-1}\mA^{\prime})^{-1}(\mA\widehat{\vbeta}-\va),\label{W1*a}\\
W_{2a}^{\ast} &  =  \bm{\vartheta}_{n}^{\prime}\left(  \mGamma_{n}\mPhi\mGamma_{n}^{\prime}\right)  ^{+}\bm{\vartheta}_{n}\label{W2*a}.
\end{align}
%\textcolor{red}{Federico, I suggest that we choose (\ref{W*a}) and omit (\ref{W*a1}); see the proof of the theorem} 
The following result generalises Theorem \ref{simplechi2}.

\begin{theorem}
\label{lrestrictionchi2} Let $T$ be any test statistic defined in equations
\eqref{D1*a} to \eqref{LM*a} and \eqref{W2*a}. If Assumptions \ref{ass1} to \ref{ass6} are
satisfied and in addition $(\mA\mH^{-1}\mA^{\prime})^{-1}\mA\mH^{-1}$ is convergent as $n\rightarrow\infty$, then under the null hypothesis (\ref{lrestrictions}) $T$
is asymptotically chi square distributed with $p$ degrees of freedom. If in addition $\lim r_{\min}\mH^{-1}=\mXi$ is nonsingular then $W_{1a}^{\ast}$ from (\ref{W1*a}) is also asymptotically chi square distributed with $p$ degrees of freedom. 
\end{theorem}

\section{Numerical results}\label{numerical}
%Throughout the paper, we have described the unfeasible version of the proposed statistics. 
In this section, 
%we will show how to practically compute the feasible version of the proposed statistics by replacing certain unknown quantities with suitable plug-in estimators. 
we first describe how to compute feasible versions of the proposed statistics by plugging in estimators for the unknown quantities.
Moreover, we investigate the statistics' finite sample properties in terms of size and power in the context of two data generating processes (DGPs). DGP1 is similar to that in \cite{MO22}. The model in DGP2 can be seen as a stylised version of a Cobb-Douglas production function with two endogenous variables corresponding to production factors, say labour and capital, as in \cite{NR08} \citep[see also][for a DGP with two endogenous variables]{CMS21}. Finally, we discuss the application of our tests to a real data example that employs genetic variants as instruments to study the effect of alcohol consumption on body mass index. The results of our tests are compared against the AR tests of \cite{CMS21} and \cite{MS22} with both naive and cross-fitted variance estimators. % for DGP1. In the case of DGP2, we compare the trinity only against the AR statistic with the naive variance estimator.\footnote{DGP2 includes two endogenous variables and, to the best of our knowledge, there is no cross-variance estimator that contemplates this scenario.} 
The general expression of the AR test is 
\begin{align*}
    AR(\vbeta)=\frac{1}{\sqrt k}\frac{\vvarepsilon'\mC \vvarepsilon}{\sqrt{\omega(\vbeta)}}
\end{align*}
where $\vvarepsilon=\vy-\mX\vbeta$, while $\mC$ in equation \eqref{jackmat} produces the AR test of \cite{CMS21} and $\mC$ in equation \eqref{jackmat2} produces the AR test of \cite{MS22}. 
The naive variance estimator for $\omega(\vbeta)$ is 
\begin{align*}
    \widehat\omega_{naive}(\vbeta)=\frac{2}{k}\vvarepsilon^{(2)'}\mC^{(2)} \vvarepsilon^{(2)}
\end{align*}
while the cross-fitted alternative is 
\begin{align*}
    \widehat\omega_{cf}(\vbeta)=\frac{2}{k}\vvarepsilon'\mB\mD_{\vvarepsilon}\mM\mD_{\vvarepsilon}\mB\vvarepsilon
\end{align*}
with 
\begin{align}\label{cfweights}
\mM=\left(\mbox{diag}(\mB)\mbox{diag}(\mB)'+\mB^{(2)}\right)^{(-1)}*\mC^{(2)},
\end{align}
where $\mbox{diag}(\mB)$ is a diagonal matrix containing the diagonal entries of $\mB$ in the main diagonal. For the version of \cite{CMS21} $\mB$ is as in equation \eqref{jackmat}, while for the version of \cite{MS22}, $\mB=\mI-\mP$. 

\subsection{Implementation}\label{implem}
To make the tests operational, we first need to estimate the relevant quantities $\mH$,  $\sigma^2$, $\mPhi$. Estimators for $\mH$,  $\sigma^2$, $\mPhi$ are already defined in \cite{BC15} \citep[see also][]{HNWCS12}. Specifically, we have $\widehat\mH(\vbeta)=\mX'\widehat\mC(\vbeta)\mX$, $\widehat\sigma^2(\vbeta)=\frac{1}{\tr(\mB)}(\vy-\mX\vbeta)'\mB(\vy-\mX\vbeta)$ and
{\small
\begin{align}\label{Phihat}
\widehat\mPhi(\vbeta)=\frac{1}{k}\left(\widehat\mF(\vbeta)+\widehat\mG(\vbeta)\right)=
\frac{1}{k}\left(\widetilde\mX(\vbeta)'\mC\mD^2_{\vvarepsilon}\mC\widetilde\mX(\vbeta)+ \widetilde\mX(\vbeta)'\mD_{\vvarepsilon}\mC^{(2)}\mD_{\vvarepsilon}\widetilde\mX(\vbeta)\right)
\end{align}
}%
with $\widetilde\mX(\vbeta)=\mX-\frac{\vvarepsilon\widehat\vsigma_{12}(\vbeta)}{\widehat\sigma^2(\vbeta)}$ and $\mD_{\vvarepsilon}$ a diagonal matrix with $\vvarepsilon=\vy-\mX\vbeta$ along the main diagonal.  We further define the estimator for the variance covariance matrix $\mSigma$ as 
 \begin{align*}
\widehat\mSigma(\vbeta)=
\begin{pmatrix}
1 & -\vbeta'\\
 \vzeros& \mI_g
 \end{pmatrix}
\widehat\mOmega
\begin{pmatrix}
1 & \vzeros\\
-\vbeta & \mI_g
\end{pmatrix}=
\frac{1}{k}
\begin{pmatrix}
1 & -\vbeta'\\
 \vzeros& \mI_g
 \end{pmatrix}
(\vy,\mX)'\mB(\vy,\mX)
\begin{pmatrix}
1 & \vzeros\\
-\vbeta & \mI_g
\end{pmatrix}
\end{align*}
where $\widehat\mOmega$ is an estimator of $\mOmega=\E[(\vy,\mX)'\mB(\vy,\mX)]$ \citep[see][for additional details and a proof of consistency for the estimator in equation \eqref{Phihat}]{BC15}. Thus, $\widehat\vsigma_{12}(\vbeta)$ and $\widehat\sigma^2(\vbeta)$ correspond to the off-diagonal entry and the top left entry of $\widehat\mSigma(\vbeta)$, respectively.
These quantities depend on the parameter vector $\vbeta$, which needs to be replaced by a suitable plug-in estimator. Wald statistics use the unrestricted estimator $\widehat\vbeta$, Lagrange multiplier statistics use either $\vbeta_0$ or the restricted estimator $\widetilde\vbeta$, depending on the type of null hypothesis being tested. Finally, distance statistics use the objective function evaluated at both the restricted and unrestricted estimators \citep[see, for example,][]{Engle84}; however, the multiplicative coefficient $r_{\min}\sigma^2$ is evaluated at the unrestricted estimator, say, $\widehat r_{\min}\widehat \sigma^2(\widehat \vbeta)$ with $\widehat r_{\min}=\lambda_{\min}\left(\widehat \mH(\widehat\vbeta)\right)$. The correction term in the distance statistic, when the test is chi-square distributed, is evaluated at $\widehat\vbeta$. 

When the asymptotic distribution is a weighted average of chi-squares, the calculation of the weights $\vvarphi$ is based on $\widehat\mH(\widehat\vbeta)$ and $\widehat\mPhi(\widehat\vbeta)$ for distance and Wald statistics and on $\widehat\mH(\vbeta_0)$ and $\widehat\mPhi(\vbeta_0)$ or $\widehat\mH(\widetilde\vbeta)$ and $\widehat\mPhi(\widetilde\vbeta)$ for the Lagrange multiplier statistic. Similarly, for distance and Wald statistics, the minimum eigenvalue $r_{\min}$ is estimated, as seen above, by $\widehat r_{\min}=\lambda_{\min}\left(\widehat \mH(\widehat\vbeta)\right)$, while for the $LM$ statistic we find $\widetilde r_{\min}=\lambda_{\min}\left(\widehat \mH(\vbeta_0)\right)$ or $\widetilde r_{\min}=\lambda_{\min}\left(\widehat \mH(\widetilde\vbeta)\right)$, depending on the type of null hypothesis under study.

The plug-in estimators presented in this section are general in the sense that apply both to tests based on quadratic forms (JIVE1, JIVE2) and tests based on ratios of quadratic forms (SJIVE, HLIM). However, for the former case and in light of the discussion in Section \ref{JIVE}, some simplifications apply. First of all, $\widehat\mH(\vbeta)=\widehat\mH=\mX'\mC\mX$, $\widehat r_{\min}=\widetilde r_{\min}=\lambda_{\min}\left(\widehat\mH\right)$ and $\widetilde\mX(\vbeta)=\mX$. This implies that the variance covariance matrix $\widehat\mPhi(\vbeta)$ depends on $\vbeta$ only through the matrix $\mD_{\vvarepsilon}$. Hence,
\begin{align*}
\widehat\mPhi(\vbeta)=\frac{1}{k}\left(\widehat\mF(\vbeta)+\widehat\mG(\vbeta)\right)=
\frac{1}{k}\left(
\mX'\mC\mD^2_{\vvarepsilon}\mC \mX+ \mX'\mD_{\vvarepsilon}\mC^{(2)}\mD_{\vvarepsilon}\mX\right).
\end{align*}
Tests that require a plug-in for $r_{\min}$ use $\widehat r_{\min}(=\widetilde r_{\min})$.
Moreover, $\sigma^2=1$ and no other quantity related to $\mSigma$ is required. It is also easy to see that the correction factors in the chi-square distributed distance statistics are zero. As explained in Section \ref{JIVE}, certain tests produce the same value. However, their corresponding p values are computed using different plug-ins for $\vvarphi$, thus generating different rejection rates in the numerical exercises. Specifically, recall from Theorem \ref{basicresult} that $\vvarphi$ is the vector of eigenvalues of $\mXi\mPhi$. Thus, for the $LM$ statistic $\vvarphi$ is computed by evaluating the estimators of $\mXi$ and $\mPhi$ at $\vbeta_0$ or $\widetilde\vbeta$. For distance and Wald statistics the estimators of $\mXi$ and $\mPhi$ are evaluated at $\widehat\vbeta$.   
Furthermore, the chi-square tests for the general linear restrictions require the generalised inverse of $\mGamma_n\mPhi\mGamma_n'$, for example,
{\small
\begin{align}\label{ginv}
\left(\mGamma_n \mPhi \mGamma_n'\right)^+&=\mA'(\mA\mA')^{-1}\left(\mA\mH^{-1}\mA'\right)\left(\mA\mH^{-1}\mPhi \mH^{-1}\mA'\right)^{-1}\left(\mA\mH^{-1}\mA'\right)(\mA\mA')^{-1}\mA\\
&=\mA'(\mA\mA')^{-1}\left(\left(\mA\mH^{-1}\mA'\right)^{-1}\mA\mH^{-1}\mPhi \mH^{-1}\mA'\left(\mA\mH^{-1}\mA'\right)^{-1}\right)^{-1}(\mA\mA')^{-1}\mA.
\end{align}
}%
Estimation of $\mH$ and $\mPhi$ follows the same guidelines as described above in this section. In Appendix \ref{completesim}, we also consider the case in which the test statistics use a cross-fit estimator of the variance instead of the expression in equation \eqref{Phihat}. The cross-fit estimator of $\mPhi$ is adapted from \cite{MS24}:
\begin{align}\label{Phicf}
\widehat\mPhi_{cf}(\vbeta)=\frac{1}{k}\left(
\widetilde{\mX}(\vbeta)'\mC\mD_{\vvarepsilon}\diag(\mB)^{-1}\mD_{\vv}\mC\widetilde{\mX}(\vbeta)
+
\widetilde{\mX}(\vbeta)'\mD_{\vv}\mM\mD_{\vv}\widetilde{\mX}(\vbeta)
\right)
\end{align}
where $\vv=\mB\vvarepsilon$. The definition of $\widetilde{\mX}(\vbeta)$ follows that of the standard case.\footnote{A proof of consistency for the estimator in equation \eqref{Phicf} is available upon request.} 

Finally, to make our statistics operational, we need  feasible counterparts of $\bm\vartheta_n$ and $\vxi_n$ (or $\vpsi_n$ and $\vtau_n$). It is easy to see that also in this case we can use suitable plug-in estimators for $\mH$ and $\vbeta$.

\subsection{Monte Carlo simulations}\label{MonteCarlo}
This section investigates the finite sample properties of the proposed test statistics for four specifications of the objective function and in the context of two DGPs. DGP1 considers a model with one endogenous variable, while DGP2 focusses on a model with two endogenous variables. Since under the alternative we allow for the presence of estimable nuisance parameters, the test statistics used in the simulations are those described in Theorem \ref{basicresulta} and Theorem \ref{lrestrictionchi2}.\footnote{To avoid notation clutter, in tables, figures and the discussion of the results we drop the subscript $a$.}
We study both size and power properties. For size, each experiment uses 5000 replications, whereas power experiments use 1000 replications. For the sake of conciseness, the power comparisons in Figure \ref{fig:mainDGP1_200_0.05_32} to Figure \ref{fig:mainDGP1_200_0.1_0.1} involve tests based on SJIVE and JIVE1. As mentioned later in this section, the competing alternatives (based on HLIM and JIVE2) tend to be dominated by SJIVE and JIVE1 statistics, particularly in the case of Lagrange multiplier statistics. 
Since there are no noticeable differences between $W_1/W_1^*$ and $W_2/W_2^*$, Table \ref{tab:DGP1_200} and Table \ref{tab:DGP2_200} feature only $W_1/W_1^*$. There are some very small differences between $D_1^*$ and $D_2^*$ that do not change the interpretation of the results, so in the tables and figures of this section we drop $D_2^*$. It is worth noting that the results in Section \ref{JIVE} on the behaviour of JIVE1/JIVE2-based tests find confirmation in our simulations. Furthermore, by direct inspection of the respective formulae, it is easy to see that, in general, the proposed statistics are computationally simpler to evaluate than AR statistics mainly due to the Hadamard square in the expression of the variance.  
The complete set of results can be found in Appendix \ref{completesim}. 
In what follows, we provide a description of the DGPs and some comments on the simulation results.

\subsubsection{DGP1}
This simulation setup considers the model 
\begin{align}\label{mcstruct}
\vy&=\vx_1\beta+\mX_2\vbeta_2+\vvarepsilon\\\label{mcfirst}
\vx_1&=\mZ\vpi+\vv=\mZ_1\vpi_1+\mZ_2\vpi_2+\vv.
\end{align} 
Here, $\vx_1$ and $\vy$ are $n$-dimensional vectors with $n=200$. Moreover, $\mX_2=\mZ_2$ is a
$(n\times g_2)$-dimensional matrix with $g=g_1+g_2=1+g_2$ being the number of regressors in equation \eqref{mcstruct}. We choose $g_2=5$; the first column of $\mX_2$ contains a vector of ones and the remaining columns are sampled independently from a standard normal. The $n\times k_1$ instrument matrix $\mZ_1$ contains in the first to third column the variables $\vz_1$, $\vz_1^{(2)}$ and $\vz_1^{(3)}$ where $\vz_1\sim\mathcal N(\vzeros,\mI_n)$. The
remaining columns are sampled from a standard normal. The number of
instruments $k=k_1+g_2$ with $k_1=\alpha n$. Clearly, the parameter $\alpha$ is the share of instruments with respect to the sample size and is chosen as
$\alpha\in\{0.05, 0.1\}$. The disturbance couple $(\vvarepsilon, \vv)$
is defined by the relationship \begin{align*}
\vv&=\rho_1\vvarepsilon+\sqrt{1-\rho_1^2}\vu_1\\
\vvarepsilon&=\left(\viota_n+\rho_2 \vz_1^{(2)}\right)*\vu_2%\\ % \rho_2 used to be \psi
\end{align*} 
with $\vu_1$ and $\vu_2$
being mutually independent and distributed as a standard normal. We set
$\beta=1$, $\vbeta_2=\viota_5$, $\rho_1=0.3$ and $\rho_2=0.2$. We also define
$\vpi=\pi\viota_k$ for $\viota_k$ a scalar $\pi$ chosen as \begin{align*}
\pi=\sqrt{\frac{(1+3\rho_1^2\rho_2^2)r}{k}}
\end{align*} with $r\in\{32,64\}$ being a measure of instrument strength. The null hypothesis corresponds to the linear restrictions 
\begin{align*}
H_0:\mA\vbeta=\va %\beta_1=\beta_{10}
\end{align*} 
with $\vbeta=(\beta,\vbeta_2')'$, $\mA=(1,\vzeros')$, a row vector, and $\va=\beta_{0}=1$.

\begin{table}[ht]
\centering
\caption{Size results for DGP1, $5\%$ nominal level and $n=200$. Results based on 5000 repetitions.}
\label{tab:DGP1_200}
\begin{footnotesize}
\begin{tabular}{lrrrrrrrrrrr}
\toprule
Method & $\alpha$ & $r$ & $D$ & $W_1$ & $LM$ & $D_1^{*}$ & $W_1^{*}$ & $LM^{*}$ & $AR_{naive}$ & $AR_{cf}$ \\
\midrule
\rowcolor{gray!6} SJIVE & 0.05 & 32 & 0.082 & 0.055 & 0.054 & 0.028 & 0.055 & 0.054 &       &      \\
SJIVE & 0.05 & 64 & 0.080 & 0.066 & 0.051 & 0.035 & 0.066 & 0.051 &       &      \\
\rowcolor{gray!6} SJIVE & 0.10 & 32 & 0.073 & 0.051 & 0.050 & 0.030 & 0.051 & 0.050 &       &      \\
SJIVE & 0.10 & 64 & 0.065 & 0.059 & 0.048 & 0.036 & 0.059 & 0.048 &       &      \\
\addlinespace
\rowcolor{gray!6} HLIM & 0.05 & 32 & 0.076 & 0.054 & 0.050 & 0.026 & 0.054 & 0.050 &       &      \\
HLIM & 0.05 & 64 & 0.073 & 0.062 & 0.051 & 0.033 & 0.062 & 0.051 &       &      \\
\rowcolor{gray!6} HLIM & 0.10 & 32 & 0.069 & 0.049 & 0.046 & 0.026 & 0.049 & 0.046 &       &      \\
HLIM & 0.10 & 64 & 0.060 & 0.052 & 0.046 & 0.031 & 0.052 & 0.046 &       &      \\
\addlinespace
\rowcolor{gray!6} JIVE1 & 0.05 & 32 & 0.028 & 0.028 & 0.054 & 0.054 & 0.028 & 0.054 & 0.008 & 0.019 \\
JIVE1 & 0.05 & 64 & 0.049 & 0.049 & 0.052 & 0.052 & 0.049 & 0.052 & 0.008 & 0.019 \\
\rowcolor{gray!6} JIVE1 & 0.10 & 32 & 0.019 & 0.019 & 0.051 & 0.051 & 0.019 & 0.051 & 0.015 & 0.038 \\
JIVE1 & 0.10 & 64 & 0.036 & 0.036 & 0.048 & 0.048 & 0.036 & 0.048 & 0.015 & 0.038 \\
\addlinespace
\rowcolor{gray!6} JIVE2 & 0.05 & 32 & 0.023 & 0.023 & 0.057 & 0.057 & 0.023 & 0.057 & 0.007 & 0.011 \\
JIVE2 & 0.05 & 64 & 0.043 & 0.043 & 0.050 & 0.050 & 0.043 & 0.050 & 0.007 & 0.011 \\
\rowcolor{gray!6} JIVE2 & 0.10 & 32 & 0.016 & 0.016 & 0.054 & 0.054 & 0.016 & 0.054 & 0.014 & 0.018 \\
JIVE2 & 0.10 & 64 & 0.031 & 0.031 & 0.047 & 0.047 & 0.031 & 0.047 & 0.014 & 0.018 \\
\bottomrule
\end{tabular}
\end{footnotesize}
\end{table}

\begin{figure}[ht]
    \centering
    % Replace 'chibar2' with 'chi2' for the second set of figures as needed
%    \begin{subfigure}[b]{0.47\textwidth}
%        \includegraphics[width=\textwidth]{mainTrinity_chibar2_DGP1_200_0.05_32_SJIVE.pdf}
%        \caption{SJIVE, $\bar\chi^2$}
%    \end{subfigure}
%    \hfill
%    \begin{subfigure}[b]{0.47\textwidth}
%        \includegraphics[width=\textwidth]{mainTrinity_chibar2_DGP1_200_0.05_32_JIVE1.pdf}
%        \caption{JIVE1, $\bar\chi^2$}
%    \end{subfigure}
    
    \begin{subfigure}[b]{0.47\textwidth}
        \includegraphics[width=\textwidth]{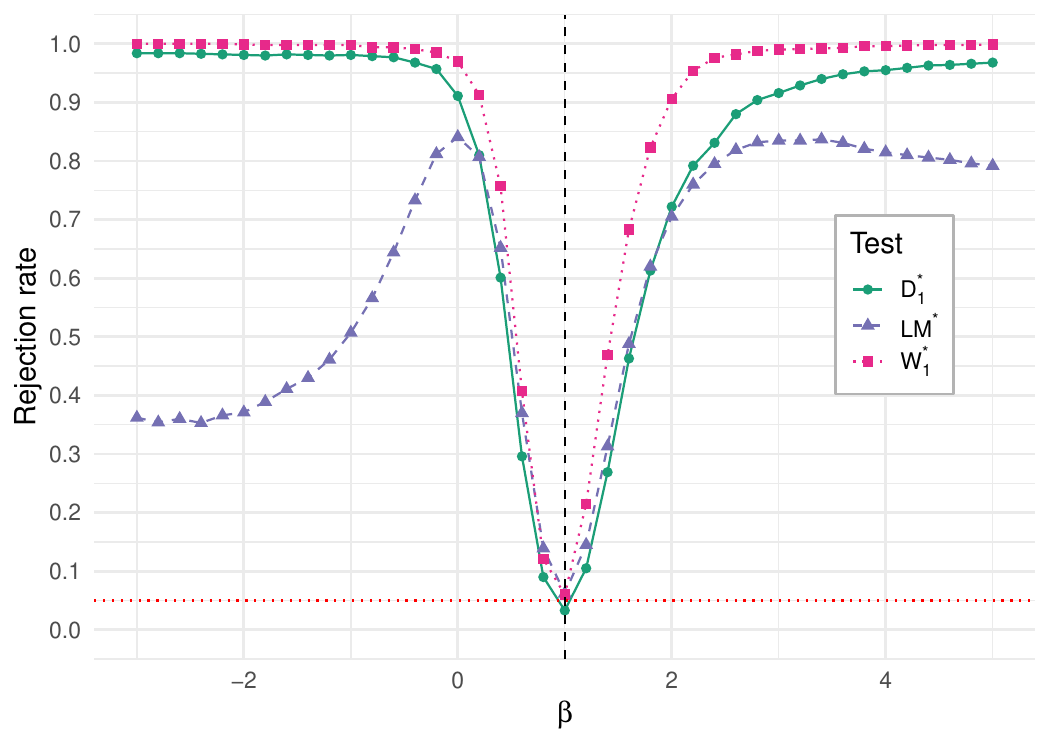}
        \caption{SJIVE}
    \end{subfigure}
    \hfill
    \begin{subfigure}[b]{0.47\textwidth}
       \includegraphics[width=\textwidth]{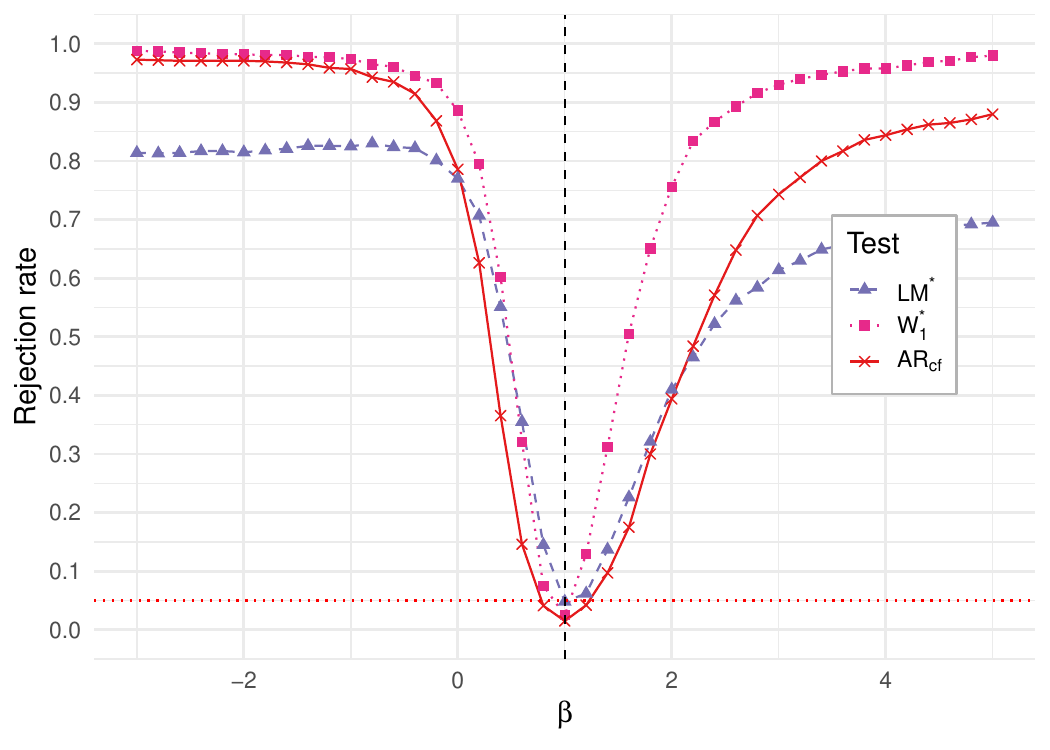}
       \caption{JIVE1}
    \end{subfigure}
    %\begin{subfigure}[b]{0.47\textwidth}
    %    \includegraphics[width=\textwidth]{mainTrinity_chi2_DGP1_200_0.05_32_JIVE1.pdf}
    %    \caption{JIVE1, $\chi^2$}
    %\end{subfigure}

    %\begin{subfigure}[b]{0.47\textwidth}
    %   \includegraphics[width=\textwidth]{mainAR_N_DGP1_200_0.05_32_JIVE1.pdf}
    %    \caption{JIVE1, $\mathcal{N}(0,1)$}
    %\end{subfigure}
    
    \caption{Power curves for DGP1 ($n=200$, $\alpha = 0.05$, $r = 32$). Results based on 1000 repetitions. The horizontal dotted red line denotes the $5\%$ nominal rejection level, while the vertical dotted black line corresponds to $\beta=1$. Panel (a) plots $D_1^*$, $LM^*$, $W_1^*$ based on the SJIVE objective function; panel (b) plots $W_1^*$, $LM^*$, $AR_{cf}$ based on the JIVE1 objective function. The statistics are all asymptotically $\chi^2_1$ distributed, besides $AR_{cf}$ which is $\mathcal N(0,1)$ in the limit.
    }
    \label{fig:mainDGP1_200_0.05_32}
\end{figure}

\begin{figure}[ht]
    \centering
    % Replace 'chibar2' with 'chi2' for the second set of figures as needed
    %\begin{subfigure}[b]{0.47\textwidth}
    %    \includegraphics[width=\textwidth]{mainTrinity_chibar2_DGP1_200_0.1_32_SJIVE.pdf}
    %    \caption{SJIVE, $\bar\chi^2$}
    %\end{subfigure}
    %\hfill
    %\begin{subfigure}[b]{0.47\textwidth}
    %    \includegraphics[width=\textwidth]{mainTrinity_chibar2_DGP1_200_0.1_32_JIVE1.pdf}
    %    \caption{JIVE1, $\bar\chi^2$}
    %\end{subfigure}
    
    \begin{subfigure}[b]{0.47\textwidth}
        \includegraphics[width=\textwidth]{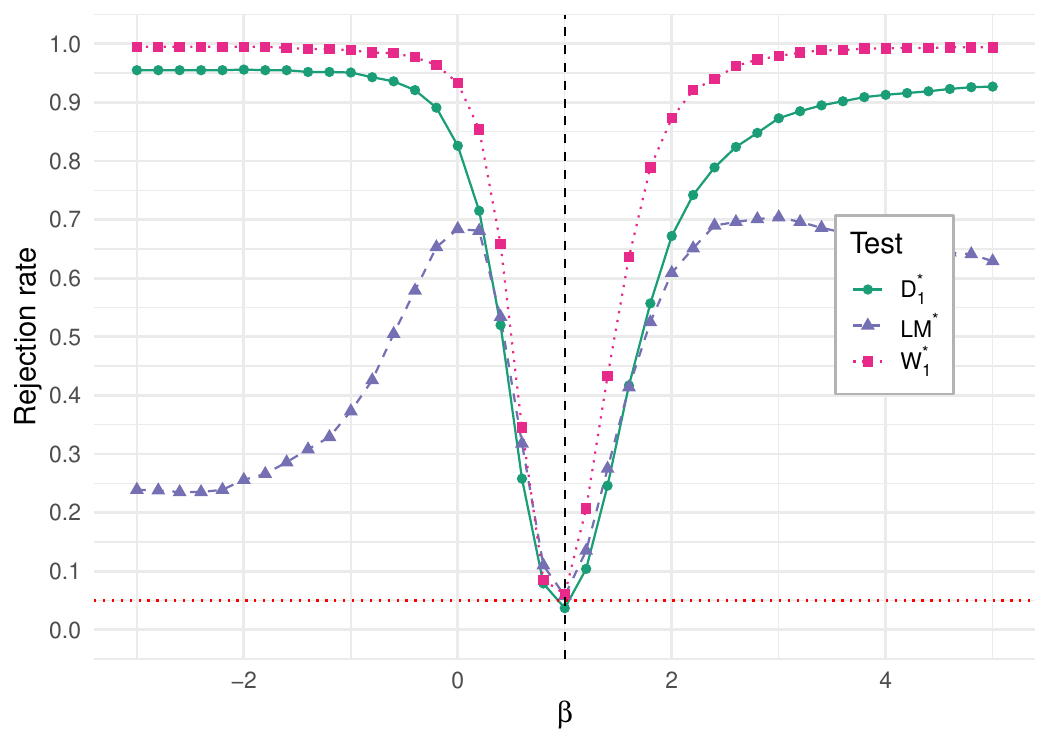}
        \caption{SJIVE}
    \end{subfigure}
    \hfill
    \begin{subfigure}[b]{0.47\textwidth}
       \includegraphics[width=\textwidth]{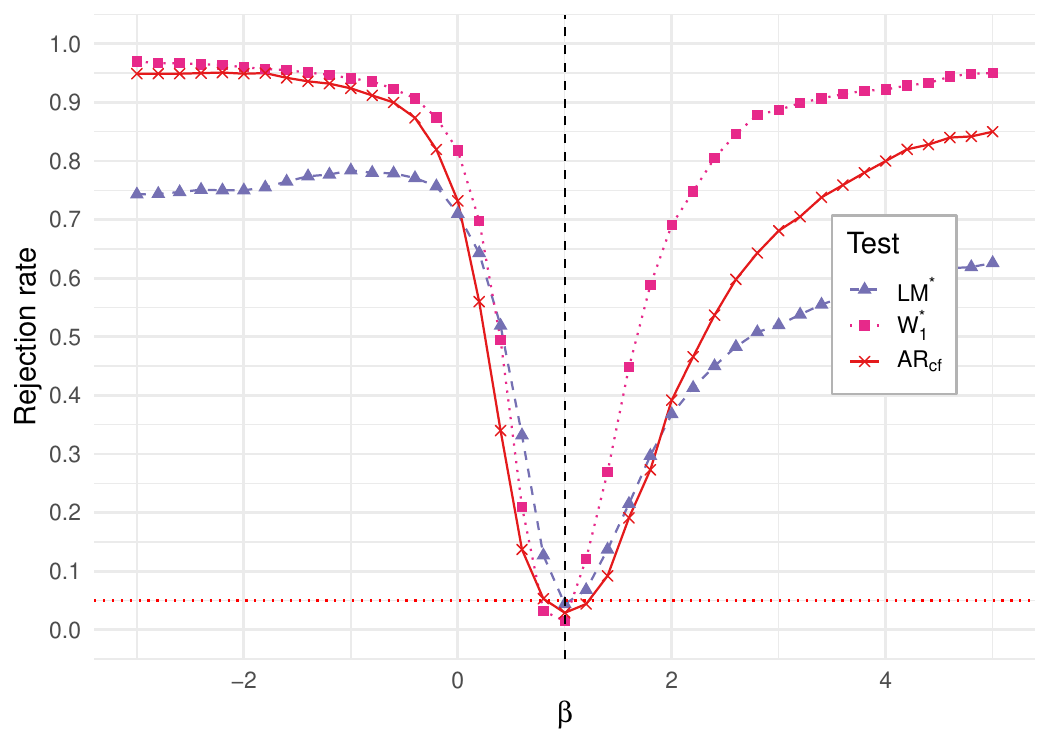}
       \caption{JIVE1}
    \end{subfigure}
    %\begin{subfigure}[b]{0.47\textwidth}
    %    \includegraphics[width=\textwidth]{mainTrinity_chi2_DGP1_200_0.1_32_JIVE1.pdf}
    %    \caption{JIVE1, $\chi^2$}
    %\end{subfigure}

    %\begin{subfigure}[b]{0.47\textwidth}
    %   \includegraphics[width=\textwidth]{mainAR_N_DGP1_200_0.1_32_JIVE1.pdf}
    %    \caption{JIVE1, $\mathcal{N}(0,1)$}
    %\end{subfigure}
    
    \caption{Power curves for DGP1 ($n=200$, $\alpha = 0.1$, $r = 32$). Results based on 1000 repetitions. The horizontal dotted red line denotes the $5\%$ nominal rejection level, while the vertical dotted black line corresponds to $\beta=1$. Panel (a) plots $D_1^*$, $LM^*$, $W_1^*$ based on the SJIVE objective function; panel (b) plots $W_1^*$, $LM^*$, $AR_{cf}$ based on the JIVE1 objective function. The statistics are all asymptotically $\chi^2_1$ distributed, besides $AR_{cf}$ which is $\mathcal N(0,1)$ in the limit.
    }
    \label{fig:mainDGP1_200_0.1_32}
\end{figure}

\clearpage

\subsubsection{DGP2}
The model in the second DGP consists of two endogenous variables and reads
\begin{align}
\vy = \vx_1\beta_1 + \vx_2\beta_2 + \mX_3\vbeta_3 + \vvarepsilon
\end{align}
where  $\mX_3$ is a $n\times 2$ matrix with the first column being a vector of ones and the second being sampled from a standard normal, $\vvarepsilon=0.2\vu+\ve$ with $\ve\sim\mathcal N(\vzeros,0.3^2\mI_n)$ and $\vu\sim\mathcal N(\vzeros,\mI_n)$. Furthermore, 
\begin{align}
\vx_j = \mZ_j\vpi_j + \delta_j\vu+\vv_j, \;\;\; j=1,2
\end{align}
where $\delta_1=0.5$, $\delta_2=0.4$, $\vv_j\sim\mathcal N(\vzeros, \sigma_{v_j}^2\mI_n)$ with $\sigma_{v_j}=0.3$. Furthermore, $\mZ_j\sim\mathcal N(\vzeros,\mI_n)$. Finally, we have $\vpi_j=\sqrt{\frac{r}{k_j(1-r)}}\viota_{k_j}$ and $k_j=\alpha n$, $n=200$ and $\alpha\in\{0.05,0.1\}$. In this case, the null hypothesis 
\begin{align*}
H_0:\mA\vbeta=\va 
\end{align*}
has $\mA=(1,1,\vzeros')$, $\vbeta=(\beta_1,\beta_2,\vbeta_3')'$ and $\va=1$. The true values for the parameters of the endogenous variables are $\beta_1=0.3$ and $\beta_2=0.7$, while $\vbeta_3=\viota_2$.

\begin{table}[ht]
\centering
\caption{Size results for DGP2, $5\%$ nominal level and $n=200$. Results based on 5000 repetitions.}
\label{tab:DGP2_200}
\begin{footnotesize}
\begin{tabular}{lrrrrrrrrrrr}
\toprule
Method & $\alpha$ & $r$ & $D$ & $W_1$ & $LM$ & $D^*_1$ & $W_1^*$ & $LM^*$ & $AR_{naive}$ & $AR_{cf}$ \\
\midrule
\cellcolor{gray!6}SJIVE & \cellcolor{gray!6}0.05 & \cellcolor{gray!6}0.1 & \cellcolor{gray!6}0.066 & \cellcolor{gray!6}0.055 & \cellcolor{gray!6}0.054 & \cellcolor{gray!6}0.047 & \cellcolor{gray!6}0.055 & \cellcolor{gray!6}0.054 & \\
SJIVE & 0.05 & 0.2 & 0.062 & 0.056 & 0.056 & 0.051 & 0.056 & 0.056 & \\
\cellcolor{gray!6}SJIVE & \cellcolor{gray!6}0.10 & \cellcolor{gray!6}0.1 & \cellcolor{gray!6}0.067 & \cellcolor{gray!6}0.052 & \cellcolor{gray!6}0.051 & \cellcolor{gray!6}0.045 & \cellcolor{gray!6}0.052 & \cellcolor{gray!6}0.051 & \\
SJIVE & 0.10 & 0.2 & 0.061 & 0.052 & 0.051 & 0.048 & 0.052 & 0.051 & \\
\addlinespace
\cellcolor{gray!6}HLIM & \cellcolor{gray!6}0.05 & \cellcolor{gray!6}0.1 & \cellcolor{gray!6}0.063 & \cellcolor{gray!6}0.052 & \cellcolor{gray!6}0.052 & \cellcolor{gray!6}0.044 & \cellcolor{gray!6}0.052 & \cellcolor{gray!6}0.052 & \\
HLIM & 0.05 & 0.2 & 0.060 & 0.056 & 0.051 & 0.046 & 0.056 & 0.051 & \\
\cellcolor{gray!6}HLIM & \cellcolor{gray!6}0.10 & \cellcolor{gray!6}0.1 & \cellcolor{gray!6}0.072 & \cellcolor{gray!6}0.062 & \cellcolor{gray!6}0.057 & \cellcolor{gray!6}0.044 & \cellcolor{gray!6}0.062 & \cellcolor{gray!6}0.057 & \\
HLIM & 0.10 & 0.2 & 0.063 & 0.061 & 0.055 & 0.050 & 0.061 & 0.055 & \\
\addlinespace
\cellcolor{gray!6}JIVE1 & \cellcolor{gray!6}0.05 & \cellcolor{gray!6}0.1 & \cellcolor{gray!6}0.030 & \cellcolor{gray!6}0.030 & \cellcolor{gray!6}0.056 & \cellcolor{gray!6}0.056 & \cellcolor{gray!6}0.030 & \cellcolor{gray!6}0.056 & \cellcolor{gray!6}0.028& \cellcolor{gray!6}0.045 \\
JIVE1 & 0.05 & 0.2 & 0.044 & 0.044 & 0.053 & 0.053 & 0.044 & 0.053 & 0.028 & 0.046 \\
\cellcolor{gray!6}JIVE1 & \cellcolor{gray!6}0.10 & \cellcolor{gray!6}0.1 & \cellcolor{gray!6}0.035 & \cellcolor{gray!6}0.035 & \cellcolor{gray!6}0.060 & \cellcolor{gray!6}0.060 & \cellcolor{gray!6}0.035 & \cellcolor{gray!6}0.060 & \cellcolor{gray!6}0.036 & \cellcolor{gray!6}0.085 \\
JIVE1 & 0.10 & 0.2 & 0.040 & 0.040 & 0.056 & 0.056 & 0.040 & 0.056 & 0.030 & 0.085 \\
\addlinespace
\cellcolor{gray!6}JIVE2 & \cellcolor{gray!6}0.05 & \cellcolor{gray!6}0.1 & \cellcolor{gray!6}0.030 & \cellcolor{gray!6}0.030 & \cellcolor{gray!6}0.056 & \cellcolor{gray!6}0.056 & \cellcolor{gray!6}0.030 & \cellcolor{gray!6}0.056 & \cellcolor{gray!6}0.028 & \cellcolor{gray!6}0.034 \\
JIVE2 & 0.05 & 0.2 & 0.043 & 0.043 & 0.053 & 0.053 & 0.043 & 0.053 & 0.028 & 0.034 \\
\cellcolor{gray!6}JIVE2 & \cellcolor{gray!6}0.10 & \cellcolor{gray!6}0.1 & \cellcolor{gray!6}0.035 & \cellcolor{gray!6}0.035 & \cellcolor{gray!6}0.059 & \cellcolor{gray!6}0.059 & \cellcolor{gray!6}0.035 & \cellcolor{gray!6}0.059 & \cellcolor{gray!6}0.037  & \cellcolor{gray!6}0.046\\
JIVE2 & 0.10 & 0.2 & 0.041 & 0.041 & 0.057 & 0.057 & 0.041 & 0.057 & 0.031 & 0.046\\
\bottomrule
\end{tabular}
\end{footnotesize}

\vspace{1ex}
%{\raggedright \footnotesize {\it Note}: Only $AR_{naive}$ is reported as corrected Anderson--Rubin statistics are unavailable when two endogenous variables are present. \par}
\end{table}

\begin{figure}[ht]
    \centering
    % Replace 'chibar2' with 'chi2' for the second set of figures as needed
    %\begin{subfigure}[b]{0.47\textwidth}
    %    \includegraphics[width=\textwidth]{mainTrinity_chibar2_DGP2_200_0.05_0.1_SJIVE.pdf}
    %    \caption{SJIVE, $\bar\chi^2$}
    %\end{subfigure}
    %\hfill
    %\begin{subfigure}[b]{0.47\textwidth}
    %    \includegraphics[width=\textwidth]{mainTrinity_chibar2_DGP2_200_0.05_0.1_JIVE1.pdf}
    %    \caption{JIVE1, $\bar\chi^2$}
    %\end{subfigure}
    
    \begin{subfigure}[b]{0.47\textwidth}
        \includegraphics[width=\textwidth]{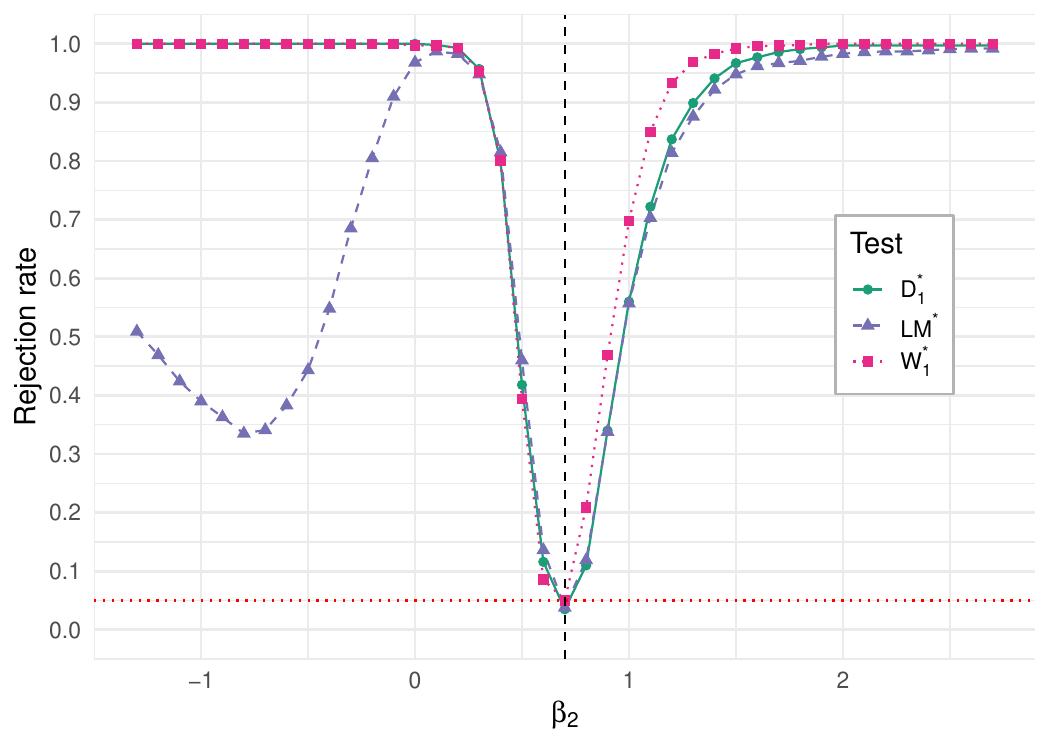}
        \caption{SJIVE}
    \end{subfigure}
    \hfill
    \begin{subfigure}[b]{0.47\textwidth}
        \includegraphics[width=\textwidth]{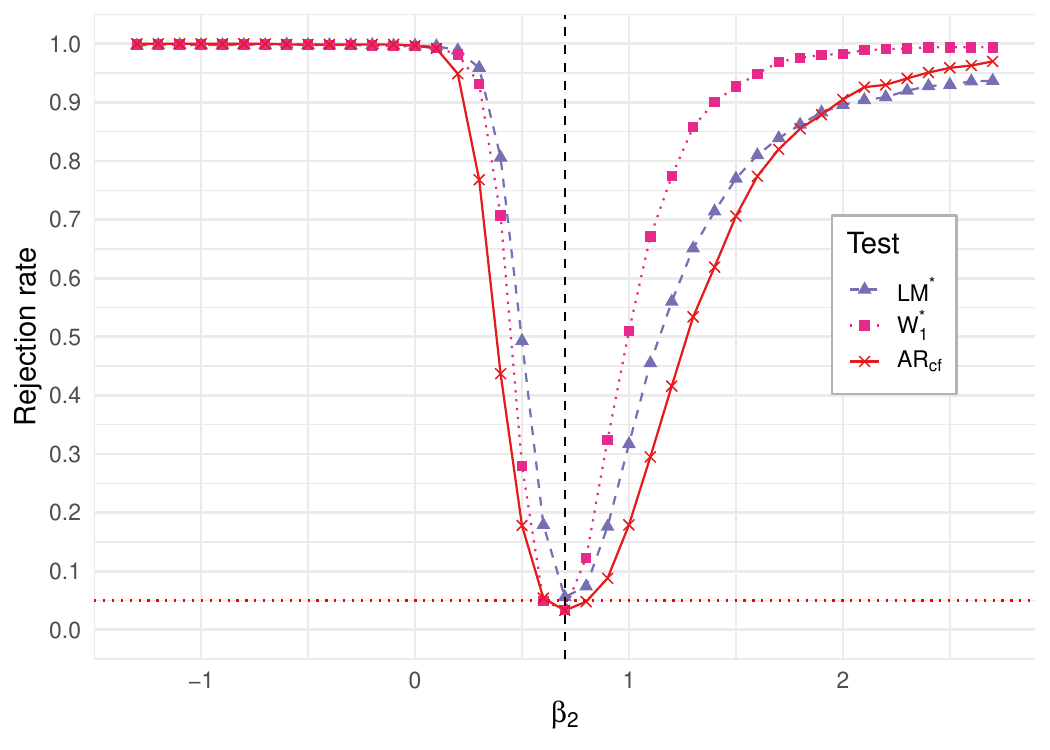}
        \caption{JIVE1}
    \end{subfigure}
    %\begin{subfigure}[b]{0.47\textwidth}
    %    \includegraphics[width=\textwidth]{mainTrinity_chi2_DGP2_200_0.05_0.1_JIVE1.pdf}
    %    \caption{JIVE1, $\chi^2$}
    %\end{subfigure}

    %\begin{subfigure}[b]{0.47\textwidth}
    %   \includegraphics[width=\textwidth]{mainAR_N_DGP2_200_0.05_0.1_JIVE1.pdf}
    %    \caption{JIVE1, $\mathcal{N}(0,1)$}
    %\end{subfigure}
    
    \caption{Power curves for DGP2 ($n=200$, $\alpha = 0.05$, $r = 0.1$). Results based on 1000 repetitions. The horizontal dotted red line denotes the $5\%$ nominal rejection level, while the vertical dotted black line corresponds to $\beta_2=0.7$. Panel (a) plots $D_1^*$, $LM^*$, $W_1^*$ based on the SJIVE objective function; panel (b) plots $W_1^*$, $LM^*$, $AR_{cf}$ based on the JIVE1 objective function. The statistics are all asymptotically $\chi^2_1$ distributed, besides $AR_{cf}$ which is $\mathcal N(0,1)$ in the limit.
    }
    \label{fig:mainDGP2_200_0.05_0.1}
\end{figure}

\begin{figure}[ht]
    \centering
    % Replace 'chibar2' with 'chi2' for the second set of figures as needed
    %\begin{subfigure}[b]{0.47\textwidth}
    %    \includegraphics[width=\textwidth]{mainTrinity_chibar2_DGP2_200_0.1_0.1_SJIVE.pdf}
    %    \caption{SJIVE, $\bar\chi^2$}
    %\end{subfigure}
    %\hfill
    %\begin{subfigure}[b]{0.47\textwidth}
    %    \includegraphics[width=\textwidth]{mainTrinity_chibar2_DGP2_200_0.1_0.1_JIVE1.pdf}
    %    \caption{JIVE1, $\bar\chi^2$}
    %\end{subfigure}
    
    \begin{subfigure}[b]{0.47\textwidth}
        \includegraphics[width=\textwidth]{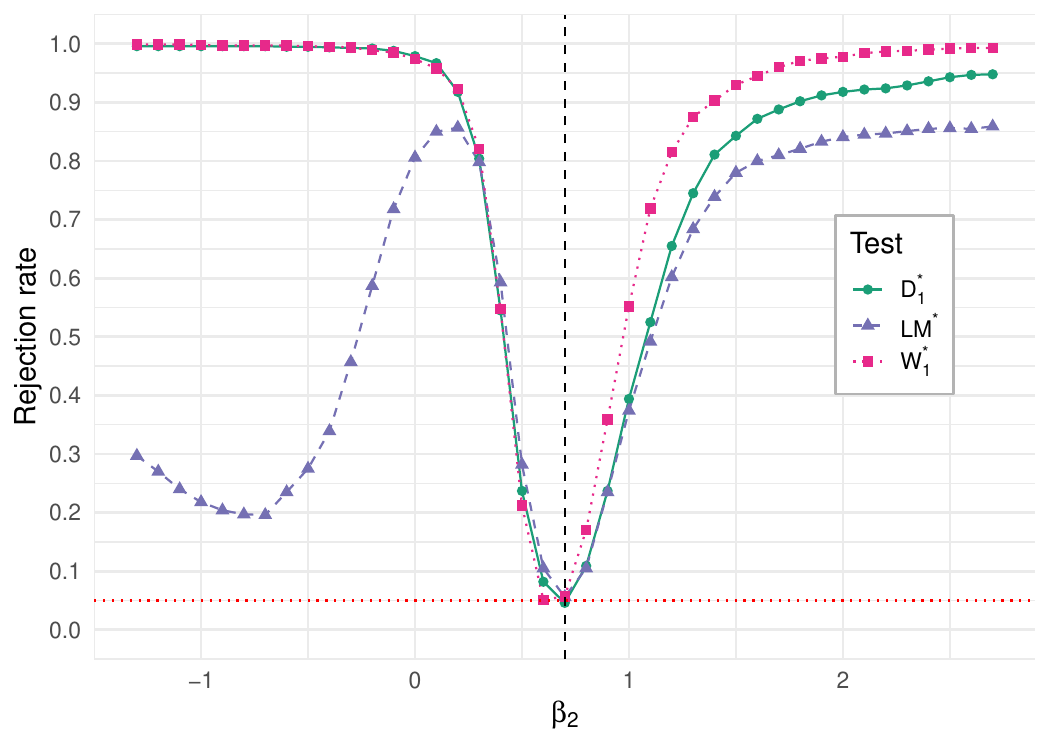}
        \caption{SJIVE}
    \end{subfigure}
    \hfill
    \begin{subfigure}[b]{0.47\textwidth}
        \includegraphics[width=\textwidth]{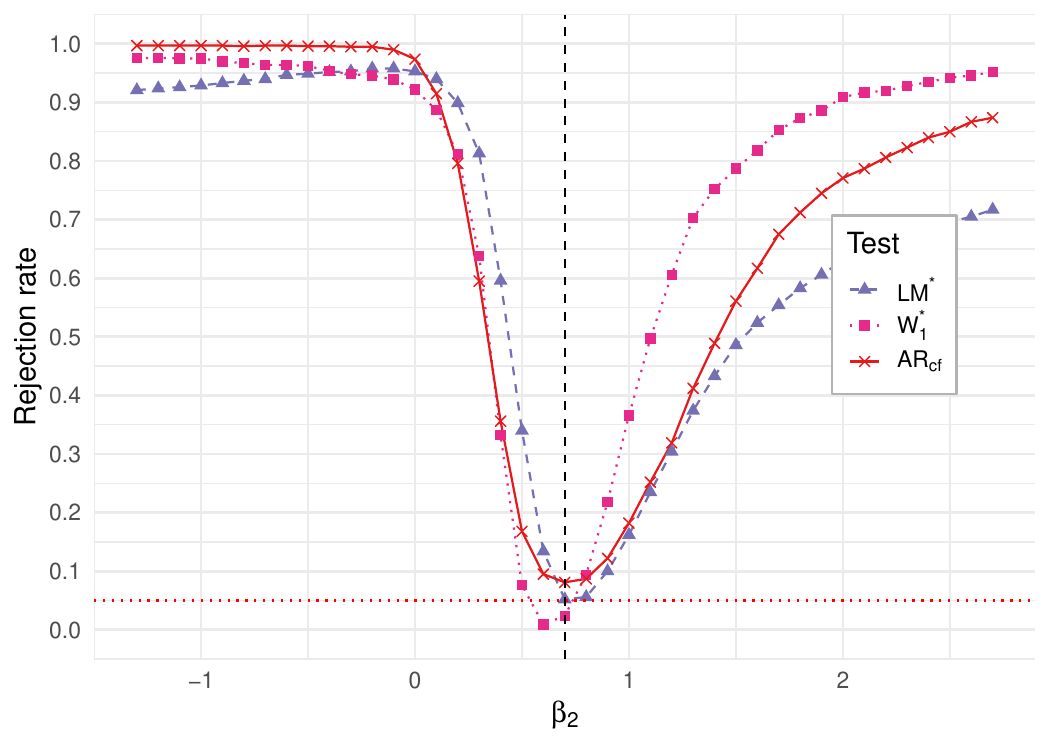}
        \caption{JIVE1}
    \end{subfigure}    
    %\begin{subfigure}[b]{0.47\textwidth}
    %    \includegraphics[width=\textwidth]{mainTrinity_chi2_DGP2_200_0.1_0.1_JIVE1.pdf}
    %    \caption{JIVE1, $\chi^2$}
    %\end{subfigure}

    %\begin{subfigure}[b]{0.47\textwidth}
    %   \includegraphics[width=\textwidth]{mainAR_N_DGP2_200_0.1_0.1_JIVE1.pdf}
    %    \caption{JIVE1, $\mathcal{N}(0,1)$}
    %\end{subfigure}
    
    \caption{Power curves for DGP2 ($n=200$, $\alpha = 0.1$, $r = 0.1$). Results based on 1000 repetitions. The horizontal dotted red line denotes the $5\%$ nominal rejection level, while the vertical dotted black line corresponds to $\beta=0.7$. Panel (a) plots $D_1^*$, $LM^*$, $W_1^*$ based on the SJIVE objective function; panel (b) plots $W_1^*$, $LM^*$, $AR_{cf}$ based on the JIVE1 objective function. The statistics are all asymptotically $\chi^2_1$ distributed, besides $AR_{cf}$ which is $\mathcal N(0,1)$ in the limit.}
    \label{fig:mainDGP1_200_0.1_0.1}
\end{figure}

\subsubsection{Comments on simulations}\label{comments}

The performance of the tests is measured in terms of rejection rates compared to a 5\% nominal rate. The asymptotically chi-square distributed tests follow a $\chi^2_1$, while the estimated version of $\vvarphi$ has a positive number in the first position and $g-1$ zeros in the remaining entries. The p values for the chi-bar-squared tests are computed using the {\tt R} function {\tt farebrother} \citep{DL10}. For power,  we focus on weak instruments. Namely, $r=32$ for DGP1 and $r=0.1$ for DGP2. The power plots for experiments with stronger instruments can be found in Appendix \ref{completesim} along with results based on tests that use cross-fit variances. As previously mentioned in this section, since some test statistics show the same type of behaviour, they are excluded from the plots in Figures \ref{fig:mainDGP1_200_0.05_32} to \ref{fig:mainDGP1_200_0.1_0.1}.
Also in this case, the full set of simulation results can be found in Appendix \ref{completesim}.

\paragraph{Size}
The size results in Table \ref{tab:DGP1_200} and Table \ref{tab:DGP2_200} show similar scenarios. 
In general we notice that size distortions oscillate between about $+3\%$ to about $-4\%$ with respect to the $5\%$ nominal size. The former result is observed for the $D$ statistic while the latter is observed for the $AR$ statistic with naive variance estimator. Both cases pertain to DGP1.

More specifically, we notice that the $LM$ and $LM^*$ tests show the same rejection rates. They consistently maintain good size properties relative to other statistics under the SJIVE/HLIM specification, particularly in the case of DGP1; under DGP2, they have a slight tendency to overreject.  Wald statistics perform similarly under both asymptotic approximations, yet they tend to overreject under the SJIVE/HLIM objective functions. The opposite occurs when JIVE1 or JIVE2 is used. As in the case of Lagrange multiplier tests both Wald statistics show the same rejection rates.
In contrast, statistics $D$ and $D^*$ display noticeable differences: $D$ tends to overreject relative to the nominal 5\% level, whereas $D^*$ remains conservative, but close to the 5\% nominal level. The proposed tests are also compared to Anderson-Rubin statistics $AR_{naive}$ and $AR_{cf}$ \citep{MS22, CMS21}. We observe that the Anderson-Rubin statistics consistently underreject. For DGP1, the tests with cross-fit variance ($AR_{cf}$) remain conservative but show some size improvements with respect to the naive counterpart.%\footnote{Table \ref{tab:DGP2_200} does not feature a $AR_{cf}$ test because, to the best of our knowledge, there are no such tests for models with multiple endogenous regressors.}

\paragraph{Power}
The power results are shown in Figure \ref{fig:mainDGP1_200_0.05_32} to Figure \ref{fig:mainDGP1_200_0.1_0.1}. The figures include results for tests based on the SJIVE and JIVE1 objective functions only, as they are similar to and, in some cases, improve over tests based on HLIM and JIVE2. For similar reasons, we exclude tests that are asymptotically distributed as a chi-bar-square. In our experiments, these power differences are noticeable but generally small (see figures in Appendix \ref{completesim}), with the exception of the Lagrange multiplier tests. It is known that Lagrange multiplier tests may lose power away from the true value.\footnote{\cite{Andrews16} reports this phenomenon in the context of combination of tests.} In our simulations, we observe a significant power drop for $LM$ tests based on both the SJIVE and HLIM objective functions and for both asymptotic distributions (see panels (a) and (b) of the figures in Appendix \ref{completesim}). $LM$ tests based on JIVE1 and JIVE2 show some power loss but not as large as the aforementioned cases (see panels (c) and (d) of the figures in Appendix \ref{completesim}). Interestingly, the power loss observed in the case of JIVE1 and in particular in the case of SJIVE is visibly smaller than the corresponding cases for JIVE2 and HLIM, respectively. Intuitively, this may be attributed to the different jackknife scheme used for SJIVE/JIVE1 tests compared to HLIM/JIVE2 tests.

Of the three types of statistics in Figure \ref{fig:mainDGP1_200_0.05_32} to Figure \ref{fig:mainDGP1_200_0.1_0.1}, we notice that the Wald tests tend to have the largest power in general as we move away from the truth. The Lagrange multiplier tests, despite having good size properties in some cases, tend to have the worst power performance. The distance statistics seem to be somewhere in the middle, showing in a number of cases power properties comparable to those of the Wald tests and good size properties in the case of $D^*$.

Interestingly, while Wald statistics generally demonstrate good size properties and high power away from the true parameter, their power may be low near the true value. This is particularly evident for Wald tests based on JIVE1 and JIVE2 objective functions (see panels (b) and (d) in Figure \ref{fig:mainDGP1_200_0.05_32} to Figure \ref{fig:mainDGP1_200_0.1_0.1}. This phenomenon is much less visible for SJIVE and HLIM objective functions. In general, the $D^*$ statistic appears to strike a good balance between size control and power. When compared to Anderson-Rubin statistics, apart from the pathological case of $LM$ and $LM^*$, our proposed tests generally show clear improvements in terms of power. 

It is useful to contrast our results with some of the findings of the weak instrument literature developed in the 1990s and early 2000s (and related recent contributions). A key result is that conventional IV estimation and inference can be unreliable under weak identification. In particular, under different definitions of weakness -- whether near non-identification as in \cite{Dufour97} or local-to-zero as in \cite{SS97} -- Wald statistics fail to provide valid inference. This has set the ground for the development of alternative statistics that are robust to the presence of weak instruments.
Under the local-to-zero asymptotic framework of \cite{SS97}, \cite{WZ98} show that certain LM and LR statistics can be boundedly pivotal, yielding conservative yet asymptotically valid inference. A complementary approach is to construct statistics with pivotal limiting distributions as in the case of Kleibergen's K statistic \citep{Kleibergen02} that also address the limitations of the AR test in the overidentified case. 
More recently, \cite{KN23} highlight additional failures of 2SLS-based t-tests and recommend robust alternatives such as the AR test in exactly identified models and LIML estimation and CLR tests in overidentified settings. Our contribution can be also interpreted as a complement to this literature in a many instruments context.

\subsection{Empirical application: the effect of alcohol consumption on BMI in the UK Biobank}

\subsubsection{Data}

We apply our proposed test statistics to examine the causal effect of alcohol consumption on body mass index (BMI) using data from the UK Biobank (UKB). The UK Biobank is a large-scale biomedical database containing extensive genetic and health information for approximately 500,000 individuals aged 40--69 in the United Kingdom \citep{Sudlow2015, Bycroft2018}. 

Our starting sample consists of 337,260 participants who were included in the genetic principal component analysis and classified as genetically Caucasian, selected from the 501,932 individuals available in the UK Biobank. Mendelian randomisation exploits genetic variants as instrumental variables to identify causal effects. Because genetic variants are randomly allocated at conception, they provide a source of exogenous variation that mimics experimental randomisation and can help overcome confounding and reverse causality that typically bias observational estimates of alcohol's health effects \citep{DaveySmith2003,Burgess2015}. In this application, genetic variants associated with alcohol consumption serve as instruments for alcohol intake.

We initially extract 98 single nucleotide polymorphisms (SNPs) previously identified as associated with alcohol consumption. To reduce linkage disequilibrium and ensure instrument independence, we perform a pruning procedure that removes SNPs with pairwise correlation exceeding 0.15, yielding a final set of 91 instruments.

Following \cite{Topiwala2022}, alcohol consumption is constructed as total weekly intake in grams of ethanol. Participants report consumption either weekly or monthly across several beverage categories, including red wine, champagne or white wine, beer or cider, spirits, fortified wine, and other alcoholic beverages. Monthly quantities are converted to weekly consumption by dividing by 4.3. Beverage-specific unit conversions are then applied, using factors such as 1.7 units per glass for wine, 2.4 units per pint for beer or cider, and one unit per measure of spirits. Alcohol units are converted to grams using the standard conversion of one UK unit corresponding to eight grams of ethanol.

The dataset undergoes several cleaning steps. Observations with missing or unknown alcohol consumption information are removed, resulting in the exclusion of 74,717 participants. Outliers in alcohol consumption are then removed using Tukey's fences with parameter $k = 1.5$, where the first and third quartiles are $Q_1 = 54.4$ and $Q_3 = 200$. This procedure excludes an additional 12,991 observations. After dropping all remaining observations with missing values, the final estimation sample consists of 97,987 individuals.

The final dataset used for estimation includes the outcome variable BMI, the endogenous variable total weekly alcohol consumption measured in grams per week, the 91 genetic variants used as instruments, and the following control variables: sex, age at recruitment, age squared, indicators for the UK Biobank assessment center, and the first ten genetic principal components to account for population stratification.

\subsubsection{Results}
The results in Table \ref{tab:jivi_ar_results} show uniform and extremely strong rejection of the null hypothesis across all methods and test statistics. 
All JIVI-based procedures yield p values at or below $0.0002$, irrespective of whether inference is conducted using the $\bar{\chi}^2$ or the standard $\chi^2$ reference distribution. 
This uniformity indicates that inference in the full sample is largely insensitive to both the choice of statistic and the choice of asymptotic approximation.

Among the jackknife-based estimators, SJIVE and HLIM deliver virtually identical outcomes across all statistics. 
Similarly, JIVE1 and JIVE2 produce indistinguishable results, with p values of $0.0001$ for all trinity statistics and $0.0000$ for both Anderson--Rubin implementations. %\textcolor{red}{\bf the size of the estimates is different between the two types of objective functions.}

Comparisons across statistics within each method reveal negligible differences between $D$, $W_1$, $W_2$, and $LM$, as well as their starred counterparts. 
Switching from the $\bar{\chi}^2$ distribution to the standard $\chi^2$ distribution has no material effect on inference.

Finally, the Anderson--Rubin tests (columns 12--13) yield p values of $0.0000$ for JIVE1 and JIVE2, reinforcing the same qualitative conclusion of strong rejection. 
Overall, the table documents a high degree of robustness across estimators, statistics, and reference distributions in the full sample.\footnote{Under the null, we use the OLS estimator for computational convenience. Although OLS is consistent, it does not preserve the finite-sample equality among the test statistics mentioned in Section \ref{MonteCarlo} and established in Appendix \ref{JIVE}; this discrepancy is purely numerical and has no asymptotic consequence.} 

%\textcolor{red}{\bf the results for JIVE1/JIVE2 in the data example do not match the theoretical results for the equality of tests, this could be due to the fact that the estimator under the null is OLS, we should perhaps mention this in a note, eg. ``For computational simplicity under the null we use the OLS estimator. This is a consistent estimator but it does not deliver the equality of tests discussed in the previous section and in the appendix''. ChatGPT assisted: ``Under the null, we use the OLS estimator for computational convenience. Although OLS is consistent, it does not preserve the finite-sample equality among the test statistics mentioned in the previous section and established in the appendix; this discrepancy is purely numerical and has no asymptotic consequence.''}

\begin{table}[ht]
\tiny \centering
\caption{Trinity and Anderson-Rubin test results.}
\label{tab:jivi_ar_results}
\begin{tabular}{lccccccccccccc}
\toprule
 & & \multicolumn{4}{c}{$\bar{\chi}^2$} & \multicolumn{5}{c}{ $\chi^2$} & \multicolumn{2}{c}{$\mathcal N(0,1)$} \\
\cmidrule(lr){3-6} \cmidrule(lr){7-11} \cmidrule(lr){12-13}
Method & Estimate &$D$ & $W_1$ & $W_2$ & $LM$ & $D_1^*$ & $D_2^*$ & $W_1^*$ & $W_2^*$ & $LM^*$ & $AR_{naive}$& $AR_{cf}$ \\
\midrule
SJIVE& 0.0141&  4.5226 &  4.1727 &  4.1727 &  4.5195 & 30.3355 & 30.3386 & 14.2885 & 14.2599 & 14.7589 & -- & -- \\
P value & --&  0.0001 &  0.0002 &  0.0002 &  0.0001 &  0.0000 &  0.0000 &  0.0001 &  0.0002 &  0.0002
& -- & -- \\ 
HLIM & 0.0140& 4.5226 &  4.1727 &  4.1727 &  4.5195 & 30.3355 & 30.3386 & 14.2885& 14.2599 & 14.7589 & -- & -- \\
P value & --&  0.0001 &  0.0001 &  0.0001 &  0.0001 &  0.0000 &  0.0000 &  0.0001 &  0.0001 &  0.0001 & -- & -- \\ 
JIVE1 & 0.0072&  2.2265 &  2.2265 &  2.2265 &  2.2265 & 14.8291 & --  & 14.4355 & 14.4066 & 14.7995 & 28.4865 & 28.8247 \\
P value & --&  0.0001 &  0.0001 &  0.0001 &  0.0001 &  0.0001 & --  &  0.0001 &  0.0001 &  0.0001&
  0.0000 & 0.0000 \\
JIVE2 &0.0072&  2.1857 &  2.1857 &  2.1857 &  2.1857 & 14.8292 & -- & 14.4354 & 14.5121 & 14.9080
  & 28.4903 & 28.6366 \\
P value & --&  0.0001 &  0.0001 &  0.0001 &  0.0001 &  0.0001 & --  &  0.0001 &  0.0001 &  0.0001
 & 0.0000 & 0.0000 \\ 
\bottomrule
\end{tabular}
\begin{tablenotes}
\small
\item \textit{Notes:} P values for testing $H_0:\beta=0$. Column 2 reports the point estimate. Columns 3--6 report trinity statistics using the $\bar{\chi}^2$ reference distribution. Columns 7--11 report trinity statistics using the standard $\chi^2$ reference distribution. Columns 12--13 report Anderson--Rubin test results.
\end{tablenotes}
\end{table}

To assess whether the uniformity of inference in the full sample is driven by large-sample size, we repeat the analysis using a random subsample of 10{,}000 observations.

\begin{table}[ht]
\tiny \centering
\caption{Trinity and Anderson--Rubin test results (random sample of 10{,}000).}
\label{tab:jivi_ar_results_10k}
\begin{tabular}{lccccccccccccc}
\toprule
 & &\multicolumn{4}{c}{$\bar{\chi}^2$} & \multicolumn{5}{c}{$\chi^2$} & \multicolumn{2}{c}{$\mathcal N(0,1)$} \\
\cmidrule(lr){3-6} \cmidrule(lr){7-11} \cmidrule(lr){12-13}
Method & Estimate & $D$ & $W_1$ & $W_2$ & $LM$ & $D_1^*$ & $D_2^*$ & $W_1^*$ & $W_2^*$ & $LM^*$ & $AR_{naive}$ & $AR_{cf}$ \\
\midrule

SJIVE & 0.0404 & 0.2053 & 0.1139 & 0.1140 & 0.2239 & 16.2205 & 15.9947 & 5.6537 & 5.2307 & 6.3915 & -- & -- \\
P value & -- & 0.0014 & 0.0174 & 0.0174 & 0.0087 & 0.0001 & 0.0001 & 0.0115 & 0.0174 & 0.0222 & -- & -- \\

HLIM  & 0.0400 & 0.1445 & 0.0839 & 0.0840 & 0.1599 & 14.9948 & 14.9374 & 5.1821 & 4.8634 & 5.7860 & -- & -- \\
P value & -- & 0.0028 & 0.0228 & 0.0227 & 0.0130 & 0.0001 & 0.0001 & 0.0162 & 0.0228 & 0.0274 & -- & -- \\

JIVE1 & 0.0244 & 0.0781 & 0.0781 & 0.0781 & 0.0781 & 6.2307 & -- & 4.1337 & 3.8164 & 5.7524 & 3.6399 & 3.7173 \\
P value & -- & 0.0419 & 0.0420 & 0.0419 & 0.0125 & 0.0126 & -- & 0.0165 & 0.0420 & 0.0508 & 0.0001 & 0.0001 \\

JIVE2 & 0.0237 & 0.0579 & 0.0577 & 0.0579 & 0.0579 & 5.9893 & -- & 4.0755 & 3.8230 & 5.6183 & 3.7171 & 3.7294 \\
P value & -- & 0.0432 & 0.0435 & 0.0432 & 0.0143 & 0.0144 & -- & 0.0178 & 0.0435 & 0.0506 & 0.0001 & 0.0001 \\

\bottomrule
\end{tabular}
\begin{tablenotes}
\small
\item \textit{Notes:} P values for testing $H_0:\beta=0$. Column 2 reports the point estimate. Columns 3--6 report trinity statistics using the $\bar{\chi}^2$ reference distribution. Columns 7--11 report trinity statistics using the standard $\chi^2$ reference distribution. Columns 12--13 report Anderson--Rubin test results.
\end{tablenotes}
\end{table}

Table \ref{tab:jivi_ar_results_10k} shows that, in contrast to the full sample, nontrivial differences across test statistics and reference distributions emerge when inference is based on a smaller sample.

Among the jackknife-based estimators, SJIVE and HLIM continue to display very similar behavior. 
Both yield small p values for the $D$ statistic, while the remaining statistics, $W_1$, $W_2$, and $LM$ produce larger but still statistically significant p values in the 0.01--0.03 range. 
Under the standard $\chi^2$ distribution, all starred statistics remain significant, with $LM^*$ delivering the smallest p-values for both methods.

For JIVE1 and JIVE2, dispersion across statistics is more pronounced. 
The weighted-$D$ statistic yields p values around 0.04, whereas $W_2$ and $LM$ are close to 0.05, implying less pronounced rejection under conventional thresholds. 
Under the standard $\chi^2$ distribution, $D_1^*$, $W_1^*$, and $W_2^*$ cluster around 0.04, while $LM^*$ remains around 0.01--0.02.

The Anderson--Rubin tests (columns 12--13) yield p values close to $0.0001$ for both JIVE1 and JIVE2, which are substantially smaller than those produced by the trinity statistics. 
As in the full sample, this illustrates that Anderson--Rubin and trinity procedures can generate materially different degrees of rejection in finite samples.

Overall, the sub-sample results indicate that the near-complete uniformity observed in the full sample does not extend to smaller samples, and that inference can depend on the particular statistic and reference distribution employed.

\clearpage

%\appendix
%\setcounter{table}{0}
%\setcounter{figure}{0}
%\renewcommand\thetable{\Alph{section}.\arabic{table}}
%\renewcommand\thefigure{\Alph{section}.\arabic{figure}}

\section{Conclusions}\label{conclusion}
In analogy with classical likelihood theory, we introduce a unified approach to inference for the parameters of linear regression models in the presence of endogeneity, heteroskedasticity and many, potentially weak, instruments. The test statistics are based on objective functions commonly encountered in the modern literature on many instruments \citep[see][among others]{HNWCS12,BC15}. Our tests can be used for inference on simple hypotheses as well as on hypotheses expressed in terms of linear restrictions. We show that their asymptotic distribution is chi-bar-square but a more conventional chi square limit can be found after appropriate modifications, a result that, in our opinion, is more appealing to practitioners. In an extensive simulation experiment we find that, generally, our methods tend to have better finite sample properties than competing Anderson-Rubin statistics both in terms of size and power. 

We also illustrate the practical relevance of our methodology through an empirical application that studies the causal effect of alcohol consumption on body mass index using genetic variants as instrumental variables in the UK Biobank. This setting naturally gives rise to a large number of instruments, some of which may be weak, making it well suited for the framework developed in this paper. The application demonstrates how the proposed tests can be implemented in a realistic empirical environment with many potentially weak instruments and heteroskedastic disturbances. The results highlight the feasibility of our procedures in large-scale genetic datasets and illustrate how the proposed statistics can provide reliable inference in settings where standard IV methods may perform poorly.

It is reasonable to believe that our methodology can be adapted to models with many instruments with fixed effects as in \cite{CSW23} as well as to models with cluster dependence \cite{FLM25}. A further important extension is to develop confidence sets by inverting our tests. This is especially relevant under weak identification, where informative procedures may yield confidence sets with nonstandard geometry (e.g., unbounded or disconnected regions). An important direction is to characterise these shapes by adapting the geometric results of \cite{DT05}. Some of these problems are part of our ongoing research agenda.

\clearpage

\appendix
\setcounter{table}{0}
\setcounter{figure}{0}
\setcounter{equation}{0}
\setcounter{footnote}{0}
\renewcommand\thetable{\Alph{section}.\arabic{table}}
\renewcommand\thefigure{\Alph{section}.\arabic{figure}}
\renewcommand\theequation{\Alph{section}.\arabic{equation}}
\renewcommand\thefootnote{\Alph{section}.\arabic{footnote}}

\section{Useful lemmata}
For this study it is crucial to have a suitable CLT, the most natural reference in this case is Lemma A2 in \cite{CSHNW12}. 

%\begin{lemma}\label{Hadamard}
%Let $\mA$ be a $n\times n$ matrix, then $\mA\odot\mA=\mK_1\left\{ \diag(\vec(\mA))^2\right\}\mK_2$ where $\mK_1=\viota_n'\otimes \mI_n$ and $\mK_2=\mI_n\otimes \viota_n$. \textcolor{red}{\bf I do not know whether this result is new or not, I guess it must have appeared somewhere before, yet I could not find it anywhere else. The generalization to $n\times m$ matrices seems straightforward, I do not see any particular requirement for matrix $\mA$. I should come up with a formal proof.}
%\end{lemma}

\begin{lemma}
[Lemma A2 in \cite{CSHNW12}]\label{CLT} Assume that the following conditions hold.

\begin{enumerate}
\item $\mC$\ is a symmetric matrix with zero main diagonal elements whose
$\left(  i,j\right)  $ element denoted by $C_{ij}$ satisfies $|C_{ij}|\leq
c_{u}|P_{ij}|$\ for any $i\neq j$, where $P_{ij}$ is element $\left(
i,j\right)  $ of the projection matrix $\mP=\mZ\left(  \mZ^{\prime}\mZ\right)
^{-1}\mZ^{\prime}$;

\item $(\vw_{1n},\vu_{1},\varepsilon_{1}),\dots,(\vw_{nn},\vu_{n}%
,\varepsilon_{n})$ are independent and $\mD_{n}=\sum_{i=1}^{n}\E\left[
\vw_{in}\vw_{in}^{\prime}\right]  $ satisfies $\Vert\mD_{n}\Vert\leq c_{u}$;

\item $\E\left[  \vw_{in}\right]  =\vzeros$, $\E\left[  \vu_{i}\right]
=\vzeros$, $\E\left[  \varepsilon_{i}\right]  =0$ and $\E\left[  \Vert
\vu_{i}\Vert^{4}\right]  \leq c_{u}$ and $\E\left[  \varepsilon_{i}%
^{4}\right]  \leq c_{u}$;

\item $\sum_{i=1}^{n}\E\left[  \Vert\vw_{in}\Vert^{4}\right]  \rightarrow0$;

\item $k\rightarrow\infty$ as $n\rightarrow\infty$.
\end{enumerate}

\noindent Then for
\[
\bar{\mSigma}_{n}=\frac{1}{k}\sum_{i,j=1}^{n}C_{ij}^{2}\left(  \E\left[
\vu_{i}\vu_{i}^{\prime}\right]  \E\left[  \varepsilon_{j}^{2}\right]
+\E\left[  \vu_{i}\varepsilon_{i}\right]  \E\left[  \varepsilon_{j}%
\vu_{j}^{\prime}\right]  \right)
\]
and any sequences $\vc_{1n}$ and $\vc_{2n}$ of conformable vectors with
bounded norm and $\mXi_{n}=\vc_{1n}^{\prime}\mD_{n}\vc_{1n}+\vc_{2n}^{\prime
}\bar{\mSigma}_{n}\vc_{2n}>\frac{1}{c_{u}}$,
\[
Y_{n}=\mXi_{n}^{-1/2}\left(  \vc_{1n}^{\prime}\sum_{i=1}^{n}\vw_{in}%
+\frac{\vc_{2n}^{\prime}}{\sqrt{k}}\sum_{i,j=1}^{n}\vu_{i}C_{ij}%
\varepsilon_{j}\right)  \rightarrow_{d}\mathcal{N}(0,1).
\]
%\textcolor{red}{Why is this needed? }That is, $P(Y_{n}\leq y)\rightarrow\Phi(y)$ for all y. \textcolor{blue}{I think it's not needed, if I remember correctly it's verbatim from Chao's paper}
\end{lemma}

%\begin{lemma}[Lemma A3 in \cite{CSHNW12}]\label{Lemma A3}
%Let $(w_i,u_i)$, $i=1,\dots,n$ be independent, then there exists a positive constant $c_u$ such that
%\begin{align*}
%\left\|\sum_{i\ne j}P_{ij}^2w_iu_j-\E\left[\sum_{i\ne j}P_{ij}^2w_iu_j\right]\right\|^2\le c_uB_n \;\; a.s.
%\end{align*}
%where 
%\begin{align*}
%B_n=k\left\{
%\overline\Var[w_i]\overline\Var[u_i]+
%\overline\Var[w_i]\overline\E[u_i]^2+
%\overline\E[w_i]^2\overline\Var[u_i]
%\right\}
%\end{align*}
%for $\overline\Var[x_i]=\max_i\Var[x_i]$ and $\overline\E[x_i]=\max_i\E[x_i]$ given that $x_i\in\{w_i,u_i\}$.
%\end{lemma}
%
%\begin{lemma}[Lemma A4 in \cite{CSHNW12}]\label{Lemma A4}
%Assume there exists $c_u>0$ such that $(w_1,u_1,\eta_1),\dots,(w_n,u_n,\eta_n)$ are independent with $\E[w_i]=\frac{a_i}{\sqrt n}$, $\E[u_i]=\frac{b_i}{\sqrt n}$, $a_i<c_u$, $b_i<c_u$, $\E[\eta_i^2]c_u$, $\Var[w_i]\le\frac{c_u}{r_{\min}}$, $\Var[u_i]\le\frac{c_u}{r_{\min}}$ and there exists $\vpi_n$ such that $\max_i|a_i-\mZ_i'\vpi_n|\to 0$ and $\frac{\sqrt k}{r_{\min}}\to 0$. Then,
%\begin{align*}
%A_n=\E\left[
%\sum_{i\ne j\ne k} w_iP_{ik}\eta_kP_{kj}u_j
%\right] = O_p(1), \;\;
%\sum_{i\ne j\ne k} w_iP_{ik}\eta_kP_{kj}u_j-A_n\to_p0.
%\end{align*}
%\end{lemma}

There are some simplifications in Lemma \ref{CLT}, in particular, it is omitted the conditioning on the instruments. This is because we assume the instruments to be fixed. However, if it were convenient to assume the instruments as stochastic we may interpret the expectations as conditional expectations. To apply Lemma \ref{CLT} we need to adapt it to our jackknife structure, defined by matrix $\mC$. This is done, e.g., in \cite{BC15} and \cite{CMS21}. 

The next result shows that the distribution of a quadratic form for a normally distributed random vector is chi-bar-square (see Definition 1 in \cite{Hansen21}). The result is a generalisation of Lemma 1 in \cite{Hansen21} to the case when the matrix in the quadratic form is not necessarily invertible.

\begin{lemma}\label{chimix}
If $q=\vx'\mTheta\vx$ where $\vx\sim \mathcal N(\vzeros,\mV)$ and $\mTheta$ symmetric positive semidefinite, then $q\sim \bar\chi^2(\boldsymbol\varphi)$, where $\boldsymbol\varphi$ is the vector of eigenvalues of $\mTheta\mV$. %and the elements of $\boldsymbol\varphi$ are nonnegative. \textcolor{red}{\bf the last sentence is omitted in Hansen's lemma but it is clear from the definition of weighted chi square he gives.}
\end{lemma}
\begin{proof}
Let $\mW$ be a square matrix such that $\mW\mW'=\mV$. Then $\vx$ can be written as $\vx=\mW\vz$, where $\vz$ is a standard normal vector of the same dimension as $\mV$, so $q=\vx'\mTheta\vx=\vz'\mS'\mLambda\mS\vz$, where $\mS'\mLambda\mS$ is the eigenvalue decomposition of the symmetric positive semidefinite matrix $\mW'\mTheta\mW$ with $\mS$ an orthonormal matrix and $\mLambda$ a diagonal matrix having the eigenvalues of $\mW'\mTheta\mW$ on the main diagonal. Since $\mS$ is orthonormal, $\vt=\mS\vz$ is standard normally distributed, so $q=\vt'\mLambda\vt=\sum_{i} \lambda_i t_{i}^2$, where $t_{i}^2$ are independent $\chi_{1}^2$-distributed. Note that the vector of eigenvalues of $\mW'\mTheta\mW$ is the same as the vector of eigenvalues of $\mTheta\mW\mW'=\mTheta\mV$, so $q$ is weighted chi square according to Definition 1 in \cite{Hansen21}.
\end{proof}

\begin{lemma}[Lemma 9.7 in \cite{NM94}]\label{chisqqf}
If $\vz\sim\mathcal N(\mR\vrho,\mR)$ with $\rk(\mR)=r$ and $\mR^-$ is any symmetric matrix satisfying $\mR\mR^-\mR=\mR$, then $\vz'\mR^-\vz$ is chi square distributed with $r$ degrees of freedom and noncentrality parameter $\vrho'\mR\vrho$. 
\end{lemma}

\begin{lemma} \label{LA4}
Under Assumptions \ref{ass1}, \ref{ass2}, \ref{ass5}, 
\newline 1. ${\mX^{\prime}\mC\mX}/{r_{\min}}=
{\mH}/{r_{\min}}+o_{p}\left(  1\right)$, $\mH^{-1}\mX^{\prime}\mC\mX  =\mI_{g}+o_{p}(1)$.
\newline 2. $\E\left[\mX^{\prime}\mC\vvarepsilon\right]  =\vzeros,~\Var\left[\mX^{\prime}\mC\vvarepsilon\right]  \leq c_{u}r_{\max}\mI_{g}+c_{u}k\mI_{g}$.
\newline 3. $\E\left[\vvarepsilon^{\prime}\mC\vvarepsilon\right]  =\vzeros,~\Var\left[\vvarepsilon^{\prime}\mC\vvarepsilon\right]  \leq c_{u}k$.
\end{lemma}
\begin{proof}
The proof is similar to that of Lemma A.4 from Supplement to \cite{CMS21}.
\end{proof}

\begin{lemma}
\label{B}Suppose that Assumptions \ref{ass1} to \ref{ass3} hold. 

\begin{enumerate}
\item ${\mV^{\prime}\mB\mV}/{\tr\left(  \mB\right)  }\rightarrow_{p}\mSigma_{22}$, $\mX^{\prime}\mB\mX/\tr\left(  \mB\right)=O_{p}\left(  1\right)  $, $\widehat{\sigma}^{2}\left(  \vbeta_{n}\right)  \rightarrow_{p}\sigma^{2}$, \newline
$\mX^{\prime}\mB\left(  \vy-\mX\vbeta_{n}\right)  /\tr\left(  \mB\right)
\rightarrow_{p}\sigma_{21}$ for any $\vbeta_{n}$ consistent estimator of $\vbeta_{0}$.

\item $\widehat{\sigma}^{2}\left(  \vbeta\right)  \geq\sigma^{2}-\vsigma
_{12}\mSigma_{22}^{-1}\vsigma_{21}+o_{p}\left(  1\right)  $ for any $\vbeta$, where $\sigma
^{2}-\vsigma_{12}\mSigma_{22}^{-1}\vsigma_{21}>0$. 
\end{enumerate}
\end{lemma}

\begin{proof}

\begin{enumerate}
\item Consider first $\mX^{\prime}\mB\mX/\tr\left(  \mB\right)  $. From $\mX=\mZ\mPi+\mV$ we
obtain%
\begin{align*}
\dfrac{\mX^{\prime}\mB\mX}{\tr\left(  \mB\right)  } &  =\dfrac{\left(  \mZ\mPi+\mV\right)
^{\prime}\mB\left(  \mZ\mPi+\mV\right)  }{\tr\left(  \mB\right)  }\\
&  =\dfrac{\left(  \mZ\mPi\right)  ^{\prime}\mB\left(  \mZ\mPi\right)  }{\tr\left(
\mB\right)  }+\dfrac{\left(  \mZ\mPi\right)  ^{\prime}\mB\mV}{\tr\left(  \mB\right)
}+\dfrac{\mV^{\prime}\mB\mZ\mPi}{\tr\left(  \mB\right)  }+\dfrac{\mV^{\prime}\mB\mV}{\tr\left(
\mB\right)  }.
\end{align*}
Recall that $\mB$ is either $\mB$ from \cite{BC15} or $\mI_{n}$. In the
former case $\mB\mZ=\mO$, so
\[
\dfrac{\mX^{\prime}\mB\mX}{\tr\left(  \mB\right)  }=\dfrac{\mV^{\prime}\mB\mV}{\tr\left(
\mB\right)  }.
\]
In order to determine the limit of this, note that by (\ref{Sigmai})%
\[
\E\left[  \mV^{\prime}\mB\mV\right]  =\sum_{i=1}^{n}B_{ii}\mSigma_{i22}.
\]
Let $\vv$ be an arbitrary column of $\mV$. Then%
\[
\E\left[  \mV^{\prime}\mB\vv\right]  =\sum_{i=1}^{n}B_{ii}\sigma_{vi22},
\]
where $\sigma_{vi22}$ is the column of $\mSigma_{i22}$ corresponding to $\vv$.
Further,%
\[
\Var\left[  \mV^{\prime}\mB\vv\right]  =\E\left[  \mV^{\prime}\mB\vv\vv^{\prime}\mB\mV\right]
-\E\left[  \mV^{\prime}\mB\vv\right]  \E\left[  \mV^{\prime}\mB\vv\right]  ^{\prime},
\]
where{\small
\begin{align*}
\E\left[  \mV^{\prime}\mB\vv\vv^{\prime}\mB\mV\right]   &  =\sum_{i,j,k,\ell}%
B_{ij}B_{k\ell}\E\left[  \mV_{i}v_{j}v_{k}\mV_{\ell}^{\prime}\right]  \\
&  =\sum_{i,k}B_{ii}B_{kk}\E\left[  \mV_{i}v_{i}v_{k}\mV_{k}^{\prime}\right]
+\sum_{i\neq j}B_{ij}^{2}E\left[  \mV_{i}v_{j}v_{i}\mV_{j}^{\prime}\right]
+\sum_{i\neq j}B_{ij}B_{ji}E\left[  \mV_{i}v_{j}v_{j}\mV_{i}^{\prime}\right]  \\
&  =\sum_{i}B_{ii}^{2}\E\left[  v_{i}^{2}\mV_{i}\mV_{i}^{\prime}\right]  -\sum
_{i}B_{ii}^{2}\vsigma_{vi22}\vsigma_{vi22}^{\prime}\\
&  +\sum_{i,k}B_{ii}B_{kk}\vsigma_{vi22}\vsigma_{vk22}^{\prime}+\sum_{i\neq
j}B_{ij}^{2}\vsigma_{vi22}\vsigma_{vj22}^{\prime}+\sum_{i\neq j}B_{ij}^{2}%
\sigma_{vj22}^{2}\mSigma_{i22}.
\end{align*}
}In the last equality we have used (\ref{Sigmai}) and the symmetry of $\mB$;
$\sigma_{vj22}^{2}$ is the main diagonal element of $\mSigma_{j22}$
corresponding to $\vv$. Note that the third term is just $\E\left[  \mV^{\prime
}\mB\vv\right]  \E\left[  \mV^{\prime}\mB\vv\right]  ^{\prime}$, so%
\[
\Var\left[  \mV^{\prime}\mB\vv\right]  =\sum_{i}B_{ii}^{2}\E\left[  v_{i}^{2}%
\mV_{i}\mV_{i}^{\prime}\right]  -\sum_{i}B_{ii}^{2}\vsigma_{vi22}\vsigma
_{vi22}^{\prime}+\sum_{i\neq j}B_{ij}^{2}\vsigma_{vi22}\vsigma_{vj22}^{\prime
}+\sum_{i\neq j}B_{ij}^{2}\vsigma_{vj22}^{2}\mSigma_{i22}.
\]
Assumption \ref{ass2} implies that $\E\left[  v_{i}^{2}\mV_{i}\mV_{i}^{\prime
}\right]  \leq c_{u}I_{g}$, $\vsigma_{vi22}\vsigma_{vj22}^{\prime}\leq
c_{u}I_{g}$ and $\sigma_{vj22}^{2}\mSigma_{i22}\leq c_{u}I_{g}$ for any $i,j$,
so
\[
\Var\left[  \mV^{\prime}\mB\vv\right]  \leq c_{u}\left(  \sum_{i,j}B_{ij}%
^{2}\right)  I_{g}.
\]
It is known that $\sum_{i,j}B_{ij}^{2}=\tr\left(  \mB^{2}\right)  $ and by the
definition of $\mB$
\begin{align*}
\tr\left(  \mB^{2}\right)    & =\tr\left(  \left(  \mI-\mP\right)  \mD\left(
\mI-\mD\right)  ^{-1}\left(  \mI-\mP\right)  ^{2}\mD\left(  \mI-\mD\right)  ^{-1}\left(
\mI-\mP\right)  \right)  \\
& =\tr\left(  \left(  \mI-\mP\right)  ^{2}\mD\left(  \mI-\mD\right)  ^{-1}\left(
\mI-\mP\right)  ^{2}\mD\left(  \mI-\mD\right)  ^{-1}\right)  =\tr\left(  \mB\mD\left(
\mI-\mD\right)  ^{-1}\right)  ,
\end{align*}
where the latter equality follows from $\left(  \mI-\mP\right)  $ idempotent. By
direct computation we obtain that%
\[
\tr\left(  \mB\mD\left(  \mI-\mD\right)  ^{-1}\right)  =\sum_{i=1}^{n}B_{ii}%
\dfrac{P_{ii}}{1-P_{ii}},
\]
so Assumption \ref{ass1} implies that%
\[
\tr\left(  \mB^{2}\right)  \leq c_{u}tr\left(  \mB\right)  .
\]
Therefore, because $tr\left(  \mB\right)  =k$,%
\[
\Var\left[  \dfrac{\mV^{\prime}\mB\vv}{\tr\left(  \mB\right)  }\right]  \leq\dfrac
{c_{u}}{k}\mI_{g}\rightarrow\mO.
\]
Using Assumption \ref{ass3},%
\[
\E\left[  \dfrac{\mV^{\prime}\mB\mV}{\tr\left(  \mB\right)  }\right]  \rightarrow
\mSigma_{22},
\]
so we obtain the probablity limit%
\[
\dfrac{\mV^{\prime}\mB\mV}{\tr\left(  \mB\right)  }\rightarrow_{p}\mSigma_{22}.
\]
This also shows that
\begin{equation}
\dfrac{\mX^{\prime}\mB\mX}{\tr\left(  \mB\right)  }=O_{p}\left(  1\right)
.\label{xbx0}%
\end{equation}

Take now the case $\mB=\mI_{n}$; then $\tr\left(  \mB\right)  =n$ and%
%\begin{subequations}\label{0}%
\begin{equation}
\dfrac{\mX^{\prime}\mB\mX}{\tr\left(  \mB\right)  }=\dfrac{\mH}{n}+\dfrac{\left(
\mZ\mPi\right)  ^{\prime}\mV}{n}+\dfrac{\mV^{\prime}\mZ\mPi}{n}+\dfrac{\mV^{\prime}%
\mB\mV}{\tr\left(  \mB\right)  }.\label{xbx1}%
\end{equation}
The first term on the right hand side is $O\left(  1\right)  $ by Assumption
\ref{ass5}. Let now $\vv$ be an arbitrary column of $\mV$. Note that%
%\end{subequations}
\[
\E\left[  \dfrac{\left(  \mZ\mPi\right)  ^{\prime}\mV}{n}\right]  =\mO\text{ \ and
\ }\Var\left[  \dfrac{\left(  \mZ\mPi\right)  ^{\prime}\vv}{n}\right]  =\dfrac
{1}{n^{2}}\left(  \mZ\mPi\right)  ^{\prime}\E\left[  \vv\vv^{\prime}\right]  \mZ\mPi
\leq\dfrac{c_{u}}{n^{2}}\mH\rightarrow\mO,
\]
where the inequality follows by Assumption \ref{ass2} and the limit by Assumption \ref{ass5}. Consequently, $\left(
\mZ\mPi\right)  ^{\prime}\mV/n=o_{p}\left(  1\right)  $ and, similarly, the third
term $\mV^{\prime}\mZ\mPi/n=o_{p}\left(  1\right)  $. In order to determine the
limit of the last term, note that by (\ref{Sigmai})%
\[
\E\left[  \mV^{\prime}\mB\mV\right]  =\sum_{i=1}^{n}\mSigma_{i22}\text{ \ and
\ }\E\left[  \mV^{\prime}\mB\vv\right]  =\sum_{i=1}^{n}\vsigma_{vi22},
\]
where $\vv$ is an arbitrary column of $\mV$. Further,%
\[
\Var\left[  \mV^{\prime}\mB\vv\right]  =\E\left[  \mV^{\prime}\mB\vv\vv^{\prime}\mB\mV\right]
-\E\left[  \mV^{\prime}\mB\vv\right]  \E\left[  \mV^{\prime}\mB\vv\right]  ^{\prime},
\]
where%
{\small
\begin{align*}
\E\left[  \mV^{\prime}\mB\vv\vv^{\prime}\mB\mV\right]    & =\sum_{i,j,k,\ell}B_{ij}%
B_{k\ell}E\left[  \mV_{i}v_{j}v_{k}\mV_{\ell}^{\prime}\right]  \\
& =\sum_{i,k}B_{ii}B_{kk}E\left[  \mV_{i}v_{i}v_{k}\mV_{k}^{\prime}\right]
+\sum_{i,j}B_{ij}^{2}\E\left[  \mV_{i}v_{j}v_{i}\mV_{j}^{\prime}\right]
+\sum_{i,j}B_{ij}B_{ji}\E\left[  \mV_{i}v_{j}v_{j}\mV_{i}^{\prime}\right]  \\
& =\sum_{i,k}\vsigma_{vi22}\vsigma_{vk22}^{\prime}+\sum_{i}B_{ii}^{2}%
\vsigma_{vi22}\vsigma_{vj22}^{\prime}+\sum_{i}B_{ii}^{2}\sigma_{vj22}^{2}%
\mSigma_{i22}.
\end{align*}
}%
In the last equality we have used (\ref{Sigmai}) and that the off-diagonal
elements of $\mB$ are 0. Since the first term is just $\E\left[  \mV^{\prime
}\mB\vv\right]  \E\left[  \mV^{\prime}\mB\vv\right]  ^{\prime}$,
\[
\Var\left[  \mV^{\prime}\mB\vv\right]  =\sum_{i}B_{ii}^{2}\vsigma_{vi22}\vsigma
_{vj22}^{\prime}+\sum_{i}B_{ii}^{2}\sigma_{vj22}^{2}\mSigma_{i22}\leq
c_{u}n\mI_{g},
\]
where the inequality follows from Assumption \ref{ass2}. Therefore, because
$\tr\left(  \mB\right)  =n$,%
\[
\Var\left[  \dfrac{\mV^{\prime}\mB\vv}{\tr\left(  \mB\right)  }\right]  \leq\dfrac
{c_{u}}{n}\mI_{g}\rightarrow \mO.
\]
Using Assumption \ref{ass3},%
\[
\E\left[  \dfrac{\mV^{\prime}\mB\mV}{\tr\left(  \mB\right)  }\right]  \rightarrow
\mSigma_{22},
\]
so we obtain the probability limit%
\begin{equation}
\dfrac{\mV^{\prime}\mB\mV}{\tr\left(  \mB\right)  }\rightarrow_{p}\mSigma_{22}%
.\label{vbv}%
\end{equation}
Hence
\begin{equation}
\dfrac{\mX^{\prime}\mB\mX}{\tr\left(  \mB\right)  }=O_{p}\left(  1\right)  .\label{xbx}%
\end{equation}
Next we consider $\mX^{\prime}\mB\left(  \vy-\mX\vbeta_{n}\right)  /\tr\left(  \mB\right)
$. Since $\vy-\mX\vbeta_{n}=\mX\left(  \vbeta_{0}-\vbeta_{n}\right)  +\vvarepsilon$, we
obtain%
\[
\dfrac{\mX^{\prime}\mB\left(  \vy-\mX\vbeta_{n}\right)  }{\tr\left(  \mB\right)  }%
=\dfrac{\mX^{\prime}\mB\mX\left(  \vbeta_{0}-\vbeta_{n}\right)  }{\tr\left(  \mB\right)
}+\dfrac{\mX^{\prime}\mB\vvarepsilon}{\tr\left(  \mB\right)  }.
\]
Based on (\ref{xbx0}) and (\ref{xbx}) and the consistency of $\beta_{n}$, the
first term on the right hand side is $O_{p}\left(  1\right)  \cdot
o_{p}\left(  1\right)  =o_{p}\left(  1\right)  $. The second term is%
\begin{equation}
\dfrac{\mX^{\prime}\mB\vvarepsilon}{\tr\left(  \mB\right)  }=\dfrac{\left(
\mZ\mPi\right)  ^{\prime}\mB\vvarepsilon}{\tr\left(  \mB\right)  }+\dfrac{\mV^{\prime
}\mB\vvarepsilon}{\tr\left(  \mB\right)  }=o_{p}\left(  1\right)  +\vsigma
_{21},\label{xbe0}%
\end{equation}
where we treat the terms $\left(  \mZ\mPi\right)  ^{\prime}\mB\vvarepsilon/\tr\left(
\mB\right)  $ and $\mV^{\prime}\mB\vvarepsilon/\tr\left(  \mB\right)  $ respectively as
$\left(  \mZ\mPi\right)  ^{\prime}\mB\mV/\tr\left(  \mB\right)  $ and $\mV^{\prime
}\mB\mV/\tr\left(  \mB\right)  $ above, and invoke Assumption \ref{ass3}. Hence we
obtain%
\begin{equation}
\dfrac{\mX^{\prime}\mB\left(  \vy-\mX\vbeta_{n}\right)  }{\tr\left(  \mB\right)  }%
=\vsigma_{21}+o_{p}\left(  1\right)  .\label{xbe}%
\end{equation}
Finally we consider $\widehat{\sigma}^{2}\left(  \vbeta_{n}\right)
\equiv\left(  \vy-\mX\vbeta_{n}\right)  ^{\prime}\mB\left(  \vy-\mX\vbeta_{n}\right)
/\tr\left(  \mB\right)  $. Using $\vy-\mX\vbeta_{n}=\mX\left(  \vbeta_{0}-\vbeta
_{n}\right)  +\vvarepsilon$, we obtain%
\[
\dfrac{\left(  \vy-\mX\vbeta_{n}\right)  ^{\prime}\mB\left(  \vy-\mX\vbeta_{n}\right)
}{\tr\left(  \mB\right)  }=\dfrac{\left(  \vbeta_{0}-\vbeta_{n}\right)  ^{\prime
}\mX^{\prime}\mB\left(  \vy-\mX\vbeta_{n}\right)  }{\tr\left(  \mB\right)  }%
+\dfrac{\vvarepsilon^{\prime}\mB\mX\left(  \vbeta_{0}-\vbeta_{n}\right)  }{\tr\left(
\mB\right)  }+\dfrac{\vvarepsilon^{\prime}\mB\vvarepsilon}{\tr\left(  \mB\right)  }.
\]
By (\ref{xbe0}) and the consistency of $\vbeta_{n}$ the first two terms are
$o_{p}\left(  1\right)  $. By Assumption \ref{ass3} and treating
$\vvarepsilon^{\prime}\mB\vvarepsilon/\tr\left(  \mB\right)  $ as $\mV^{\prime
}\mB\mV/\tr\left(  \mB\right)  $ above we have%
\begin{equation}
\dfrac{\vvarepsilon^{\prime}\mB\vvarepsilon}{\tr\left(  \mB\right)  }=\sigma
^{2}+o_{p}\left(  1\right)  ,\label{ebe}%
\end{equation}
so the result follows.

\item Using $\vy=\mX\vbeta_{0}+\varepsilon$ we have
\begin{equation}
\widehat{\sigma}^{2}\left(  \vbeta\right) =\dfrac{\vvarepsilon^{\prime}\mB\vvarepsilon}{\tr\left(
\mB\right)  }-2\dfrac{\vvarepsilon^{\prime}\mB\mX\left(  \vbeta-\vbeta_{0}\right)
}{\tr\left(  \mB\right)  }+\dfrac{\left(  \vbeta-\vbeta_{0}\right)  ^{\prime
}\mX^{\prime}\mB\mX\left(  \vbeta-\vbeta_{0}\right)  }{\tr\left(  \mB\right)
}.\label{sg2}%
\end{equation}
In the case when $\mB$ is as in \cite{BC15}, $\mB\mZ=\mO$, so%
\[
\widehat{\sigma}^{2}\left(  \vbeta\right) =\dfrac{\vvarepsilon^{\prime}\mB\vvarepsilon}{\tr\left(
\mB\right)  }-2\dfrac{\vvarepsilon^{\prime}\mB\mV\left(  \vbeta-\vbeta_{0}\right)
}{\tr\left(  \mB\right)  }+\dfrac{\left(  \vbeta-\vbeta_{0}\right)  ^{\prime
}\mV^{\prime}\mB\mV\left(  \vbeta-\vbeta_{0}\right)  }{\tr\left(  \mB\right)  }.
\]
Using %part 1 with $\vbeta_{n}=\vbeta_{0}$
\ref{ebe}, \ref{xbe0} and \ref{vbv}, we get
\begin{equation}
\widehat{\sigma}^{2}\left(  \vbeta\right)  =\sigma^{2}-2\vsigma_{12}\left(
\vbeta-\vbeta_{0}\right)  +\left(  \vbeta-\vbeta_{0}\right)  ^{\prime}\mSigma
_{22}\left(  \vbeta-\vbeta_{0}\right)  +o_{p}\left(  1\right)  .
\label{s2b}
\end{equation}
Note that the quadratic function in $\vt$%
\[
\sigma^{2}-2\vsigma_{12}\vt+\vt^{\prime}\mSigma_{22}\vt\geq\sigma^{2}-\vsigma
_{12}\mSigma_{22}^{-1}\vsigma_{21},
\]
where the right hand side is attained for $\vt=\mSigma_{22}^{-1}\vsigma_{21}$.
Since by Assumption \ref{ass3}  $\mSigma=\left(
\begin{array}
[c]{cc}%
\sigma^{2} & \vsigma_{12}\\
\vsigma_{21} & \mSigma_{22}%
\end{array}
\right)  $ is positive definite, its determinant%
\[
\left\lvert \mSigma \right\rvert =\left\lvert \mSigma_{22} \right\rvert 
\left(\sigma^{2}-\vsigma_{12}\mSigma_{22}^{-1}\vsigma_{21}\right)  >0,
\]
so
\[
\widehat{\sigma}^{2}\left(  \vbeta\right)  \geq\sigma^{2}-\vsigma_{12}%
\mSigma_{22}^{-1}\vsigma_{21}+o_{p}\left(  1\right)  ,
\]
where $\sigma^{2}-\vsigma_{12}\mSigma_{22}^{-1}\vsigma_{21}>0$. 

In the case when $\mB=\mI_{n}$, using part 1 and $\dfrac{\mH}{n}\geq0$, we have%
\[
\dfrac{\mX^{\prime}\mB\mX}{\tr\left(  \mB\right)  }\geq\mSigma_{22}+o_{p}\left(
1\right)  .
\]
Therefore, employing the results from part 1 to equation (\ref{sg2}) we obtain%
\[
\widehat{\sigma}^{2}\geq\sigma^{2}-2\vsigma_{12}\left(  \vbeta-\vbeta_{0}\right)
+\left(  \vbeta-\vbeta_{0}\right)  ^{\prime}\mSigma_{22}\left(  \vbeta-\vbeta
_{0}\right)  +o_{p}\left(  1\right)  ,
\]
so%
\[
\widehat{\sigma}^{2}\left(  \vbeta\right)  \geq\sigma^{2}-\vsigma_{12}%
\mSigma_{22}^{-1}\vsigma_{21}+o_{p}\left(  1\right)
\]
in this case as well.
\end{enumerate}
\end{proof}

\begin{lemma}
\label{AHA}Let $\widetilde{\mH}$ be a $g\times g$ matrix with $\widetilde{\mH}%
^{-1}\mH=\mI_{g}+o_{p}\left(  1\right)  ,$ $\mA$ as in (\ref{lrestrictions}) and
suppose that $(\mA\mH^{-1}\mA^{\prime})^{-1}\mA\mH^{-1}=O\left(  1\right)  $.
Then $\mA^{\prime}\left(  \mA\widetilde{\mH}^{-1}\mA^{\prime}\right)  ^{-1}%
\mA\widetilde{\mH}^{-1}=\mA^{\prime}\left(  \mA\mH^{-1}\mA^{\prime}\right)  ^{-1}%
\mA\mH^{-1}+o_{p}\left(  1\right)  $.
\end{lemma}

\begin{proof}
First we prove that
\begin{equation}
\left(  \mA\widetilde{\mH}^{-1}\mA^{\prime}\right)  ^{-1}\mA\mH^{-1}\mA^{\prime}%
=\mI_{p}+o_{p}\left(  1\right)  .\label{AH1}%
\end{equation}
Indeed, we have%
\begin{align*}
\mA\widetilde{\mH}^{-1}\mA^{\prime}\left(  \mA\mH^{-1}\mA^{\prime}\right)  ^{-1}  
&=\mA \widetilde{\mH}^{-1}\mH\mH^{-1}\mA^{\prime}\left(  \mA \mH^{-1}\mA^{\prime}\right)  ^{-1}\\
& =\mA \left(  \mI_{g}+o_{p}\left(  1\right)  \right)  \mH^{-1}\mA^{\prime
}\left(  \mA \mH^{-1}\mA^{\prime}\right)  ^{-1}\\
& =\mI_{p}+\mA o_{p}\left(  1\right)  \left( \left(  \mA \mH^{-1}\mA^{\prime}\right)  ^{-1} \mA\mH^{-1} \right)^\prime\\
& = \mI_{p}+O\left(  1\right)o_{p}\left(  1\right)O\left(  1\right) =\mI_{p}+o_{p}\left(  1\right)  ,
\end{align*}
which implies (\ref{AH1}). Now note that%
{\footnotesize
\begin{align*}
\mA^{\prime}\left(  \mA\widetilde{\mH}^{-1}\mA^{\prime}\right)  ^{-1}\mA\widetilde{\mH}%
^{-1}  & =\mA^{\prime}\left(  \mA\widetilde{\mH}^{-1}\mA^{\prime}\right)  ^{-1}%
\mA\mH^{-1}\mA^{\prime}\left(  \mA\mH^{-1}\mA^{\prime}\right)  ^{-1}\mA\mH^{-1}\mH\widetilde{\mH}%
^{-1}\\
& =\mA^{\prime}\left(  \mI_{p}+o_{p}\left(  1\right)  \right)  \left(
\mA\mH^{-1}\mA^{\prime}\right)  ^{-1}\mA\mH^{-1}\left(  \mI_{g}+o_{p}\left(  1\right)  \right)
\\
& =\mA^{\prime}\left(  \mA\mH^{-1}\mA^{\prime}\right)  ^{-1}\mA\mH^{-1}+O(1)o_{p}\left(
1\right) +O(1)o_{p}\left( 1\right)O(1)+O(1)o_{p}\left( 1\right)O(1)o_{p}\left( 1\right)
\\
& =\mA^{\prime}\left(  \mA\mH^{-1}\mA^{\prime}\right)  ^{-1}\mA\mH^{-1}+o_{p}\left(
1\right)  .
\end{align*}
}%
\end{proof}

\clearpage

\section{Auxiliary results}\label{Auxres}

\subsection{Theorems and proofs of main results}

\begin{proposition}
\label{Pcons}
Under Assumptions \ref{ass1} to \ref{ass3} any estimator $\vbeta_n$ that satisfies $Q_{n}(\vbeta_n)\leq Q_{n}(\vbeta_{0})$ is
consistent and $(\vbeta_n-\vbeta_{0})^{\prime}\mH(\vbeta_n%
-\vbeta_{0})/\sqrt{k}=o_{p}\left(  1\right)  $.
\end{proposition}

\begin{proof}[Proof of Proposition \ref{Pcons}]

Using $\vy-\mX\vbeta_n=\vvarepsilon-\mX(\vbeta_n-\vbeta_{0})$, we
obtain
\begin{equation}
Q_{n}(\vbeta_n)-Q_{n}(\vbeta_{0})=\left(  \dfrac{1}{\widehat{\sigma}^{2}(\vbeta_n)}-\dfrac{1}%
{\widehat{\sigma}^{2}(\vbeta_{0})}\right) \vvarepsilon
^{\prime}\mC\vvarepsilon +\dfrac{(\vbeta_n-\vbeta_{0})^{\prime}\mX^{\prime
}\mC\mX(\vbeta_n-\vbeta_{0})}{\widehat{\sigma}^{2}(\vbeta_n)}%
-\dfrac{2\vvarepsilon^{\prime}\mC\mX(\vbeta_n-\vbeta_{0})}{\widehat{\sigma
}^{2}(\vbeta_n)},\label{dQn}%
\end{equation}
where we use the notation 
\begin{equation} 
\widehat{\sigma}^{2}\left(  \vbeta\right)\equiv 
\dfrac{\left(  \vy-\mX\vbeta\right)  ^{\prime}\mB\left(  \vy-\mX\vbeta\right)}{\tr\left(  \mB\right)}.\label{sh2}
\end{equation} 
Now we rescale (\ref{dQn}) by ${1}/{r_{\min}}$ and observe that by Lemma
\ref{LA4} and Assumptions \ref{ass1}, \ref{ass2}, \ref{ass5},%
\[
\dfrac{1}{r_{\min}}\vvarepsilon^{\prime}\mC\vvarepsilon=o_{p}\left(  1\right)
,~~\dfrac{1}{r_{\min}}\vvarepsilon^{\prime}\mC\mX=o_{p}\left(  1\right)  .
\]
Lemma \ref{B}.1 applied to $\vbeta_{n}=\vbeta_{0}$ implies $\widehat{\sigma}^{2}(\vbeta_{0})\rightarrow_{p}\sigma^{2}$
and by Lemma \ref{B}.2, $\widehat{\sigma}^{2}({\vbeta_n})$ is bounded
away from 0 in the limit, so for the first and the third terms we get%

\begin{align}
\left(  \dfrac{1}{\widehat{\sigma}^{2}(\vbeta_n)}-\dfrac{1}%
{\widehat{\sigma}^{2}(\vbeta_{0})}\right)  \dfrac{\vvarepsilon^{\prime
}\mC\vvarepsilon}{r_{\min}}  & =O_{p}\left(  1\right)  \cdot o_{p}\left(
1\right)  =o_{p}\left(  1\right)  , \label{ece} \\
\dfrac{2\vvarepsilon^{\prime}\mC\mX}{r_{\min}}\dfrac{({\vbeta_n}-\vbeta_{0}%
)}{\widehat{\sigma}^{2}({\vbeta_n})}  & =o_{p}\left(  1\right)  \cdot
O_{p}\left(  1\right)  =o_{p}\left(  1\right)  . \label{ecx}
\end{align}
With respect to the second term on the right hand side of (\ref{dQn}) we first note that from Lemma \ref{LA4}.1
\begin{equation}
\dfrac{\mX^{\prime}\mC\mX}{r_{\min}}=\dfrac{\mH}{r_{\min}}+o_{p}\left(  1\right).  \label{xcx}%
\end{equation}
Substituting this into the rescaled version of (\ref{dQn}) and using (\ref{ece}) and (\ref{ecx}) we obtain
\[
\dfrac{Q_{n}(\vbeta_{n})-Q_{n}(\vbeta_{0})}{r_{\min}}
%\geq\dfrac{1}{c_{u}}
=\dfrac{1}{\widehat{\sigma}^{2}(\vbeta_{n})}(\vbeta_{n}-\vbeta
_{0})^{\prime}\dfrac{\mH}{r_{\min}}(\vbeta_{n}-\vbeta_{0})+o_{p}\left(
1\right)  .
\]
Since $\mH/r_{\min}$ is positive definite in the limit having the eigenvalues $\geq1$, we
obtain from the Rayleigh quotient that%
\[
\dfrac{Q_{n}(\vbeta_{n})-Q_{n}(\vbeta_{0})}{r_{\min}}\geq\dfrac{\left\Vert \vbeta_{n}-\vbeta_{0}\right\Vert ^{2}}{\widehat{\sigma
}^{2}(\vbeta_{n})}+o_{p}\left(  1\right)  .
\]
Together with $Q_{n}(\vbeta_n)\leq Q_{n}(\vbeta_{0})$  this implies%
\[
\dfrac{\left\Vert \vbeta_n-\vbeta_{0}\right\Vert ^{2}}{\widehat{\sigma
}^{2}(\vbeta_n)}=o_{p}\left(  1\right)  .
\]
Since both the numerator and denominator are quadratic in $\left(
\vbeta_n-\vbeta_{0}\right)  $, this can only happen, if
\begin{equation}
\vbeta_n-\vbeta_{0}=o_{p}\left(  1\right)  ,\label{cons}%
\end{equation}
that is, $\vbeta_n$ is consistent.

Now rescaling (\ref{dQn}) by ${1}/{\sqrt{k}}$ and using $Q_{n}(\vbeta_n)\leq Q_{n}(\vbeta_{0})$, we
get%
\begin{equation}
\left(  \dfrac{1}{\widehat{\sigma}^{2}(\vbeta_n)}-\dfrac{1}%
{\widehat{\sigma}^{2}(\vbeta_{0})}\right)  \dfrac{\vvarepsilon^{\prime
}\mC\vvarepsilon}{\sqrt{k}}  %& \dfrac{\left(  \widehat{\sigma}^{2}(\vbeta_{0})-\widehat{\sigma}^{2}%
%(\vbeta_n)\right)  \vvarepsilon^{\prime}\mC\vvarepsilon}{\widehat{\sigma
%}^{2}(\vbeta_{0})\widehat{\sigma}^{2}(\vbeta_n)\sqrt{k}}
+\dfrac
{(\vbeta_n-\vbeta_{0})^{\prime}\mX^{\prime}\mC\mX(\vbeta_n-\vbeta_{0}%
)}{\widehat{\sigma}^{2}(\vbeta_n)\sqrt{k}}-\dfrac{2\vvarepsilon^{\prime
}\mC\mX(\vbeta_n-\vbeta_{0})}{\widehat{\sigma}^{2}(\vbeta_n)\sqrt{k}%
}\leq0. \label{l2p}
\end{equation}
Assumptions \ref{ass1}, \ref{ass2}, \ref{ass5} and Lemma \ref{LA4} imply%
\begin{equation}
\frac{1}{\sqrt{k}}\vvarepsilon^{\prime}\mC\vvarepsilon=O_{p}\left(  1\right)
,~~\frac{1}{\sqrt{k}}\vvarepsilon^{\prime}\mC\mX=O_{p}\left(  1\right)
,\label{1/rk}%
\end{equation}
and in addition, using Lemma \ref{LA4}, Lemma \ref{B}.1 and (\ref{cons}), we obtain that the first and the third terms on the left hand side of (\ref{l2p}) are $o_{p}(1)$. Further, by Lemma \ref{LA4}.1 (\ref{l2p}) becomes
\[
\dfrac{(\vbeta_n-\vbeta_{0})^{\prime}\mH(\vbeta_n-\vbeta_{0})}%
{\sqrt{k}}+o_{p}\left(  1\right)  \leq0.
\]
Since $(\vbeta_n-\vbeta_{0})^{\prime}\mH(\vbeta_n-\vbeta_{0}%
)/\sqrt{k}\geq0$, we find that 
%\textcolor{red}{\bf the result below is part of the proposition's thesis but it seems to be part of the old proof}
\begin{equation}
\dfrac{(\vbeta_n-\vbeta_{0})^{\prime}\mH(\vbeta_n-\vbeta_{0})}%
{\sqrt{k}}=o_{p}\left(  1\right)  .\label{rate}%
\end{equation}
This completes the proof.
\end{proof}

\begin{proof}
[Proof of Theorem \ref{basicresult}]%
%\textcolor{blue}{For JIVE $\vzeros=\nabla Q_{n}(\widehat{\vbeta})=-2\mX'\mC(\vy-\mX\widehat{\vbeta})$, so the result should be the same for $\widehat{\sigma}^{2}(\widehat{\vbeta})=1$ and $\widehat{\mC}(\widehat{\vbeta})=\mC$.} 
From the first order condition of minimising $Q_{n}$ we have $\nabla Q_{n}(\widehat{\vbeta})=\vzeros$. Since 
\begin{equation*}
\nabla Q_{n}(\vbeta) =\dfrac{-2}{\widehat{\sigma}^{2}\left(  \vbeta\right)  }\mX^{\prime
}\widehat{\mC}\left(  \vbeta\right)  \left(  \vy-\mX\vbeta\right) 
\end{equation*}

we obtain\footnote{ The gradient of $Q_{n}(\vbeta)$ is
\begin{align}
\nabla Q_{n}(\vbeta) &  =\frac{-2\mX^{\prime}\mC(\vy-\mX\vbeta)}{(\vy-\mX\vbeta)^{\prime
}\mB(\vy-\mX\vbeta)/\tr(\mB)}-Q_{n}(\vbeta)\frac{-\frac{2}{\tr(\mB)}X^{\prime}\mB(\vy-\mX\vbeta
)}{(\vy-\mX\vbeta)^{\prime}\mB(\vy-\mX\vbeta)/\tr(\mB)}\nonumber\\
&  =\frac{-2\tr(\mB)}{(\vy-\mX\vbeta)^{\prime}\mB(\vy-\mX\vbeta)}\left(  \mX^{\prime}
\mC(\vy-\mX\vbeta)-\frac{Q_{n}(\vbeta)}{\tr(\mB)}\mX^{\prime}\mB(\vy-\mX\vbeta)\right)
\nonumber\\
&  =\dfrac{-2}{\widehat{\sigma}^{2}\left(  \vbeta\right)  }\mX^{\prime
}\widehat{\mC}\left(  \vbeta\right)  \left(  \vy-\mX\vbeta\right). \nonumber
\end{align}}\label{ftn}
\[
\dfrac{-2}{\widehat{\sigma}^{2}(\widehat{\vbeta})}\mX^{\prime}\widehat{\mC}%
(\widehat{\vbeta})(\vy-\mX\widehat{\vbeta})=\vzeros.
\]
Using $\vy-\mX\widehat{\vbeta}=\vvarepsilon-\mX(\widehat{\vbeta}-\vbeta_{0})$, (\ref{C-hat}) and rescaling by $\frac{1}{\sqrt{k}}$, we obtain%
\begin{equation}
\frac{1}{\sqrt{k}}\mX^{\prime}\left(  \mC-\dfrac{Q_{n}(\widehat{\vbeta})}{\tr\left(
\mB\right)  }\mB\right)  \vvarepsilon=\frac{1}{\sqrt{k}}\mX^{\prime}\widehat{\mC}%
(\widehat{\vbeta})\mX(\widehat{\vbeta}-\vbeta_{0}).\label{foc2}%
\end{equation}
From (\ref{dQn}) rescaled by $\frac{1}{\sqrt{k}}$, by (\ref{1/rk}), Lemma \ref{LA4}.1, 
(\ref{rate}), Lemma \ref{B}.1 and the consistency of $\widehat{\vbeta}$, we obtain that
\begin{align}
\frac{1}{\sqrt{k}}\left(  Q_{n}(\widehat{\vbeta})-Q_{n}(\vbeta_{0})\right)   &
=o_{p}\left(  1\right)  \text{ \ and}\label{cdQ}\\
\frac{1}{\sqrt{k}}Q_{n}(\widehat{\vbeta}) &  =O_{p}\left(  1\right)
.\label{cQn}%
\end{align}
Together with Lemma \ref{B}.1, \ref{cQn} implies that the expressions on the two sides of (\ref{foc2}) are
$O_{p}\left(  1\right)  $. Substituting $Q_{n}(\widehat{\vbeta})$ from (\ref{cdQ}) into the left hand side
%and (\ref{cQn}) in the right hand side 
of (\ref{foc2}) and using (\ref{tau}), (\ref{C-hat}) and Lemma \ref{B}.1, we obtain that the left hand side of (\ref{foc2}) is equal to
\[
\frac{1}{\sqrt{k}}\mX^{\prime}\left(  \mC-\dfrac{Q_{n}(\widehat{\vbeta})}{\tr\left(
\mB\right)  }\mB\right)  \vvarepsilon  =\vtau_{n}+o_{p}\left(  1\right)  .
\]
With respect to the right hand side of (\ref{foc2}) we note that from 
(\ref{C-hat}) we have%
\[
\mX^{\prime}\widehat{\mC}%
(\widehat{\vbeta})\mX\mH^{-1}=\mX^{\prime}\left(  \mC-\dfrac{Q_{n}(\widehat{\vbeta})}{\tr\left(
\mB\right)  }\mB\right)  \mX\mH^{-1}=\mX^{\prime}\mC\mX\mH^{-1}-\mX^{\prime}\dfrac
{Q_{n}(\widehat{\vbeta})}{tr\left(  \mB\right)  }\mB\mX\mH^{-1},
\]
where by Lemma \ref{LA4}.1 $\mX^{\prime}\mC\mX\mH^{-1}=\mI_{g}+o_{p}\left(  1\right)  $ and by
Assumption \ref{ass5}, (\ref{cQn}) and Lemma \ref{B}.1 % Assumption \ref{ass6}
\[
\mX^{\prime}\dfrac{Q_{n}(\widehat{\vbeta})}{\tr\left(  \mB\right)  }\mB\mX\mH^{-1}%
=\dfrac{\sqrt{k}}{r_{\min}}\dfrac{Q_{n}(\widehat{\vbeta})}{\sqrt{k}}\dfrac{\mX^{\prime}\mB\mX}{\tr\left(
\mB\right)  }r_{\min}\mH^{-1}=o\left(  1\right)  \cdot O_{p}\left(  1\right)  \cdot O_{p}\left(  1\right)
\cdot O\left(  1\right)  =o_{p}\left(  1\right),
\]
which implies that
\begin{equation}
\mX^{\prime}\widehat{\mC}%
(\widehat{\vbeta})\mX\mH^{-1}=\mI_{g}+o_{p}\left(  1\right)  .\label{HhHi}%
\end{equation}
Therefore, the right hand side of (\ref{foc2}) is 
\begin{equation*}
\frac{1}{\sqrt{k}}\mX^{\prime}\widehat{\mC}%
(\widehat{\vbeta})\mX(\widehat{\vbeta}-\vbeta_{0})
=\frac{1}{\sqrt{k}}\mH(\widehat{\vbeta}-\vbeta_{0})+o_{p}\left(  1\right).
\end{equation*}
Consequently, by using (\ref{psi}), (\ref{foc2}) becomes
\begin{equation}
\vtau_{n}=\vpsi_{n}+o_{p}\left(  1\right)  =O_{p}\left(  1\right)  .\label{foc1}%
\end{equation}
Next we show that $\vtau_{n}$ is asymptotically normally distributed. Notice
that%
\[
\vtau_{n}=\frac{1}{\sqrt{k}}\mPi^{\prime}\mZ^{\prime}\mC\vvarepsilon+\frac{1}%
{\sqrt{k}}\left(  \mV^{\prime}-\dfrac{\mV^{\prime}\mB\vvarepsilon}{\tr\left(
\mB\right)  }\dfrac{\tr\left(  \mB\right)  }{\vvarepsilon^{\prime}\mB\vvarepsilon
}\vvarepsilon^{\prime}\right)  \mC\vvarepsilon.
\]
By the continuous mapping theorem and Lemma \ref{B}.1 
%under Assumptions \ref{ass1} - \ref{ass5} 
we obtain
\[
\dfrac{\mV^{\prime}\mB\vvarepsilon}{\tr\left(  \mB\right)  }\dfrac{\tr\left(  \mB\right)
}{\vvarepsilon^{\prime}\mB\vvarepsilon}=\dfrac{\vsigma_{21}}{\sigma^{2}}%
+o_{p}\left(  1\right)  .
\]
Based on $\frac{1}{\sqrt{k}}\vvarepsilon^{\prime}\mC\vvarepsilon=O_{p}\left(
1\right)  $, which follows from Lemma \ref{LA4}, we obtain%
\[
\vtau_{n}=\frac{1}{\sqrt{k}}\mPi^{\prime}\mZ^{\prime}\mC\vvarepsilon+\frac{1}%
{\sqrt{k}}\left(  \mV^{\prime}-\dfrac{\vsigma_{21}}{\sigma^{2}}\vvarepsilon
^{\prime}\right)  \mC\vvarepsilon+o_{p}\left(  1\right)  .
\]
This converges in distribution by the central limit theorem from \citet[][Lemma A2]{CSHNW12} in the version stated in Lemma \ref{CLT}. In order to see
this, let $\vw_{in}$ be such that
\[
\frac{1}{\sqrt{k}}\mPi^{\prime}\mZ^{\prime}\mC\vvarepsilon=\sum_{i=1}^{n}\vw_{in},
\]
that is, 
\[
\vw_{in}=\frac{1}{\sqrt{k}}\mPi^{\prime}\mZ^{\prime}\mC_{i}\vvarepsilon_{i}.
\]
Next we verify that $w_{in}$ satisfies the conditions of Lemma \ref{CLT}.
Condition (1) is satisfied based on equation (A.1) from the Supplement to
\cite{CMS21}. Condition (2) is satisfied because%
\[
\sum_{i=1}^{n}\E\left[  \vw_{in}\vw_{in}^{\prime}\right]  =\dfrac{1}{k}\mPi^{\prime
}\mZ^{\prime}\mC\mD_{\sigma^{2}}\mC\mZ\mPi\equiv\dfrac{1}{k}\mF
\]
and this converges by Assumption \ref{ass6} to a finite matrix. By Assumption
\ref{ass2} and the fact that $\E\left[  \vw_{in}\right]  =\vzeros$, Condition (3) is
also satisfied. Condition (4) follows by Assumption \ref{ass2} from observing
that%
\[
\sum_{i=1}^{n}\E\left[  \left\Vert \vw_{in}\right\Vert ^{4}\right]  \leq
\dfrac{c_{u}}{k^{2}}\sum_{i=1}^{n}\left\Vert \mPi^{\prime}\mZ^{\prime}%
\mC_{i}\right\Vert ^{4}%
\]
where the right hand side goes to $0$ by Assumption \ref{ass5n}. Now applying
Lemma \ref{CLT} we obtain that%
\begin{equation}
\vtau_{n}\rightarrow_{d}\mathcal{N}(\vzeros,\mPhi).\label{cd_tau}%
\end{equation}
Since from (\ref{foc1})%
\[
\vpsi_{n}=\vtau_{n}+o_{p}(1)\rightarrow_{d}\mathcal{N}(0,\mPhi),
\]
we obtain that%
\begin{equation}
\vpsi_{n}\rightarrow_{d}\mathcal{N}(\vzeros,\mPhi).\label{cd_psi}%
\end{equation}
Consequently, for $T\in\{LM,W\}$ the result follows from Lemma \ref{chimix}
and the continuous mapping theorem. 
%\textcolor{red}{!! In Lemma \ref{chimix} the matrix $\mXi$ appears in the inverse form, which suggests that it should be invertible, while in our case $\mXi$ may not be invertible when the rates of $r_{\min}$ and $r_{\max}$ are different !!}

Now we turn to studying the statistic $D$. The first term on the right hand
side of (\ref{dQn}) rescaled by ${r_{\min}}/{k}$ is%
\begin{align*} \label{rm/k}
\dfrac{r_{\min}}{k}\dfrac{\left(  \widehat{\sigma}^{2}(\vbeta_{0}%
)-\widehat{\sigma}^{2}(\widehat{\vbeta})\right)  \vvarepsilon^{\prime
}\mC\vvarepsilon}{\widehat{\sigma}^{2}(\vbeta_{0})\widehat{\sigma}^{2}%
(\widehat{\vbeta})} &  =\dfrac{r_{\min}}{k}\dfrac{\vvarepsilon^{\prime
}\mC\vvarepsilon}{\widehat{\sigma}^{2}(\vbeta_{0})\widehat{\sigma}^{2}%
(\widehat{\vbeta})\tr\left(  \mB\right)  }2\vvarepsilon^{\prime}\mB\mX(\widehat{\vbeta
}-\vbeta_{0})\\
&  -\dfrac{r_{\min}}{k}\dfrac{\vvarepsilon^{\prime}\mC\vvarepsilon}%
{\widehat{\sigma}^{2}(\vbeta_{0})\widehat{\sigma}^{2}(\widehat{\vbeta})\tr\left(
\mB\right)  }(\widehat{\vbeta}-\vbeta_{0})^{\prime}\mX^{\prime}\mB\mX(\widehat{\vbeta
}-\vbeta_{0})\\
&  =\dfrac{r_{\min}}{k}\dfrac{1}{\widehat{\sigma}%
^{2}(\widehat{\vbeta})}\dfrac{Q_{n}\left(  \vbeta_{0}\right)  }{\tr\left(
\mB\right)  }2\vvarepsilon^{\prime}\mB\mX(\widehat{\vbeta}-\vbeta_{0})\\
&  -\dfrac{\vvarepsilon^{\prime}\mC\vvarepsilon}{\sqrt{k}}\dfrac{1}%
{\widehat{\sigma}^{2}(\vbeta_{0})\widehat{\sigma}^{2}(\widehat{\vbeta}%
)}(\widehat{\vbeta}-\vbeta_{0})^{\prime}\dfrac{\mX^{\prime}\mB\mX}{\tr\left(  \mB\right)
}\dfrac{r_{\min}}{\sqrt{k}}(\widehat{\vbeta}-\vbeta_{0}).
\end{align*}
By Lemma \ref{LA4}, Lemma \ref{B}, (\ref{cd_psi}) and the consistency of $\widehat{\vbeta}$ it follows that the second
term is%
\[
\dfrac{\vvarepsilon^{\prime}\mC\vvarepsilon}{\sqrt{k}}\dfrac{1}{\widehat{\sigma
}^{2}(\vbeta_{0})\widehat{\sigma}^{2}(\widehat{\vbeta})}(\widehat{\vbeta}%
-\vbeta_{0})^{\prime}\dfrac{\mX^{\prime}\mB\mX}{\tr\left(  \mB\right)  } \left( \dfrac{r_{\min}%
}{\sqrt{k}}(\widehat{\vbeta}-\vbeta_{0}) \right)=O_{p}\left(  1\right)  \cdot
O_{p}\left(  1\right)  \cdot o_{p}\left(  1\right)  \cdot O_{p}\left(
1\right)  \cdot O_{p}\left(  1\right)  =o_{p}\left(  1\right)  .
\]
Using in addition (\ref{cdQ}) and (\ref{cQn}) we can see that the first term
is
\begin{align*}
\dfrac{r_{\min}}{k}\dfrac{1}{\widehat{\sigma}%
^{2}(\widehat{\vbeta})}\dfrac{Q_{n}\left(  \vbeta_{0}\right)  }{\tr\left(
\mB\right)  }2\vvarepsilon^{\prime}\mB\mX(\widehat{\vbeta}-\vbeta_{0}) &  =\dfrac
{Q_{n}\left(  \vbeta_{0}\right)  }{\sqrt{k}}\dfrac{2%
}{\widehat{\sigma}^{2}(\widehat{\vbeta})}\dfrac{\vvarepsilon^{\prime}%
\mB\mX}{\tr\left(  \mB\right)  }\dfrac{r_{\min}}{\sqrt{k}}(\widehat{\vbeta}-\vbeta
_{0})\\
&  =O_{p}\left(  1\right)  \cdot O_{p}\left(  1\right)  \cdot O_{p}\left(
1\right)  \cdot O_{p}\left(  1\right)  =O_{p}\left(  1\right)  .
\end{align*}
Therefore, using (\ref{psi}) and Lemma \ref{B}.1, we find that the first term on the right hand
side of (\ref{dQn}) rescaled by ${r_{\min}}/{k}$ is
\begin{equation}
\dfrac{r_{\min}}{k}\dfrac{\left(  \widehat{\sigma}^{2}(\vbeta_{0}%
)-\widehat{\sigma}^{2}(\widehat{\vbeta})\right)  \vvarepsilon^{\prime
}\mC\vvarepsilon}{\widehat{\sigma}^{2}(\vbeta_{0})\widehat{\sigma}^{2}%
(\widehat{\vbeta})}=\dfrac{2r_{\min}}{\sigma^{2}\sqrt{k}}\dfrac{Q_{n}\left(
\vbeta_{0}\right)  }{\tr\left(  \mB\right)  }\vvarepsilon^{\prime}\mB\mX\mH^{-1}\vpsi
_{n}+o_{p}\left(  1\right)  .\label{dQ1}%
\end{equation}
The second term on the right hand side of (\ref{dQn}) rescaled by
${r_{\min}}/{k}$ is%
\[
\dfrac{r_{\min}}{k}\dfrac{(\widehat{\vbeta}-\vbeta_{0})^{\prime}\mX^{\prime
}\mC\mX(\widehat{\vbeta}-\vbeta_{0})}{\widehat{\sigma}^{2}(\widehat{\vbeta})}%
=\dfrac{r_{\min}}{k}\dfrac{(\widehat{\vbeta}-\vbeta_{0})^{\prime}\mH\mH^{-1}%
\mX^{\prime}\mC\mX\mH^{-1}\mH(\widehat{\vbeta}-\vbeta_{0})}{\widehat{\sigma}%
^{2}(\widehat{\vbeta})}.
\]
Using $\mX^{\prime}\mC\mX\mH^{-1}=\mI_{g}+o_{p}\left(  1\right)  $ from Lemma \ref{LA4}.1,
(\ref{psi}) and Lemma \ref{B}.1, we get%
\begin{equation}
\dfrac{r_{\min}}{k}\dfrac{(\widehat{\vbeta}-\vbeta_{0})^{\prime}\mX^{\prime
}\mC\mX(\widehat{\vbeta}-\vbeta_{0})}{\widehat{\sigma}^{2}(\widehat{\vbeta})}%
=\dfrac{r_{\min}}{\sigma^{2}}\vpsi_{n}^{\prime}\mH^{-1}\vpsi_{n}+o_{p}\left(
1\right)  .\label{dQ2}%
\end{equation}
The third term on the right hand side of (\ref{dQn}) rescaled by
${r_{\min}}/{k}$ is%
\[
-\dfrac{2r_{\min}}{k}\dfrac{\vvarepsilon^{\prime}\mC\mX(\widehat{\vbeta}-\vbeta_{0}%
)}{\widehat{\sigma}^{2}(\widehat{\vbeta})}=-\dfrac{2r_{\min}}{\sqrt{k}}%
\dfrac{\vvarepsilon^{\prime}\mC\mX}{\widehat{\sigma}^{2}(\widehat{\vbeta})}%
\mH^{-1}\frac{1}{\sqrt{k}}\mH(\widehat{\vbeta}-\vbeta_{0}).
\]
Using (\ref{psi}) and Lemma \ref{B}.1, we obtain%
\begin{equation}
-\dfrac{2r_{\min}}{k}\dfrac{\vvarepsilon^{\prime}\mC\mX(\widehat{\vbeta}-\vbeta_{0}%
)}{\widehat{\sigma}^{2}(\widehat{\vbeta})}=-\dfrac{2r_{\min}}{\sigma^{2}%
\sqrt{k}}\vvarepsilon^{\prime}\mC\mX\mH^{-1}\vpsi_{n}+o_{p}\left(  1\right)
.\label{dQ3}%
\end{equation}
Collecting the results in (\ref{dQ1}), (\ref{dQ2}), (\ref{dQ3}) and using (\ref{C-hat}), we find that
(\ref{dQn}) rescaled by ${r_{\min}}/{k}$ is%
\begin{align*}
\dfrac{r_{\min}}{k}\left(  Q_{n}(\widehat{\vbeta})-Q_{n}(\vbeta_{0})\right)   &
=\dfrac{2r_{\min}}{\sigma^{2}\sqrt{k}}\dfrac{Q_{n}\left(  \vbeta_{0}\right)
}{\tr\left(\mB\right)  }\vvarepsilon^{\prime}\mB\mX\mH^{-1}\vpsi_{n}+\dfrac{r_{\min}%
}{\sigma^{2}}\vpsi_{n}^{\prime}\mH^{-1}\vpsi_{n}\\
&  -\dfrac{2r_{\min}}{\sigma^{2}\sqrt{k}}\vvarepsilon^{\prime}\mC\mX\mH^{-1}\vpsi
_{n}+o_{p}\left(  1\right)  \\
&  =-\dfrac{2r_{\min}}{\sigma^{2}\sqrt{k}}\vvarepsilon^{\prime}\widehat{\mC}%
\left(  \vbeta_{0}\right)  \mX\mH^{-1}\vpsi_{n}+\dfrac{r_{\min}}{\sigma^{2}}\vpsi
_{n}^{\prime}\mH^{-1}\vpsi_{n}+o_{p}\left(  1\right)  .
\end{align*}
From (\ref{tau}), (\ref{foc1}) and (\ref{W}) this becomes%
\[
\dfrac{r_{\min}}{k}\left(  Q_{n}(\widehat{\vbeta})-Q_{n}(\vbeta_{0})\right)
=\dfrac{-2}{\sigma^{2}}W+\dfrac{1}{\sigma^{2}}W+o_{p}\left(  1\right)
=-\dfrac{1}{\sigma^{2}}W+o_{p}\left(  1\right)  ,
\]
so from (\ref{D}) we have%
\[
D=W+o_{p}\left(  1\right)  .
\]
Consequently, the result also follows for $D$.
\end{proof}

\begin{proof}[Proof of Corollary \ref{feasibleresult}]
First notice that%
\[
\widehat{D}=\dfrac{\widehat{r}_{\min}}{r_{\min}}\dfrac{\widehat{\sigma}^{2}%
}{\sigma^{2}}D,
\]
where $\widehat{r}_{\min}=\lambda_{\min}(\widehat{\mH})$. From (\ref{HhHi}) and (\ref{Hh_s2}) we have%
%\[
%\widehat{\mH}\mH^{-1}=\mX^{\prime}\left(  \mC-\dfrac{Q_{n}(\widehat{\vbeta})}{\tr\left(
%\mB\right)  }\mB\right)  \mX\mH^{-1}=\mX^{\prime}\mC\mX\mH^{-1}-\mX^{\prime}\dfrac
%{Q_{n}(\widehat{\vbeta})}{tr\left(  \mB\right)  }\mB\mX\mH^{-1}.
%\]
%By Lemma \ref{LA4} $\mX^{\prime}\mC\mX\mH^{-1}=\mI_{g}+o_{p}\left(  1\right)  $ and by
%Assumption \ref{ass5}, (\ref{cQn}), Lemma \ref{B}.1 and Assumption \ref{ass6}
%\[
%\mX^{\prime}\dfrac{Q_{n}(\widehat{\vbeta})}{\tr\left(  \mB\right)  }\mB\mX\mH^{-1}%
%=\dfrac{\sqrt{k}}{r_{\min}}\dfrac{Q_{n}(\widehat{\vbeta})}{\sqrt{k}}\dfrac{\mX^{\prime}\mB\mX}{\tr\left(
%\mB\right)  }r_{\min}\mH^{-1}=o\left(  1\right)  \cdot O_{p}\left(  1\right)  \cdot O_{p}\left(  1\right)
%\cdot O\left(  1\right)  =o_{p}\left(  1\right)  .
%\]
%Hence%
\begin{equation}
\widehat{\mH}\mH^{-1}=\mI_{g}+o_{p}\left(  1\right)  .\label{HhHi1}%
\end{equation}
By the continuity of the minimal eigenvalue as a function of the matrix \citep[e.g.,][Appendix D]{HJ12},
\begin{equation}
\dfrac{\widehat{r}_{\min}}{r_{\min}}\rightarrow_{p}1\label{rhr}%
\end{equation}
and further, $\dfrac{\widehat{\sigma}^{2}}{\sigma^{2}}\rightarrow_{p}1$ by
Lemma \ref{B}.1 applied to $\vbeta_{n}=\widehat{\vbeta}$. Therefore,
\begin{equation}
\widehat{D}=D+o_{p}\left(  1\right)  .\label{Dh}%
\end{equation}

Regarding the statistic $\widehat{LM}$ we observe that
\[
\widehat{r}_{\min}\widehat{\mH}^{-1}=\dfrac{\widehat{r}_{\min}}{r_{\min}}\left(
r_{\min}\mH^{-1}\right)  \mH\widehat{\mH}^{-1}.
\]
%We notice that (\ref{HhHi}) implies%
%\[
%H\widehat{H}^{-1}=I_{g}+o_{p}\left(  1\right)  .
%\]
Therefore, by (\ref{HhHi1}) and (\ref{rhr})
\begin{equation}
\widehat{r}_{\min}\widehat{\mH}^{-1}=r_{\min}\mH^{-1}+o_{p}\left(  1\right)
\label{rhHhi}%
\end{equation}
and consequently%
\begin{equation}
\widehat{LM}=\vtau_{n}^{\prime}r_{\min}\mH^{-1}\vtau_{n}+o_{p}\left(  1\right)
=LM+o_{p}\left(  1\right)  .\label{LMh}%
\end{equation}
Finally, with respect to the statistic $\widehat{W}=\widehat{r}_{\min
}\widehat{\vpsi}_{n}^{\prime}\widehat{\mH}^{-1}\widehat{\vpsi}_{n}$ first we note
that%
\[
\widehat{\vpsi}_{n}=\dfrac{1}{\sqrt{k}}\widehat{\mH}(\widehat{\vbeta}-\vbeta
_{0})=\dfrac{1}{\sqrt{k}}\widehat{\mH}\mH^{-1}\mH(\widehat{\vbeta}-\vbeta_{0}%
)=\vpsi_{n}+o_{p}\left(  1\right)  ,
\]
where the latter equality follows from (\ref{psi}) and (\ref{HhHi1}). Now
applying (\ref{rhHhi}) we have%
\begin{equation}
\widehat{W}=\left(  \vpsi_{n}^{\prime}+o_{p}\left(  1\right)  \right)  \left(
r_{\min}\mH^{-1}+o_{p}\left(  1\right)  \right)  \left(  \vpsi_{n}+o_{p}\left(
1\right)  \right)  =W+o_{p}\left(  1\right)  \label{Wh}%
\end{equation}
because $\vpsi_{n}=O_{p}\left(  1\right)  $ from (\ref{cd_psi}) and $r_{\min
}\mH^{-1}=O\left(  1\right)  $. Based on (\ref{Dh}), (\ref{LMh}), (\ref{Wh}) and
Theorem \ref{basicresult} the result follows.
\end{proof}

\begin{proof}[Proof of Corollary \ref{basicresult3}]
%From the first order condition of $Q_n(\vbeta)$, a mean value expansion argument and the central limit theorem we find
%\begin{align}\label{exp}
%\frac{1}{\sqrt k} \nabla^2Q_n( \beta)  \left(\widehat\beta-\beta \right)=-\frac{1}{\sqrt k} \nabla Q_n(\beta)+o_p(1)
%\to_d\mathcal N(0,\varphi^2).
%\end{align}
%Through a second order mean value expansion, rescaling and the central limit theorem we find
The proof is similar to that of Theorem \ref{basicresult}. By the same arguments we find
\begin{align*}
-r\frac{\sigma^2}{k}\left(Q_n(\widehat\beta)-Q_n(\beta)\right) =  \frac{r^2}{k}(\widehat\beta-\beta)^2+o_p(1) \to_d\zeta^2.
\end{align*}
The result for the LM test follows from similar arguments. The chi square limit follows from a well known relationship between the normal and chi square distributions.
\end{proof}

\begin{proof}
[Proof of Theorem \ref{basicresulta}]First we observe that $\widetilde{\vbeta}$
is consistent. This follows from Proposition \ref{Pcons} because
$Q_{n}(\widetilde{\vbeta})\leq Q_{n}(\vbeta_{0})$, which is implied by
$\mathcal{L}(\widetilde{\vbeta},\widetilde{\vgamma})\leq\mathcal{L}(\vbeta
_{0},\vgamma_{0})$ and the fact that the true $\vbeta_{0}$ satisfies $H_{0}$.
Next we concentrate $\widetilde{\vgamma}$ out from the FOCs. The FOC (\ref{dLb}) with
respect to $\vbeta$ implies%
\begin{equation}
\mX^{\prime}\widehat{\mC}(\widetilde{\vbeta})y-\widetilde{\mH}\widetilde{\vbeta
}=\widehat{\sigma}^{2}(\widetilde{\vbeta})\mA^{\prime}\widetilde{\vgamma
},\label{dLb1}%
\end{equation}
where we use the notation $\widetilde{\mH}=\mX^{\prime}\widehat{\mC}%
(\widetilde{\vbeta})\mX$. Similar to (\ref{HhHi1}) we have
\begin{equation}
\widetilde{\mH}\mH^{-1}=\mI_{g}+o_{p}\left(  1\right), \label{HhHi2}%
\end{equation}
which implies that $\widetilde{\mH}$ is invertible in the limit. Multiplied by $\mA\widetilde{\mH}^{-1}$ from the left,
(\ref{dLb1}) becomes%
\[
\mA\widetilde{\mH}^{-1}\mX^{\prime}\widehat{\mC}(\widetilde{\vbeta})\vvarepsilon
=\widehat{\sigma}^{2}(\widetilde{\vbeta})\mA\widetilde{\mH}^{-1}\mA^{\prime
}\widetilde{\vgamma},
\]
where we have also used $\vy=\mX\vbeta_{0}+\vvarepsilon$ and the linear restriction.
From this we express
\[
\widehat{\sigma}^{2}(\widetilde{\vbeta})\widetilde{\vgamma}=\left(
\mA\widetilde{\mH}^{-1}\mA^{\prime}\right)^{-1}\mA\widetilde{\mH}^{-1}\mX^{\prime
}\widehat{\mC}(\widetilde{\vbeta})\vvarepsilon
\]
and substitute it into (\ref{dLb1}); we obtain%
\[
\mX^{\prime}\widehat{\mC}(\widetilde{\vbeta})\vy-\widetilde{\mH}\widetilde{\vbeta
}=\mA^{\prime}\left(  \mA\widetilde{\mH}^{-1}\mA^{\prime}\right)  ^{-1}\mA\widetilde{\mH}%
^{-1}\mX^{\prime}\widehat{\mC}(\widetilde{\vbeta})\vvarepsilon.
\]
Using $\vy=\mX\vbeta_{0}+\vvarepsilon$ and rescaling by $\frac{1}{\sqrt{k}}$ we
obtain%
\begin{equation}
\frac{1}{\sqrt{k}}\mX^{\prime}\widehat{\mC}(\widetilde{\vbeta})\vvarepsilon=\frac
{1}{\sqrt{k}}\widetilde{\mH}(\widetilde{\vbeta}-\vbeta_{0})+\frac{1}{\sqrt{k}%
}\mA^{\prime}\left(  \mA\widetilde{\mH}^{-1}\mA^{\prime}\right)^{-1}\mA\widetilde{\mH}%
^{-1}\mX^{\prime}\widehat{\mC}(\widetilde{\vbeta})\vvarepsilon.\label{dLb2}%
\end{equation}
Based on (\ref{C-hat}) and the fact that $\frac{1}{\sqrt{k}}\left(
Q_{n}(\widetilde{\vbeta})-Q_{n}(\vbeta_{0})\right)  =o_{p}\left(  1\right)  $
(which follows in the same way as (\ref{cdQ}) in the proof of Theorem
\ref{basicresult}), the left hand side of (\ref{dLb2}) is
\begin{equation}
\frac{1}{\sqrt{k}}\mX^{\prime}\widehat{\mC}(\widetilde{\vbeta})\vvarepsilon=\frac
{1}{\sqrt{k}}\mX^{\prime}\widehat{\mC}(\vbeta_{0})\vvarepsilon+o_{p}\left(
1\right)  \equiv\vtau_{n}+o_{p}\left(  1\right).  \label{taua}%
\end{equation}
Based on the fact that the limit of $%\mA^{\prime}
(\mA\mH^{-1}\mA^{\prime})^{-1}\mA\mH^{-1}$ is finite, by Lemma \ref{AHA} the second term on the right hand side of (\ref{dLb2}) is
$\mA^{\prime}\left(  \mA\mH^{-1}\mA^{\prime}\right)  ^{-1}\mA\mH^{-1}\vtau_{n}+o_{p}\left(
1\right)  $. These facts imply that the first term on the right hand side of
(\ref{dLb2}) is $O_{p}\left(  1\right)  $. At the same time, from (\ref{ksi})
and (\ref{taua}),%
\begin{equation}
\frac{1}{\sqrt{k}}\widetilde{\mH}(\widetilde{\vbeta}-\vbeta_{0})=\vtau_{n}-\vxi
_{n}+o_{p}\left(  1\right)  =O_{p}\left(  1\right)  .\label{t-ks}%
\end{equation}
Therefore, (\ref{dLb2}) becomes%
\[
\vtau_{n}=\vtau_{n}-\vxi_{n}+\mA^{\prime}\left(  \mA\mH^{-1}\mA^{\prime}\right)
^{-1}\mA\mH^{-1}\vtau_{n}+o_{p}\left(  1\right)  ,
\]
that is,
\begin{equation}
\vxi_{n}=\mA^{\prime}\left(  \mA\mH^{-1}\mA^{\prime}\right)  ^{-1}\mA\mH^{-1}\vtau_{n}%
+o_{p}\left(  1\right)  .\label{ksi1}%
\end{equation}
This together with (\ref{LMa}) and (\ref{lim}) implies that
\[
LM_{a}=\vtau_{n}^{\prime}\mXi_{a}\vtau_{n}+o_{p}\left(  1\right)  .
\]
Since from (\ref{cd_tau}) $\vtau_{n}\rightarrow_{d}\mathcal{N}(\vzeros,\mPhi)$, Lemma
\ref{chimix} implies that $LM_{a}$ converges in distribution to $\bar{\chi
}^{2}(\boldsymbol{\varphi})$. Regarding the statistic $W_{1a}$ from \eqref{W1a}we note that
(\ref{psi}), (\ref{lim}) and $H_0$ imply that%
\[
W_{1a}=\vpsi_{n}^{\prime}\mXi_{a}\vpsi_{n}+o_{p}\left(  1\right)  .
\]
Therefore, similar to $LM_{a}$, (\ref{cd_psi}) and Lemma \ref{chimix} imply
that $W_{1a}$ converges in distribution to $\bar{\chi}^{2}(\boldsymbol{\varphi
})$.

Now we turn to the statistic $D_{a}$, which involves%
\begin{align}
\frac{r_{\min}}{k}\left(  Q_{n}(\widehat{\vbeta})-Q_{n}(\widetilde{\vbeta
})\right)   &  =\dfrac{r_{\min}}{k}\dfrac{\left(  \widehat{\sigma}%
^{2}(\widetilde{\vbeta})-\widehat{\sigma}^{2}(\widehat{\vbeta})\right)
\vvarepsilon^{\prime}\mC\vvarepsilon}{\widehat{\sigma}^{2}(\widetilde{\vbeta
})\widehat{\sigma}^{2}(\widehat{\vbeta})}+\dfrac{r_{\min}}{k}\dfrac
{(\widehat{\vbeta}-\vbeta_{0})^{\prime}\mX^{\prime}\mC\mX(\widehat{\vbeta}-\vbeta_{0}%
)}{\widehat{\sigma}^{2}(\widehat{\vbeta})}\label{D1,2}\\
&  -\dfrac{r_{\min}}{k}\dfrac{(\widetilde{\vbeta}-\vbeta_{0})^{\prime}\mX^{\prime
}\mC\mX(\widetilde{\vbeta}-\vbeta_{0})}{\widehat{\sigma}^{2}(\widetilde{\vbeta}%
)}-\dfrac{2r_{\min}}{k}\dfrac{\vvarepsilon^{\prime}\mC\mX(\widehat{\vbeta}-\vbeta
_{0})}{\widehat{\sigma}^{2}(\widehat{\vbeta})}\nonumber\\
&  +\dfrac{2r_{\min}}{k}\dfrac{\vvarepsilon^{\prime}\mC\mX(\widetilde{\vbeta}%
-\vbeta_{0})}{\widehat{\sigma}^{2}(\widetilde{\vbeta})}\equiv A_{1}+A_{2}%
+A_{3}+A_{4}+A_{5}.\nonumber
\end{align}
From $\widehat{\sigma}^{2}\left(  \vbeta\right)  \equiv\left(  \vy-\mX\vbeta\right)
^{\prime}\mB\left(  \vy-\mX\vbeta\right)  /\tr\left(  \mB\right)  $\
\begin{align}
A_{1}  & =\dfrac{2r_{\min}}{k}\dfrac{\vvarepsilon^{\prime}\mB\mX(\widehat{\vbeta
}-\widetilde{\vbeta})\vvarepsilon^{\prime}\mC\vvarepsilon}{\widehat{\sigma}%
^{2}(\widetilde{\vbeta})\widehat{\sigma}^{2}(\widehat{\vbeta})\tr\left(
\mB\right)  }\label{1_1}\\
& +\dfrac{r_{\min}}{k}\dfrac{(\widetilde{\vbeta}-\vbeta_{0})^{\prime}\mX^{\prime
}\mB\mX(\widetilde{\vbeta}-\vbeta_{0})\vvarepsilon^{\prime}\mC\vvarepsilon
}{\widehat{\sigma}^{2}(\widetilde{\vbeta})\widehat{\sigma}^{2}(\widehat{\vbeta
})\tr\left(  \mB\right)  }\label{1_2}\\
& -\dfrac{r_{\min}}{k}\dfrac{(\widehat{\vbeta}-\vbeta_{0})^{\prime}\mX^{\prime
}\mB\mX(\widehat{\vbeta}-\vbeta_{0})\vvarepsilon^{\prime}\mC\vvarepsilon}%
{\widehat{\sigma}^{2}(\widetilde{\vbeta})\widehat{\sigma}^{2}(\widehat{\vbeta
})\tr\left(  \mB\right)  }.\label{1_3}%
\end{align}
Since by (\ref{t-ks}), (\ref{foc1}), (\ref{psi}) and $\widetilde{\mH}%
\mH^{-1}=\mI_{g}+o_{p}\left(  1\right)  $, it holds that
\begin{equation}
\frac{1}{\sqrt{k}}\mH(\widehat{\vbeta}-\widetilde{\vbeta})=\vxi_{n}+o_{p}\left(
1\right)  =O_{p}\left(  1\right)  ,\label{ksi2}%
\end{equation}
therefore,
\[
\dfrac{r_{\min}}{\sqrt{k}}(\widehat{\vbeta}-\widetilde{\vbeta})=O_{p}\left(
1\right)  .
\]
Using this fact as well as Lemma \ref{LA4}, Lemma \ref{B} and the consistency
of $\widehat{\vbeta}$ and $\widetilde{\vbeta}$, we get that the right hand side
of (\ref{1_1}) is $O_{p}\left(  1\right)  $ while (\ref{1_2}), (\ref{1_3}) are
$o_{p}\left(  1\right)  $. Since from Lemma \ref{B} the right hand side of
(\ref{1_1}) is
\[
\dfrac{r_{\min}}{k}\dfrac{\vvarepsilon^{\prime}\mB\mX(\widehat{\vbeta}%
-\widetilde{\vbeta})\vvarepsilon^{\prime}\mC\vvarepsilon}{\widehat{\sigma}%
^{2}(\widetilde{\vbeta})\widehat{\sigma}^{2}(\widehat{\vbeta})\tr\left(
\mB\right)  }=\dfrac{r_{\min}}{k}\dfrac{\vvarepsilon^{\prime}\mB\mX(\widehat{\vbeta
}-\widetilde{\vbeta})Q_{n}\left(  \vbeta_{0}\right)  }{\sigma^{2}\tr\left(
\mB\right)  }+o_{p}\left(  1\right)  ,
\]
we have
\begin{equation}
A_{1}=\dfrac{2r_{\min}}{k}\dfrac{\vvarepsilon^{\prime}\mB\mX(\widehat{\vbeta
}-\widetilde{\vbeta})Q_{n}\left(  \vbeta_{0}\right)  }{\sigma^{2}\tr\left(
\mB\right)  }+o_{p}\left(  1\right)  .\label{A1}%
\end{equation}
From Lemma \ref{LA4}, Lemma \ref{B} and (\ref{psi}),
\begin{equation}
A_{2}=\dfrac{r_{\min}}{\sigma^{2}}\vpsi_{n}^{\prime}\mH^{-1}\vpsi_{n}+o_{p}\left(
1\right)  .\label{A2}%
\end{equation}
From Lemma \ref{LA4}, Lemma \ref{B}, (\ref{t-ks}), $\widetilde{\mH}\mH^{-1}%
=\mI_{g}+o_{p}\left(  1\right)  $, (\ref{lim}) and (\ref{ksi1}),
\begin{equation}
A_{3}=-\dfrac{r_{\min}}{\sigma^{2}}\vtau_{n}^{\prime}\mH^{-1}\vtau_{n}+\dfrac
{1}{\sigma^{2}}\vtau_{n}^{\prime}\mXi_{a}\vtau_{n}+o_{p}\left(  1\right)
.\label{A3}%
\end{equation}
By Lemma \ref{B}%
\begin{equation}
A_{4}+A_{5}=-\dfrac{2r_{\min}}{k}\dfrac{\vvarepsilon^{\prime}\mC\mX(\widehat{\vbeta
}-\widetilde{\vbeta})}{\sigma^{2}}+o_{p}\left(  1\right)  \label{A45}%
\end{equation}
Now collecting the results from (\ref{A1}), (\ref{A2}), (\ref{A3}),
(\ref{A45}), we obtain%
\begin{align*}
\frac{r_{\min}}{k}\left(  Q_{n}(\widehat{\vbeta})-Q_{n}(\widetilde{\vbeta
})\right)    & =\dfrac{2r_{\min}}{k}\dfrac{\vvarepsilon^{\prime}%
\mB\mX(\widehat{\vbeta}-\widetilde{\vbeta})Q_{n}\left(  \vbeta_{0}\right)  }%
{\sigma^{2}\tr\left(  \mB\right)  }+\dfrac{r_{\min}}{\sigma^{2}}\vpsi_{n}^{\prime
}\mH^{-1}\vpsi_{n}\\
& -\dfrac{r_{\min}}{\sigma^{2}}\vtau_{n}^{\prime}\mH^{-1}\vtau_{n}+\dfrac
{1}{\sigma^{2}}\vtau_{n}^{\prime}\mXi_{a}\vtau_{n}-\dfrac{2r_{\min}}{k}%
\dfrac{\vvarepsilon^{\prime}\mC\mX(\widehat{\vbeta}-\widetilde{\vbeta})}{\sigma^{2}%
}+o_{p}\left(  1\right)  \\
& =-\dfrac{2r_{\min}}{k\sigma^{2}}\vvarepsilon^{\prime}\widehat{\mC}\left(
\vbeta_{0}\right)  \mX(\widehat{\vbeta}-\widetilde{\vbeta})+\dfrac{1}{\sigma^{2}%
}\vtau_{n}^{\prime}\mXi_{a}\vtau_{n}+o_{p}\left(  1\right)  ,
\end{align*}
where the latter equality follows from (\ref{C-hat}) and (\ref{foc1}). Using
(\ref{tau}), (\ref{ksi2}) and (\ref{ksi1}) for the first term, we obtain%
\[
\frac{r_{\min}}{k}\left(  Q_{n}(\widehat{\vbeta})-Q_{n}(\widetilde{\vbeta
})\right)  =-\dfrac{2}{\sigma^{2}}\vtau_{n}^{\prime}\mXi_{a}\vtau_{n}+\dfrac
{1}{\sigma^{2}}\vtau_{n}^{\prime}\mXi_{a}\vtau_{n}+o_{p}\left(  1\right)  ,
\]
so we obtain that%
\[
D_{a}=\vtau_{n}^{\prime}\mXi_{a}\vtau_{n}+o_{p}\left(  1\right)  .
\]
Therefore, similar to $LM_{a}$ and $W_{1a}$, $D_{a}$ converges in distribution
to $\bar{\chi}^{2}(\boldsymbol{\varphi})$.
Finally, consider $W_{2a}$ from \eqref{W2a}. From \eqref{theta} and \eqref{ksi2} we have that $\bm\vartheta_{n}=\vxi_{n}+o_{p}\left(1\right)$, which together with \eqref{LMa} and \eqref{W2a} implies that $W_{2a}=LM_a+o_{p}\left(1\right)$. Therefore, $W_{2a}$ also converges in distribution
to $\bar{\chi}^{2}(\boldsymbol{\varphi})$.
\end{proof}

\begin{proof}
[Proof of Theorem \ref{simplechi2}]From (\ref{objfun*})%
\begin{equation}
\dfrac{1}{k}\left(  Q_{n}^{\ast}(\widehat{\vbeta})-Q_{n}^{\ast}(\vbeta
_{0})\right)  =\dfrac{\vvarepsilon^{\prime}\mJ\vvarepsilon}{k\widehat{\sigma}%
^{2}(\widehat{\vbeta})}-\dfrac{2(\widehat{\vbeta}-\vbeta_{0})^{\prime}\mX^{\prime
}\mJ\vvarepsilon}{k\widehat{\sigma}^{2}(\widehat{\vbeta})}+\dfrac{(\widehat{\vbeta
}-\vbeta_{0})^{\prime}\mX^{\prime}\mJ\mX(\widehat{\vbeta}-\vbeta_{0})}{k\widehat{\sigma
}^{2}(\widehat{\vbeta})}-\dfrac{\vvarepsilon^{\prime}\mJ\vvarepsilon}%
{k\widehat{\sigma}^{2}(\vbeta_{0})}.\label{1/kQ}%
\end{equation}
Note that from Lemma \ref{LA4}.2  $\mX^{\prime}\mC\vvarepsilon/\sqrt{k}=O_{p}(1)$ and from Lemma \ref{LA4}.1 and (\ref{foc1})  

\noindent $\mX^{\prime}\mC\mX(\widehat{\vbeta}-\vbeta_{0})/\sqrt{k}=O_{p}\left(  1\right)$, so
\begin{align*}
\dfrac{\vvarepsilon^{\prime}\mJ\vvarepsilon}{k}  & =\dfrac{\vvarepsilon^{\prime}%
\mC\mX}{\sqrt{k}}\mPhi^{-1}\dfrac{\mX^{\prime}\mC\vvarepsilon}{\sqrt{k}} =O_{p}\left(
1\right)  ,\\
\dfrac{(\widehat{\vbeta}-\vbeta_{0})^{\prime}\mX^{\prime}\mJ\vvarepsilon}{k}  &
=\dfrac{(\widehat{\vbeta}-\vbeta_{0})^{\prime}\mX^{\prime}\mC\mX}{\sqrt{k}}\mPhi
^{-1}\dfrac{\mX^{\prime}\mC\vvarepsilon}{\sqrt{k}}=O_{p}\left(  1\right), \label{2_1/kQ}\\
\dfrac{(\widehat{\vbeta}-\vbeta_{0})^{\prime}\mX^{\prime}\mJ\mX(\widehat{\vbeta}%
-\vbeta_{0})}{k}  & =\dfrac{(\widehat{\vbeta}-\vbeta_{0})^{\prime}\mX^{\prime}%
\mC\mX}{\sqrt{k}}\mPhi^{-1}\dfrac{\mX^{\prime}\mC\mX(\widehat{\vbeta}-\vbeta_{0})}{\sqrt
{k}}=O_{p}\left(  1\right).
\end{align*}
Therefore, by Lemma \ref{B} the first and fourth terms on the right hand side of
(\ref{1/kQ}) are%
\begin{equation}
\dfrac{\vvarepsilon^{\prime}\mJ\vvarepsilon}{k\widehat{\sigma}^{2}(\widehat{\vbeta
})}-\dfrac{\vvarepsilon^{\prime}\mJ\vvarepsilon}{k\widehat{\sigma}^{2}(\vbeta_{0}%
)}=\dfrac{\vvarepsilon^{\prime}\mJ\vvarepsilon}{k}\left(  \dfrac{1}%
{\widehat{\sigma}^{2}(\widehat{\vbeta})}-\dfrac{1}{\widehat{\sigma}^{2}%
(\vbeta_{0})}\right)  =O_{p}\left(  1\right)  \cdot o_{p}\left(  1\right)
=o_{p}\left(  1\right)  .\label{1,4t}%
\end{equation}
Further, expressing $\mX^{\prime}\mC\vvarepsilon/\sqrt{k}$ from
\begin{equation}
\dfrac{\mX^{\prime}\mC(  \vy-\mX\widehat{\vbeta})  }{\sqrt{k}}%
=\dfrac{\mX^{\prime}\mC\vvarepsilon}{\sqrt{k}}-\dfrac{\mX^{\prime}\mC\mX\mH^{-1}\mH(
\widehat{\vbeta}-\vbeta_{0})  }{\sqrt{k}}=\dfrac{\mX^{\prime}\mC\vvarepsilon
}{\sqrt{k}}-\vpsi_{n}+o_{p}\left(  1\right)  ,\label{XC}
\end{equation}
which follows from Lemma \ref{LA4}.1 and (\ref{psi}), then substituting it into the second term on
the right hand side of (\ref{1/kQ}) and using Lemma \ref{B}, Lemma \ref{LA4}
and (\ref{psi}) again, we obtain%
\begin{equation}
\dfrac{2(\widehat{\vbeta}-\vbeta_{0})^{\prime}\mX^{\prime}\mJ\vvarepsilon
}{k\widehat{\sigma}^{2}(\widehat{\vbeta})}=2\vpsi_{n}^{\prime}\mPhi^{-1}%
\dfrac{\mX^{\prime}\mC(  \vy-\mX\widehat{\vbeta})  }{\sigma^{2}\sqrt{k}%
}+\dfrac{2\vpsi_{n}^{\prime}\mPhi^{-1}\vpsi_{n}}{\sigma^{2}}+o_{p}\left(
1\right)  .\label{2t}%
\end{equation}
Finally, from Lemma \ref{B}, Lemma \ref{LA4} and (\ref{psi}) the third term on
the right hand side of (\ref{1/kQ}) is%
\begin{equation}
\dfrac{(\widehat{\vbeta}-\vbeta_{0})^{\prime}\mX^{\prime}\mJ\mX(\widehat{\vbeta}%
-\vbeta_{0})}{k\widehat{\sigma}^{2}(\widehat{\vbeta})}=\dfrac{\vpsi_{n}^{\prime
}\mPhi^{-1}\vpsi_{n}}{\sigma^{2}}+o_{p}\left(  1\right)  .\label{3t}%
\end{equation}
Collecting the results in (\ref{1,4t}), (\ref{2t}), (\ref{3t}) we have that%
\[
\dfrac{1}{k}\left(  Q_{n}^{\ast}(\widehat{\vbeta})-Q_{n}^{\ast}(\vbeta
_{0})\right)  =-2\vpsi_{n}^{\prime}\mPhi^{-1}\dfrac{\mX^{\prime}\mC(
\vy-\mX\widehat{\vbeta})  }{\sigma^{2}\sqrt{k}}-\dfrac{\vpsi_{n}^{\prime}%
\mPhi^{-1}\vpsi_{n}}{\sigma^{2}}+o_{p}\left(  1\right)  ,
\]
which implies that%
\[
D^{\ast}=\vpsi_{n}^{\prime}\mPhi^{-1}\vpsi_{n}+o_{p}\left(  1\right)  \equiv
W^{\ast}+o_{p}\left(  1\right)  \equiv LM^{\ast}+o_{p}(1).
\]
The chi square limit follows from Lemma \ref{chisqqf} and (\ref{cd_psi}).
\end{proof}

\begin{proof}
[Proof of Theorem \ref{lrestrictionchi2}]From (\ref{objfun*a})%
\begin{align}
\dfrac{1}{k}\left(  Q_{an}^{\ast}(\widetilde{\vbeta})-Q_{an}^{\ast
}(\widehat{\vbeta})\right)   &  =\dfrac{\vvarepsilon^{\prime}\mJ_{a}\vvarepsilon
}{k\widehat{\sigma}^{2}(\widetilde{\vbeta})}-\dfrac{\vvarepsilon^{\prime}%
J_{a}\vvarepsilon}{k\widehat{\sigma}^{2}(\widehat{\vbeta})}\label{1/kQ1}\\
&  +\dfrac{(\widetilde{\vbeta}-\vbeta_{0})^{\prime}\mX^{\prime}\mJ_{a}%
\mX(\widetilde{\vbeta}-\vbeta_{0})}{k\widehat{\sigma}^{2}(\widetilde{\vbeta}%
)}-\dfrac{(\widehat{\vbeta}-\vbeta_{0})^{\prime}\mX^{\prime}\mJ_{a}\mX(\widehat{\vbeta
}-\vbeta_{0})}{k\widehat{\sigma}^{2}(\widehat{\vbeta})}\label{1/kQ2}\\
&  +\dfrac{2(\widehat{\vbeta}-\vbeta_{0})^{\prime}\mX^{\prime}\mJ_{a}\vvarepsilon
}{k\widehat{\sigma}^{2}(\widehat{\vbeta})}-\dfrac{2(\widetilde{\vbeta}-\vbeta
_{0})^{\prime}\mX^{\prime}\mJ_{a}\vvarepsilon}{k\widehat{\sigma}^{2}%
(\widetilde{\vbeta})}. \label{1/kQ3}%
\end{align}
Note that by Lemma \ref{LA4}.2 $\vvarepsilon^{\prime}\mC\mX/\sqrt{k}=O_{p}\left(
1\right)  $ and we know by assumption that $\mGamma_{n}=O\left(  1\right)  $
and $\left(  \mGamma_{n}\mPhi\mGamma_{n}^{\prime}\right)  ^{+}=O\left(  1\right)
$. Therefore,%
\[
\dfrac{\vvarepsilon^{\prime}\mJ_{a}\vvarepsilon}{k}=\dfrac{\vvarepsilon^{\prime}%
\mC\mX}{\sqrt{k}}\mGamma_{n}^{\prime}\left(  \mGamma_{n}\mPhi\mGamma_{n}^{\prime
}\right)  ^{+}\mGamma_{n}\dfrac{\mX^{\prime}\mC\vvarepsilon}{\sqrt{k}}=O_{p}\left(
1\right)
\]
and by Lemma \ref{B}.1 the right hand side of (\ref{1/kQ1}) is
\begin{equation}
\dfrac{\vvarepsilon^{\prime}\mJ_{a}\vvarepsilon}{k\widehat{\sigma}^{2}%
(\widetilde{\vbeta})}-\dfrac{\vvarepsilon^{\prime}\mJ_{a}\vvarepsilon
}{k\widehat{\sigma}^{2}(\widehat{\vbeta})}=o_{p}\left(  1\right)  . \label{T41}%
\end{equation}
Further, the first term of (\ref{1/kQ2}) is%

\[
\dfrac{(\widetilde{\vbeta}-\vbeta_{0})^{\prime}\mX^{\prime}\mJ_{a}\mX(\widetilde{\vbeta
}-\vbeta_{0})}{k\widehat{\sigma}^{2}(\widetilde{\vbeta})}=\dfrac
{(\widetilde{\vbeta}-\vbeta_{0})^{\prime}\mX^{\prime}\mC\mX\mGamma_{n}^{\prime}\left(
\mGamma_{n}\mPhi\mGamma_{n}^{\prime}\right)  ^{+}\mGamma_{n}X^{\prime
}\mC\mX(\widetilde{\vbeta}-\vbeta_{0})}{k\widehat{\sigma}^{2}(\widetilde{\vbeta})}.
\]
This involves%
\[
\dfrac{\mGamma_{n}\mX^{\prime}\mC\mX(\widetilde{\vbeta}-\vbeta_{0})}{\sqrt{k}}%
=\dfrac{\mGamma_{n}\mX^{\prime}\mC\mX\mH^{-1}\mH(\widetilde{\vbeta}-\vbeta_{0})}{\sqrt{k}%
}=\mGamma_{n}\vtau_{n}-\mGamma_{n}\vxi_{n}+o_{p}\left(  1\right)  ,
\]
where the latter equality follows from Lemma \ref{LA4}.1, (\ref{HhHi2}) and
(\ref{t-ks}) and the fact that $\mGamma_{n}=O\left(  1\right)  $. From
(\ref{ksi1}) we obtain that%
\begin{equation}
\dfrac{\mGamma_{n}\mX^{\prime}\mC\mX(\widetilde{\vbeta}-\vbeta_{0})}{\sqrt{k}}%
=\mGamma_{n}\vtau_{n}-\mGamma_{n}^{2}\vtau_{n}+o_{p}\left(  1\right)
=o_{p}\left(  1\right)  ,\label{GnXt}%
\end{equation}
where the last equality follows from $\mGamma_{n}^{2}=\mGamma_{n}$. Using $\left(  \mGamma_{n}\mPhi\mGamma_{n}^{\prime}\right)  ^{+}=O\left(
1\right)  $ and Lemma \ref{B}.1 we obtain that
\begin{equation}
\dfrac{(\widetilde{\vbeta}-\vbeta_{0})^{\prime}\mX^{\prime}\mJ_{a}\mX(\widetilde{\vbeta
}-\vbeta_{0})}{k\widehat{\sigma}^{2}(\widetilde{\vbeta})}=o_{p}\left(  1\right)
.\label{T421}%
\end{equation}
The second term of (\ref{1/kQ2}) is%
\[
\dfrac{(\widehat{\vbeta}-\vbeta_{0})^{\prime}\mX^{\prime}\mJ_{a}\mX(\widehat{\vbeta
}-\vbeta_{0})}{k\widehat{\sigma}^{2}(\widehat{\vbeta})}=\dfrac{(\widehat{\vbeta
}-\vbeta_{0})^{\prime}\mX^{\prime}\mC\mX\mGamma_{n}^{\prime}\left(  \mGamma_{n}%
\mPhi\mGamma_{n}^{\prime}\right)  ^{+}\mGamma_{n}\mX^{\prime}\mC\mX(\widehat{\vbeta
}-\vbeta_{0})}{k\widehat{\sigma}^{2}(\widehat{\vbeta})},
\]
which involves%
\begin{equation}
\dfrac{\mGamma_{n}\mX^{\prime}\mC\mX(\widehat{\vbeta}-\vbeta_{0})}{\sqrt{k}}%
=\dfrac{\mGamma_{n}\mX^{\prime}\mC\mX\mH^{-1}\mH(\widehat{\vbeta}-\vbeta_{0})}{\sqrt{k}%
}=\mGamma_{n}\vpsi_{n}+o_{p}\left(  1\right)  ,\label{GnX}%
\end{equation}
where the last equality follows from $\mGamma_{n}=O\left(  1\right)  $, Lemma
\ref{LA4}.1 and (\ref{psi}). Lemma \ref{B}.1 and $\left(  \mGamma_{n}\mPhi
\mGamma_{n}^{\prime}\right)  ^{+}=O\left(  1\right)  $ imply that the second
term of (\ref{1/kQ2}) is%
\begin{equation}
\dfrac{(\widehat{\vbeta}-\vbeta_{0})^{\prime}\mX^{\prime}\mJ_{a}\mX(\widehat{\vbeta
}-\vbeta_{0})}{k\widehat{\sigma}^{2}(\widehat{\vbeta})}=\dfrac{1}{\sigma^{2}%
}\vpsi_{n}^{\prime}\mGamma_{n}^{\prime}\left(  \mGamma_{n}\mPhi\mGamma_{n}^{\prime
}\right)  ^{+}\mGamma_{n}\vpsi_{n}+o_{p}\left(  1\right)  .\label{T422}%
\end{equation}
The first expression in (\ref{1/kQ3}) is
\begin{align*}
\dfrac{2(\widehat{\vbeta}-\vbeta_{0})^{\prime}\mX^{\prime}\mJ_{a}\vvarepsilon
}{k\widehat{\sigma}^{2}(\widehat{\vbeta})} &  =\dfrac{2(\widehat{\vbeta}%
-\vbeta_{0})^{\prime}\mX^{\prime}\mC\mX\mGamma_{n}^{\prime}\left(  \mGamma_{n}%
\mPhi\mGamma_{n}^{\prime}\right)  ^{+}\mGamma_{n}\mX^{\prime}\mC\vvarepsilon
}{k\widehat{\sigma}^{2}(\widehat{\vbeta})}\\
&  =\dfrac{2}{\sigma^{2}}\vpsi_{n}\mGamma_{n}^{\prime}\left(  \mGamma_{n}%
\mPhi\mGamma_{n}^{\prime}\right)  ^{+}\mGamma_{n}\dfrac{\mX^{\prime}\mC\vvarepsilon
}{\sqrt{k}}+o_{p}\left(  1\right)  ,
\end{align*}
where the last equality follows from Lemma \ref{B}.1 and (\ref{GnX}).
Expressing $\mX^{\prime}\mC\vvarepsilon/\sqrt{k}$ from (\ref{XC}) and substituting
it we obtain%
\begin{align}
\dfrac{2(\widehat{\vbeta}-\vbeta_{0})^{\prime}\mX^{\prime}\mJ_{a}\vvarepsilon
}{k\widehat{\sigma}^{2}(\widehat{\vbeta})} &  =\dfrac{2}{\sigma^{2}}\vpsi
_{n}^{\prime}\mGamma_{n}^{\prime}\left(  \mGamma_{n}\mPhi\mGamma_{n}^{\prime
}\right)  ^{+}\mGamma_{n}\vpsi_{n}\nonumber\\
&  +\dfrac{2}{\sigma^{2}}\vpsi_{n}^{\prime}\mGamma_{n}^{\prime}\left(
\mGamma_{n}\mPhi\mGamma_{n}^{\prime}\right)  ^{+}\mGamma_{n}\dfrac{\mX^{\prime
}\mC\left(  \vy-\mX\widehat{\vbeta}\right)  }{\sqrt{k}}+o_{p}\left(  1\right)
.\label{T431}%
\end{align}
The second expression in (\ref{1/kQ3}) is%
\begin{align}
\dfrac{2(\widetilde{\vbeta}-\vbeta_{0})^{\prime}\mX^{\prime}\mJ_{a}\vvarepsilon
}{k\widehat{\sigma}^{2}(\widetilde{\vbeta})} &  =\dfrac{2(\widetilde{\vbeta
}-\vbeta_{0})^{\prime}\mX^{\prime}\mC\mX\mGamma_{n}^{\prime}\left(  \mGamma_{n}%
\mPhi\mGamma_{n}^{\prime}\right)  ^{+}\mGamma_{n}\mX^{\prime}\mC\vvarepsilon
}{k\widehat{\sigma}^{2}(\widetilde{\beta})}\nonumber\\
&  =\dfrac{2}{\sigma^{2}}\left(  \mGamma_{n}\vtau_{n}-\mGamma_{n}^{2}\vtau
_{n}+o_{p}\left(  1\right)  \right)  ^{\prime}\left(  \mGamma_{n}\mPhi\mGamma
_{n}^{\prime}\right)  ^{+}\mGamma_{n}\dfrac{\mX^{\prime}\mC\vvarepsilon}{\sqrt{k}%
}+o_{p}\left(  1\right)  \nonumber\\
&  =o_{p}\left(  1\right)  ,\label{T432}%
\end{align}
where the second line follows from (\ref{GnXt}) and Lemma \ref{B}.1 while the
probability limit follows from $\mGamma_{n}^{2}=\mGamma_{n}$, $\left(
\mGamma_{n}\mPhi\mGamma_{n}^{\prime}\right)  ^{+}=O\left(  1\right)  $ and
$\mGamma_{n}\mX^{\prime}\mC\vvarepsilon/\sqrt{k}=O_{p}\left(  1\right)  $.
Collecting the results in (\ref{T41}), (\ref{T421}), (\ref{T422}),
(\ref{T431}), (\ref{T432}), we obtain%
\begin{align*}
\dfrac{1}{k}\left(  Q_{an}^{\ast}(\widetilde{\vbeta})-Q_{an}^{\ast
}(\widehat{\vbeta})\right)   &  =\dfrac{1}{\sigma^{2}}\vpsi_{n}^{\prime}%
\mGamma_{n}^{\prime}\left(  \mGamma_{n}\mPhi\mGamma_{n}^{\prime}\right)
^{+}\mGamma_{n}\vpsi_{n}\\
&  +\dfrac{2}{\sigma^{2}}\vpsi_{n}^{\prime}\mGamma_{n}^{\prime}\left(
\mGamma_{n}\mPhi\mGamma_{n}^{\prime}\right)  ^{+}\mGamma_{n}\dfrac{\mX^{\prime
}\mC\left(  \vy-\mX\widehat{\vbeta}\right)  }{\sqrt{k}}+o_{p}\left(  1\right)  .
\end{align*}
Therefore,%
\[
D_{a}^{\ast}=\vpsi_{n}^{\prime}\mGamma_{n}^{\prime}\left(  \mGamma_{n}\Phi
\mGamma_{n}^{\prime}\right)  ^{+}\mGamma_{n}\vpsi_{n}+o_{p}\left(  1\right)  .
\]
Note that $\mGamma_{n}\vpsi_{n}\rightarrow_{d}N\left(  \vzeros,\mGamma\mPhi
\mGamma^{\prime}\right)  $, where $\mGamma$ is the limit of $\mGamma_{n}$ as $n\rightarrow\infty$. Since $\mGamma_{n}$ is idempotent and $\mGamma^2_{n}\rightarrow \mGamma^2$, $\mGamma$ is also idempotent and its trace is equal to the limit of the trace of $\mGamma_n$; consequently $\mGamma$ and $\mGamma\mPhi\mGamma^{\prime}$ have rank equal to $p$. Therefore, by Lemma (\ref{chisqqf}) $\vpsi
_{n}^{\prime}\mGamma_{n}^{\prime}\left(  \mGamma_{n}\mPhi\mGamma_{n}^{\prime
}\right)  ^{+}\mGamma_{n}\vpsi_{n}\rightarrow_{d}\chi^{2}_p$ and
consequently $D_{a}^{\ast}\rightarrow_{d}\chi^{2}_p  $.

With respect to the test statistic in (\ref{LM*a}), from (\ref{ksi1}) by using
$\left(  \mGamma_{n}\mPhi\mGamma_{n}^{\prime}\right)  ^{+}=O\left(  1\right)  $
we obtain that
\[
LM_{a}^{\ast}=\vtau_{n}^{\prime}\mGamma_{n}^{\prime}\left(  \mGamma_{n}\mPhi
\mGamma_{n}^{\prime}\right)  ^{+}\mGamma_{n}\vtau_{n}+o_{p}\left(  1\right)
=\vpsi_{n}^{\prime}\mGamma_{n}^{\prime}\left(  \mGamma_{n}\mPhi\mGamma_{n}^{\prime
}\right)  ^{+}\mGamma_{n}\vpsi_{n}+o_{p}\left(  1\right),
\]
where the last equality follows from (\ref{foc1}). Consequently, $LM_{a}%
^{\ast}\rightarrow_{d}\chi^{2}_p$. Since
\[
\bm\vartheta_{n}=\dfrac{\mH(\widehat{\vbeta}-\widetilde{\vbeta})}{\sqrt{k}}%
=\dfrac{\mH(\widehat{\vbeta}-\vbeta_{0})}{\sqrt{k}}-\dfrac{\mH(\widetilde{\vbeta
}-\vbeta_{0})}{\sqrt{k}}=\vpsi_{n}-\vtau_{n}+\vxi_{n}+o_{p}\left(  1\right), 
\]
where the last equality follows from (\ref{psi}), (\ref{HhHi2}) and
(\ref{t-ks}), using (\ref{foc1}) and (\ref{ksi1}) we obtain that
$\bm\vartheta_{n}=\mGamma_{n}\vtau_{n}+o_{p}\left(  1\right)  $, so $W_{2a}^{\ast
}\rightarrow_{d}\chi^{2}_p  $. 

With respect to $W_{1a}^{\ast}$ from (\ref{W1*a}) we note that from (\ref{psi}) and 
(\ref{lrestrictions}) under the null hypothesis,%
\[
\dfrac{r_{\min}}{\sqrt{k}}(\mA\widehat{\vbeta}-\va)=\dfrac{r_{\min}}{\sqrt{k}}\mA(\widehat{\vbeta}-\vbeta_{0})=\dfrac{1}{\sqrt{k}}\mA r_{\min}\mH^{-1}%
\mH(\widehat{\vbeta}-\vbeta_{0})=\mA r_{\min}\mH^{-1}\vpsi_{n}.
\]
From (\ref{cd_psi}) and $\lim r_{\min}\mH^{-1}=\mXi$ nonsingular
%\textcolor{red}{(this extra assumption (also used for Theorem 1)  is not needed for the other 3 statistics) }
we have $\mA r_{\min}\mH^{-1}\vpsi_{n}\rightarrow_{d}N(\vzeros,\mA\mXi\mPhi\mXi \mA^{\prime})$, where $\mA\mXi\mPhi\mXi \mA^{\prime}$ has rank is $p$; 
%\textcolor{red}{!! actually $\mXi$ may not be nonsingular when the rates of $r_{\min}$ and $r_{\max}$ are different, in this case $\mA\mXi\mPhi\mXi \mA^{\prime}$ is not full rank !!}; 
therefore, by Lemma (\ref{chisqqf}) $W_{1a}^{\ast}
%=\frac{1}{k}(\mA\widehat{\vbeta}-\va)^{\prime}(\mA\mH^{-1}\mPhi \mH^{-1}\mA^{\prime})^{-1}(\mA\widehat{\vbeta}-\va)
\rightarrow_{d}\chi^{2}_p.$

\end{proof}

\subsection{The JIVE case}\label{JIVE}
The JIVE case deserves particular attention as the results that produces are, compared to the ratio of quadratic forms case, analytically simpler and easier to implement. This is easy to see by setting the denominator of equation \eqref{objfun} to one and $\widehat\mC(\vbeta)=\mC$.

This section provides results that correspond to those in Theorems \ref{basicresult}  to \ref{lrestrictionchi2}, when the objective function is that of JIVE1/JIVE2. In this case, we need the following slight modification of Assumption \ref{ass6}.

\begin{assumption}
\label{ass6'} The matrix sequence $\frac{1}{k}(\mF+\mG)$ is convergent when $n\rightarrow\infty$ and $\lim\frac{1}{k}(\mF+\mG)=\mPhi$, where
\begin{align}
\mF  &  =\mPi^{\prime}\mZ^{\prime}\mC\mD_{\sigma^{2}}\mC\mZ\mPi,\label{G}\\
\mG  &  =\sum_{i\neq j}C_{ij}^{2}\left(  \sigma_{j}^{2}\mSigma_{i22}+\vsigma_{i21}\vsigma_{j12}\right).
\end{align}
%$\mD_{\sigma^{2}}$ a diagonal matrix containing $\sigma_{i}^{2}$ in its diagonal
%entries and $\mSigma_{i}^{\ast}=\Var\left[  \left(  \varepsilon_{i},\mV_{i}^{\prime}
%-\frac{\varepsilon_{i}\vsigma_{12}}{\vsigma^{2}}\right)^{\prime}  \right]  $.
\end{assumption} 
\noindent Assumption \ref{ass6'} replaces Assumption \ref{ass6} used in Theorems \ref{basicresult}  to \ref{lrestrictionchi2}. In this section we present the feasible versions of the test statistics because they tend to take simpler forms than the unfeasible versions. In order to define the test statistics we use the following notation: 
\begin{align*}
\widehat\mH & =\mX'\mC\mX, \ \ \widehat r_{\min}=\lambda_{\min}\left(\widehat\mH\right), \ \ \widehat\mGamma=\mA^{\prime}(\mA\widehat\mH^{-1}\mA^{\prime})^{-1}\mA\widehat\mH^{-1}, 
\\
\widehat\mPhi & = \widehat\mPhi(\widehat\vbeta), \ \ \widehat\mPhi_0  = \widehat\mPhi(\vbeta_0), \ \ 
\widetilde\mPhi  = \widehat\mPhi(\widetilde\vbeta)
\\
\text{ \ \ with\ \ } & \widehat\mPhi(\vbeta)=\frac{1}{k}\left(
\mX'\mC\mD^2_{\vvarepsilon}\mC \mX+ \mX'\mD_{\vvarepsilon}\mC^{(2)}\mD_{\vvarepsilon}\mX\right),
\end{align*}
where $\widehat\vbeta$ denotes the JIVE estimator, $\vbeta_0$ denotes the null hypothesis in \eqref{simplenull}, $\widetilde\vbeta$ denotes the JIVE estimator restricted by the null hypothesis (\ref{lrestrictions}) and $\mD_{\vvarepsilon}$ is a diagonal matrix with $\vvarepsilon=\vy-\mX\vbeta$ along the main diagonal. The following results hold (proofs are omitted because they are similar to those of Theorems \ref{basicresult}  to \ref{lrestrictionchi2}). 

\begin{theorem}
\label{resJIVE} Suppose that Assumptions \ref{ass1} to \ref{ass5}, \ref{ass5n} and \ref{ass6'} are satisfied and $\mPhi$ is specified in Assumption \ref{ass6'}.
\begin{enumerate}
\item The feasible versions of the test statistics in \eqref{D}, \eqref{LM}, \eqref{W} are identical to $\widehat T=\frac{\widehat r_{\min}}{k}(\widehat{\vbeta}-\vbeta_{0})^{\prime}\widehat\mH(\widehat{\vbeta}-\vbeta_{0})$. If $r_{\min}\mH^{-1}$  is convergent as $n\rightarrow\infty$, $\mXi=\lim r_{\min}\mH^{-1}$ and $\vvarphi$ is the vector of
eigenvalues of $\mXi\mPhi$, then under the null hypothesis \eqref{simplenull} $\widehat T\rightarrow
_{d}\vzeta^{\prime}\mXi\vzeta\sim\bar{\chi}^{2}(\vvarphi)$ for
$\vzeta\sim\mathcal{N}(0,\mPhi)$.
\item The feasible versions of the test statistics in \eqref{Da} to \eqref{W2a} are identical to $\widehat T_{a}=\frac{\widehat r_{\min}}{k}(\mA\widehat{\vbeta}-\va)^{\prime}(\mA \widehat\mH^{-1}\mA^{\prime})^{-1}(\mA\widehat{\vbeta}-\va)$. If $r_{\min}\mH^{-1}$  and $(\mA\mH^{-1}\mA^{\prime})^{-1}\mA\mH^{-1}$ are convergent, $\mXi_{a}=\lim r_{\min}\mH^{-1}\mA^{\prime}(\mA\mH^{-1}\mA^{\prime})^{-1}\mA\mH^{-1}$ and $\vvarphi$ is the vector of
eigenvalues of $\mXi_{a}\mPhi$,
then under the null hypothesis \eqref{lrestrictions} $\widehat T_{a}\rightarrow_{d}\vzeta^{\prime}\mXi_{a}\vzeta\sim
\bar{\chi}^{2}(\boldsymbol{\varphi})$ for $\vzeta\sim\mathcal{N}(\vzeros,\mPhi)$.
%\item The feasible versions of the test statistics in \eqref{D*}, \eqref{LM*}, \eqref{W*} are identical to $\widehat T^{\ast}=\frac{1}{k}(\widehat{\vbeta}-\vbeta_{0})^{\prime}\widehat\mH \widehat\mPhi^{-1} \widehat\mH (\widehat{\vbeta}-\vbeta_{0})$. Under the null hypothesis \eqref{simplenull} $\widehat T^{\ast}$ is asymptotically chi square distributed with $g$ degrees of freedom.
\item For the feasible versions of the test statistics in \eqref{D*}, \eqref{LM*}, \eqref{W*}, we find that $\widehat D^*=\widehat{LM}^*\ne \widehat W^*$. Under the null hypothesis \eqref{simplenull} the tests are asymptotically chi square distributed with $g$ degrees of freedom.
%\item The feasible versions of the test statistics in \eqref{D1*a} to \eqref{W2*a} are identical to $\widehat T_{a}^{\ast}=\frac{1}{k}(\widehat{\vbeta}-\widetilde\vbeta)^{\prime}\widehat\mH \left(  \widehat\mGamma\widehat\mPhi\widehat\mGamma^{\prime}\right)  ^{+} \widehat\mH (\widehat{\vbeta}-\widetilde\vbeta)$, where $\widetilde\vbeta$ denotes the JIVE estimator restricted by the null hypothesis (\ref{lrestrictions}). 
%If $(\mA\mH^{-1}\mA^{\prime})^{-1}\mA\mH^{-1}$ is convergent as $n\rightarrow\infty$, then under the null hypothesis (\ref{lrestrictions}) $\widehat T_{a}^{\ast}$
 
\item For the feasible versions of the test statistics in \eqref{D1*a} to \eqref{W2*a} we find $\widehat{D}^*_{1a}=\widehat{D}^*_{2a}=\widehat{LM}^* _{1a}\ne \widehat{W}^*_{1a}=\widehat{W}^*_{2a}$. If $(\mA\mH^{-1}\mA^{\prime})^{-1}\mA\mH^{-1}$ is convergent as $n\rightarrow\infty$, then under the null hypothesis (\ref{lrestrictions}) the feasible test statistics are asymptotically chi square distributed with $p$ degrees of freedom. 
\end{enumerate}
\end{theorem}
In points \emph{3.} and \emph{4.} distance and Lagrange multiplier statistics differ from Wald statistics because the former use an estimator of $\mPhi$ that relies on $\vbeta_0$ or $\widetilde\vbeta$, the latter on the other hand uses an estimator of  $\mPhi$ that employs $\widehat\vbeta$ as a plug-in. Clearly, by using the same estimator for $\mPhi$ we obtain full equality.
Furthermore, we note that the condition that $r_{\min}\mH^{-1}$ is convergent is not needed when the statistics are asymptotically chi square distributed, in line with the results from Theorems \ref{basicresult}  to \ref{lrestrictionchi2}.

\clearpage

\section{Additional simulation results}\label{completesim}
This section contains the full set of simulation results for DGP1 and DGP2 as presented in Section \ref{MonteCarlo}.

\subsection{DGP1}

\begin{table}[p]
\centering
\caption{Size results for DGP1, $5\%$ nominal level and $n=200$.  Results based on 5000 repetitions.}
\label{tab:DGP1_200_full}
\resizebox{\textwidth}{!}{%
\begin{tabular}{lrr rrrr rrrr rrrrr rrrrr rr}
\toprule
Method & $\alpha$ & $r$ 
 & $D$ & $W_1$ & $W_2$ & $LM$ 
 & $D_{cf}$ & $W_{1,cf}$ & $W_{2,cf}$ & $LM_{cf}$ 
 & $D_1^{*}$ & $D_2^{*}$ & $W_1^{*}$ & $W_2^{*}$ & $LM^{*}$ 
 & $D_{1,cf}^{*}$ & $D_{2,cf}^{*}$ & $W_{1,cf}^{*}$ & $W_{2,cf}^{*}$ & $LM_{cf}^{*}$ 
 & $AR_{naive}$ & $AR_{cf}$ \\
\midrule
SJIVE  & 0.05 & 32 & 0.082 & 0.055 & 0.055 & 0.054 & 0.074 & 0.050 & 0.050 & 0.068 & 0.028 & 0.030 & 0.055 & 0.055 & 0.054 & 0.048 & 0.050 & 0.050 & 0.050 & 0.068 &       &      \\
SJIVE  & 0.05 & 64 & 0.080 & 0.066 & 0.066 & 0.051 & 0.073 & 0.059 & 0.059 & 0.070 & 0.035 & 0.037 & 0.066 & 0.066 & 0.051 & 0.055 & 0.057 & 0.059 & 0.059 & 0.070 &       &      \\
SJIVE  & 0.10 & 32 & 0.073 & 0.051 & 0.051 & 0.050 & 0.067 & 0.047 & 0.047 & 0.056 & 0.030 & 0.035 & 0.051 & 0.051 & 0.050 & 0.044 & 0.047 & 0.047 & 0.047 & 0.056 &       &      \\
SJIVE  & 0.10 & 64 & 0.065 & 0.059 & 0.059 & 0.048 & 0.061 & 0.053 & 0.053 & 0.057 & 0.036 & 0.038 & 0.059 & 0.059 & 0.048 & 0.050 & 0.052 & 0.053 & 0.053 & 0.057 &       &      \\
\addlinespace
HLIM   & 0.05 & 32 & 0.076 & 0.054 & 0.054 & 0.050 & 0.070 & 0.046 & 0.046 & 0.066 & 0.026 & 0.028 & 0.054 & 0.054 & 0.050 & 0.044 & 0.046 & 0.046 & 0.047 & 0.066 &       &      \\
HLIM   & 0.05 & 64 & 0.073 & 0.062 & 0.062 & 0.051 & 0.067 & 0.056 & 0.056 & 0.068 & 0.033 & 0.035 & 0.062 & 0.062 & 0.051 & 0.050 & 0.051 & 0.056 & 0.056 & 0.068 &       &      \\
HLIM   & 0.10 & 32 & 0.069 & 0.049 & 0.049 & 0.046 & 0.064 & 0.043 & 0.043 & 0.051 & 0.026 & 0.030 & 0.049 & 0.049 & 0.046 & 0.036 & 0.038 & 0.043 & 0.043 & 0.051 &       &      \\
HLIM   & 0.10 & 64 & 0.060 & 0.052 & 0.052 & 0.046 & 0.055 & 0.048 & 0.048 & 0.052 & 0.031 & 0.032 & 0.052 & 0.052 & 0.046 & 0.038 & 0.039 & 0.048 & 0.048 & 0.052 &       &      \\
\addlinespace
JIVE1  & 0.05 & 32 & 0.028 & 0.028 & 0.028 & 0.054 & 0.029 & 0.029 & 0.029 & 0.075 & 0.054 &       & 0.028 & 0.028 & 0.054 & 0.075 &       & 0.029 & 0.029 & 0.075 & 0.008 & 0.019 \\
JIVE1  & 0.05 & 64 & 0.049 & 0.049 & 0.049 & 0.052 & 0.051 & 0.051 & 0.051 & 0.073 & 0.052 &       & 0.049 & 0.049 & 0.052 & 0.073 &       & 0.051 & 0.051 & 0.073 & 0.008 & 0.019 \\
JIVE1  & 0.10 & 32 & 0.019 & 0.019 & 0.019 & 0.051 & 0.020 & 0.020 & 0.020 & 0.058 & 0.051 &       & 0.019 & 0.019 & 0.051 & 0.058 &       & 0.020 & 0.020 & 0.058 & 0.015 & 0.038 \\
JIVE1  & 0.10 & 64 & 0.036 & 0.036 & 0.036 & 0.048 & 0.037 & 0.037 & 0.037 & 0.057 & 0.048 &       & 0.036 & 0.036 & 0.048 & 0.057 &       & 0.037 & 0.037 & 0.057 & 0.015 & 0.038 \\
\addlinespace
JIVE2  & 0.05 & 32 & 0.023 & 0.023 & 0.023 & 0.057 & 0.022 & 0.022 & 0.022 & 0.067 & 0.057 &       & 0.023 & 0.023 & 0.057 & 0.067 &       & 0.022 & 0.022 & 0.067 & 0.007 & 0.011 \\
JIVE2  & 0.05 & 64 & 0.043 & 0.043 & 0.043 & 0.050 & 0.040 & 0.040 & 0.040 & 0.064 & 0.050 &       & 0.043 & 0.043 & 0.050 & 0.064 &       & 0.040 & 0.040 & 0.064 & 0.007 & 0.011 \\
JIVE2  & 0.10 & 32 & 0.016 & 0.016 & 0.016 & 0.054 & 0.016 & 0.016 & 0.016 & 0.054 & 0.054 &       & 0.016 & 0.016 & 0.054 & 0.054 &       & 0.016 & 0.016 & 0.054 & 0.014 & 0.018 \\
JIVE2  & 0.10 & 64 & 0.031 & 0.031 & 0.031 & 0.047 & 0.029 & 0.029 & 0.029 & 0.053 & 0.047 &       & 0.031 & 0.031 & 0.047 & 0.053 &       & 0.029 & 0.029 & 0.053 & 0.014 & 0.018 \\
\bottomrule
\end{tabular}%
}
\end{table}

%%%%%%%%%%%%%%%%%%%%%%%%%
% Power DGP1
%%%%%%%%%%%%%%%%%%%%%%%%%

%%%%%%%%%%%%%%%%%%%%%%%%%
% Trinity_200_0.05_32 

\begin{figure}[ht]
    \centering
    % Replace 'chibar2' with 'chi2' for the second set of figures as needed
    \begin{subfigure}[b]{0.47\textwidth}
        \includegraphics[width=\textwidth]{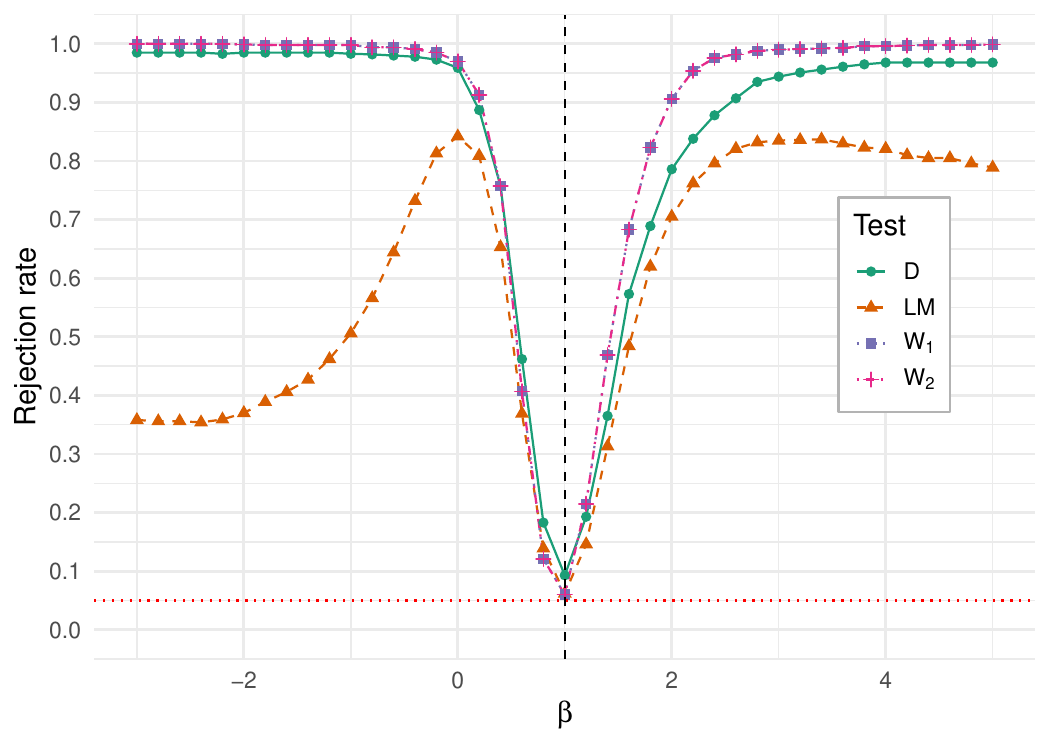}
        \caption{SJIVE}
    \end{subfigure}
    \hfill
    \begin{subfigure}[b]{0.47\textwidth}
        \includegraphics[width=\textwidth]{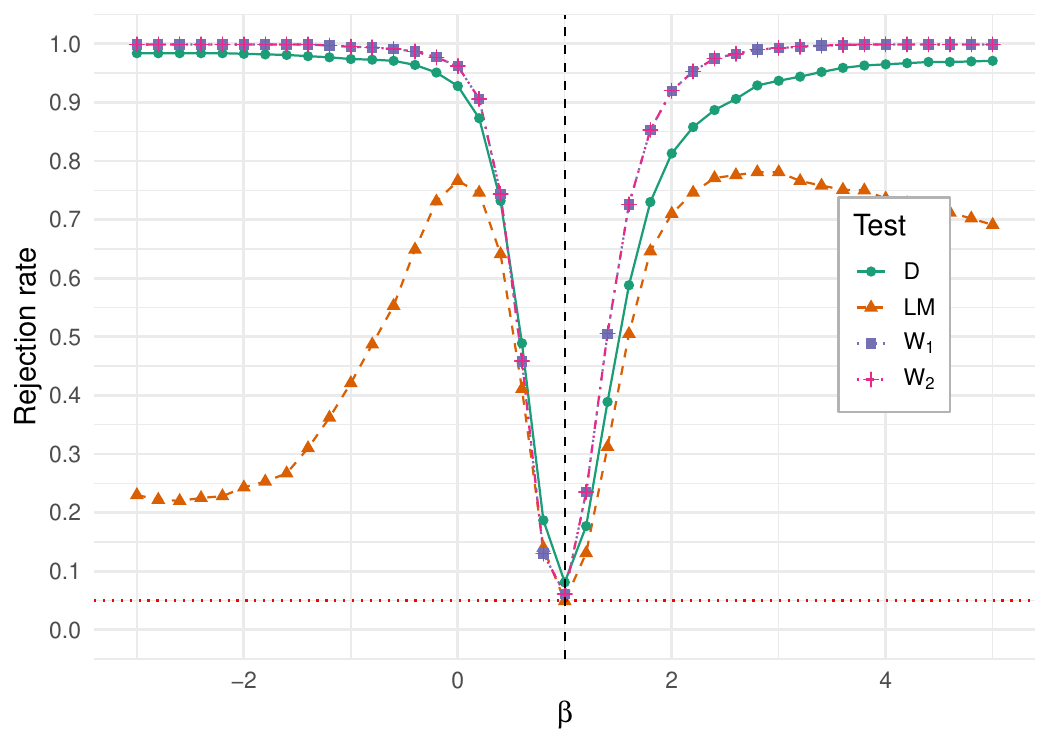}
        \caption{HLIM}
    \end{subfigure}
    
    \begin{subfigure}[b]{0.47\textwidth}
        \includegraphics[width=\textwidth]{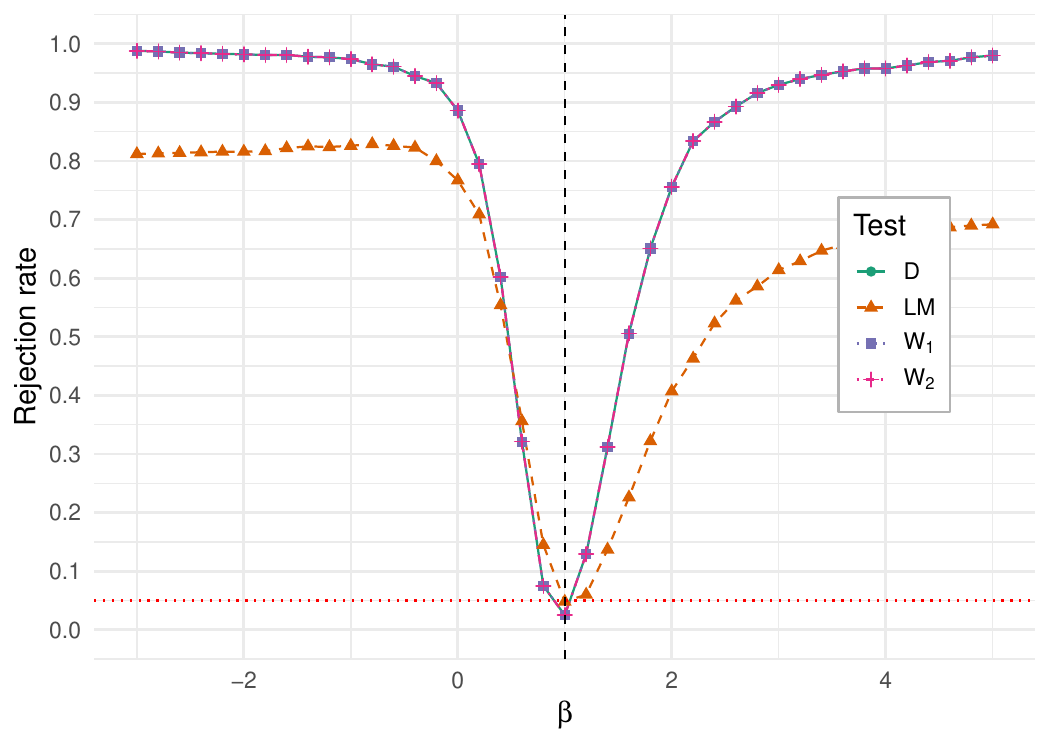}
        \caption{JIVE1}
    \end{subfigure}
    \hfill
    \begin{subfigure}[b]{0.47\textwidth}
        \includegraphics[width=\textwidth]{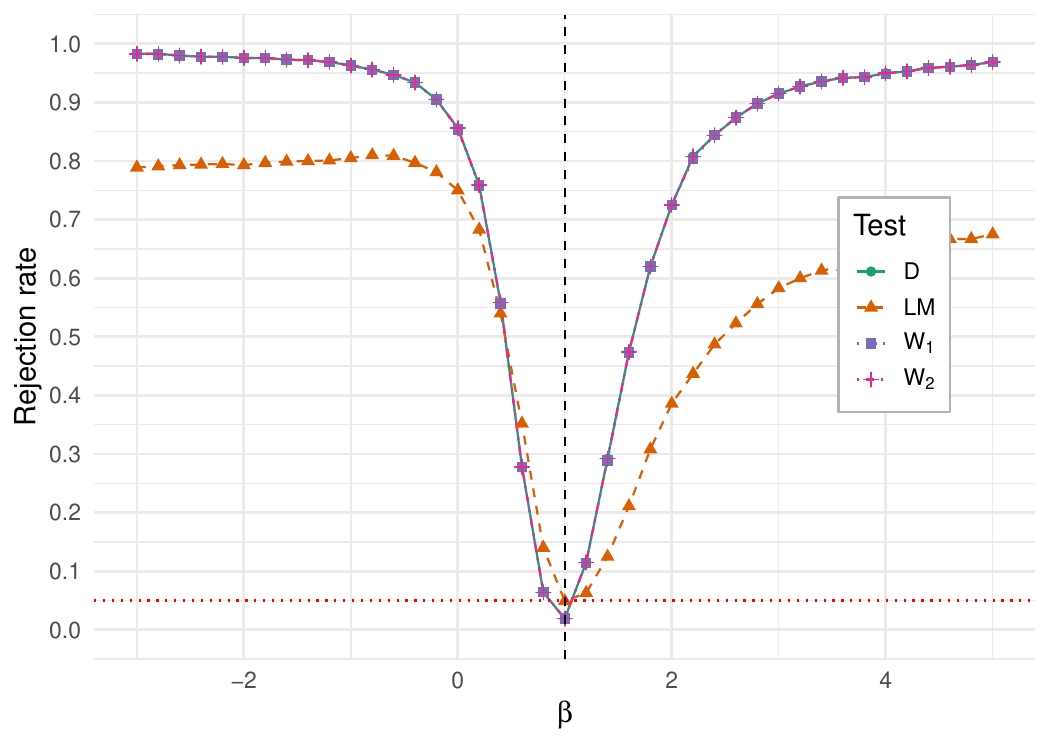}
        \caption{JIVE2}
    \end{subfigure}
    
    \caption{Power curves for DGP1 ($n=200$, $\alpha = 0.05$, $r = 32$). Trinity of test statistics distributed as a $\bar\chi^2$. Results based on 1000 repetitions. The horizontal dotted red line denotes the $5\%$ nominal rejection level, while the vertical dotted black line corresponds to $\beta=1$. Panel (a) plots statistics based on the SJIVE objective function; panel (b) plots statistics based on the HLIM objective function; panel (c) plots statistics based on the JIVE1 objective function; panel (d) plots statistics based on the JIVE2 objective function.}
    \label{fig:Trinity_chibar2_DGP1_200_0.05_32}
\end{figure}

\begin{figure}[ht]
    \centering
    % Replace 'chibar2' with 'chi2' for the second set of figures as needed
    \begin{subfigure}[b]{0.47\textwidth}
        \includegraphics[width=\textwidth]{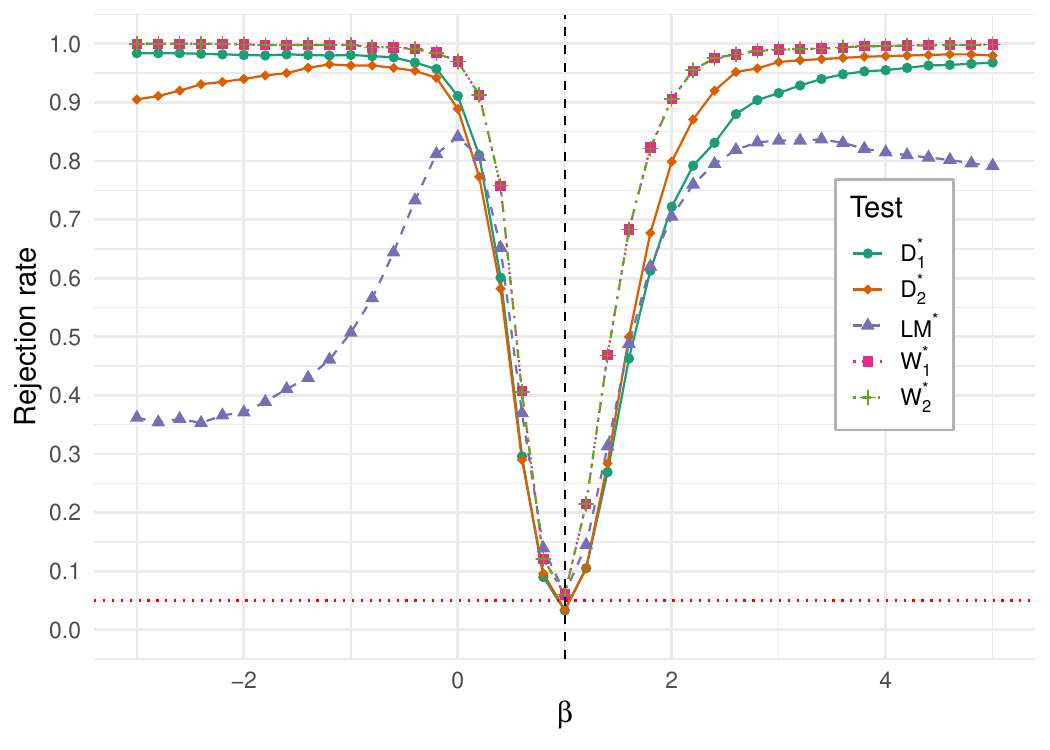}
        \caption{SJIVE}
    \end{subfigure}
    \hfill
    \begin{subfigure}[b]{0.47\textwidth}
        \includegraphics[width=\textwidth]{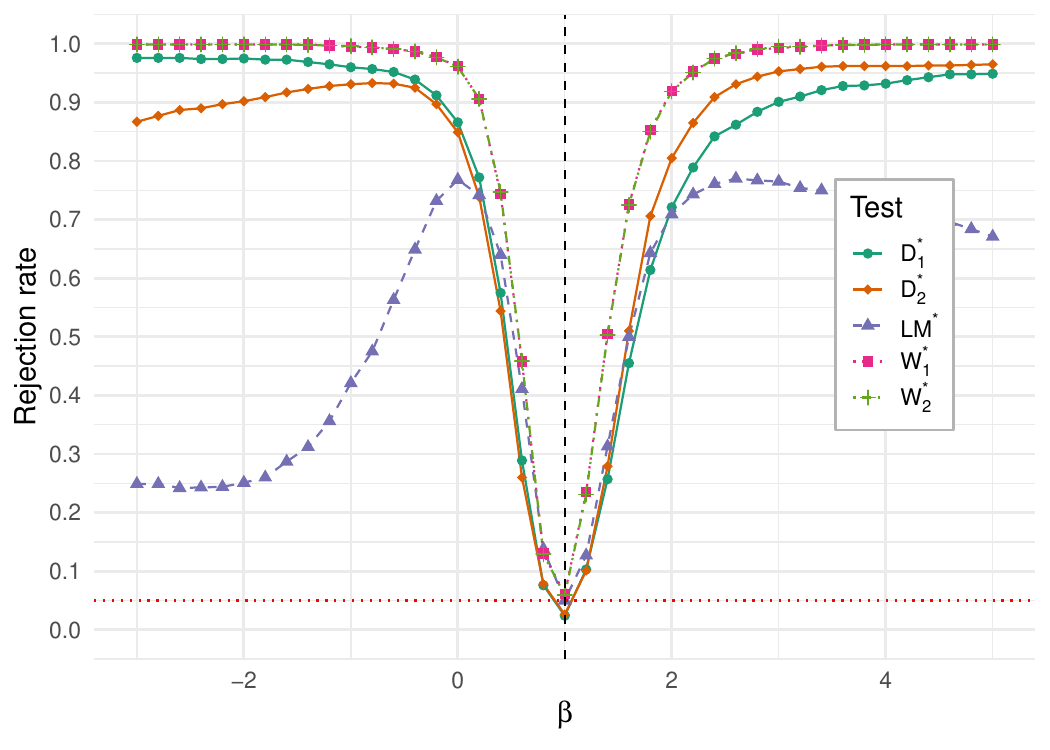}
        \caption{HLIM}
    \end{subfigure}
    
    \begin{subfigure}[b]{0.47\textwidth}
        \includegraphics[width=\textwidth]{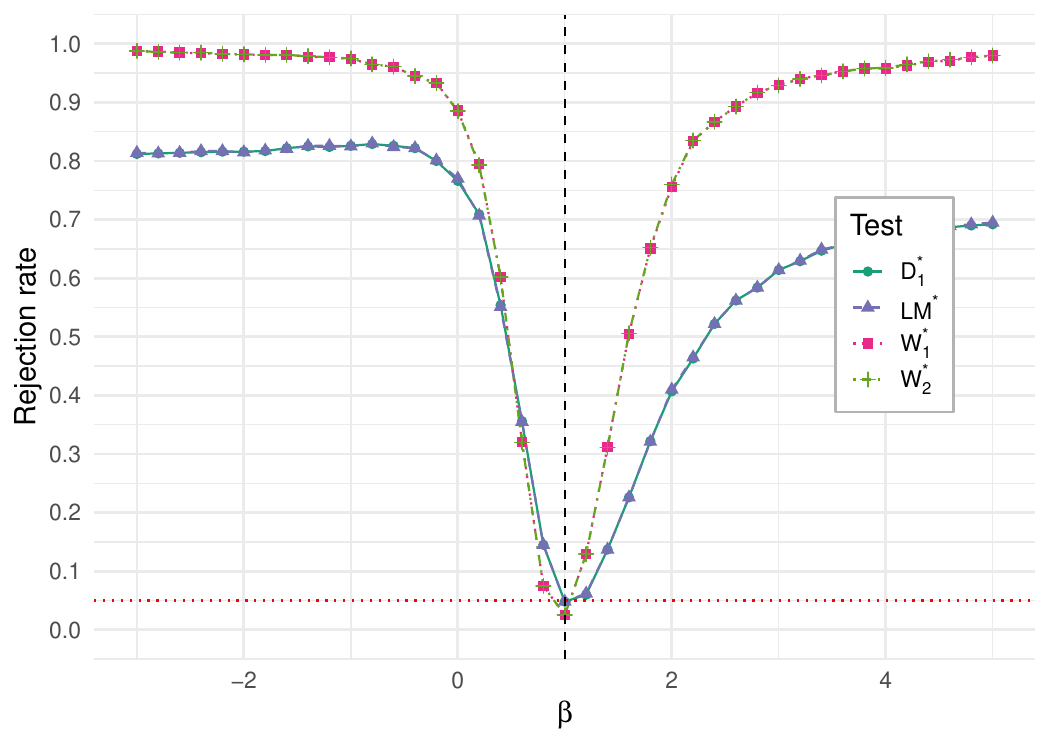}
        \caption{JIVE1}
    \end{subfigure}
    \hfill
    \begin{subfigure}[b]{0.47\textwidth}
        \includegraphics[width=\textwidth]{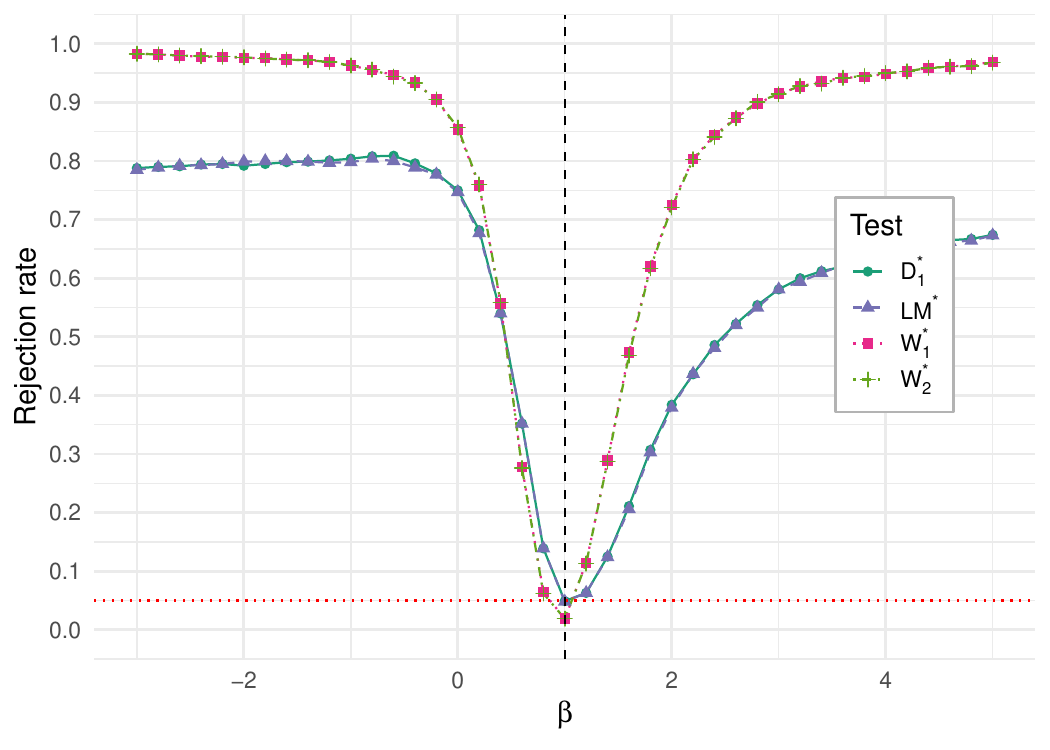}
        \caption{JIVE2}
    \end{subfigure}
    
    \caption{Power curves for DGP1 ($n=200$, $\alpha = 0.05$, $r = 32$). Trinity of test statistics distributed as a $\chi^2$. Results based on 1000 repetitions. The horizontal dotted red line denotes the $5\%$ nominal rejection level, while the vertical dotted black line corresponds to $\beta=1$. Panel (a) plots statistics based on the SJIVE objective function; panel (b) plots statistics based on the HLIM objective function; panel (c) plots statistics based on the JIVE1 objective function; panel (d) plots statistics based on the JIVE2 objective function.}
    \label{fig:Trinity_chi2_DGP1_200_0.05_32}
\end{figure}

% cf variance

\begin{figure}[ht]
    \centering
    % Replace 'chibar2' with 'chi2' for the second set of figures as needed
    \begin{subfigure}[b]{0.47\textwidth}
        \includegraphics[width=\textwidth]{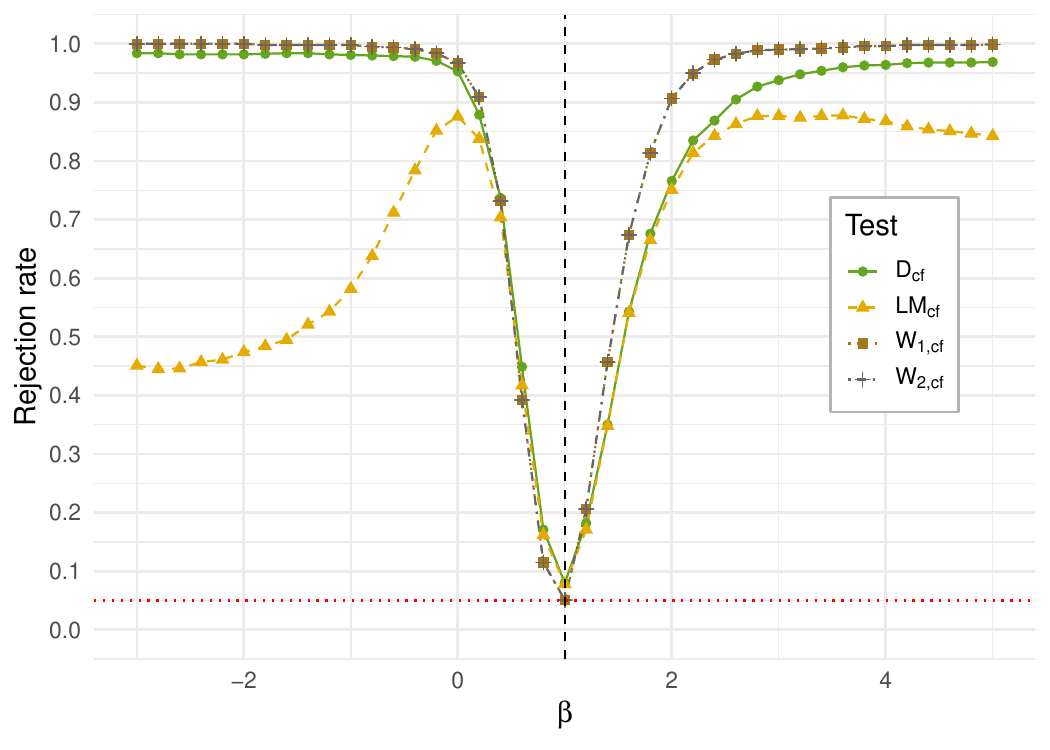}
        \caption{SJIVE}
    \end{subfigure}
    \hfill
    \begin{subfigure}[b]{0.47\textwidth}
        \includegraphics[width=\textwidth]{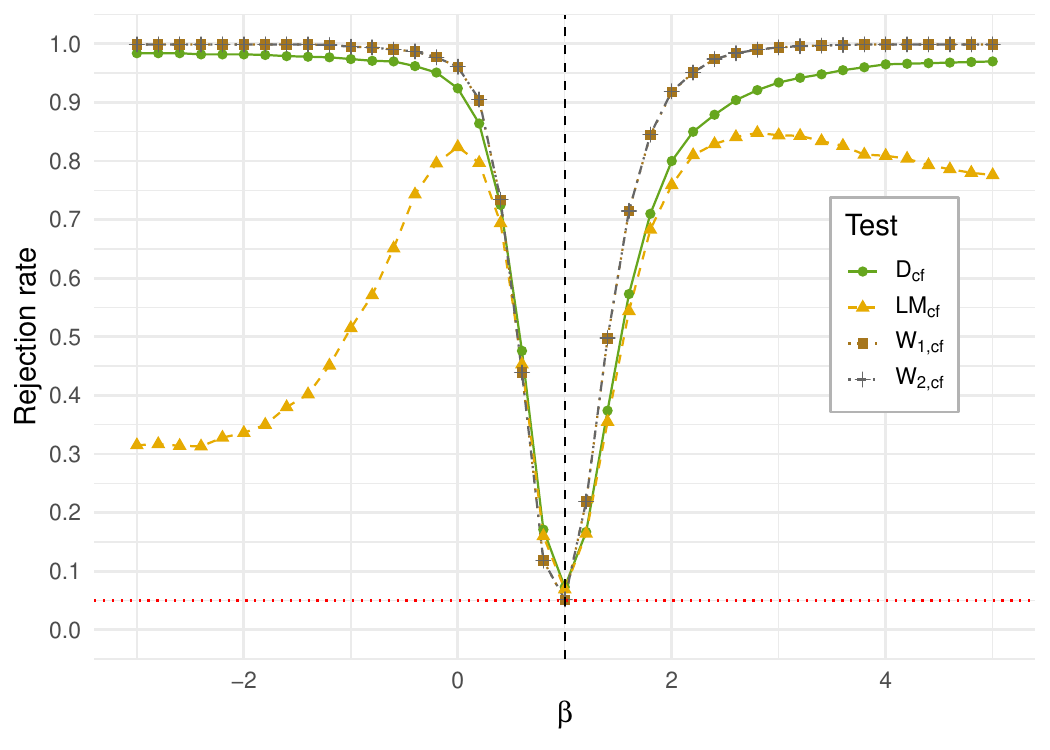}
        \caption{HLIM}
    \end{subfigure}
    
    \begin{subfigure}[b]{0.47\textwidth}
        \includegraphics[width=\textwidth]{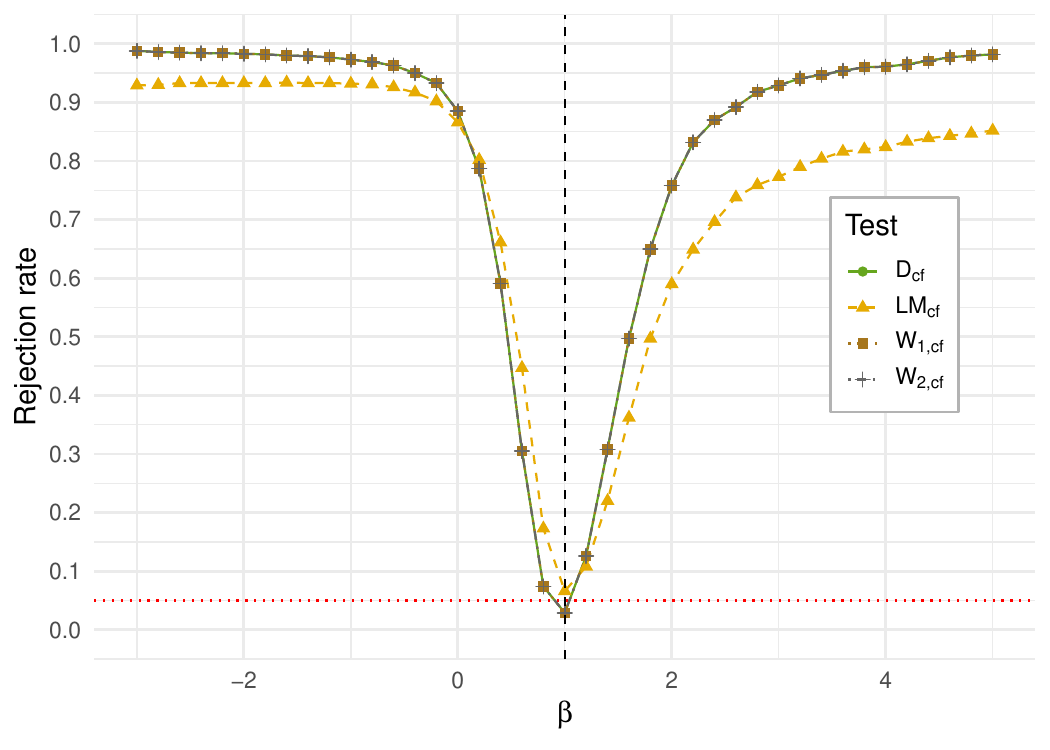}
        \caption{JIVE1}
    \end{subfigure}
    \hfill
    \begin{subfigure}[b]{0.47\textwidth}
        \includegraphics[width=\textwidth]{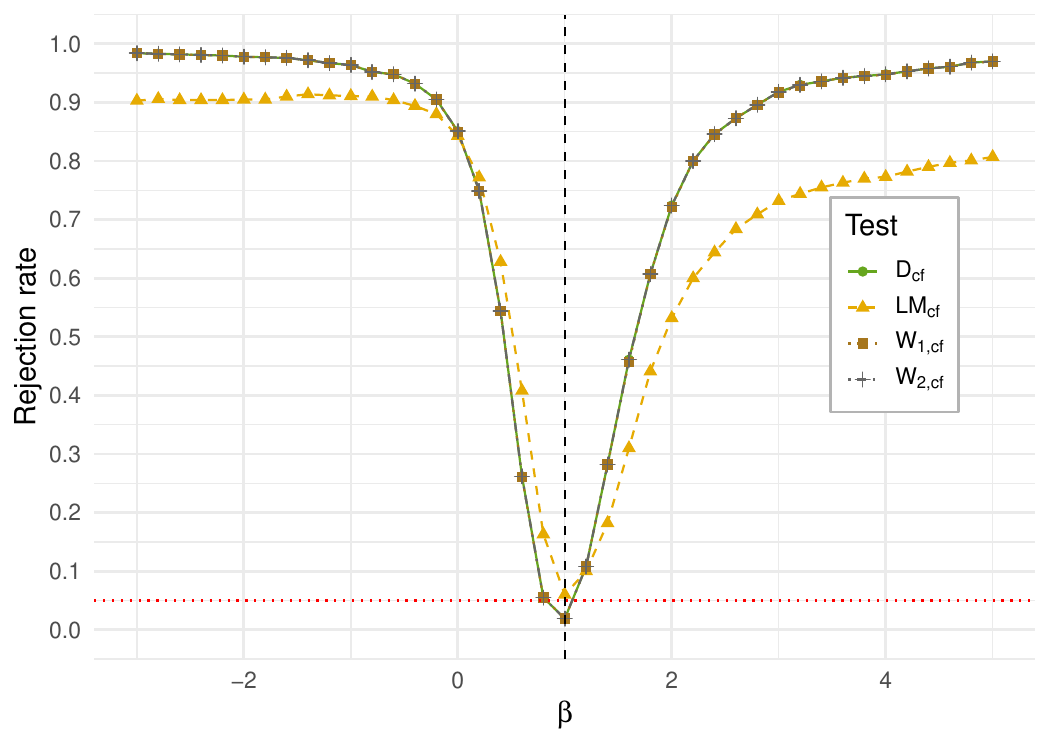}
        \caption{JIVE2}
    \end{subfigure}
    
    \caption{Power curves for DGP1 ($n=200$, $\alpha = 0.05$, $r = 32$). Trinity of test statistics distributed as a $\bar\chi^2$ with cross-fit variance. Results based on 1000 repetitions. The horizontal dotted red line denotes the $5\%$ nominal rejection level, while the vertical dotted black line corresponds to $\beta=1$. Panel (a) plots statistics based on the SJIVE objective function; panel (b) plots statistics based on the HLIM objective function; panel (c) plots statistics based on the JIVE1 objective function; panel (d) plots statistics based on the JIVE2 objective function.}
    \label{fig:Trinity_chibar2_cf_DGP1_200_0.05_32}
\end{figure}

\begin{figure}[ht]
    \centering
    % Replace 'chibar2' with 'chi2' for the second set of figures as needed
    \begin{subfigure}[b]{0.47\textwidth}
        \includegraphics[width=\textwidth]{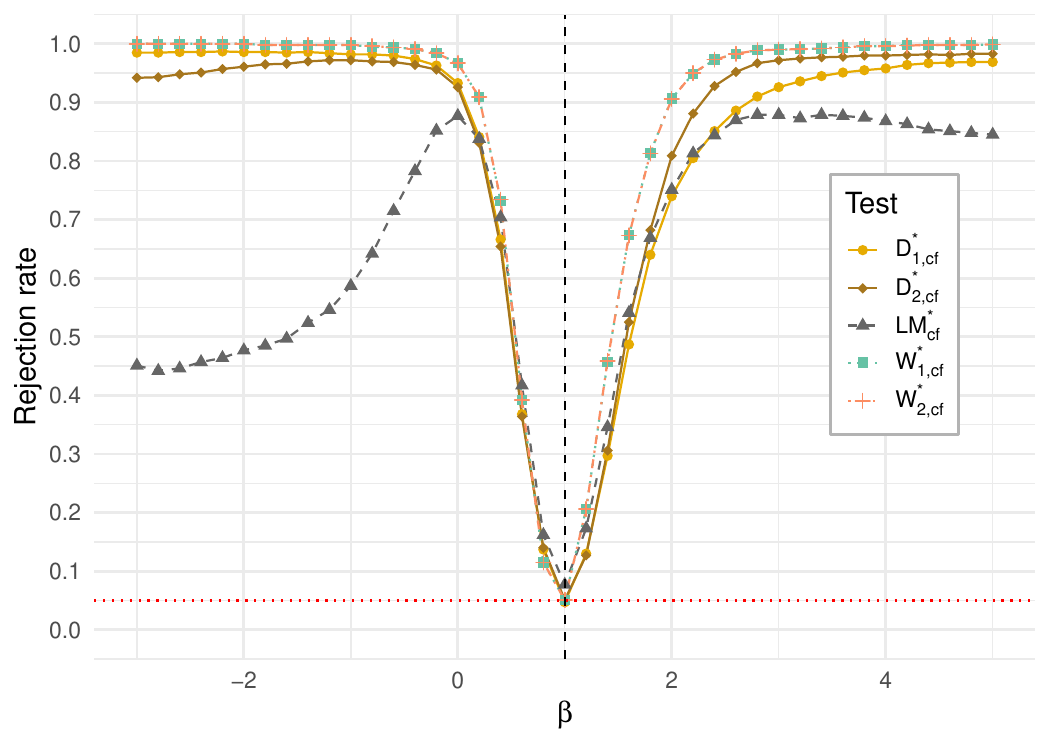}
        \caption{SJIVE}
    \end{subfigure}
    \hfill
    \begin{subfigure}[b]{0.47\textwidth}
        \includegraphics[width=\textwidth]{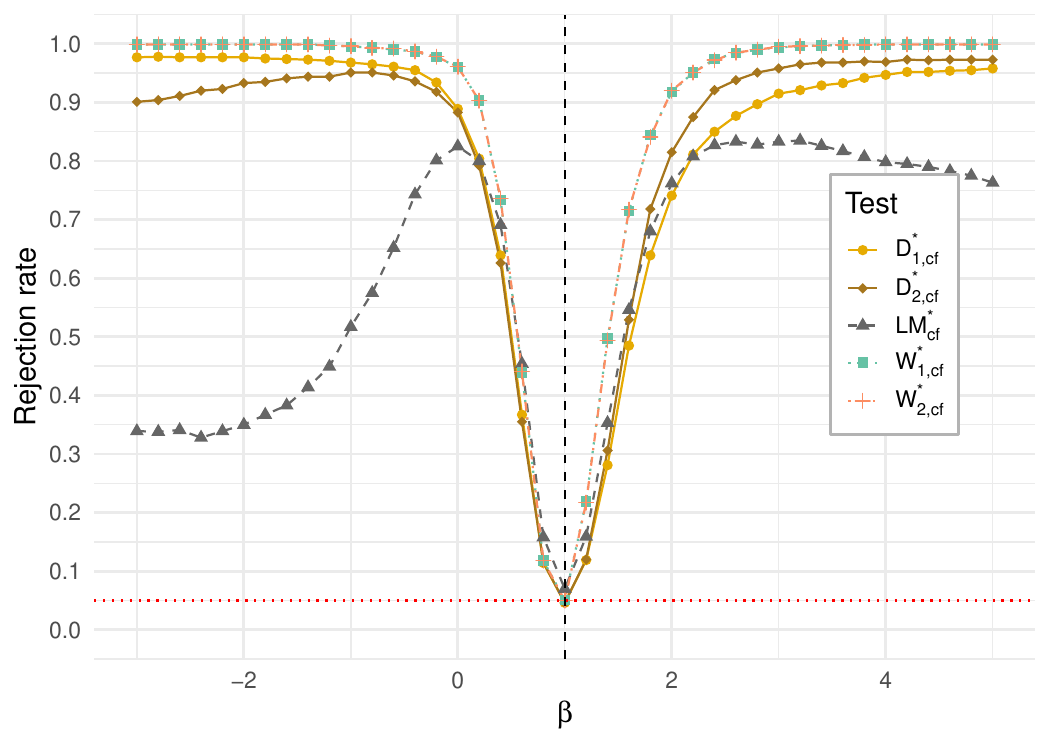}
        \caption{HLIM}
    \end{subfigure}
    
    \begin{subfigure}[b]{0.47\textwidth}
        \includegraphics[width=\textwidth]{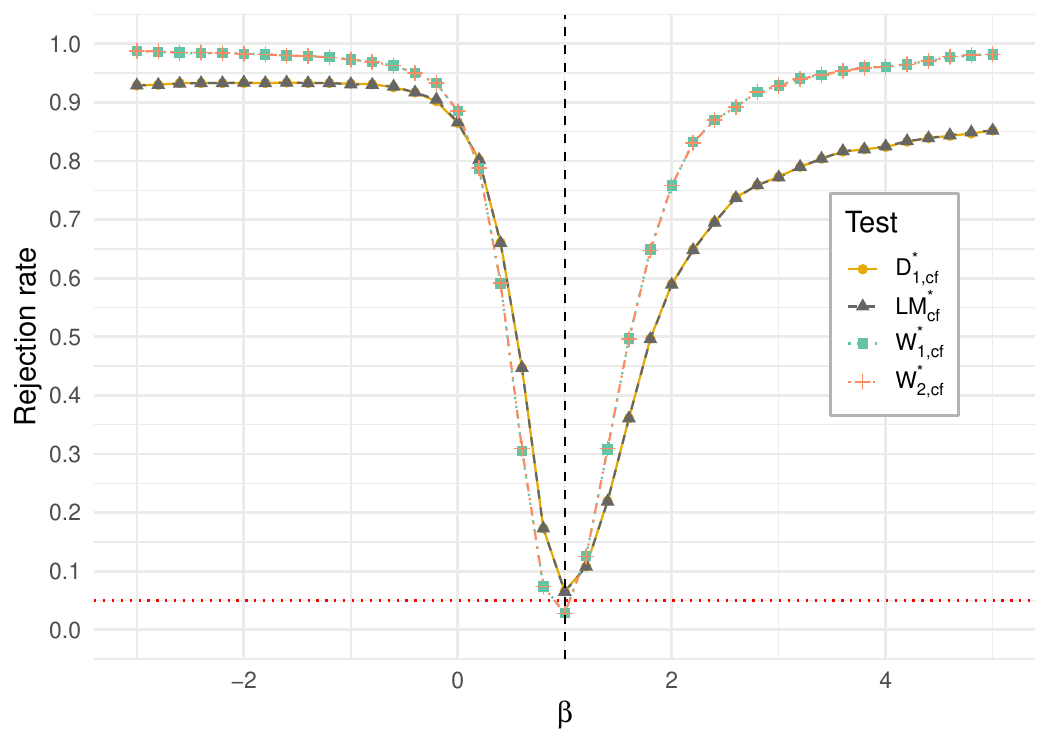}
        \caption{JIVE1}
    \end{subfigure}
    \hfill
    \begin{subfigure}[b]{0.47\textwidth}
        \includegraphics[width=\textwidth]{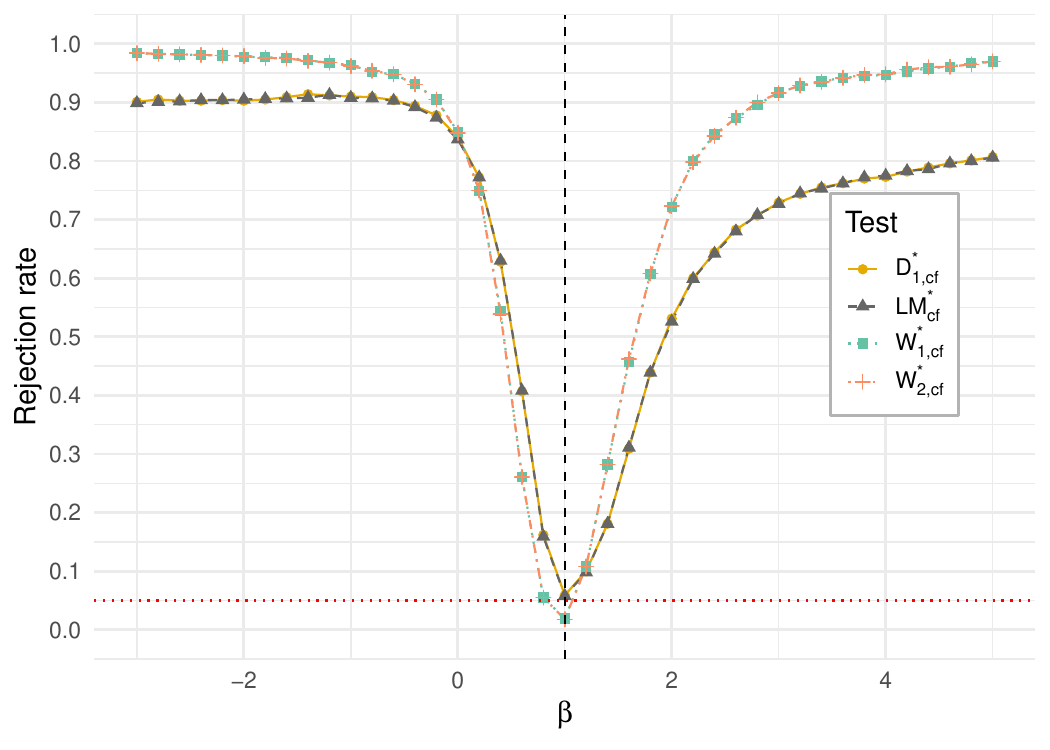}
        \caption{JIVE2}
    \end{subfigure}
    
    \caption{Power curves for DGP1 ($n=200$, $\alpha = 0.05$, $r = 32$). Trinity of test statistics distributed as a $\chi^2$ with cross-fit variance. Results based on 1000 repetitions. The horizontal dotted red line denotes the $5\%$ nominal rejection level, while the vertical dotted black line corresponds to $\beta=1$. Panel (a) plots statistics based on the SJIVE objective function; panel (b) plots statistics based on the HLIM objective function; panel (c) plots statistics based on the JIVE1 objective function; panel (d) plots statistics based on the JIVE2 objective function.}
    \label{fig:Trinity_chi2_cf_DGP1_200_0.05_32}
\end{figure}

%%%%%%%%%%%%%%%%%%%%%%%%%
% Trinity_200_0.05_64 

\begin{figure}[ht]
    \centering
    % Replace 'chibar2' with 'chi2' for the second set of figures as needed
    \begin{subfigure}[b]{0.47\textwidth}
        \includegraphics[width=\textwidth]{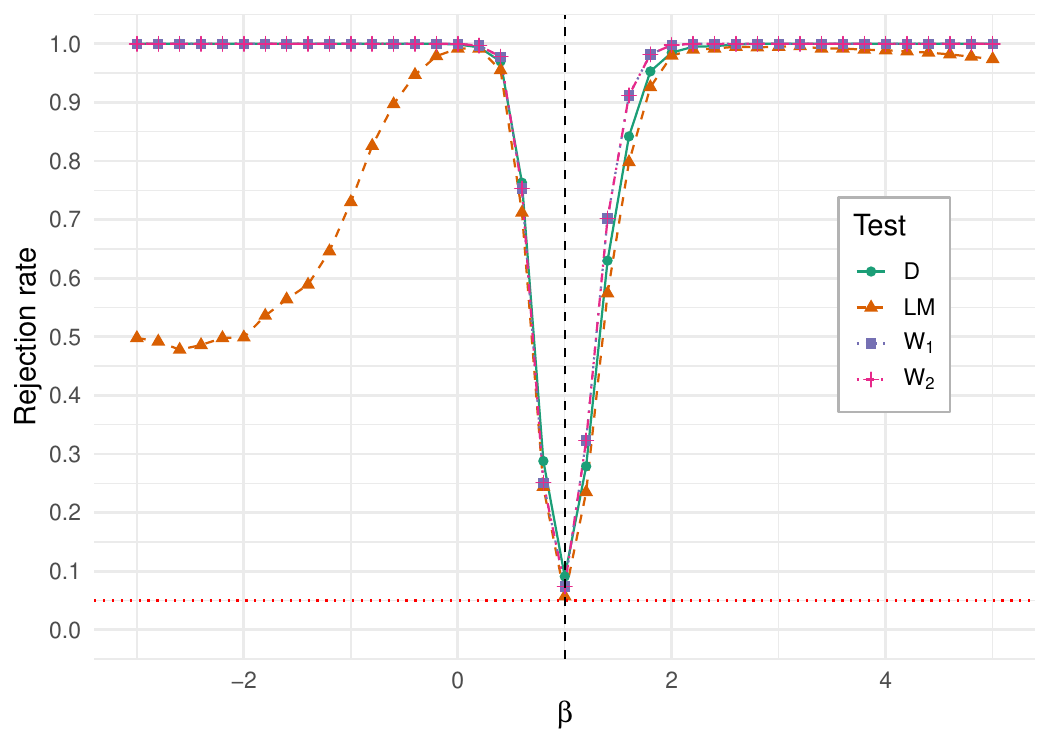}
        \caption{SJIVE}
    \end{subfigure}
    \hfill
    \begin{subfigure}[b]{0.47\textwidth}
        \includegraphics[width=\textwidth]{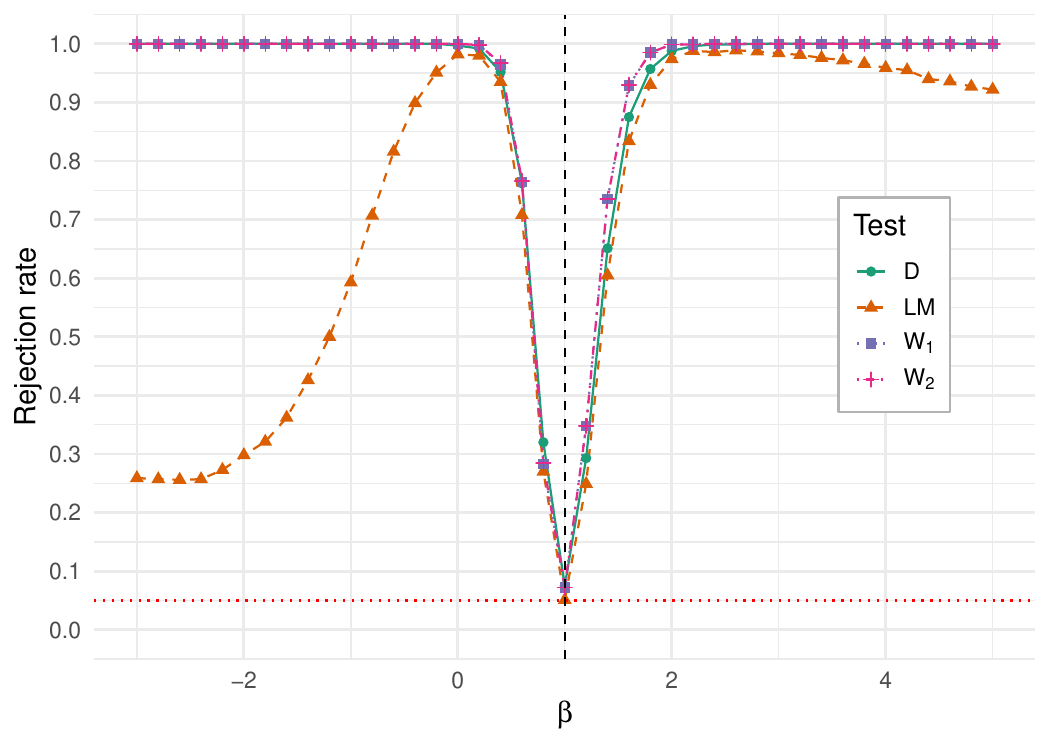}
        \caption{HLIM}
    \end{subfigure}
    
    \begin{subfigure}[b]{0.47\textwidth}
       \includegraphics[width=\textwidth]{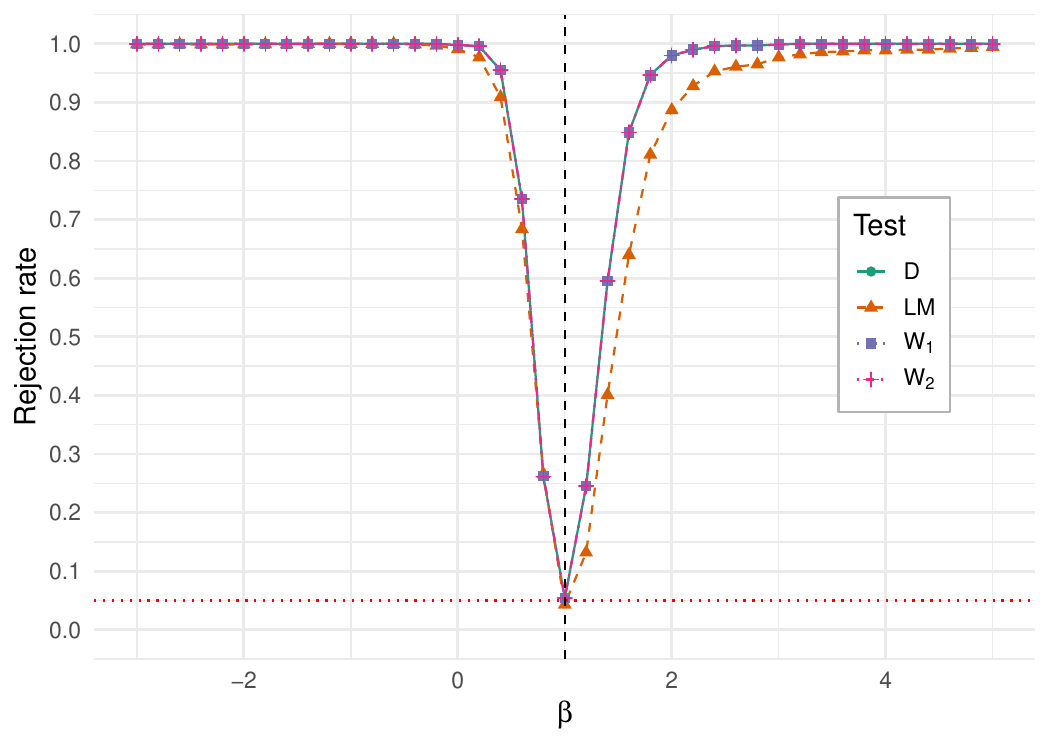}
        \caption{JIVE1}
    \end{subfigure}
    \hfill
    \begin{subfigure}[b]{0.47\textwidth}
        \includegraphics[width=\textwidth]{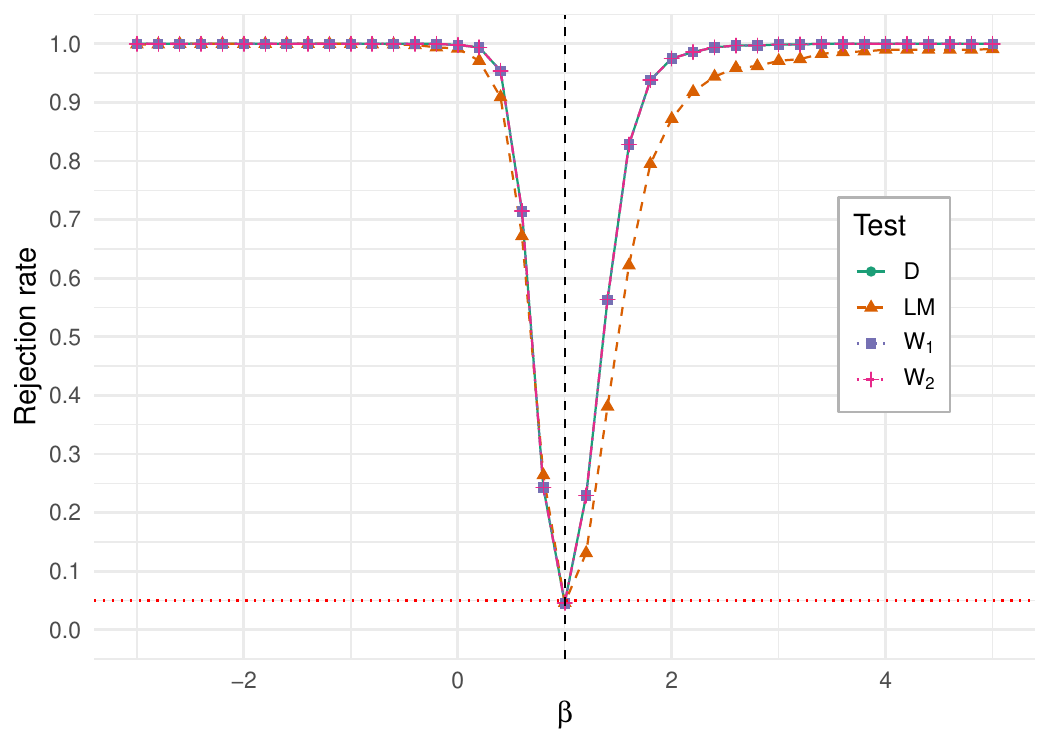}
        \caption{JIVE2}
    \end{subfigure}
    
    \caption{Power curves for DGP1 ($n=200$, $\alpha = 0.05$, $r = 64$). Trinity of test statistics distributed as a $\bar\chi^2$. Results based on 1000 repetitions. The horizontal dotted red line denotes the $5\%$ nominal rejection level, while the vertical dotted black line corresponds to $\beta=1$. Panel (a) plots statistics based on the SJIVE objective function; panel (b) plots statistics based on the HLIM objective function; panel (c) plots statistics based on the JIVE1 objective function; panel (d) plots statistics based on the JIVE2 objective function.}
    \label{fig:Trinity_chibar2_DGP1_200_0.05_64}
\end{figure}

\begin{figure}[ht]
    \centering
    % Replace 'chibar2' with 'chi2' for the second set of figures as needed
    \begin{subfigure}[b]{0.47\textwidth}
        \includegraphics[width=\textwidth]{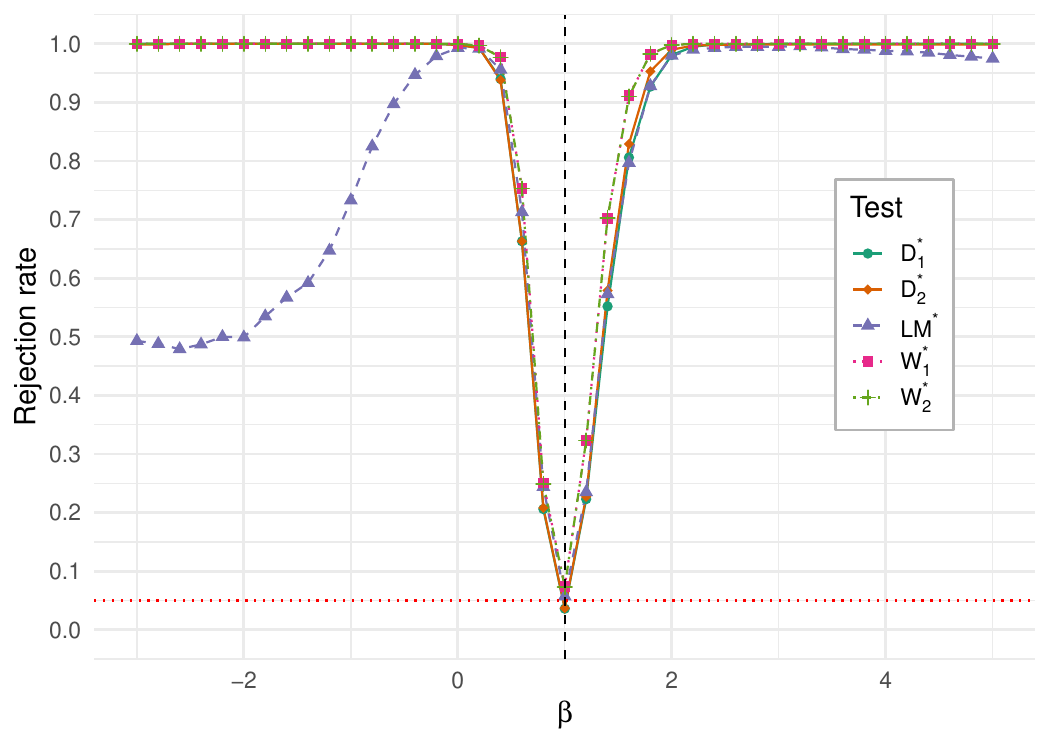}
        \caption{SJIVE}
    \end{subfigure}
    \hfill
    \begin{subfigure}[b]{0.47\textwidth}
        \includegraphics[width=\textwidth]{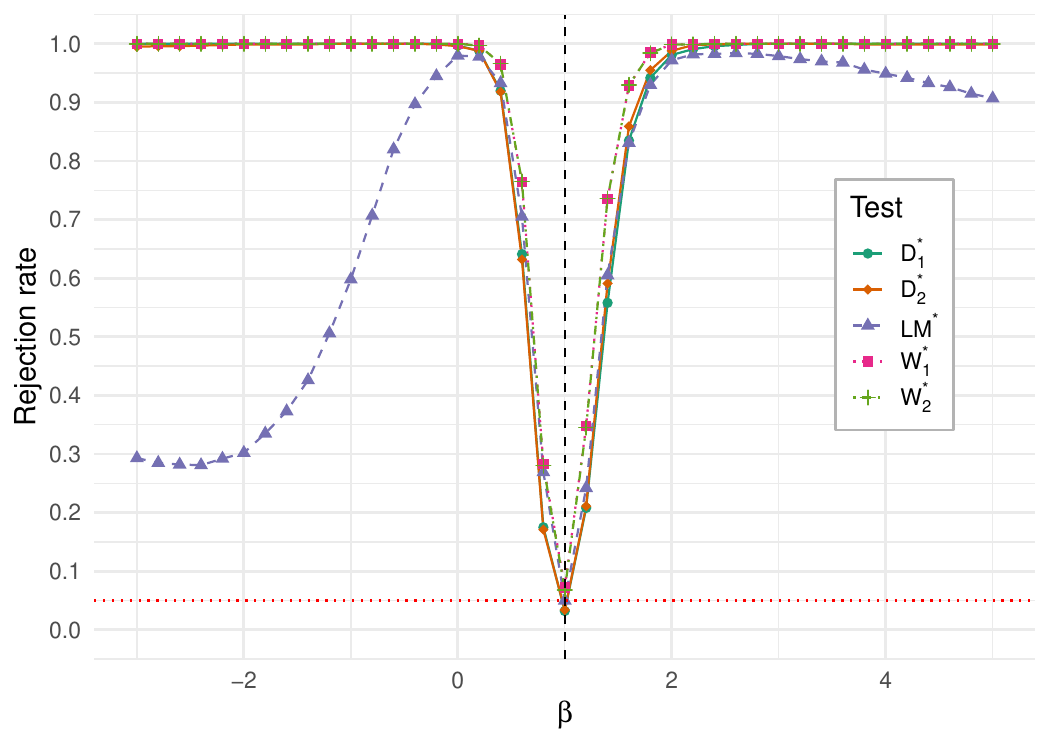}
        \caption{HLIM}
    \end{subfigure}
    
    \begin{subfigure}[b]{0.47\textwidth}
        \includegraphics[width=\textwidth]{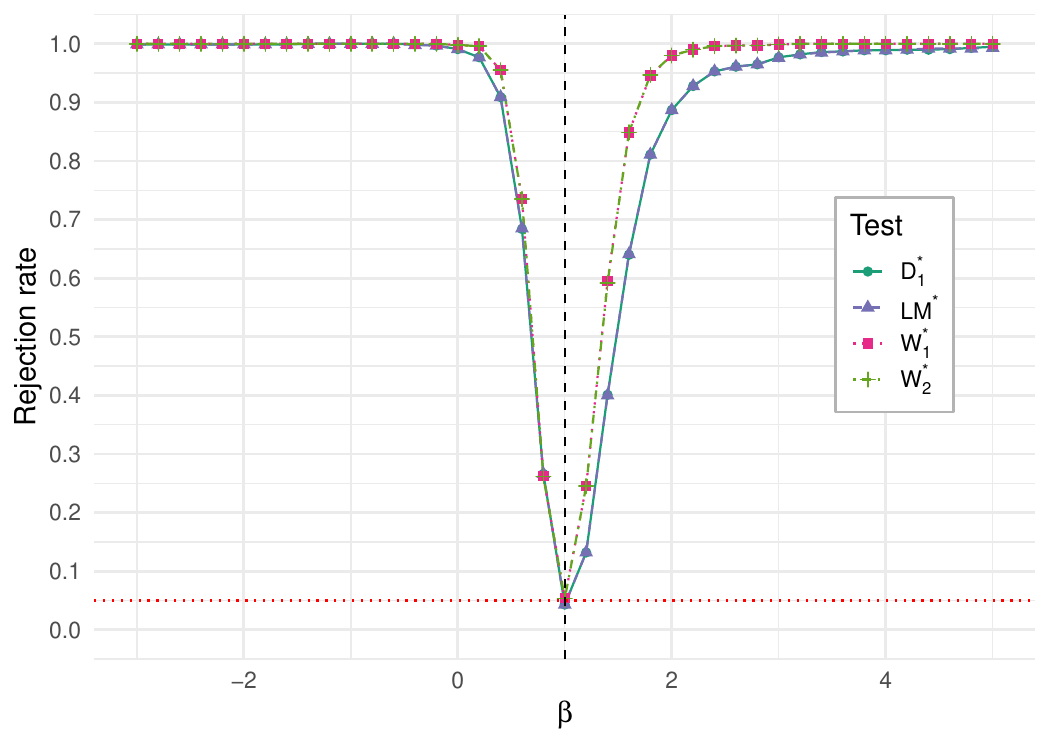}
        \caption{JIVE1}
    \end{subfigure}
    \hfill
    \begin{subfigure}[b]{0.47\textwidth}
        \includegraphics[width=\textwidth]{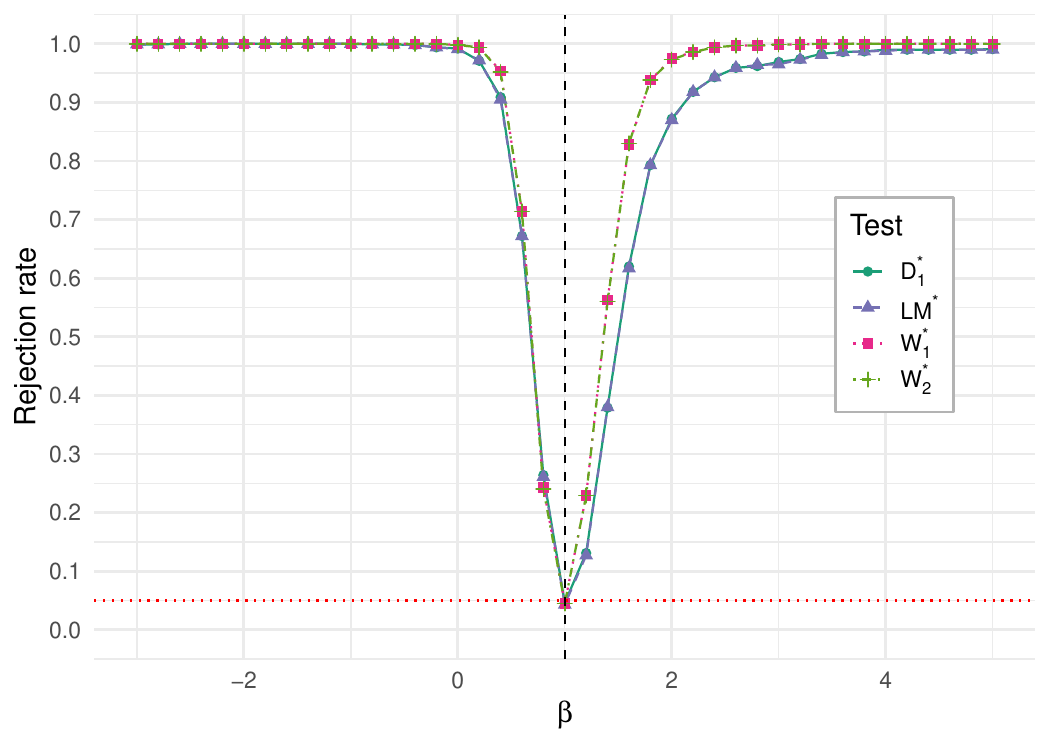}
        \caption{JIVE2}
    \end{subfigure}
    
    \caption{Power curves for DGP1 ($n=200$, $\alpha = 0.05$, $r = 64$). Trinity of test statistics distributed as a $\chi^2$. Results based on 1000 repetitions. The horizontal dotted red line denotes the $5\%$ nominal rejection level, while the vertical dotted black line corresponds to $\beta=1$. Panel (a) plots statistics based on the SJIVE objective function; panel (b) plots statistics based on the HLIM objective function; panel (c) plots statistics based on the JIVE1 objective function; panel (d) plots statistics based on the JIVE2 objective function.}
    \label{fig:Trinity_chi2_DGP1_200_0.05_64}
\end{figure}

% cf variance

\begin{figure}[ht]
    \centering
    % Replace 'chibar2' with 'chi2' for the second set of figures as needed
    \begin{subfigure}[b]{0.47\textwidth}
        \includegraphics[width=\textwidth]{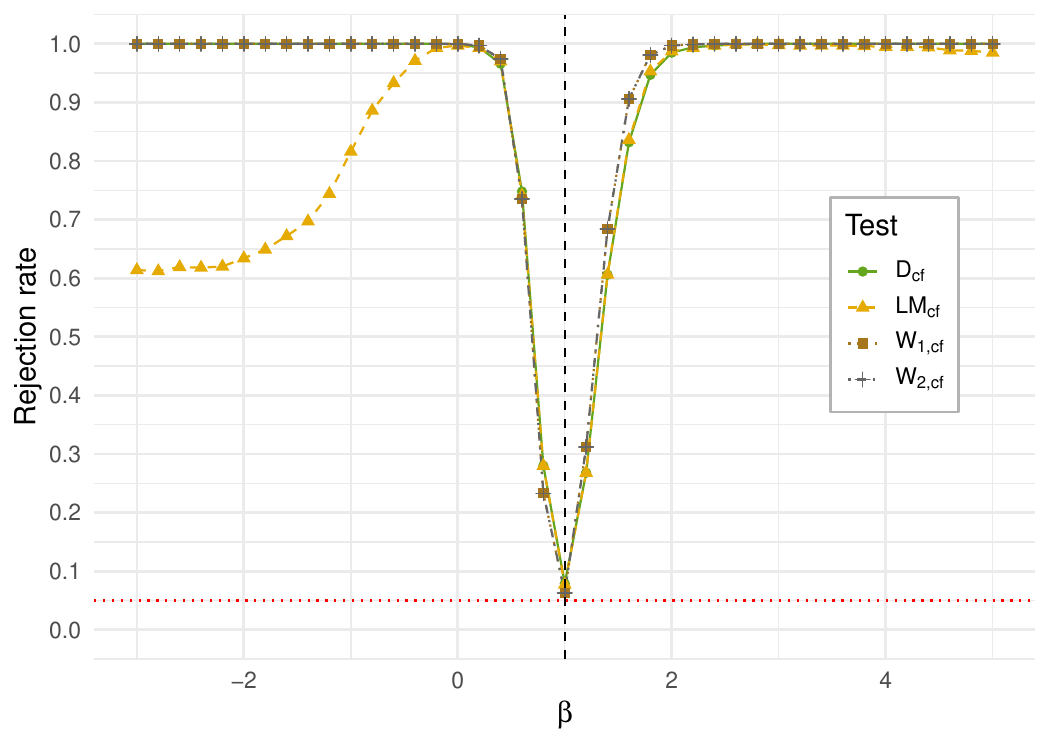}
        \caption{SJIVE}
    \end{subfigure}
    \hfill
    \begin{subfigure}[b]{0.47\textwidth}
        \includegraphics[width=\textwidth]{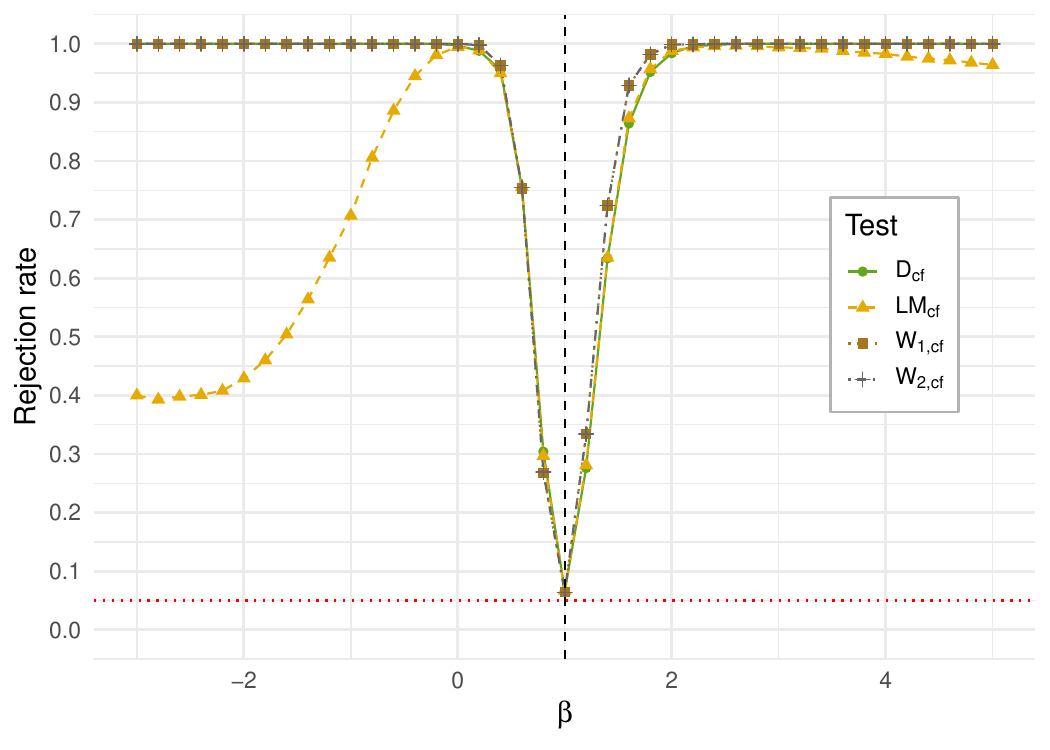}
        \caption{HLIM}
    \end{subfigure}
    
    \begin{subfigure}[b]{0.47\textwidth}
        \includegraphics[width=\textwidth]{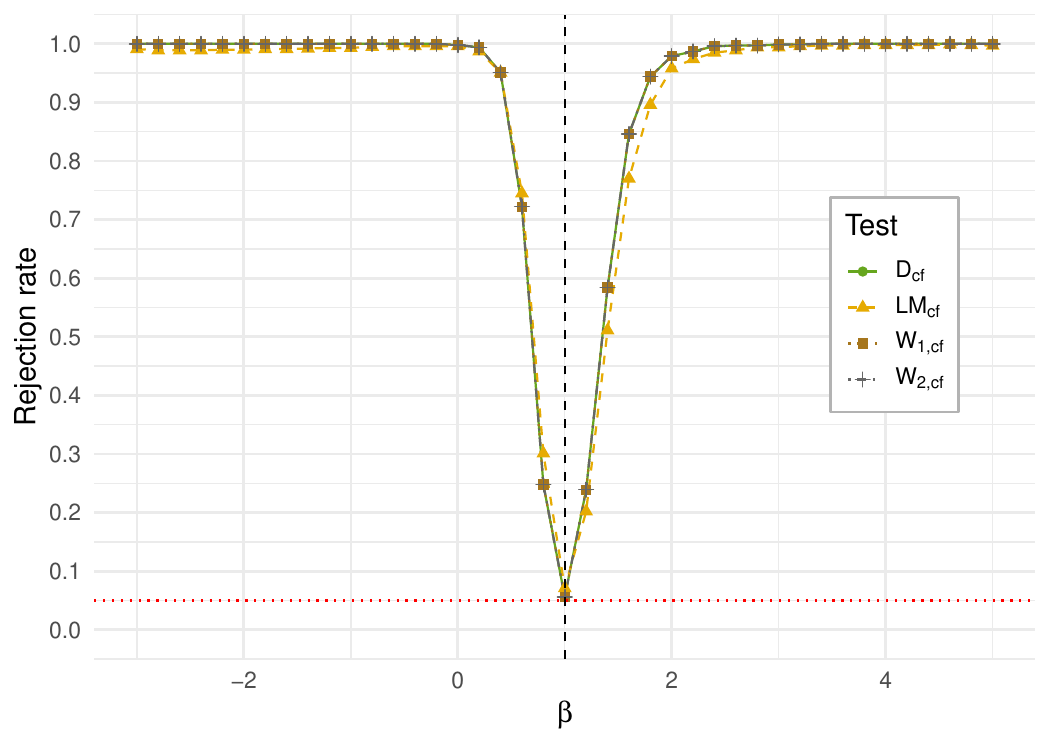}
        \caption{JIVE1}
    \end{subfigure}
    \hfill
    \begin{subfigure}[b]{0.47\textwidth}
       \includegraphics[width=\textwidth]{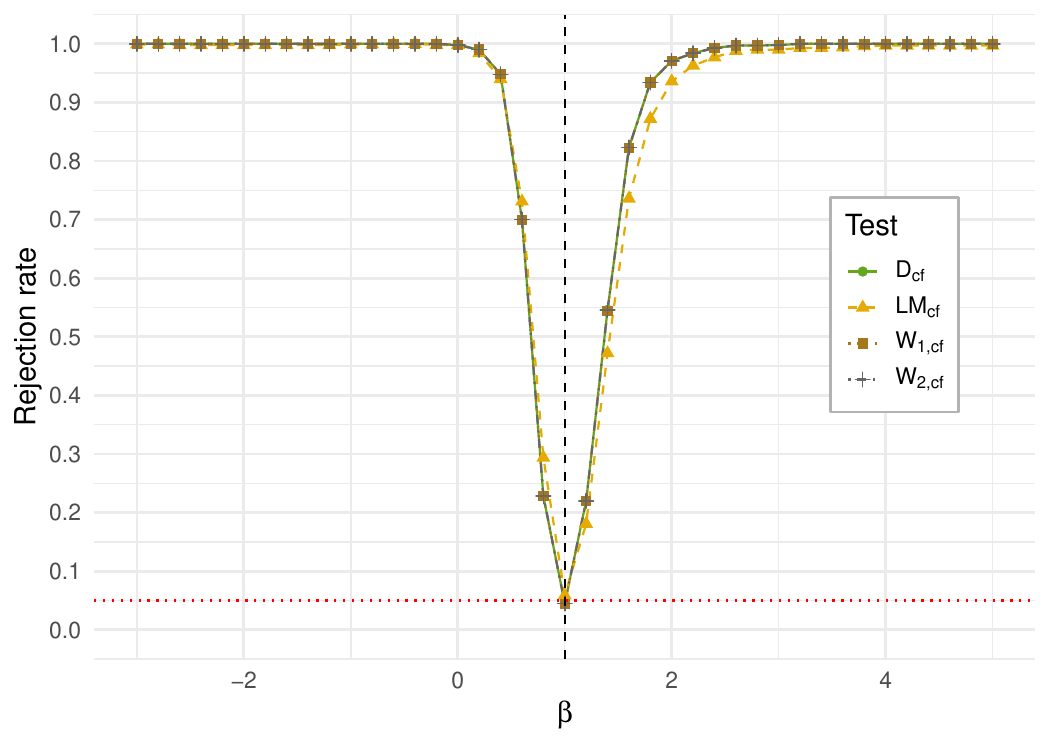}
        \caption{JIVE2}
    \end{subfigure}
    
    \caption{Power curves for DGP1 ($n=200$, $\alpha = 0.05$, $r = 64$). Trinity of test statistics distributed as a $\bar\chi^2$ with cross-fit variance. Results based on 1000 repetitions. The horizontal dotted red line denotes the $5\%$ nominal rejection level, while the vertical dotted black line corresponds to $\beta=1$. Panel (a) plots statistics based on the SJIVE objective function; panel (b) plots statistics based on the HLIM objective function; panel (c) plots statistics based on the JIVE1 objective function; panel (d) plots statistics based on the JIVE2 objective function.}
    \label{fig:Trinity_chibar2_cf_DGP1_200_0.05_64}
\end{figure}

\begin{figure}[ht]
    \centering
    % Replace 'chibar2' with 'chi2' for the second set of figures as needed
    \begin{subfigure}[b]{0.47\textwidth}
        \includegraphics[width=\textwidth]{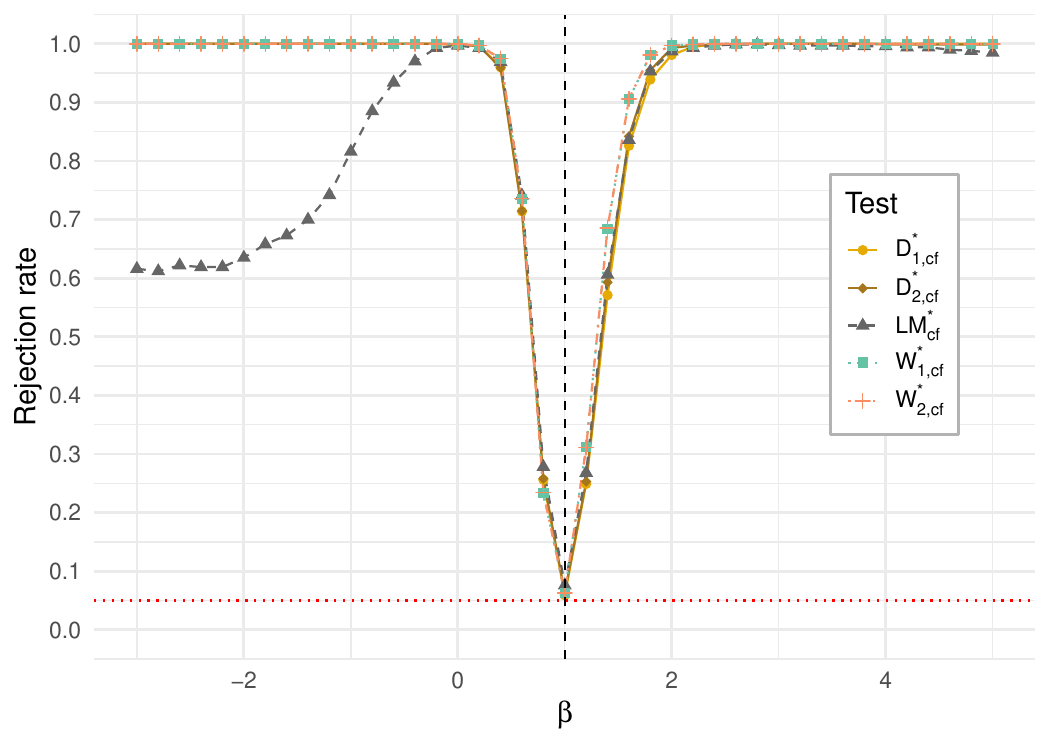}
        \caption{SJIVE}
    \end{subfigure}
    \hfill
    \begin{subfigure}[b]{0.47\textwidth}
        \includegraphics[width=\textwidth]{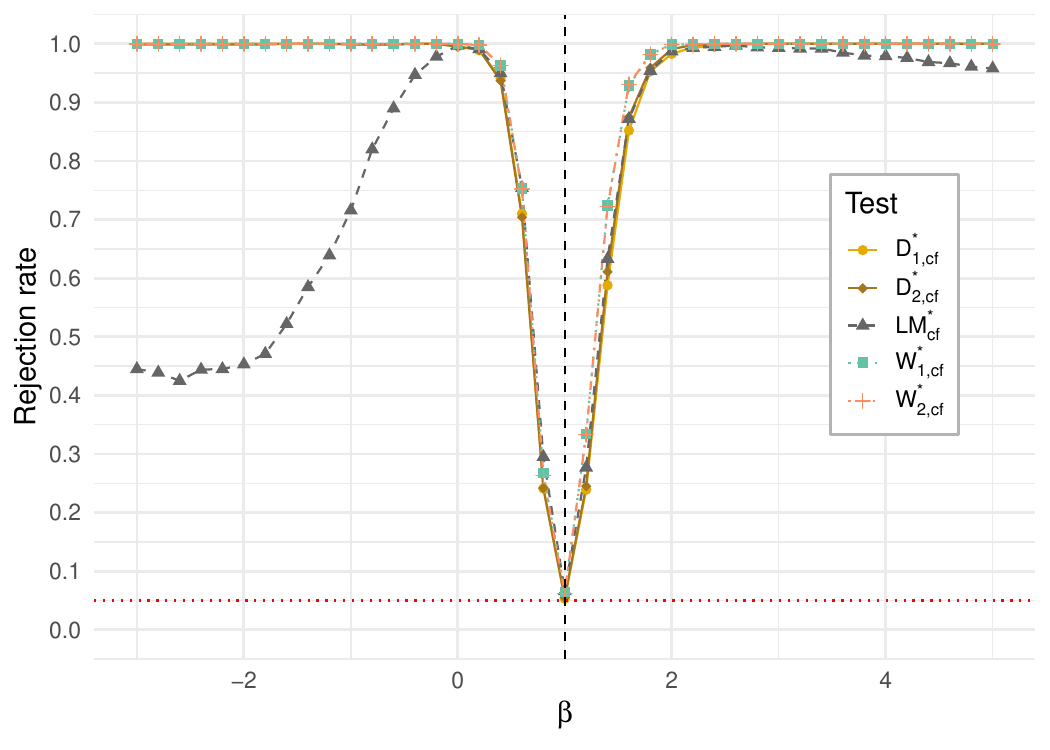}
        \caption{HLIM}
    \end{subfigure}
    
    \begin{subfigure}[b]{0.47\textwidth}
        \includegraphics[width=\textwidth]{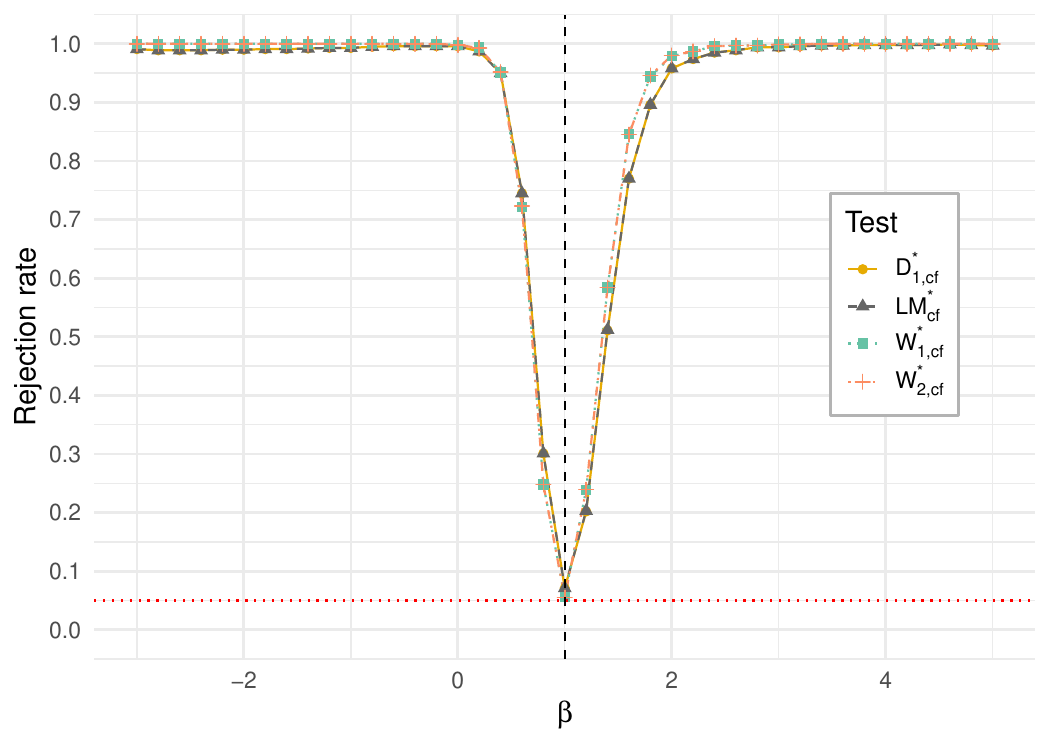}
        \caption{JIVE1}
    \end{subfigure}
    \hfill
    \begin{subfigure}[b]{0.47\textwidth}
        \includegraphics[width=\textwidth]{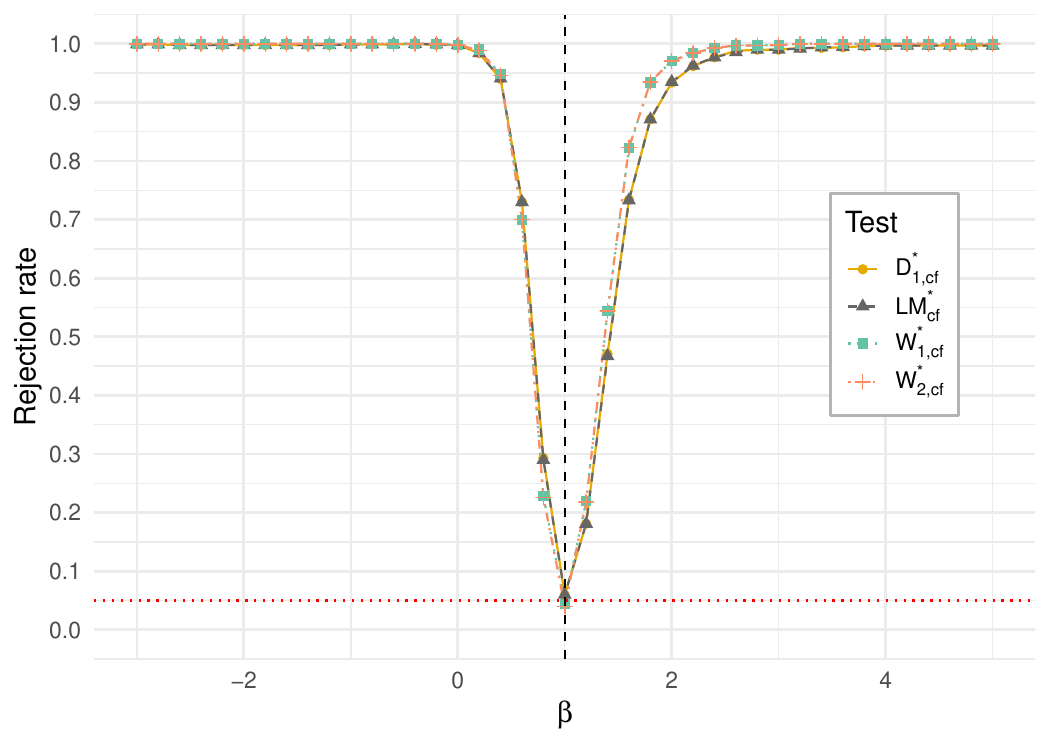}
        \caption{JIVE2}
    \end{subfigure}
    
    \caption{Power curves for DGP1 ($n=200$, $\alpha = 0.05$, $r = 64$). Trinity of test statistics distributed as a $\chi^2$ with cross-fit variance. Results based on 1000 repetitions. The horizontal dotted red line denotes the $5\%$ nominal rejection level, while the vertical dotted black line corresponds to $\beta=1$. Panel (a) plots statistics based on the SJIVE objective function; panel (b) plots statistics based on the HLIM objective function; panel (c) plots statistics based on the JIVE1 objective function; panel (d) plots statistics based on the JIVE2 objective function.}
    \label{fig:Trinity_chi2_cf_DGP1_200_0.05_64}
\end{figure}

%%%%%%%%%%%%%%%%%%%%%%%%%
% Trinity_200_0.1_32 

\begin{figure}[ht]
    \centering
    % Replace 'chibar2' with 'chi2' for the second set of figures as needed
    \begin{subfigure}[b]{0.47\textwidth}
        \includegraphics[width=\textwidth]{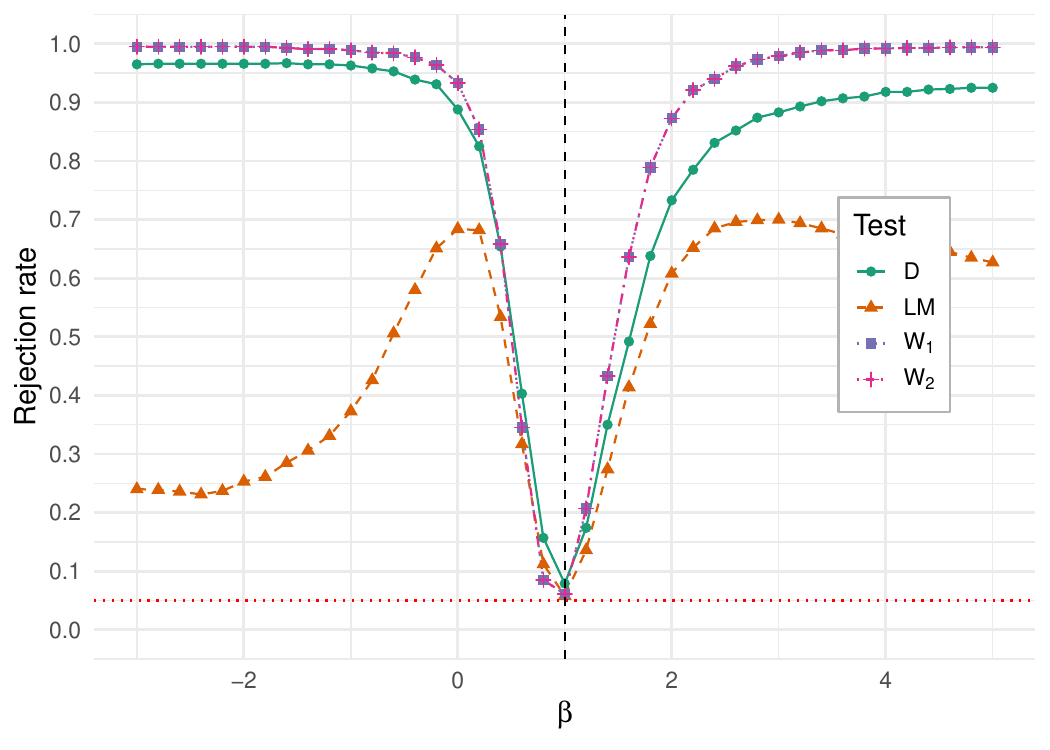}
        \caption{SJIVE}
    \end{subfigure}
    \hfill
    \begin{subfigure}[b]{0.47\textwidth}
        \includegraphics[width=\textwidth]{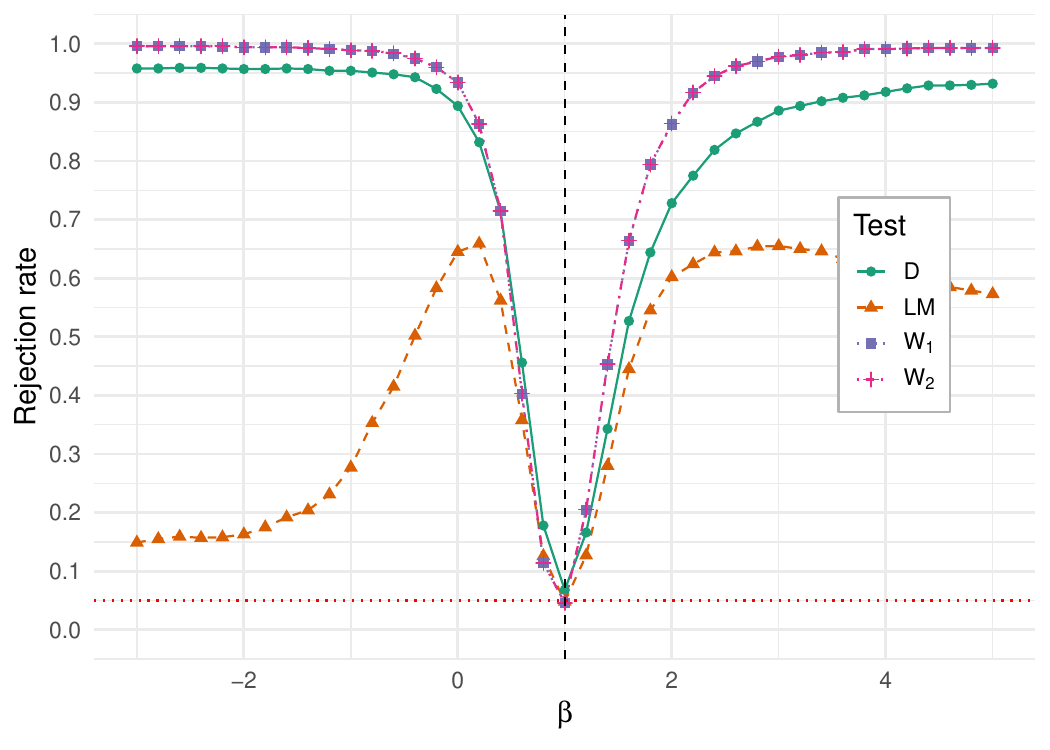}
        \caption{HLIM}
    \end{subfigure}
    
    \begin{subfigure}[b]{0.47\textwidth}
        \includegraphics[width=\textwidth]{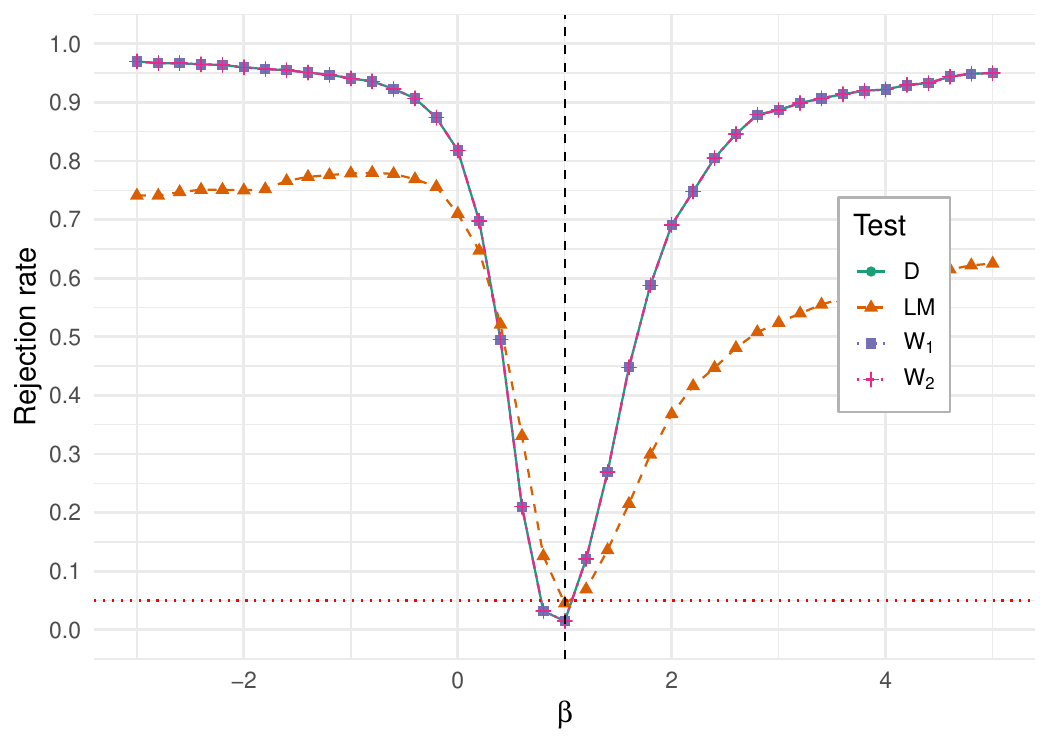}
        \caption{JIVE1}
    \end{subfigure}
    \hfill
    \begin{subfigure}[b]{0.47\textwidth}
        \includegraphics[width=\textwidth]{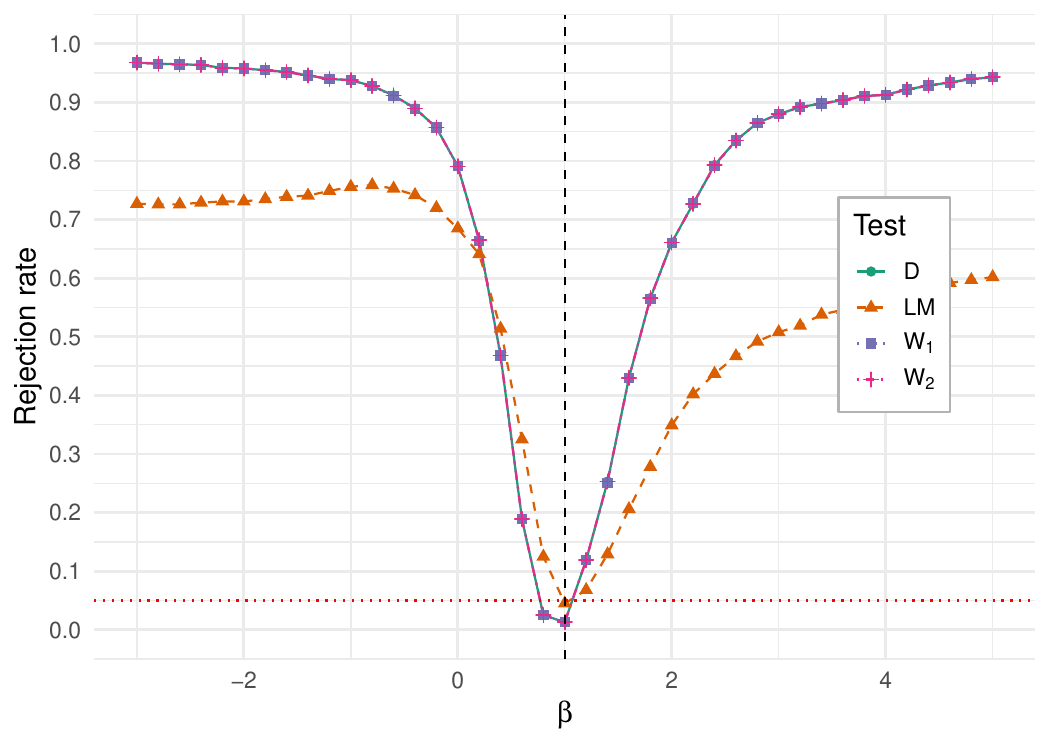}
        \caption{JIVE2}
    \end{subfigure}
    
    \caption{Power curves for DGP1 ($n=200$, $\alpha = 0.1$, $r = 32$). Trinity of test statistics distributed as a $\bar\chi^2$. Results based on 1000 repetitions. The horizontal dotted red line denotes the $5\%$ nominal rejection level, while the vertical dotted black line corresponds to $\beta=1$. Panel (a) plots statistics based on the SJIVE objective function; panel (b) plots statistics based on the HLIM objective function; panel (c) plots statistics based on the JIVE1 objective function; panel (d) plots statistics based on the JIVE2 objective function.}
    \label{fig:Trinity_chibar2_DGP1_200_0.1_32}
\end{figure}

\begin{figure}[ht]
    \centering
    % Replace 'chibar2' with 'chi2' for the second set of figures as needed
    \begin{subfigure}[b]{0.47\textwidth}
        \includegraphics[width=\textwidth]{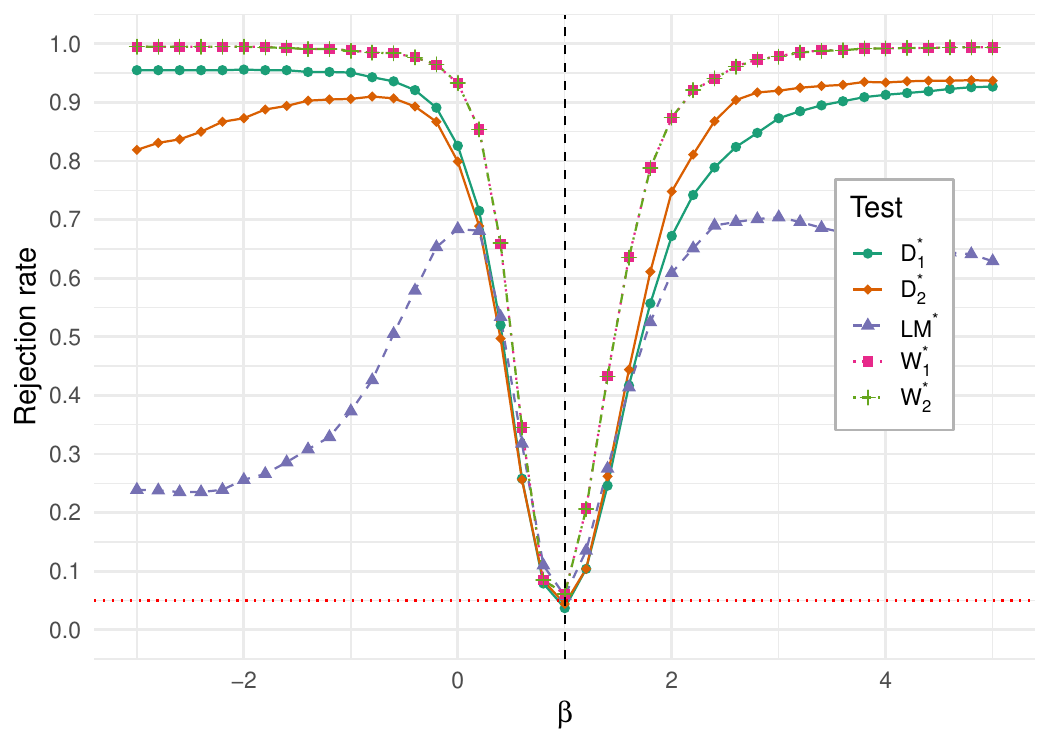}
        \caption{SJIVE}
    \end{subfigure}
    \hfill
    \begin{subfigure}[b]{0.47\textwidth}
        \includegraphics[width=\textwidth]{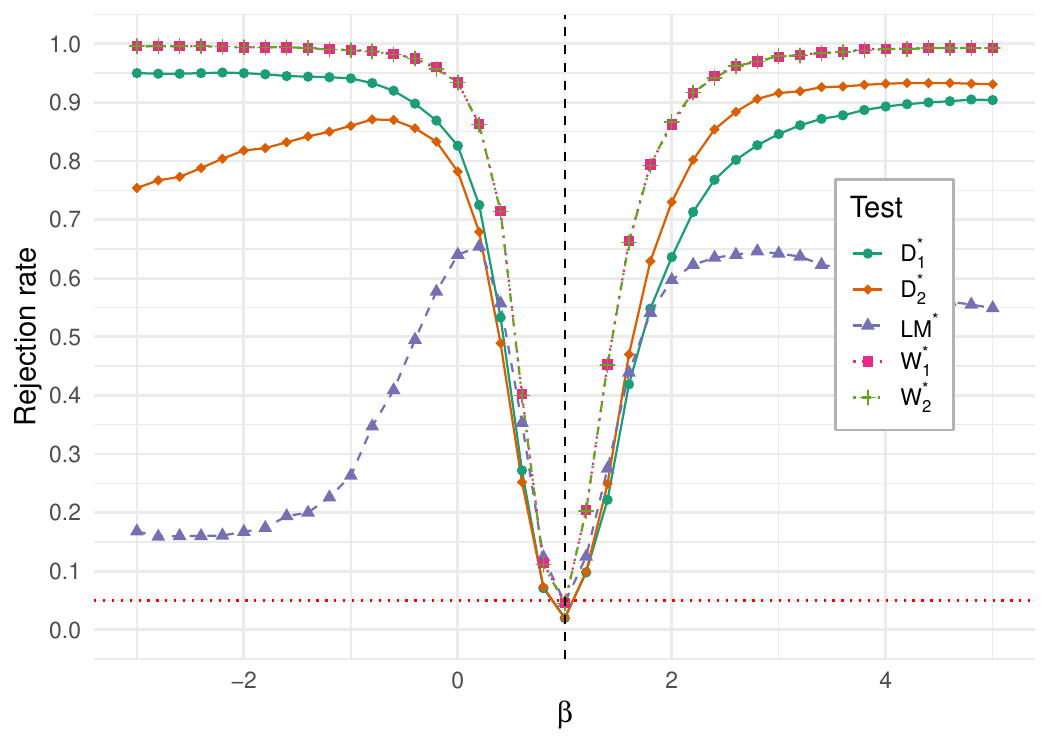}
        \caption{HLIM}
    \end{subfigure}
    
    \begin{subfigure}[b]{0.47\textwidth}
        \includegraphics[width=\textwidth]{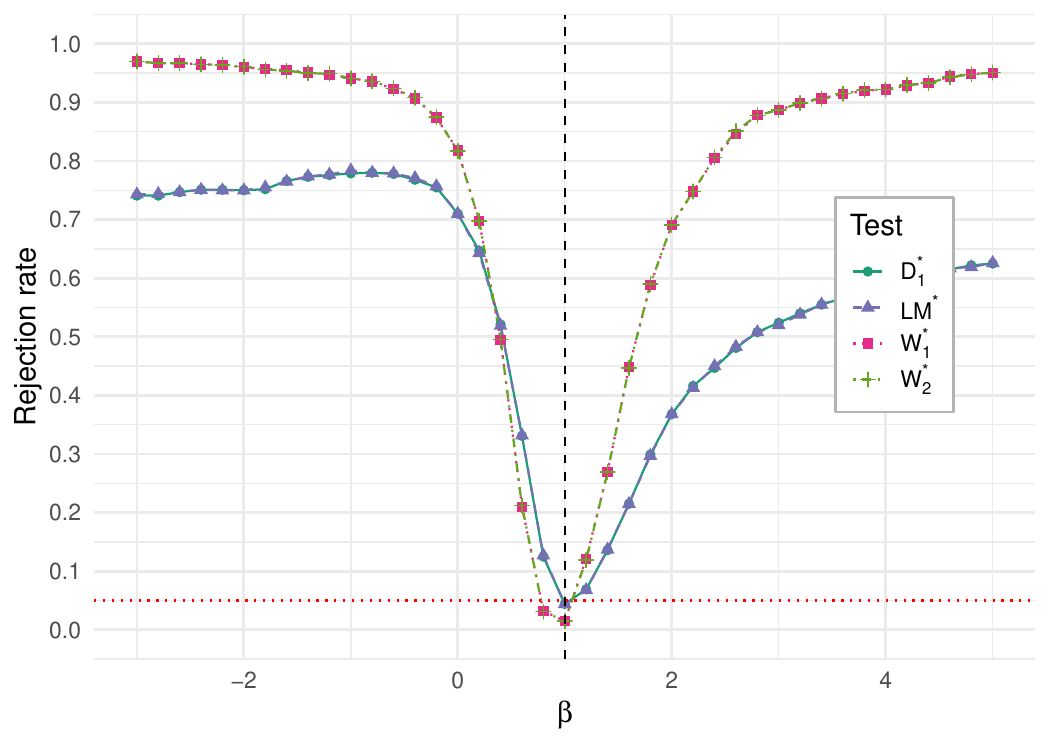}
        \caption{JIVE1}
    \end{subfigure}
    \hfill
    \begin{subfigure}[b]{0.47\textwidth}
        \includegraphics[width=\textwidth]{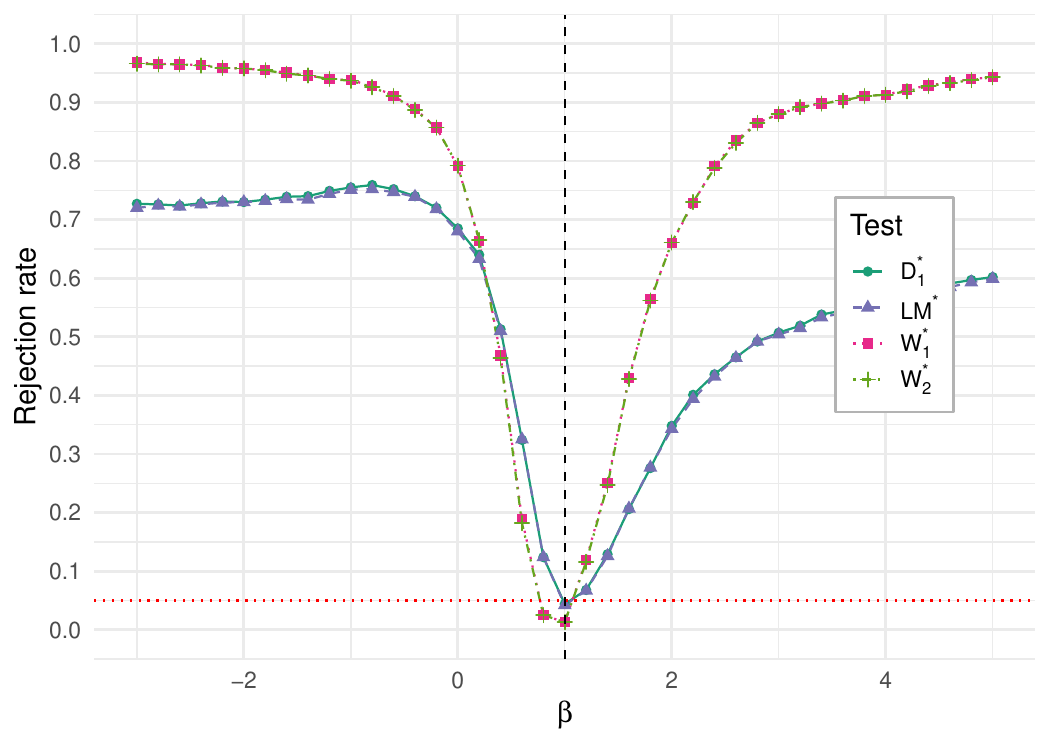}
        \caption{JIVE2}
    \end{subfigure}
    
    \caption{Power curves for DGP1 ($n=200$, $\alpha = 0.1$, $r = 32$). Trinity of test statistics distributed as a $\chi^2$. Results based on 1000 repetitions. The horizontal dotted red line denotes the $5\%$ nominal rejection level, while the vertical dotted black line corresponds to $\beta=1$. Panel (a) plots statistics based on the SJIVE objective function; panel (b) plots statistics based on the HLIM objective function; panel (c) plots statistics based on the JIVE1 objective function; panel (d) plots statistics based on the JIVE2 objective function.}
    \label{fig:Trinity_chi2_DGP1_200_0.1_32}
\end{figure}

% cf variance

\begin{figure}[ht]
    \centering
    % Replace 'chibar2' with 'chi2' for the second set of figures as needed
    \begin{subfigure}[b]{0.47\textwidth}
        \includegraphics[width=\textwidth]{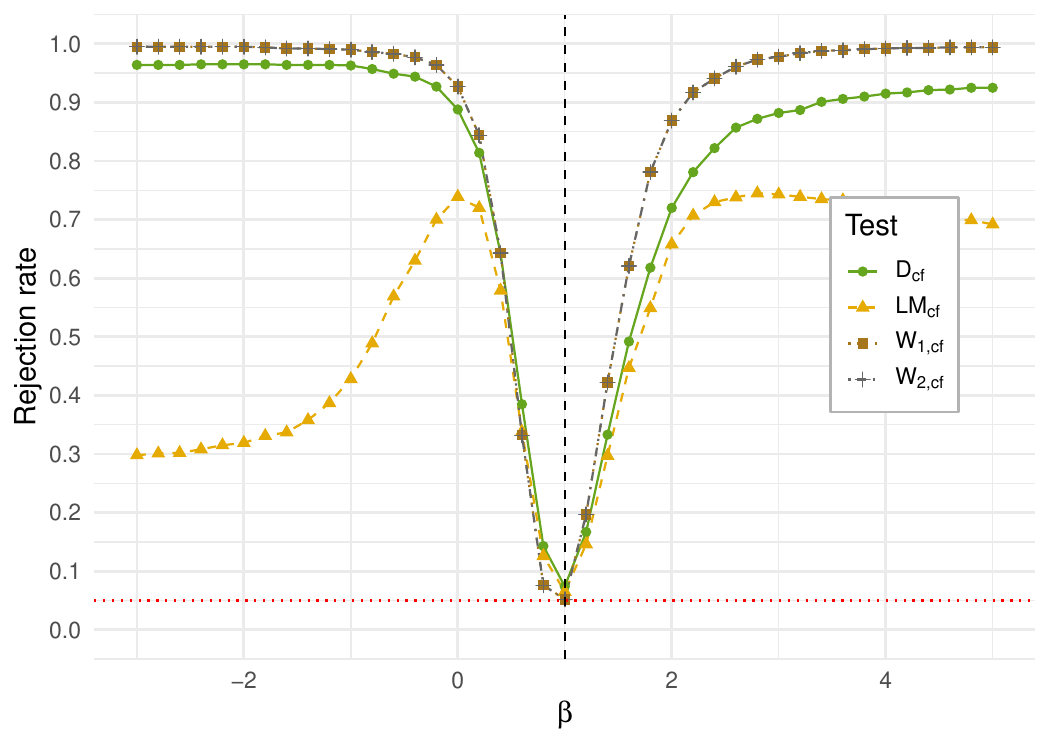}
        \caption{SJIVE}
    \end{subfigure}
    \hfill
    \begin{subfigure}[b]{0.47\textwidth}
        \includegraphics[width=\textwidth]{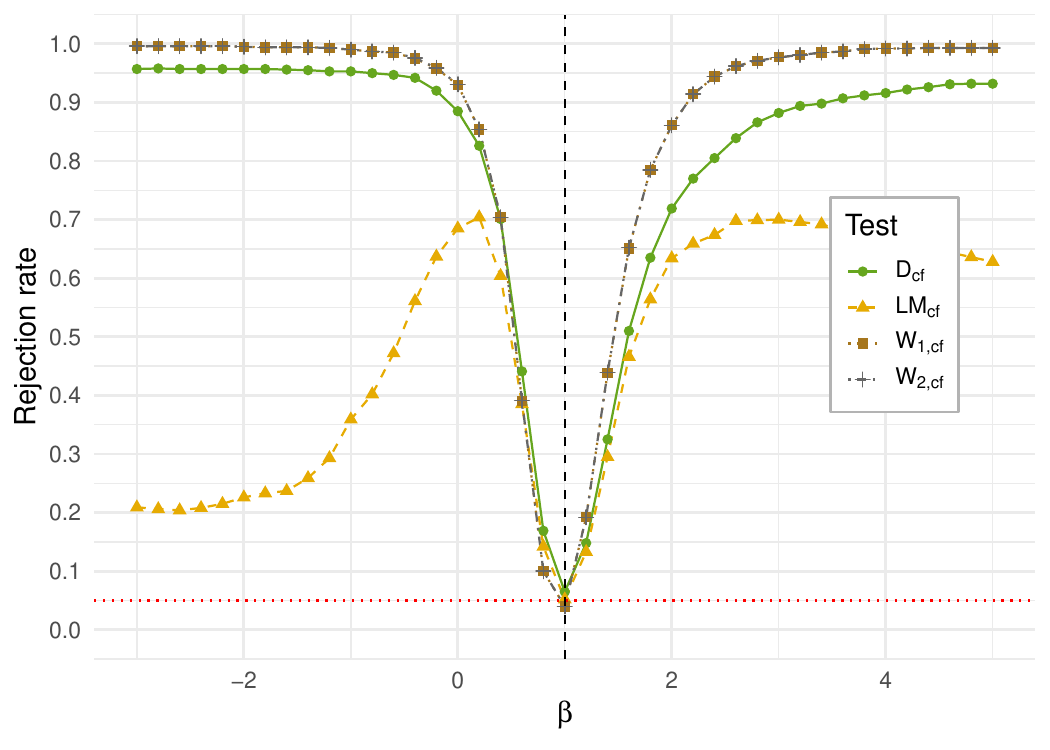}
        \caption{HLIM}
    \end{subfigure}
    
    \begin{subfigure}[b]{0.47\textwidth}
        \includegraphics[width=\textwidth]{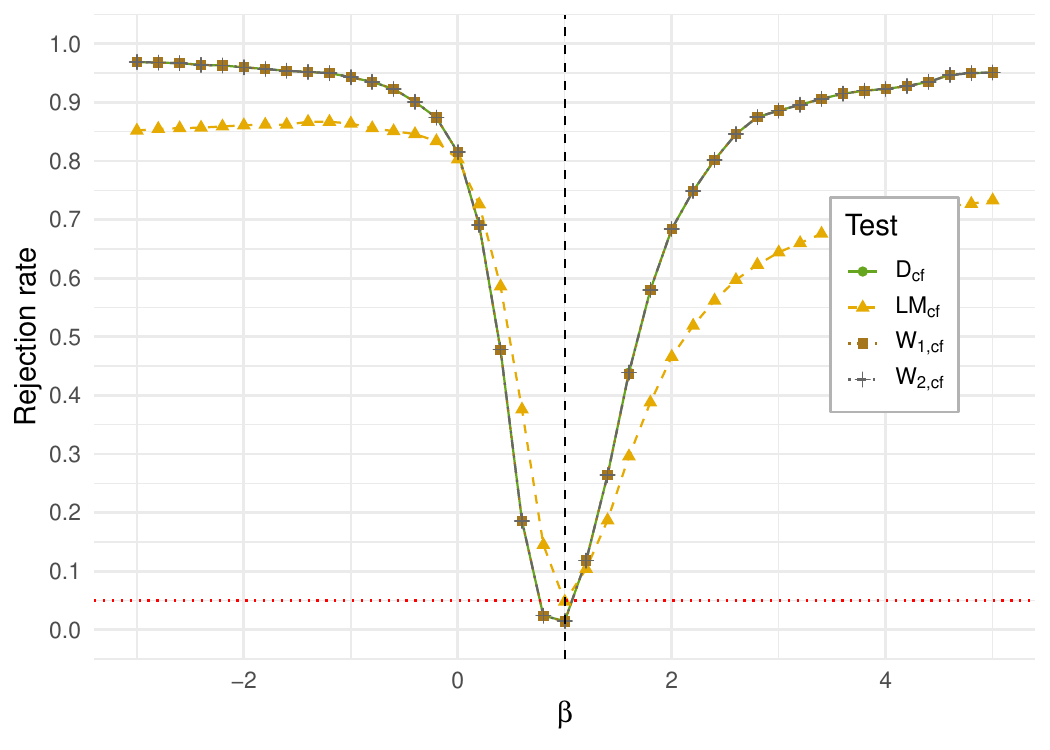}
        \caption{JIVE1}
    \end{subfigure}
    \hfill
    \begin{subfigure}[b]{0.47\textwidth}
        \includegraphics[width=\textwidth]{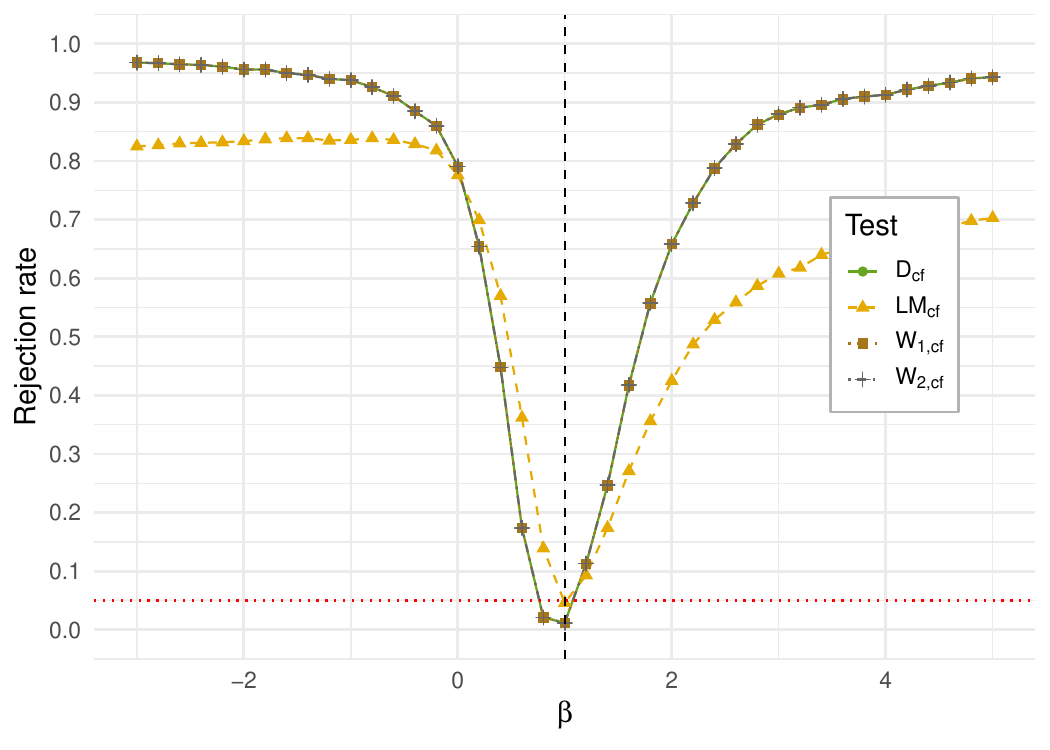}
        \caption{JIVE2}
    \end{subfigure}
    
    \caption{Power curves for DGP1 ($n=200$, $\alpha = 0.1$, $r = 32$). Trinity of test statistics distributed as a $\bar\chi^2$ with cross-fit variance. Results based on 1000 repetitions. The horizontal dotted red line denotes the $5\%$ nominal rejection level, while the vertical dotted black line corresponds to $\beta=1$. Panel (a) plots statistics based on the SJIVE objective function; panel (b) plots statistics based on the HLIM objective function; panel (c) plots statistics based on the JIVE1 objective function; panel (d) plots statistics based on the JIVE2 objective function.}
    \label{fig:Trinity_chibar2_cf_DGP1_200_0.1_32}
\end{figure}

\begin{figure}[ht]
    \centering
    % Replace 'chibar2' with 'chi2' for the second set of figures as needed
    \begin{subfigure}[b]{0.47\textwidth}
        \includegraphics[width=\textwidth]{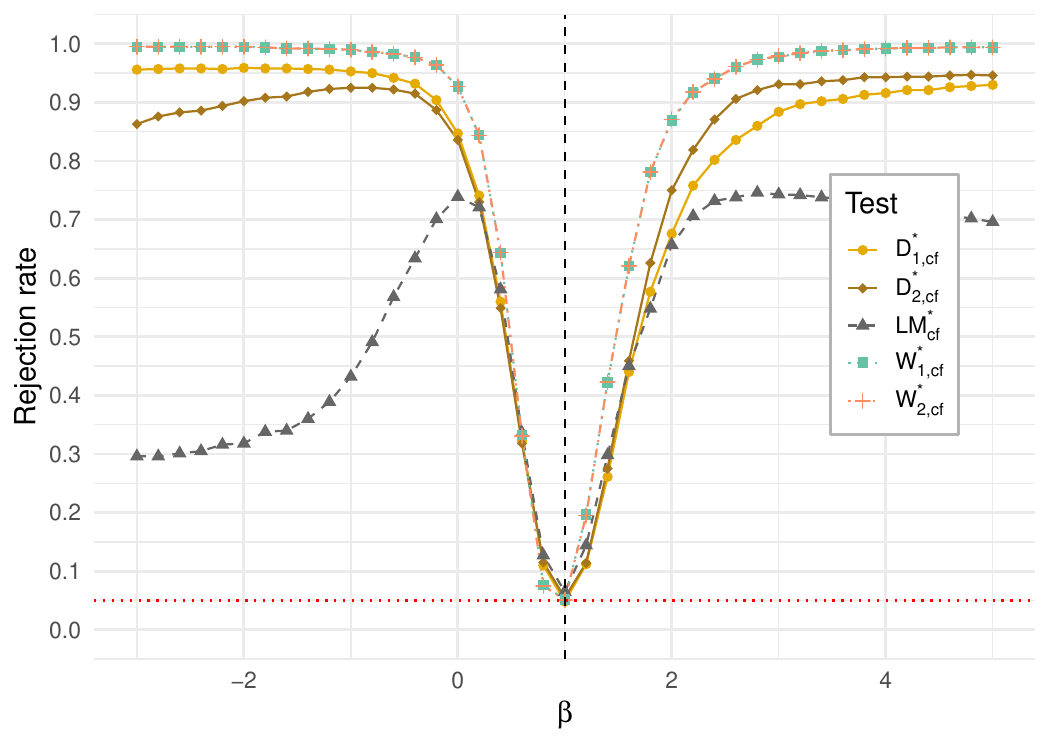}
        \caption{SJIVE}
    \end{subfigure}
    \hfill
    \begin{subfigure}[b]{0.47\textwidth}
        \includegraphics[width=\textwidth]{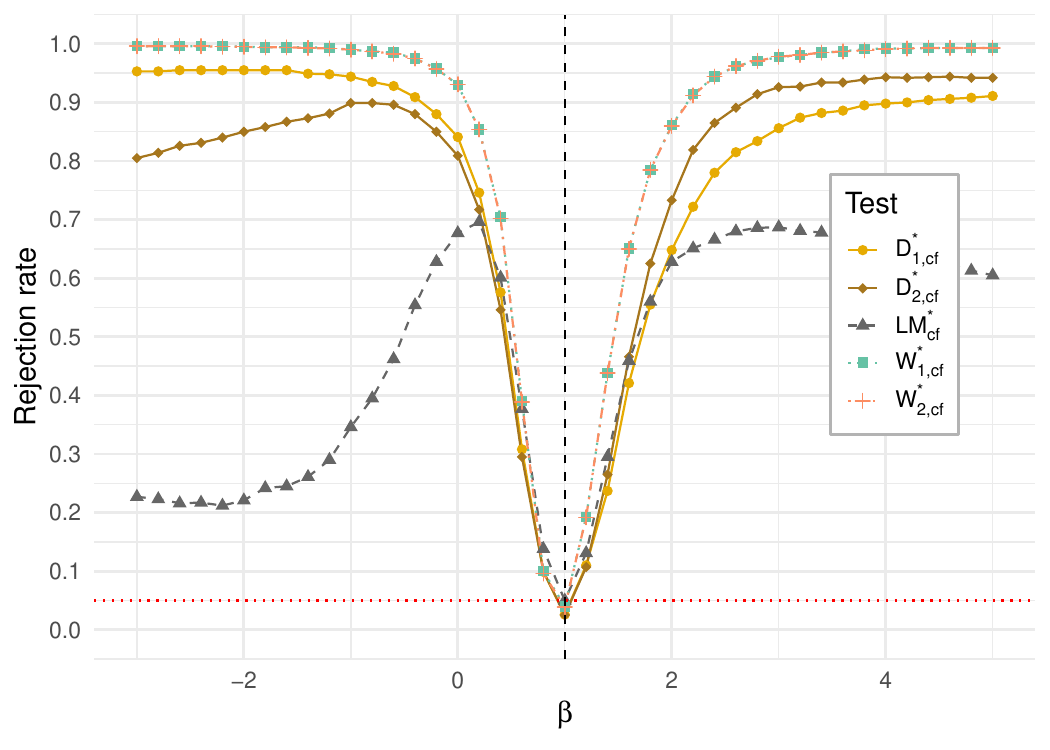}
        \caption{HLIM}
    \end{subfigure}
    
    \begin{subfigure}[b]{0.47\textwidth}
        \includegraphics[width=\textwidth]{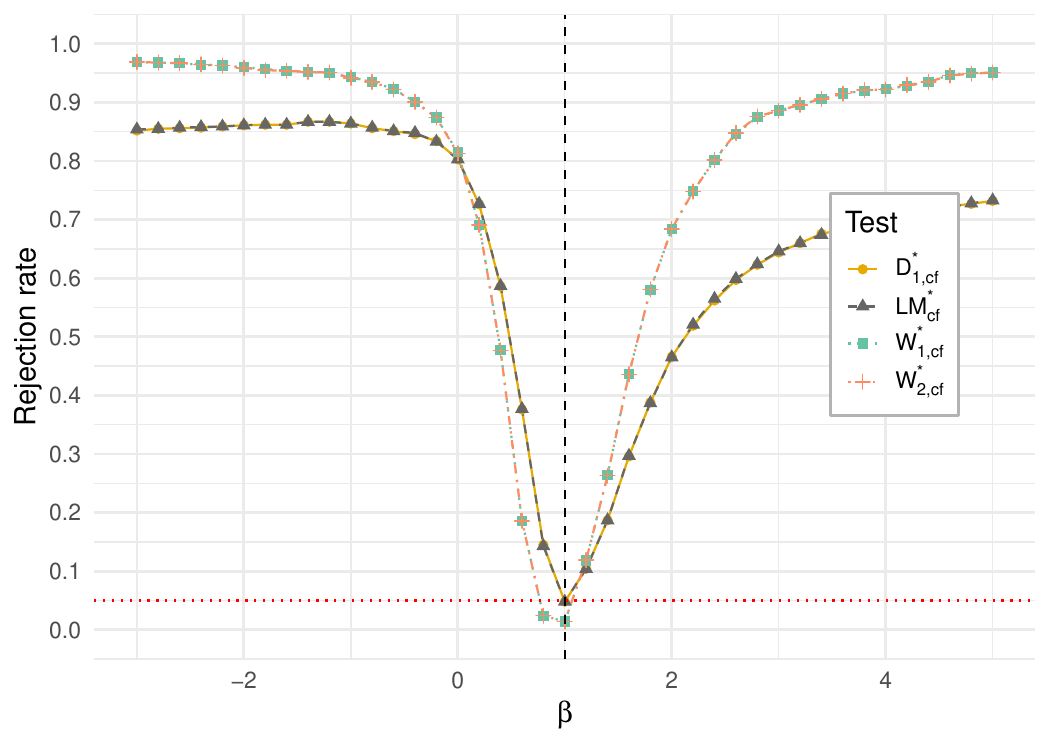}
        \caption{JIVE1}
    \end{subfigure}
    \hfill
    \begin{subfigure}[b]{0.47\textwidth}
        \includegraphics[width=\textwidth]{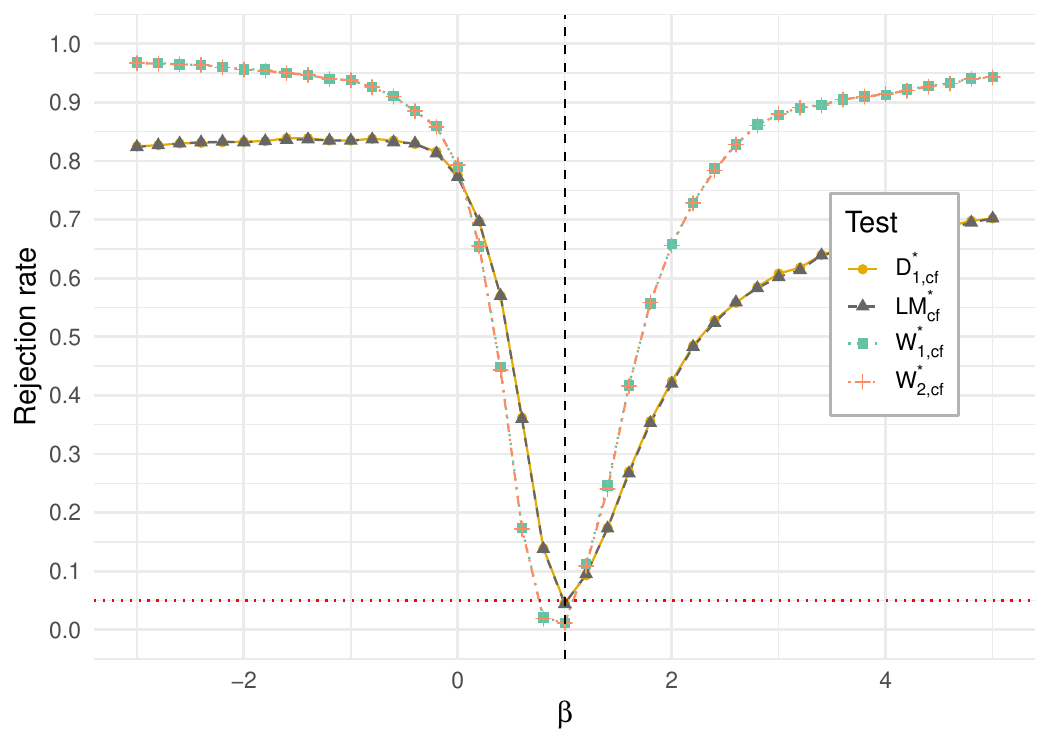}
        \caption{JIVE2}
    \end{subfigure}
    
    \caption{Power curves for DGP1 ($n=200$, $\alpha = 0.1$, $r = 32$). Trinity of test statistics distributed as a $\chi^2$ with cross-fit variance. Results based on 1000 repetitions. The horizontal dotted red line denotes the $5\%$ nominal rejection level, while the vertical dotted black line corresponds to $\beta=1$. Panel (a) plots statistics based on the SJIVE objective function; panel (b) plots statistics based on the HLIM objective function; panel (c) plots statistics based on the JIVE1 objective function; panel (d) plots statistics based on the JIVE2 objective function.}
    \label{fig:Trinity_chi2_cf_DGP1_200_0.1_32}
\end{figure}

%%%%%%%%%%%%%%%%%%%%%%%%%
% Trinity_200_0.1_64 

\begin{figure}[ht]
    \centering
    % Replace 'chibar2' with 'chi2' for the second set of figures as needed
    \begin{subfigure}[b]{0.47\textwidth}
        \includegraphics[width=\textwidth]{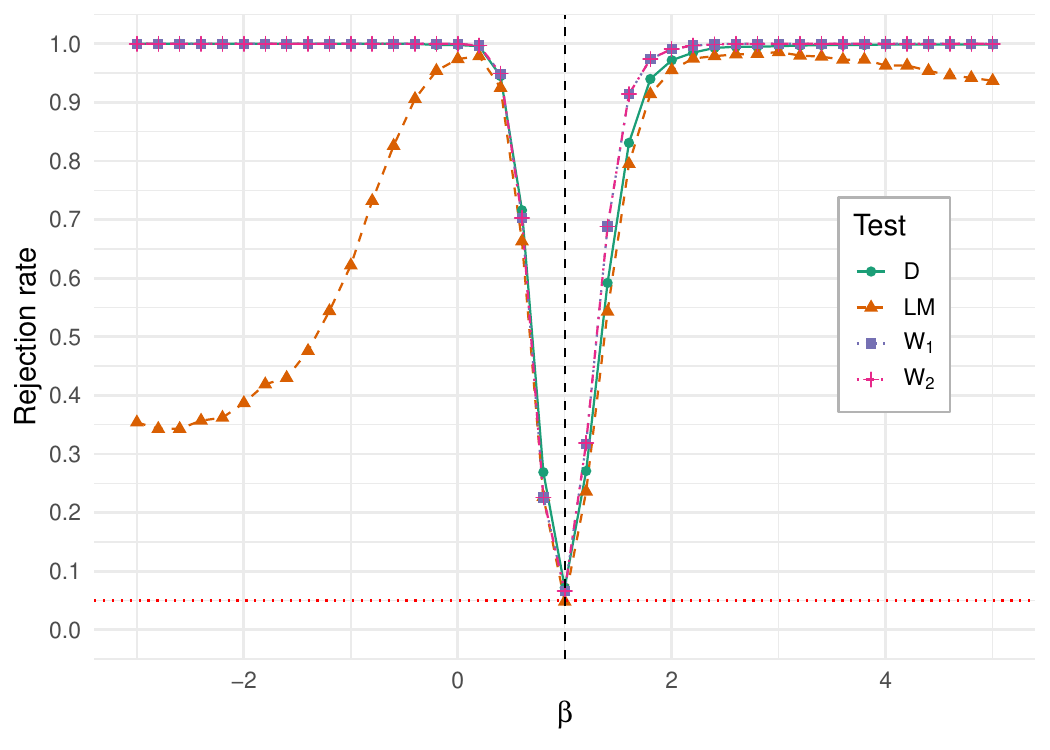}
        \caption{SJIVE}
    \end{subfigure}
    \hfill
    \begin{subfigure}[b]{0.47\textwidth}
        \includegraphics[width=\textwidth]{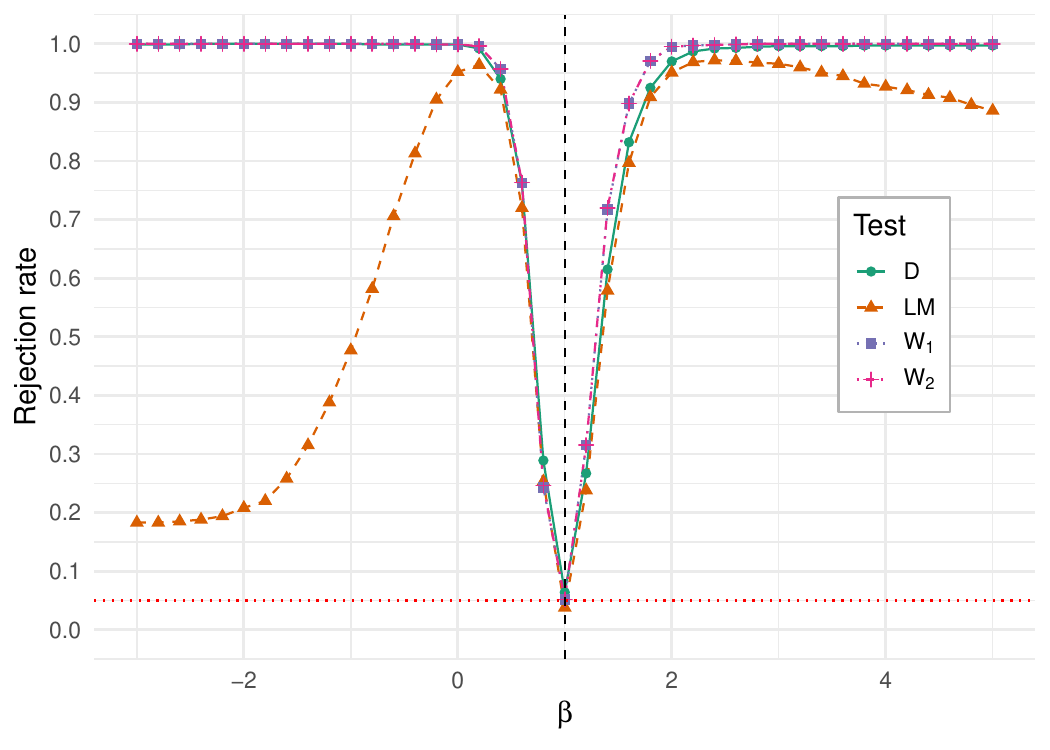}
        \caption{HLIM}
    \end{subfigure}
    
    \begin{subfigure}[b]{0.47\textwidth}
        \includegraphics[width=\textwidth]{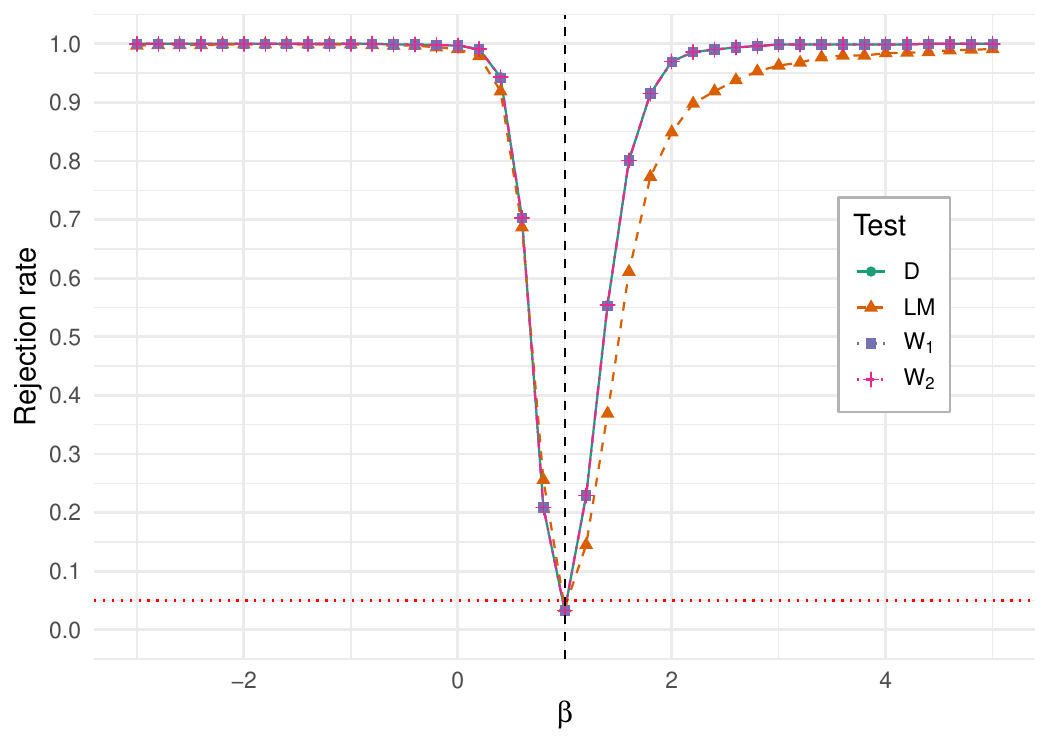}
        \caption{JIVE1}
    \end{subfigure}
    \hfill
    \begin{subfigure}[b]{0.47\textwidth}
        \includegraphics[width=\textwidth]{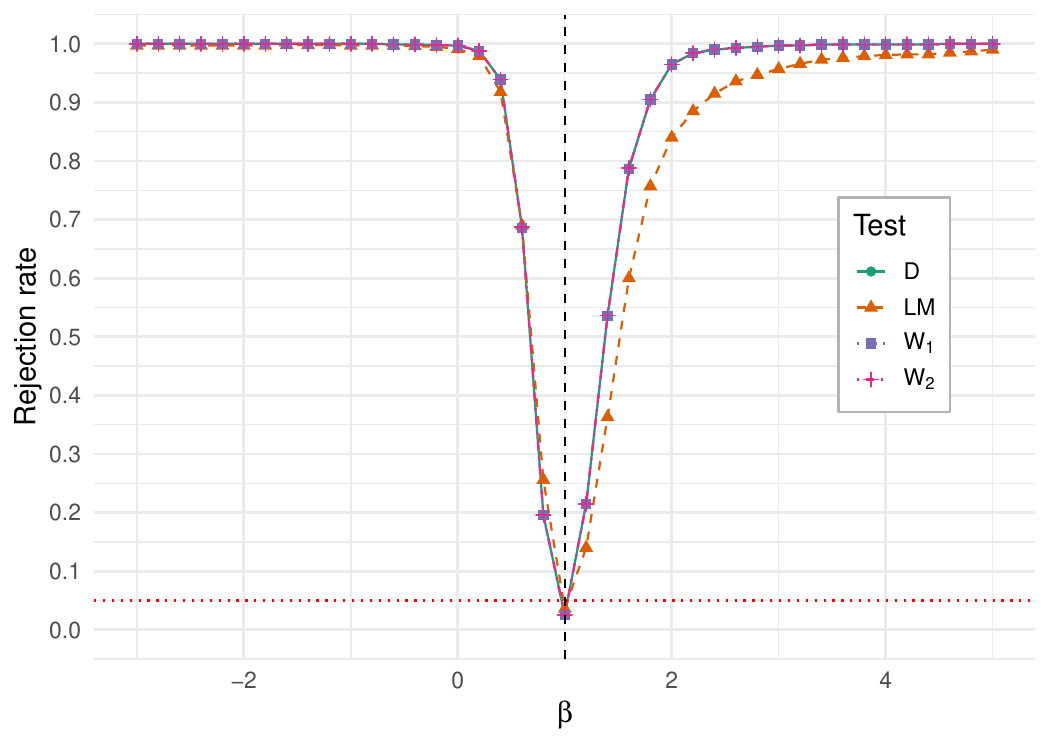}
        \caption{JIVE2}
    \end{subfigure}
    
    \caption{Power curves for DGP1 ($n=200$, $\alpha = 0.1$, $r = 64$). Trinity of test statistics distributed as a $\bar\chi^2$. Results based on 1000 repetitions. The horizontal dotted red line denotes the $5\%$ nominal rejection level, while the vertical dotted black line corresponds to $\beta=1$. Panel (a) plots statistics based on the SJIVE objective function; panel (b) plots statistics based on the HLIM objective function; panel (c) plots statistics based on the JIVE1 objective function; panel (d) plots statistics based on the JIVE2 objective function.}
    \label{fig:Trinity_chibar2_DGP1_200_0.1_64}
\end{figure}

\begin{figure}[ht]
    \centering
    % Replace 'chibar2' with 'chi2' for the second set of figures as needed
    \begin{subfigure}[b]{0.47\textwidth}
        \includegraphics[width=\textwidth]{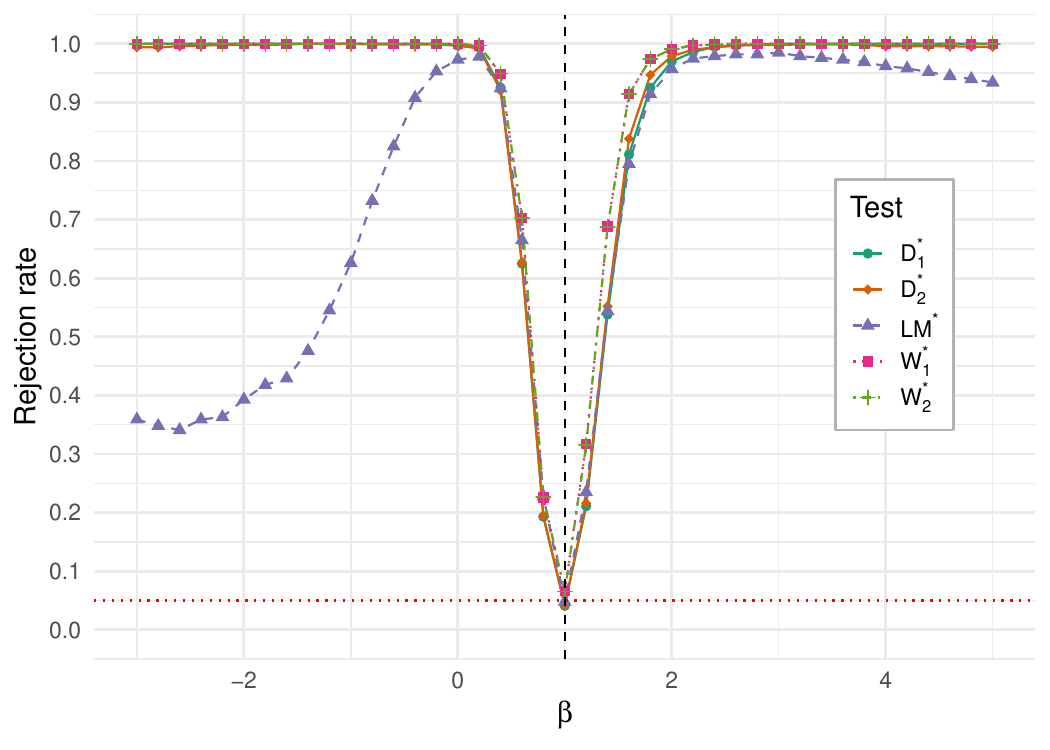}
        \caption{SJIVE}
    \end{subfigure}
    \hfill
    \begin{subfigure}[b]{0.47\textwidth}
        \includegraphics[width=\textwidth]{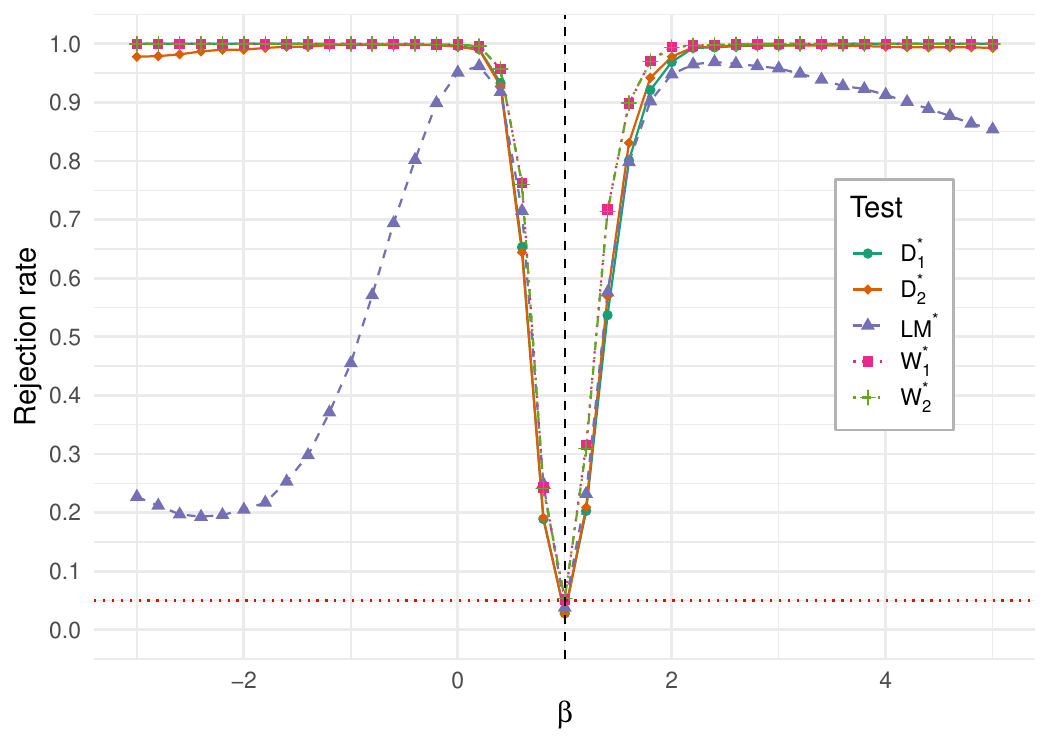}
        \caption{HLIM}
    \end{subfigure}
    
    \begin{subfigure}[b]{0.47\textwidth}
        \includegraphics[width=\textwidth]{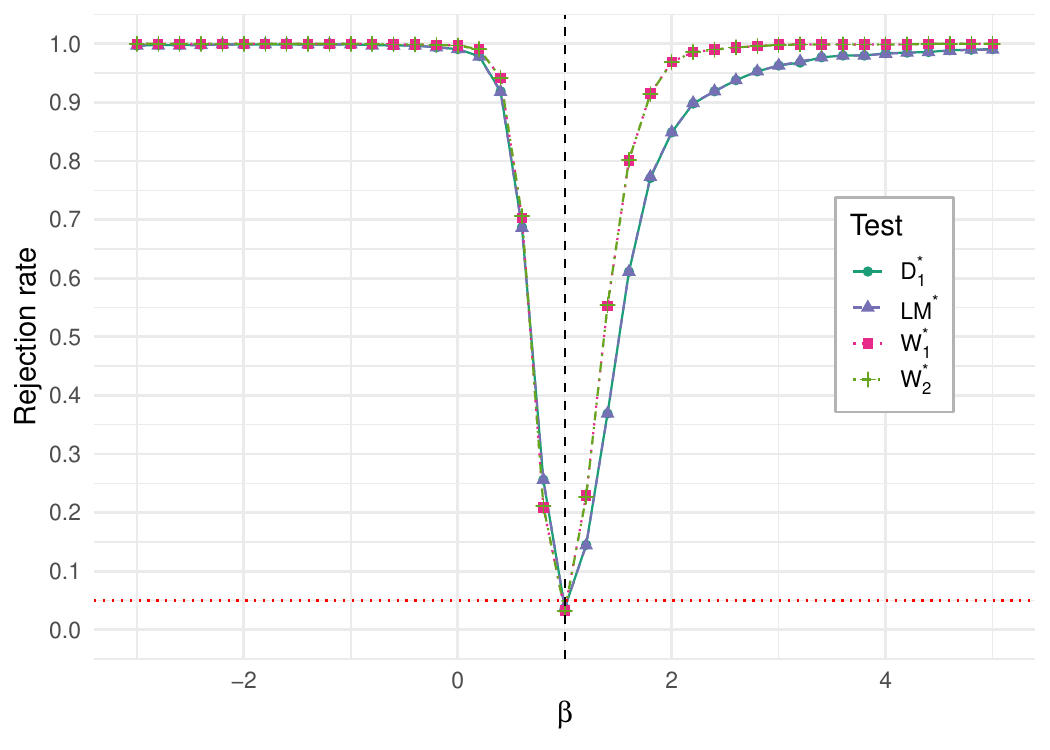}
        \caption{JIVE1}
    \end{subfigure}
    \hfill
    \begin{subfigure}[b]{0.47\textwidth}
        \includegraphics[width=\textwidth]{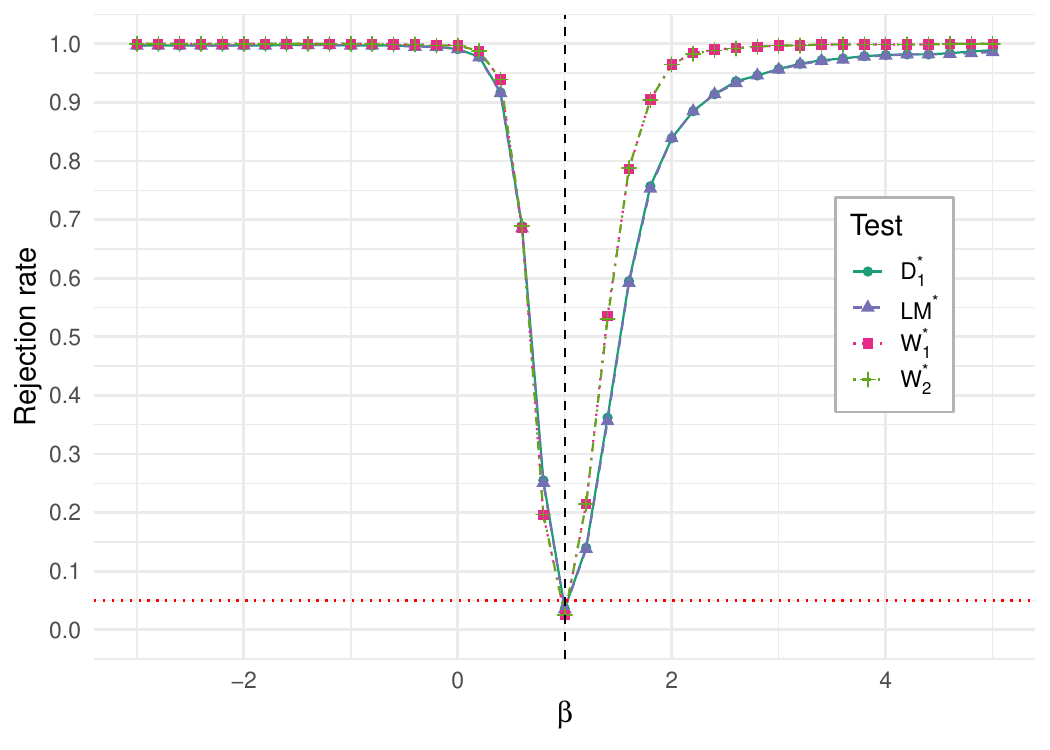}
        \caption{JIVE2}
    \end{subfigure}
    
    \caption{Power curves for DGP1 ($n=200$, $\alpha = 0.1$, $r = 64$). Trinity of test statistics distributed as a $\chi^2$. Results based on 1000 repetitions. The horizontal dotted red line denotes the $5\%$ nominal rejection level, while the vertical dotted black line corresponds to $\beta=1$. Panel (a) plots statistics based on the SJIVE objective function; panel (b) plots statistics based on the HLIM objective function; panel (c) plots statistics based on the JIVE1 objective function; panel (d) plots statistics based on the JIVE2 objective function.}
    \label{fig:Trinity_chi2_DGP1_200_0.1_64}
\end{figure}

% cf variance

\begin{figure}[ht]
    \centering
    % Replace 'chibar2' with 'chi2' for the second set of figures as needed
    \begin{subfigure}[b]{0.47\textwidth}
        \includegraphics[width=\textwidth]{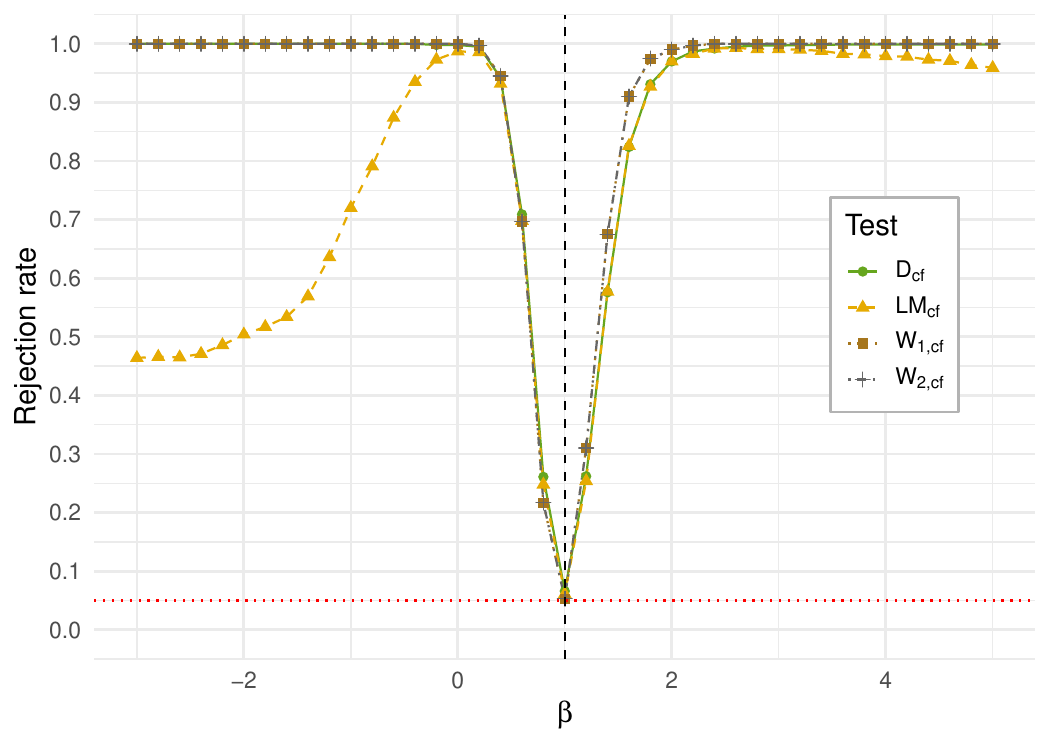}
        \caption{SJIVE}
    \end{subfigure}
    \hfill
    \begin{subfigure}[b]{0.47\textwidth}
        \includegraphics[width=\textwidth]{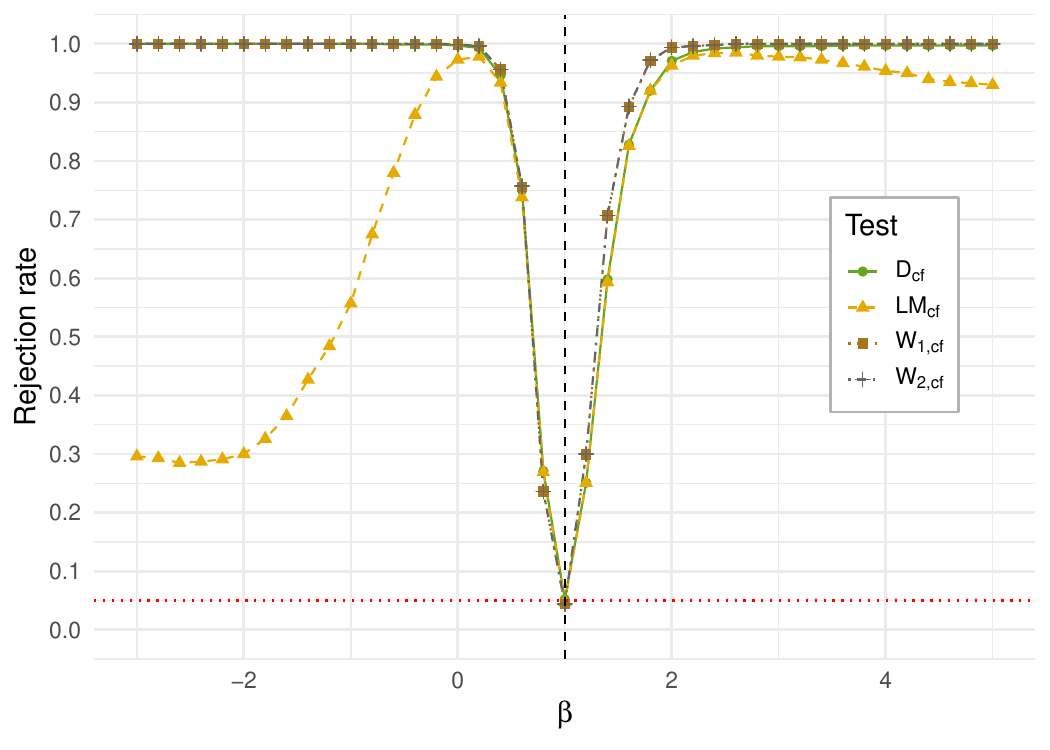}
        \caption{HLIM}
    \end{subfigure}
    
    \begin{subfigure}[b]{0.47\textwidth}
        \includegraphics[width=\textwidth]{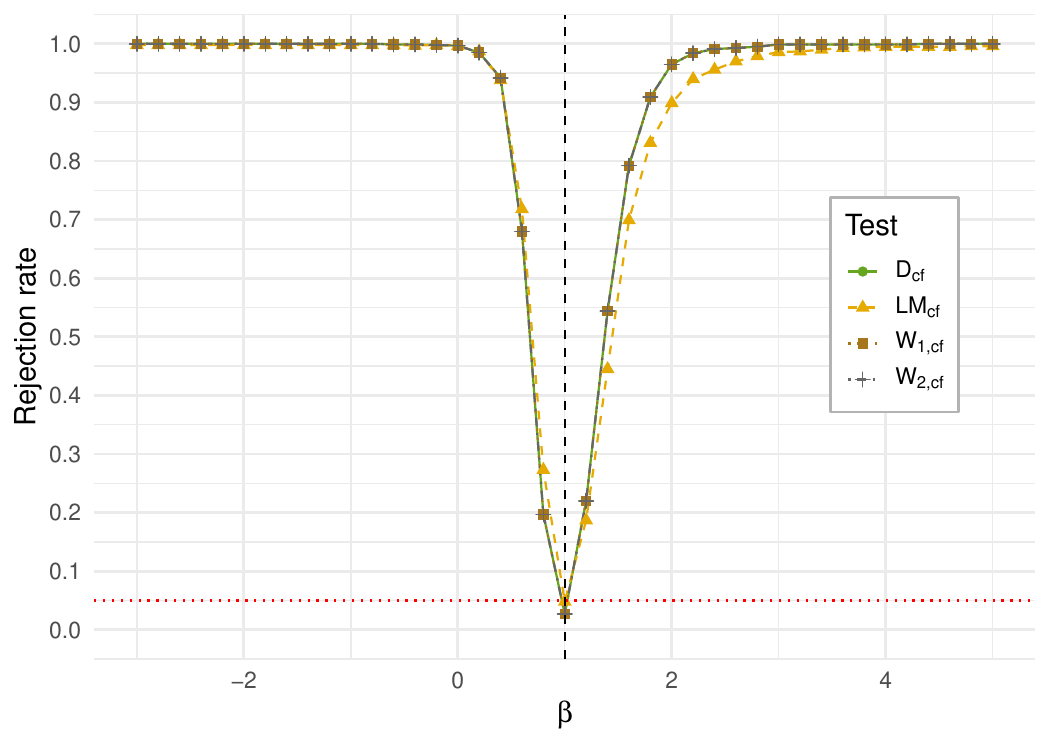}
        \caption{JIVE1}
    \end{subfigure}
    \hfill
    \begin{subfigure}[b]{0.47\textwidth}
        \includegraphics[width=\textwidth]{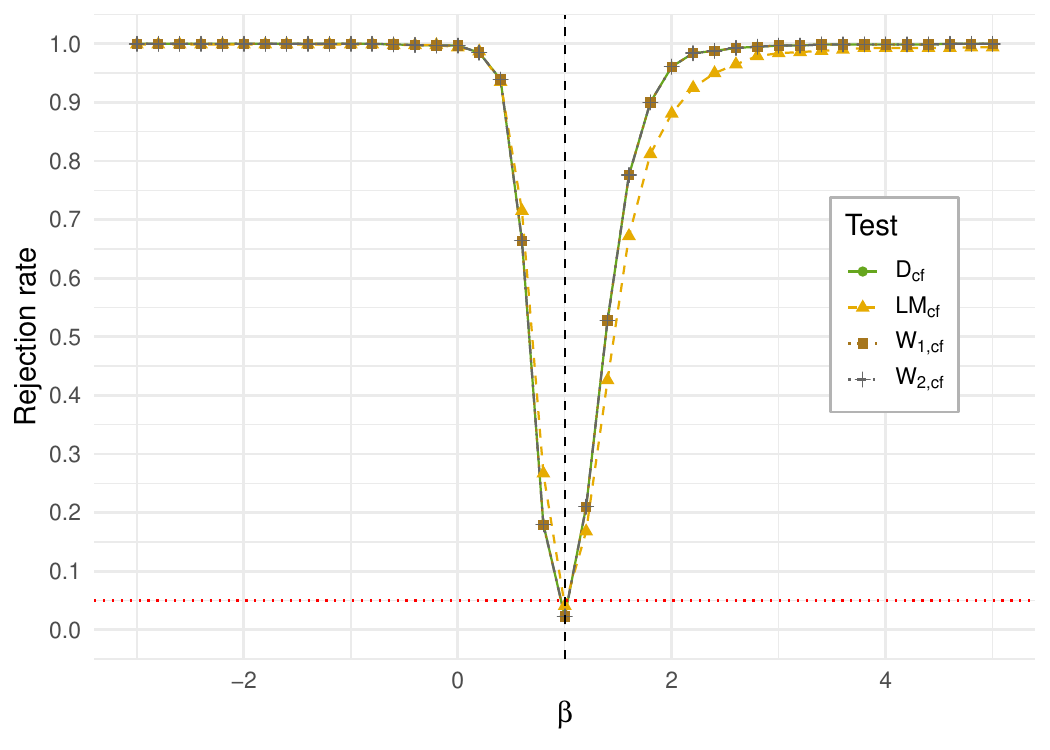}
        \caption{JIVE2}
    \end{subfigure}
    
    \caption{Power curves for DGP1 ($n=200$, $\alpha = 0.1$, $r = 64$). Trinity of test statistics distributed as a $\bar\chi^2$ with cross-fit variance. Results based on 1000 repetitions. The horizontal dotted red line denotes the $5\%$ nominal rejection level, while the vertical dotted black line corresponds to $\beta=1$. Panel (a) plots statistics based on the SJIVE objective function; panel (b) plots statistics based on the HLIM objective function; panel (c) plots statistics based on the JIVE1 objective function; panel (d) plots statistics based on the JIVE2 objective function.}
    \label{fig:Trinity_chibar2_cf_DGP1_200_0.1_64}
\end{figure}

\begin{figure}[ht]
    \centering
    % Replace 'chibar2' with 'chi2' for the second set of figures as needed
    \begin{subfigure}[b]{0.47\textwidth}
        \includegraphics[width=\textwidth]{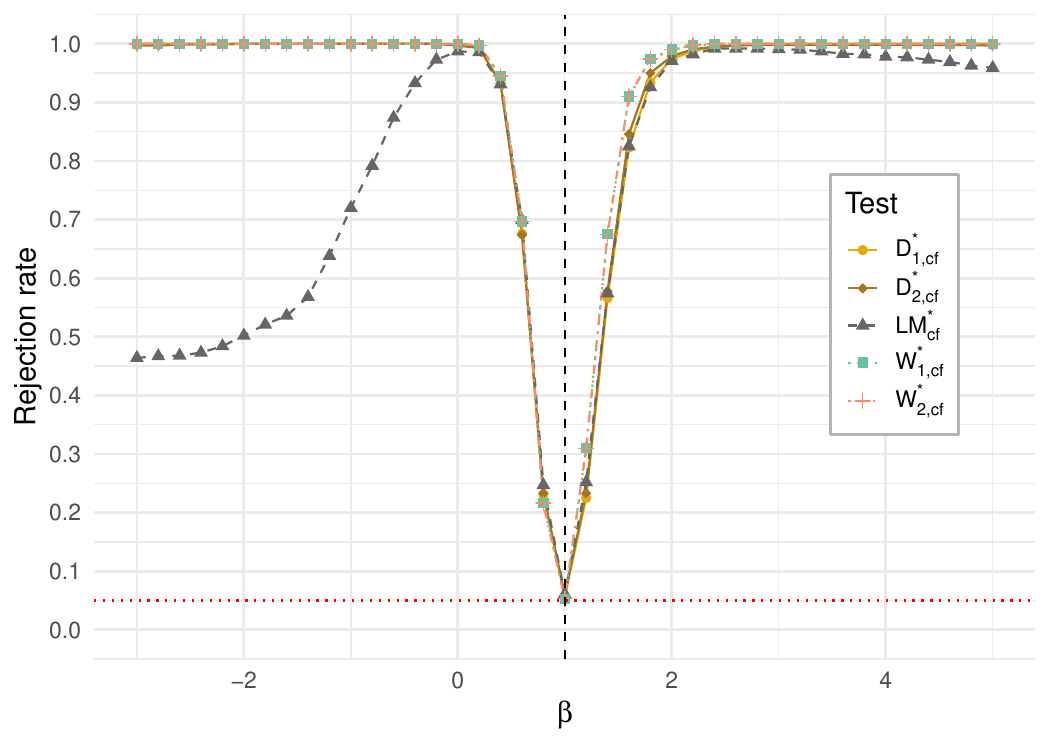}
        \caption{SJIVE}
    \end{subfigure}
    \hfill
    \begin{subfigure}[b]{0.47\textwidth}
        \includegraphics[width=\textwidth]{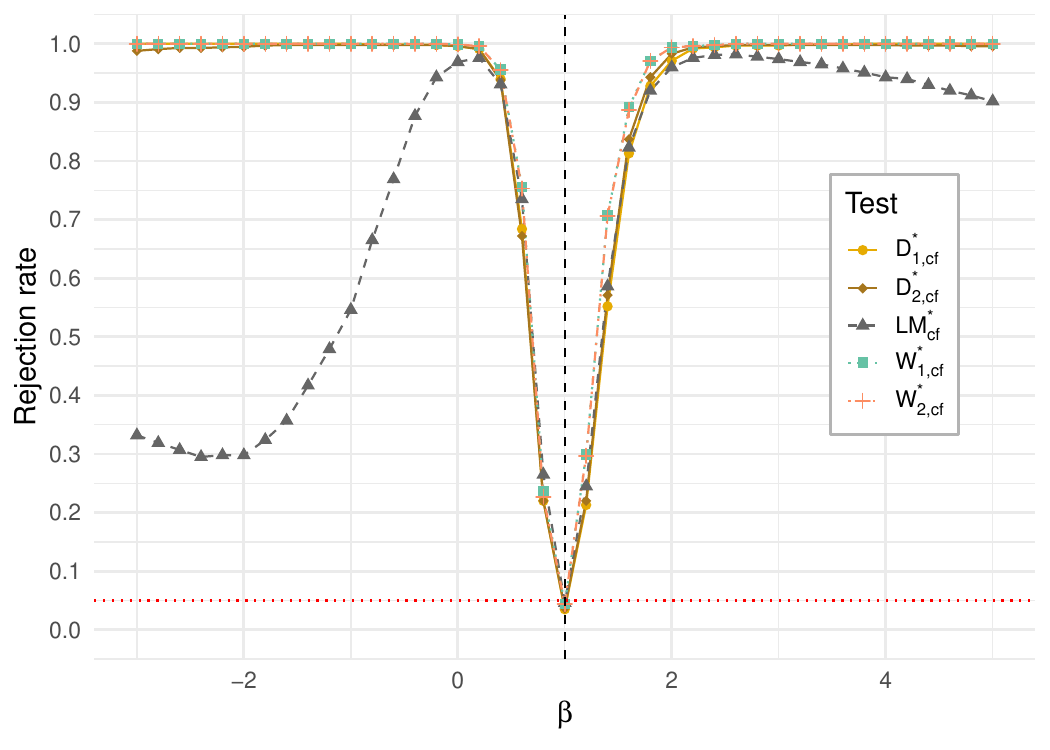}
        \caption{HLIM}
    \end{subfigure}
    
    \begin{subfigure}[b]{0.47\textwidth}
        \includegraphics[width=\textwidth]{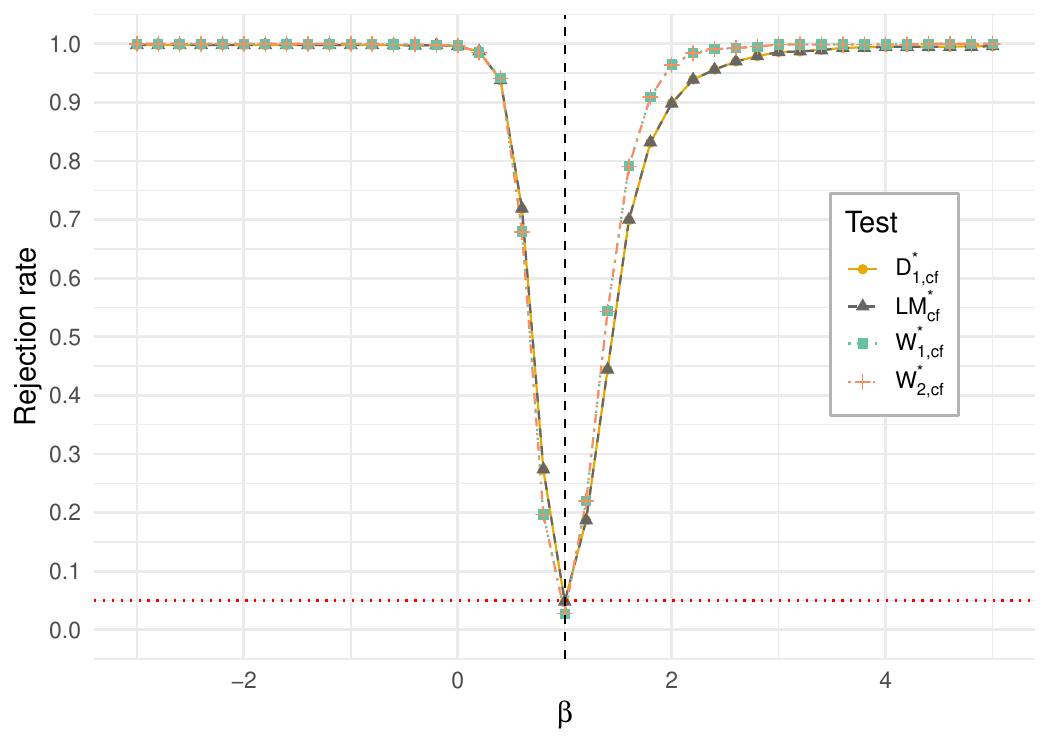}
        \caption{JIVE1}
    \end{subfigure}
    \hfill
    \begin{subfigure}[b]{0.47\textwidth}
        \includegraphics[width=\textwidth]{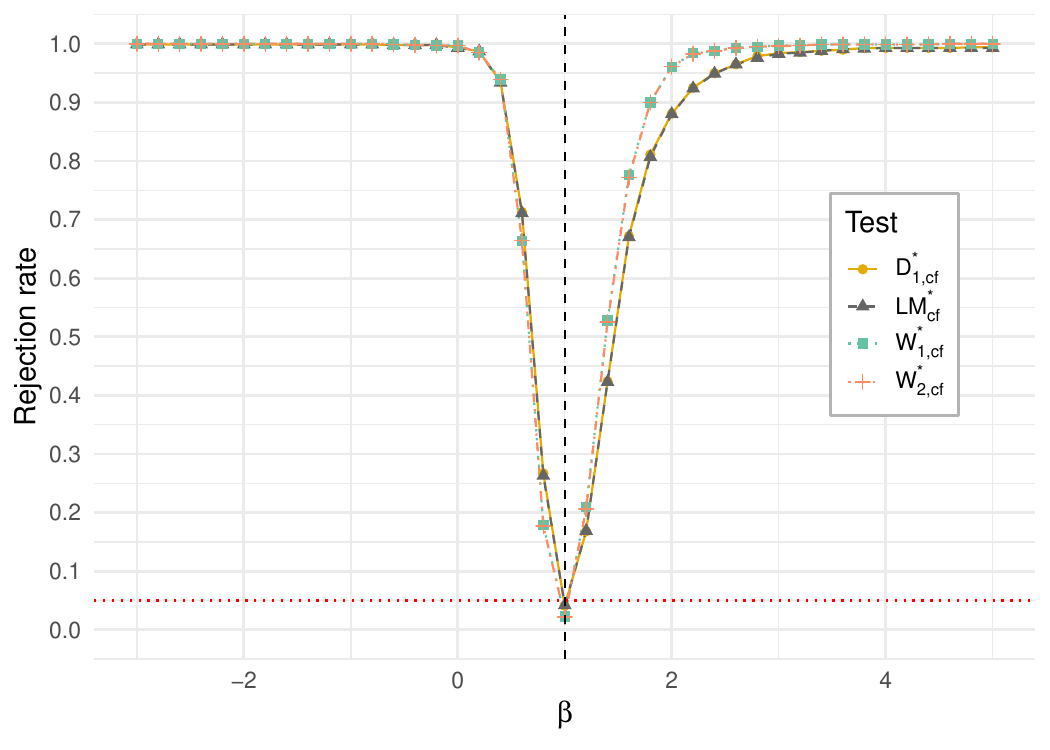}
        \caption{JIVE2}
    \end{subfigure}
    
    \caption{Power curves for DGP1 ($n=200$, $\alpha = 0.1$, $r = 64$). Trinity of test statistics distributed as a $\chi^2$ with cross-fit variance. Results based on 1000 repetitions. The horizontal dotted red line denotes the $5\%$ nominal rejection level, while the vertical dotted black line corresponds to $\beta=1$. Panel (a) plots statistics based on the SJIVE objective function; panel (b) plots statistics based on the HLIM objective function; panel (c) plots statistics based on the JIVE1 objective function; panel (d) plots statistics based on the JIVE2 objective function.}
    \label{fig:Trinity_chi2_cf_DGP1_200_0.1_64}
\end{figure}

%%%%%%%%%%%%%%%%%%%%%%%%%%%%%%%%%%%%%%%%%%%%%%%%%%%%%%%%%%%%%%%%%%%%%%%%%%%%%%%%%%
%%%%%%%%%%%%%%%%%%%%%%%%%%%%%%%%%%%%%%%%%%%%%%%%%%%%%%%%%%%%%%%%%%%%%%%%%%%%%%%%%%

% AR_200_0.05 

\begin{figure}[ht]
    \centering
    \begin{subfigure}[b]{0.47\textwidth}
        \includegraphics[width=\textwidth]{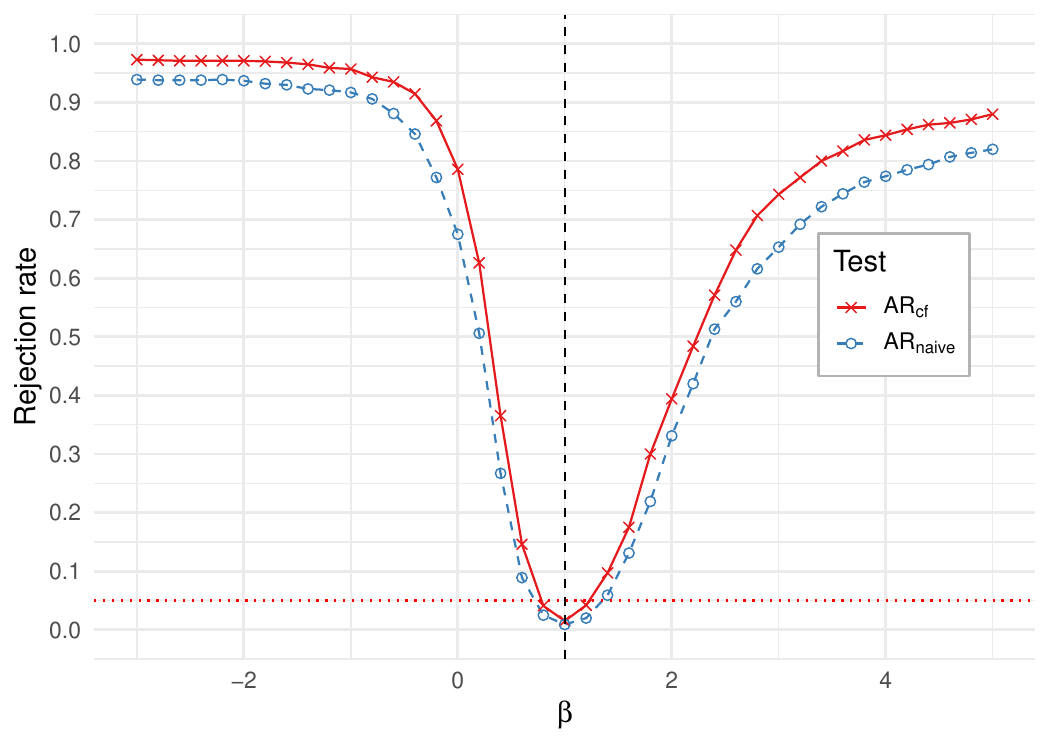}
        \caption{JIVE1, $r=32$}
    \end{subfigure}
    \hfill
    \begin{subfigure}[b]{0.47\textwidth}
        \includegraphics[width=\textwidth]{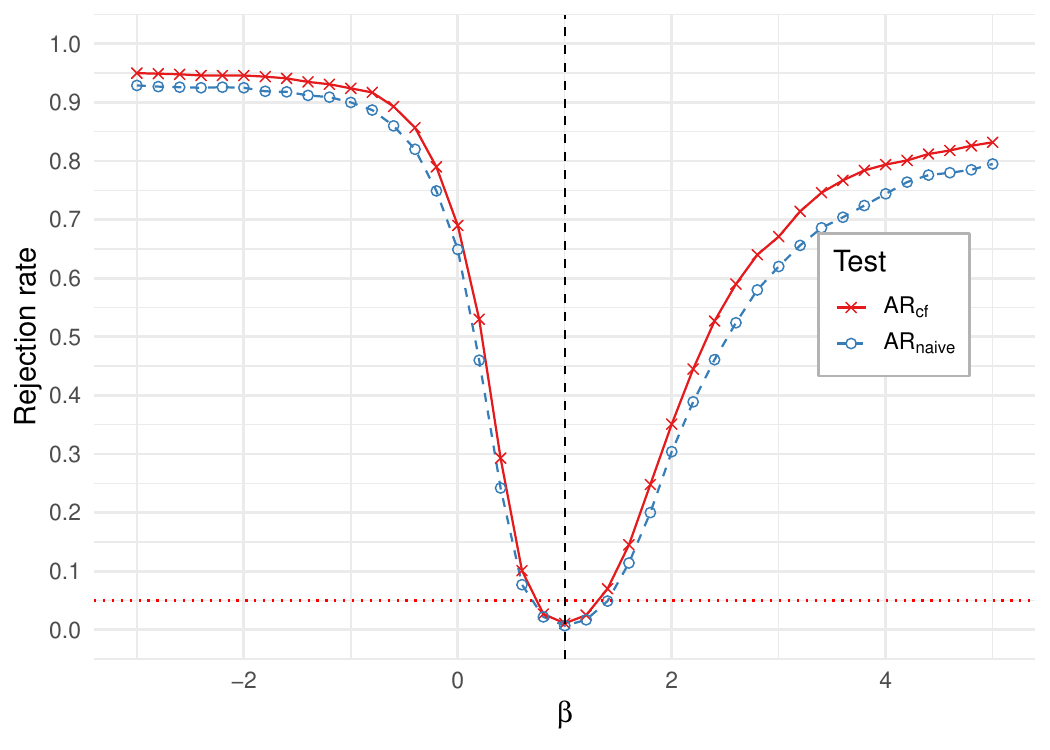}
        \caption{JIVE2, $r=32$}
    \end{subfigure}
    
    \begin{subfigure}[b]{0.47\textwidth}
        \includegraphics[width=\textwidth]{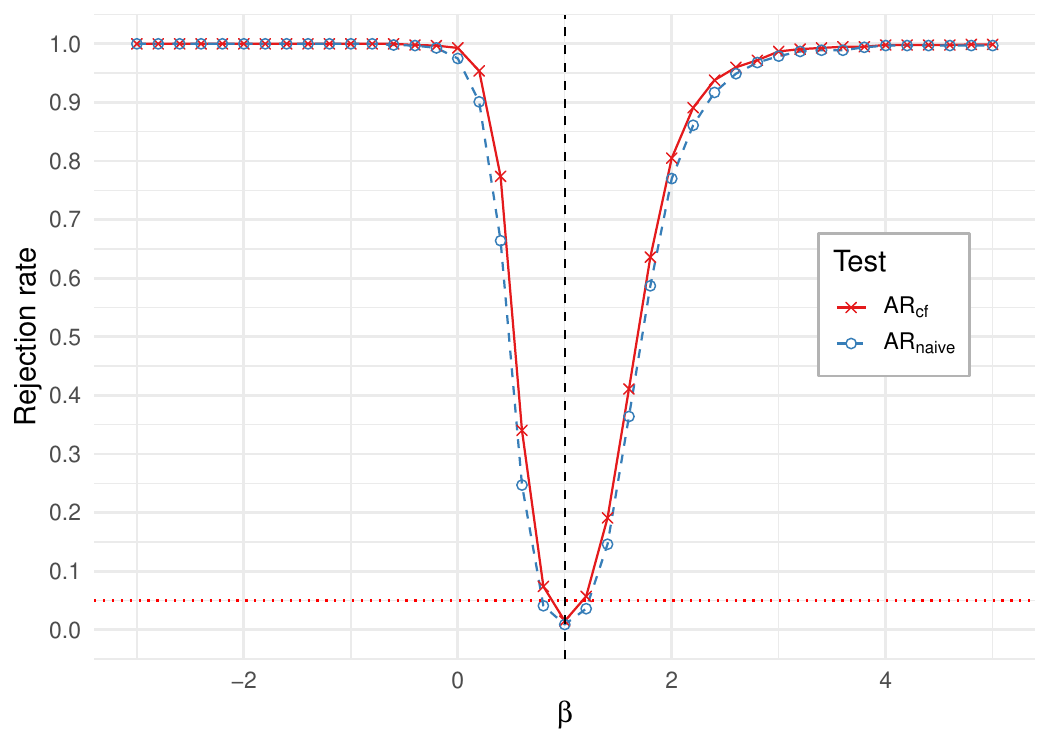}
        \caption{JIVE1, $r=64$}
    \end{subfigure}
    \hfill
    \begin{subfigure}[b]{0.47\textwidth}
        \includegraphics[width=\textwidth]{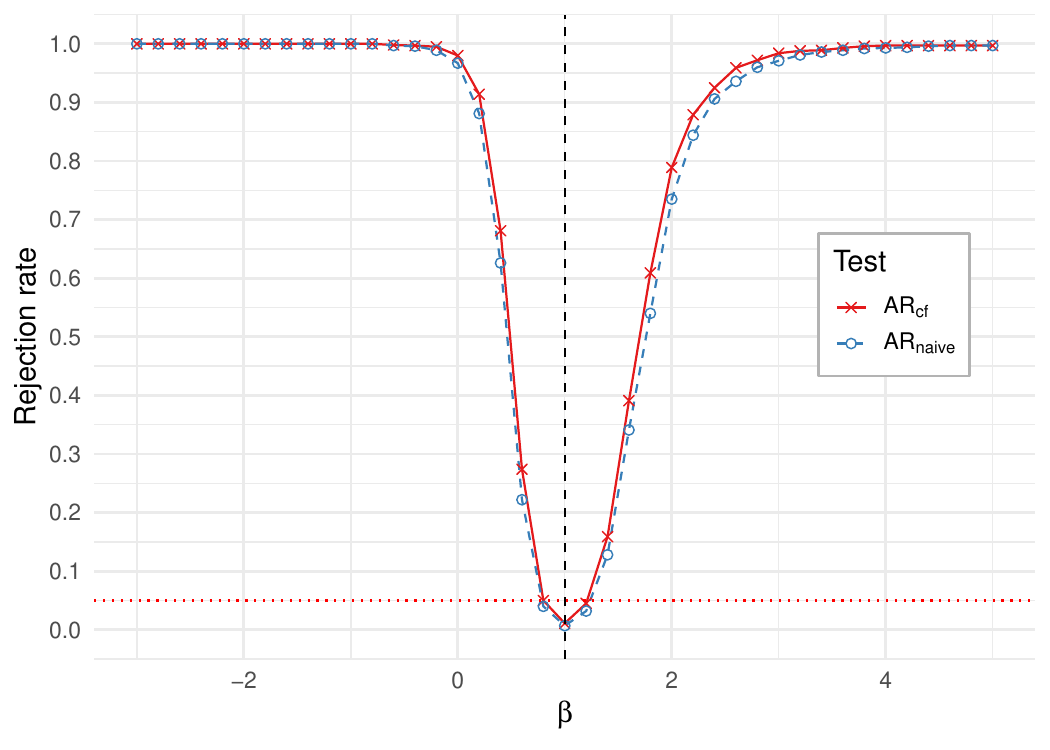}
        \caption{JIVE2, $r=64$}
    \end{subfigure}
    
    \caption{Power curves for DGP1 ($n=200$ and $\alpha = 0.05$). AR tests distributed as $\mathcal{N}(0,1)$. Results based on 1000 repetitions. The horizontal dotted red line denotes the $5\%$ nominal rejection level, while the vertical dotted black line corresponds to $\beta=1$. Panel (a) and panel (c) plot AR statistics based on the JIVE1 objective function with weak and strong instruments, respectively; panel (b)  and panel (d) plot AR statistics based on the JIVE2 objective function for weak and strong instruments, respectively.}
    \label{fig:AR_N_DGP1_200_0.05}
\end{figure}

%%%%%%%%%%%%%%%%%%%%%%%%%
% AR_200_0.1

\begin{figure}[ht]
    \centering
    \begin{subfigure}[b]{0.47\textwidth}
        \includegraphics[width=\textwidth]{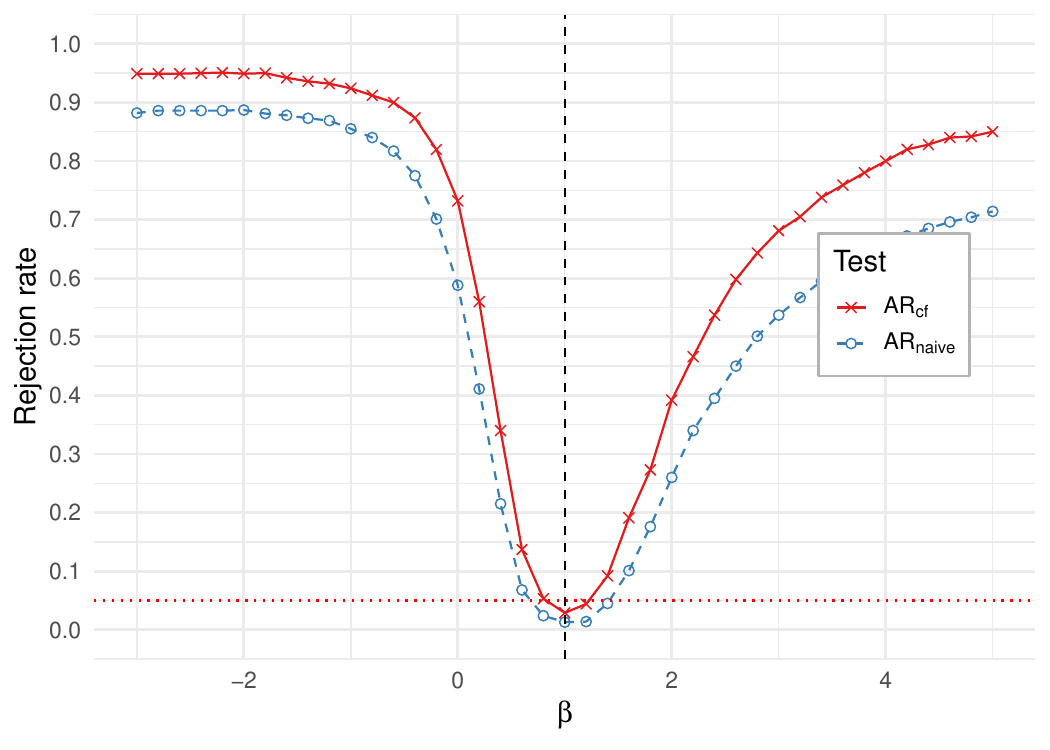}
        \caption{JIVE1, $r=32$}
    \end{subfigure}
    \hfill
    \begin{subfigure}[b]{0.47\textwidth}
        \includegraphics[width=\textwidth]{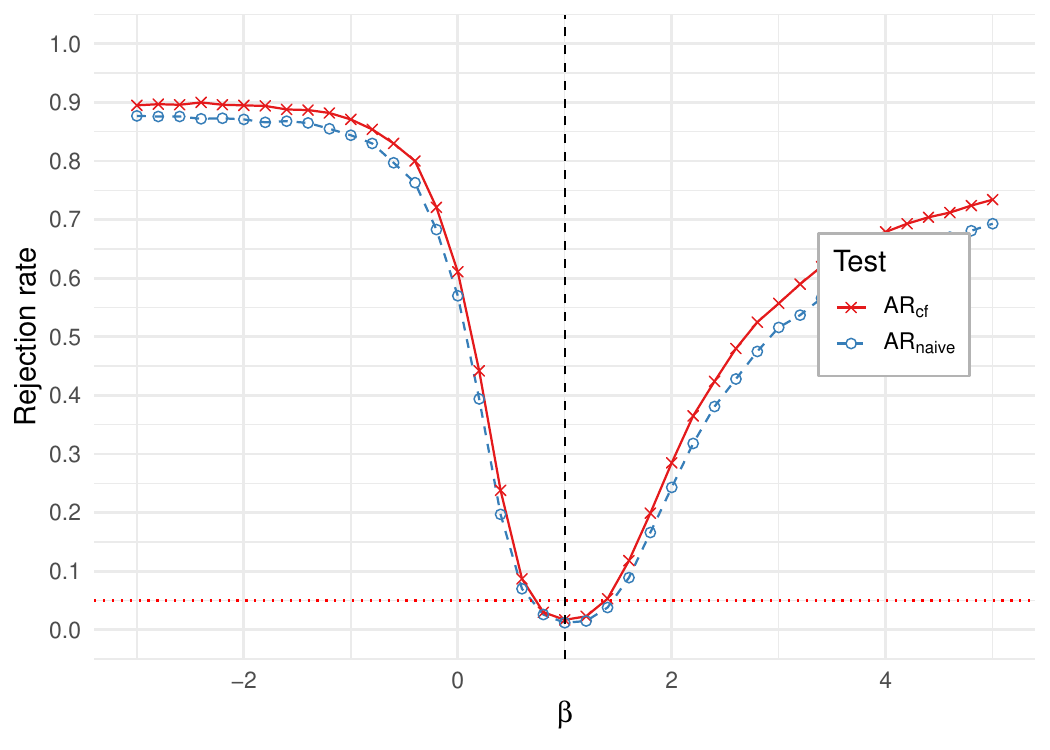}
        \caption{JIVE2, $r=32$}
    \end{subfigure}
    
    \begin{subfigure}[b]{0.47\textwidth}
        \includegraphics[width=\textwidth]{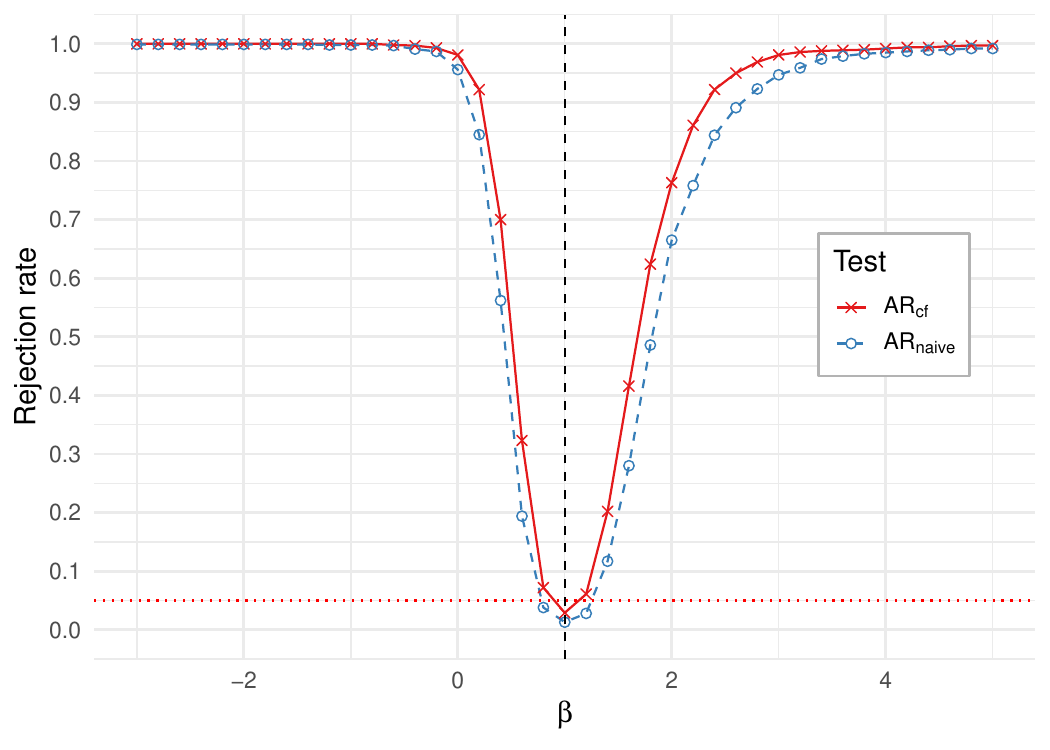}
        \caption{JIVE1, $r=64$}
    \end{subfigure}
    \hfill
    \begin{subfigure}[b]{0.47\textwidth}
        \includegraphics[width=\textwidth]{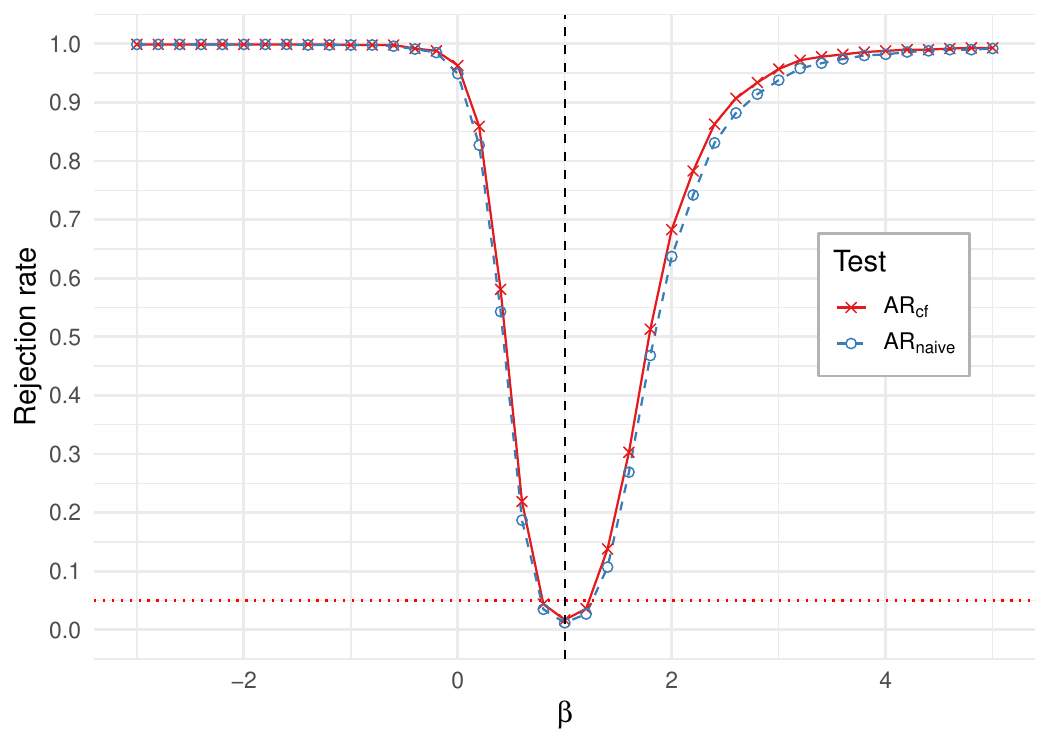}
        \caption{JIVE2, $r=64$}
    \end{subfigure}
    
    \caption{Power curves for DGP1 ($n=200$ and $\alpha = 0.1$). AR tests distributed as $\mathcal{N}(0,1)$. Results based on 1000 repetitions. The horizontal dotted red line denotes the $5\%$ nominal rejection level, while the vertical dotted black line corresponds to $\beta=1$. Panel (a) and panel (c) plot AR statistics based on the JIVE1 objective function with weak and strong instruments, respectively; panel (b)  and panel (d) plot AR statistics based on the JIVE2 objective function for weak and strong instruments, respectively.}
    \label{fig:AR_N_DGP1_200_0.1}
\end{figure}

\clearpage

\subsection{DGP2}

\begin{table}[p]
\centering
\caption{Size results for DGP2, $5\%$ nominal level and $n=200$.}
\label{tab:DGP2_200_full}
\resizebox{\textwidth}{!}{%
\begin{tabular}{lrr rrrr rrrr rrrrr rrrrr rr}
\toprule
Method & $\alpha$ & $r$ & $D$ & $W_1$ & $W_2$ & $LM$ & $D_{cf}$ & $W_{1,cf}$ & $W_{2,cf}$ & $LM_{cf}$ & $D_1^{*}$ & $D_2^{*}$ & $W_1^{*}$ & $W_2^{*}$ & $LM^{*}$ & $D_{1,cf}^{*}$ & $D_{2,cf}^{*}$ & $W_{1,cf}^{*}$ & $W_{2,cf}^{*}$ & $LM_{cf}^{*}$ & $AR_{naive}$ & $AR_{cf}$ \\
\midrule
SJIVE  & 0.05 & 0.1 & 0.066 & 0.055 & 0.055 & 0.054 & 0.062 & 0.052 & 0.052 & 0.059 & 0.047 & 0.051 & 0.055 & 0.055 & 0.054 & 0.053 & 0.058 & 0.052 & 0.052 & 0.059 &      &      \\
SJIVE  & 0.05 & 0.2 & 0.062 & 0.056 & 0.056 & 0.056 & 0.060 & 0.054 & 0.054 & 0.059 & 0.051 & 0.052 & 0.056 & 0.056 & 0.056 & 0.056 & 0.059 & 0.054 & 0.054 & 0.059 &      &      \\
SJIVE  & 0.10 & 0.1 & 0.067 & 0.052 & 0.052 & 0.051 & 0.061 & 0.048 & 0.048 & 0.058 & 0.045 & 0.049 & 0.052 & 0.052 & 0.051 & 0.053 & 0.058 & 0.048 & 0.048 & 0.058 &      &      \\
SJIVE  & 0.10 & 0.2 & 0.067 & 0.052 & 0.052 & 0.051 & 0.061 & 0.048 & 0.048 & 0.058 & 0.045 & 0.049 & 0.052 & 0.052 & 0.051 & 0.053 & 0.058 & 0.048 & 0.048 & 0.058 &      &      \\
\addlinespace
HLIM   & 0.05 & 0.1 & 0.063 & 0.052 & 0.052 & 0.052 & 0.058 & 0.047 & 0.047 & 0.055 & 0.044 & 0.048 & 0.052 & 0.052 & 0.052 & 0.049 & 0.053 & 0.047 & 0.047 & 0.055 &      &      \\
HLIM   & 0.05 & 0.2 & 0.060 & 0.056 & 0.056 & 0.051 & 0.054 & 0.051 & 0.051 & 0.055 & 0.046 & 0.047 & 0.056 & 0.056 & 0.051 & 0.051 & 0.052 & 0.051 & 0.051 & 0.055 &      &      \\
HLIM   & 0.10 & 0.1 & 0.072 & 0.062 & 0.062 & 0.057 & 0.066 & 0.056 & 0.056 & 0.056 & 0.044 & 0.050 & 0.062 & 0.062 & 0.057 & 0.049 & 0.054 & 0.056 & 0.056 & 0.056 &      &      \\
HLIM   & 0.10 & 0.2 & 0.072 & 0.062 & 0.062 & 0.057 & 0.066 & 0.056 & 0.056 & 0.056 & 0.044 & 0.050 & 0.062 & 0.062 & 0.057 & 0.049 & 0.054 & 0.056 & 0.056 & 0.056 &      &      \\
\addlinespace
JIVE1  & 0.05 & 0.1 & 0.030 & 0.030 & 0.030 & 0.056 & 0.031 & 0.031 & 0.031 & 0.057 & 0.056 &      & 0.030 & 0.030 & 0.056 & 0.057 &      & 0.031 & 0.031 & 0.057 & 0.028 & 0.045 \\
JIVE1  & 0.05 & 0.2 & 0.044 & 0.044 & 0.044 & 0.053 & 0.041 & 0.041 & 0.041 & 0.055 & 0.053 &      & 0.044 & 0.044 & 0.053 & 0.055 &      & 0.041 & 0.041 & 0.055 & 0.028 & 0.046 \\
JIVE1  & 0.10 & 0.1 & 0.035 & 0.035 & 0.035 & 0.060 & 0.035 & 0.035 & 0.035 & 0.055 & 0.060 &      & 0.035 & 0.035 & 0.060 & 0.055 &      & 0.035 & 0.035 & 0.055 & 0.036 & 0.085 \\
JIVE1  & 0.10 & 0.2 & 0.035 & 0.035 & 0.035 & 0.060 & 0.035 & 0.035 & 0.035 & 0.055 & 0.060 &      & 0.035 & 0.035 & 0.060 & 0.055 &      & 0.035 & 0.035 & 0.055 & 0.036 & 0.085 \\
\addlinespace
JIVE2  & 0.05 & 0.1 & 0.030 & 0.030 & 0.030 & 0.056 & 0.030 & 0.030 & 0.030 & 0.057 & 0.056 &      & 0.030 & 0.030 & 0.056 & 0.057 &      & 0.030 & 0.030 & 0.057 & 0.028 & 0.034 \\
JIVE2  & 0.05 & 0.2 & 0.043 & 0.043 & 0.043 & 0.053 & 0.040 & 0.040 & 0.040 & 0.054 & 0.053 &      & 0.043 & 0.043 & 0.053 & 0.054 &      & 0.040 & 0.040 & 0.054 & 0.028 & 0.034 \\
JIVE2  & 0.10 & 0.1 & 0.035 & 0.035 & 0.035 & 0.059 & 0.035 & 0.035 & 0.035 & 0.056 & 0.059 &      & 0.035 & 0.035 & 0.059 & 0.056 &      & 0.035 & 0.035 & 0.056 & 0.037 & 0.046 \\
JIVE2  & 0.10 & 0.2 & 0.035 & 0.035 & 0.035 & 0.059 & 0.035 & 0.035 & 0.035 & 0.056 & 0.059 &      & 0.035 & 0.035 & 0.059 & 0.056 &      & 0.035 & 0.035 & 0.056 & 0.037 & 0.046 \\
\bottomrule
\end{tabular}%
}
\end{table}

%%%%%%%%%%%%%%%%%%%%%%%%%
% Power DGP2
%%%%%%%%%%%%%%%%%%%%%%%%%

%%%%%%%%%%%%%%%%%%%%%%%%%
% Trinity_200_0.05_0.1 

\begin{figure}[ht]
    \centering
    % Replace 'chibar2' with 'chi2' for the second set of figures as needed
    \begin{subfigure}[b]{0.47\textwidth}
        \includegraphics[width=\textwidth]{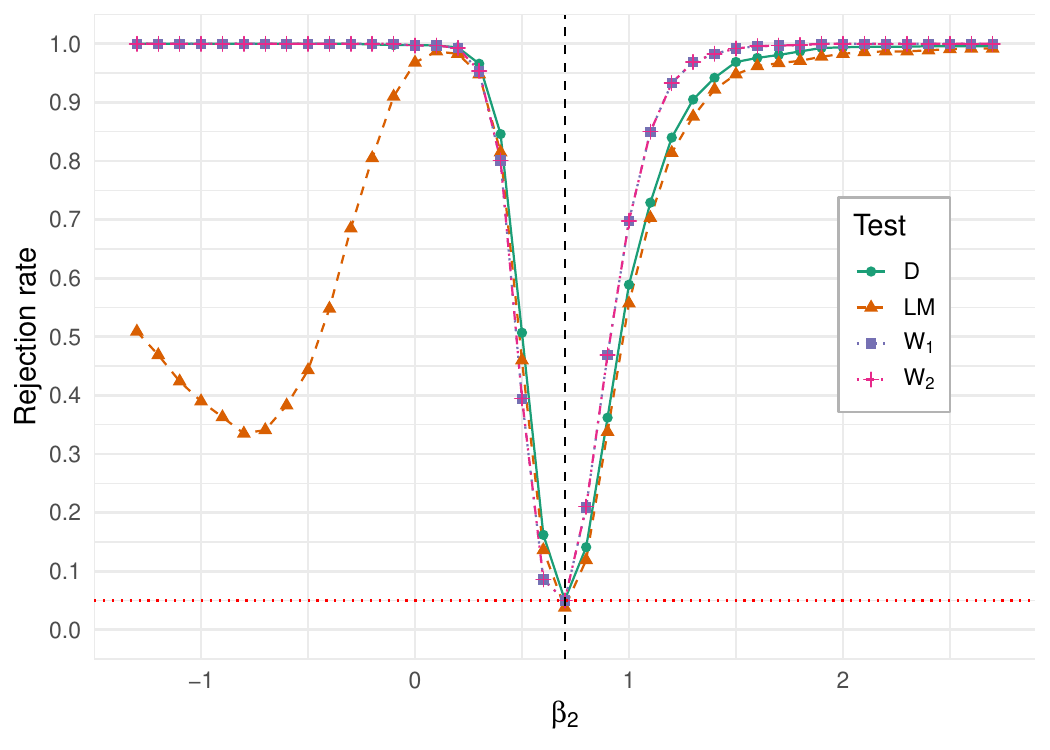}
        \caption{SJIVE}
    \end{subfigure}
    \hfill
    \begin{subfigure}[b]{0.47\textwidth}
        \includegraphics[width=\textwidth]{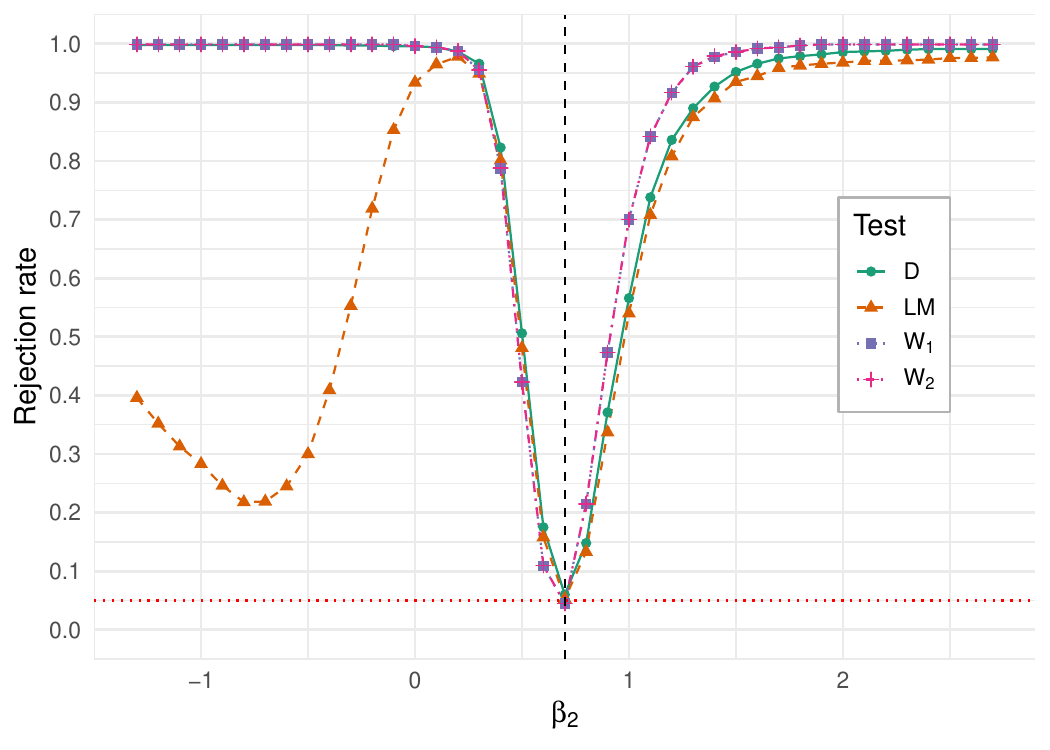}
        \caption{HLIM}
    \end{subfigure}
    
    \begin{subfigure}[b]{0.47\textwidth}
        \includegraphics[width=\textwidth]{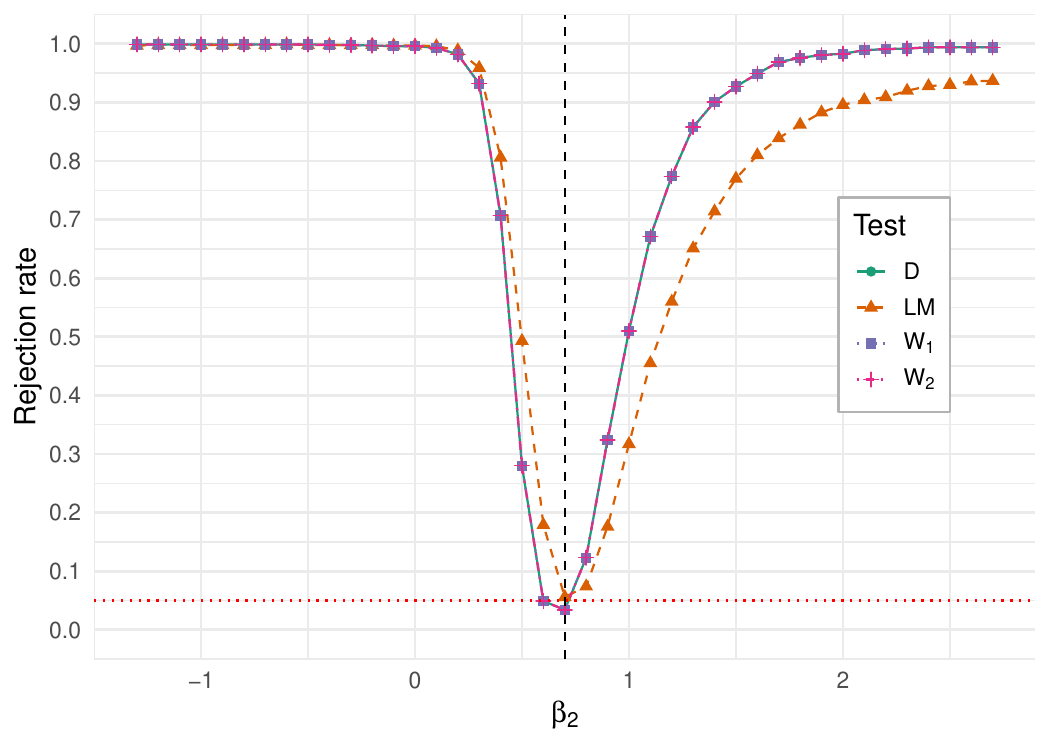}
        \caption{JIVE1}
    \end{subfigure}
    \hfill
    \begin{subfigure}[b]{0.47\textwidth}
        \includegraphics[width=\textwidth]{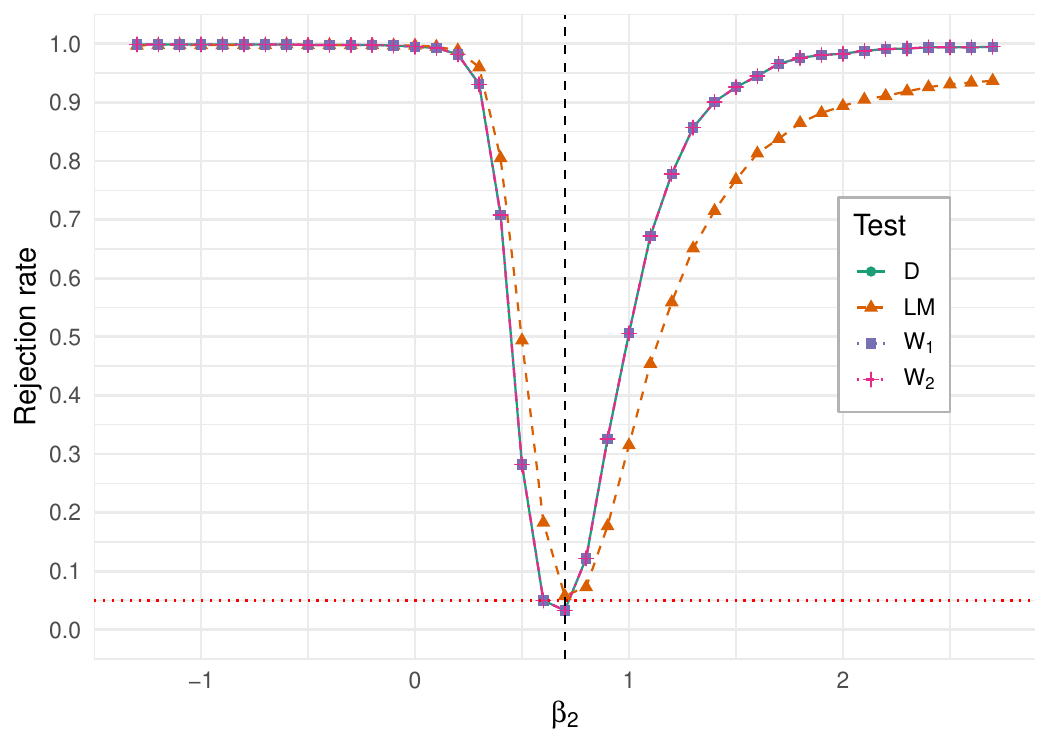}
        \caption{JIVE2}
    \end{subfigure}
    
    \caption{Power curves for DGP2 ($n=200$, $\alpha = 0.05$, $r = 0.1$). Trinity of test statistics distributed as a $\bar\chi^2$. Results based on 1000 repetitions. The horizontal dotted red line denotes the $5\%$ nominal rejection level, while the vertical dotted black line corresponds to $\beta=1$. Panel (a) plots statistics based on the SJIVE objective function; panel (b) plots statistics based on the HLIM objective function; panel (c) plots statistics based on the JIVE1 objective function; panel (d) plots statistics based on the JIVE2 objective function.}
    \label{fig:Trinity_chibar2_DGP2_200_0.05_0.1}
\end{figure}

\begin{figure}[ht]
    \centering
    % Replace 'chibar2' with 'chi2' for the second set of figures as needed
    \begin{subfigure}[b]{0.47\textwidth}
        \includegraphics[width=\textwidth]{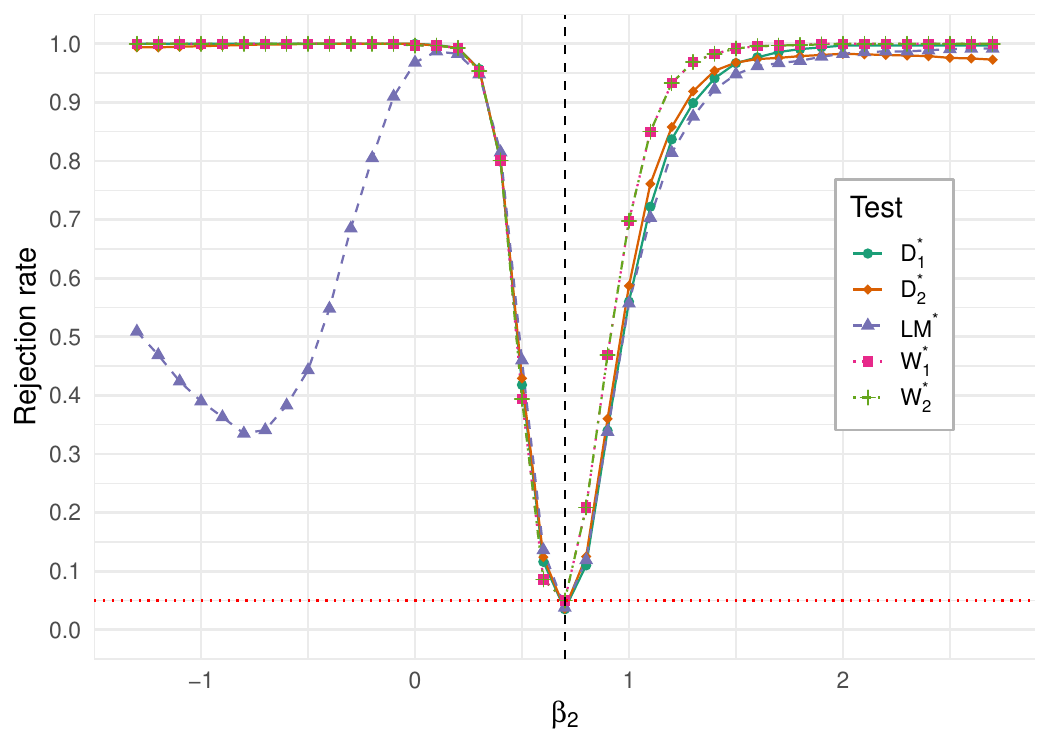}
        \caption{SJIVE}
    \end{subfigure}
    \hfill
    \begin{subfigure}[b]{0.47\textwidth}
        \includegraphics[width=\textwidth]{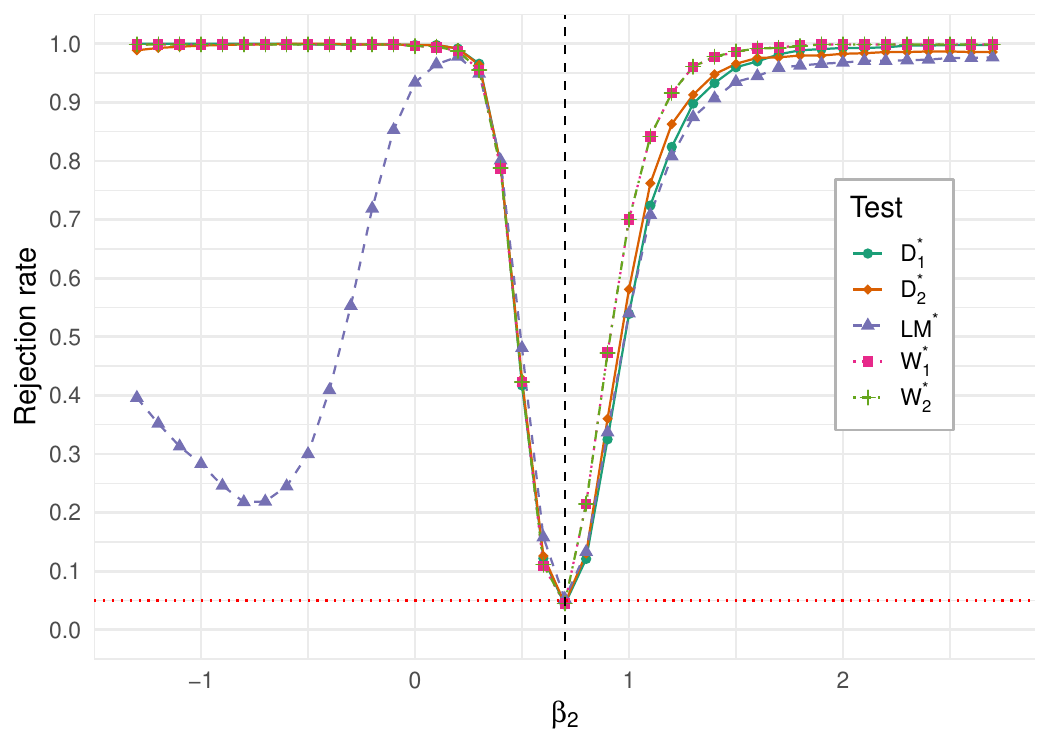}
        \caption{HLIM}
    \end{subfigure}
    
    \begin{subfigure}[b]{0.47\textwidth}
        \includegraphics[width=\textwidth]{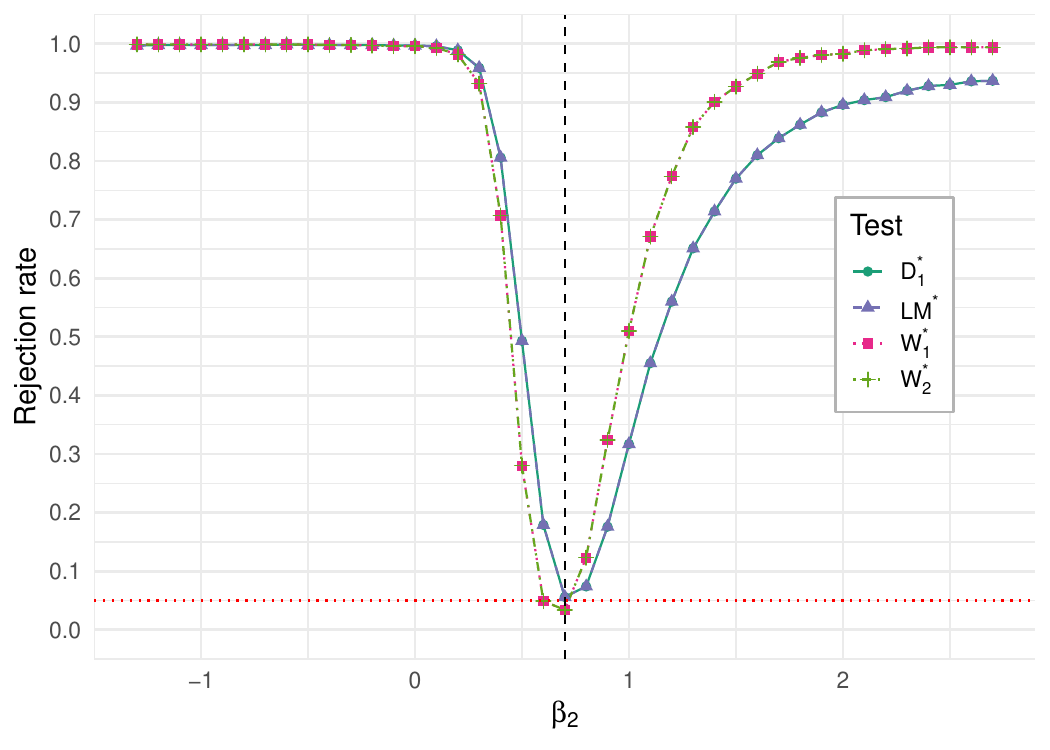}
        \caption{JIVE1}
    \end{subfigure}
    \hfill
    \begin{subfigure}[b]{0.47\textwidth}
        \includegraphics[width=\textwidth]{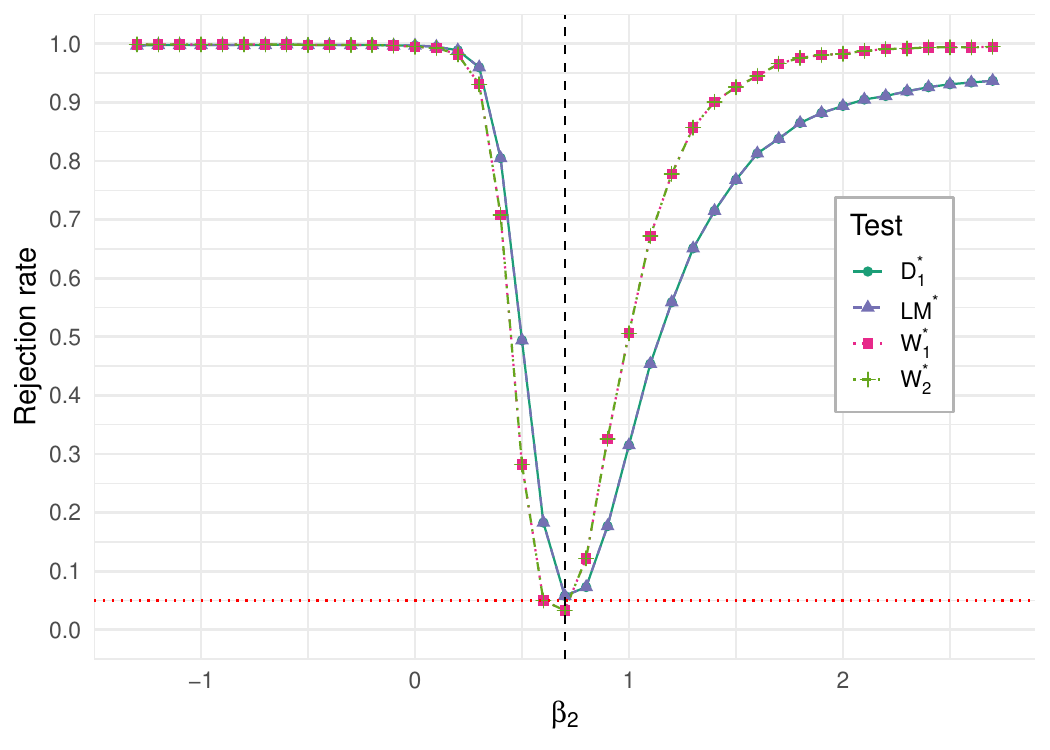}
        \caption{JIVE2}
    \end{subfigure}
    
    \caption{Power curves for DGP2 ($n=200$, $\alpha = 0.05$, $r = 0.1$). Trinity of test statistics distributed as a $\chi^2$. Results based on 1000 repetitions. The horizontal dotted red line denotes the $5\%$ nominal rejection level, while the vertical dotted black line corresponds to $\beta=1$. Panel (a) plots statistics based on the SJIVE objective function; panel (b) plots statistics based on the HLIM objective function; panel (c) plots statistics based on the JIVE1 objective function; panel (d) plots statistics based on the JIVE2 objective function.}
    \label{fig:Trinity_chi2_DGP2_200_0.05_0.1}
\end{figure}

% cf variance

\begin{figure}[ht]
    \centering
    % Replace 'chibar2' with 'chi2' for the second set of figures as needed
    \begin{subfigure}[b]{0.47\textwidth}
        \includegraphics[width=\textwidth]{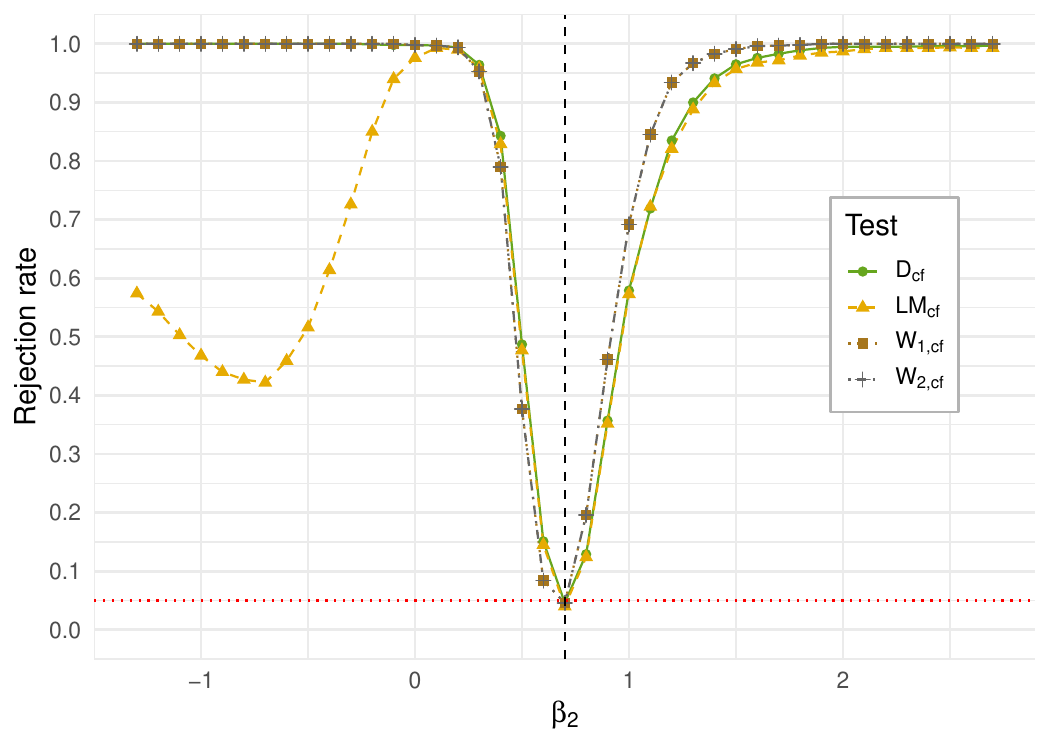}
        \caption{SJIVE}
    \end{subfigure}
    \hfill
    \begin{subfigure}[b]{0.47\textwidth}
       \includegraphics[width=\textwidth]{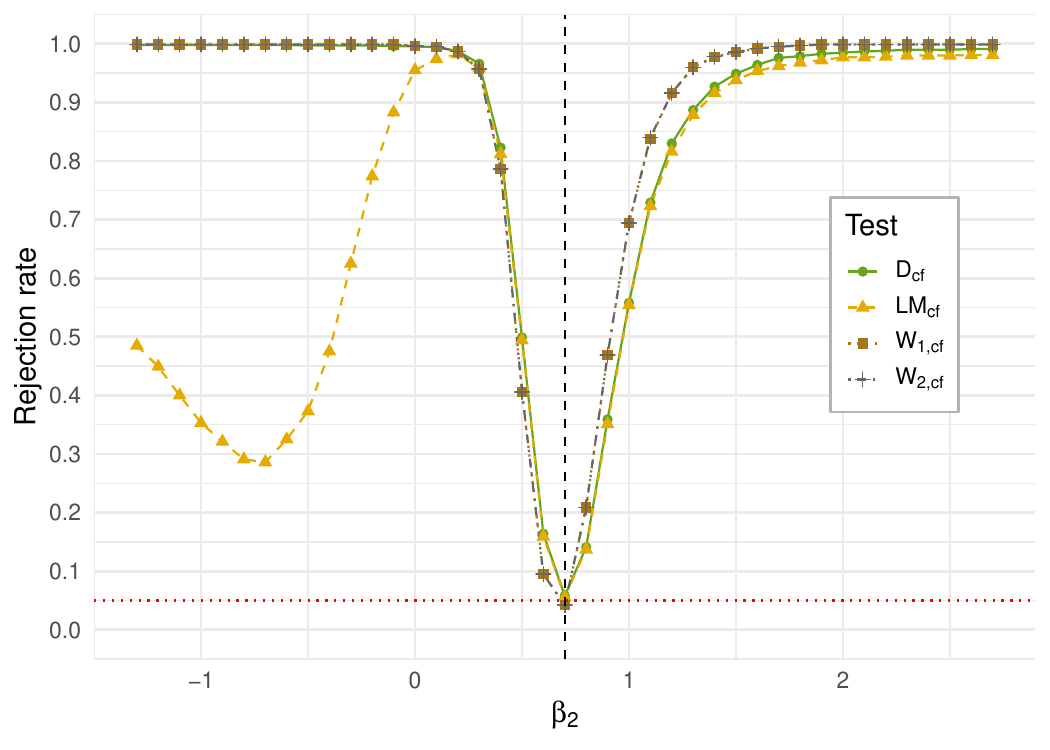}
        \caption{HLIM}
    \end{subfigure}
    
    \begin{subfigure}[b]{0.47\textwidth}
        \includegraphics[width=\textwidth]{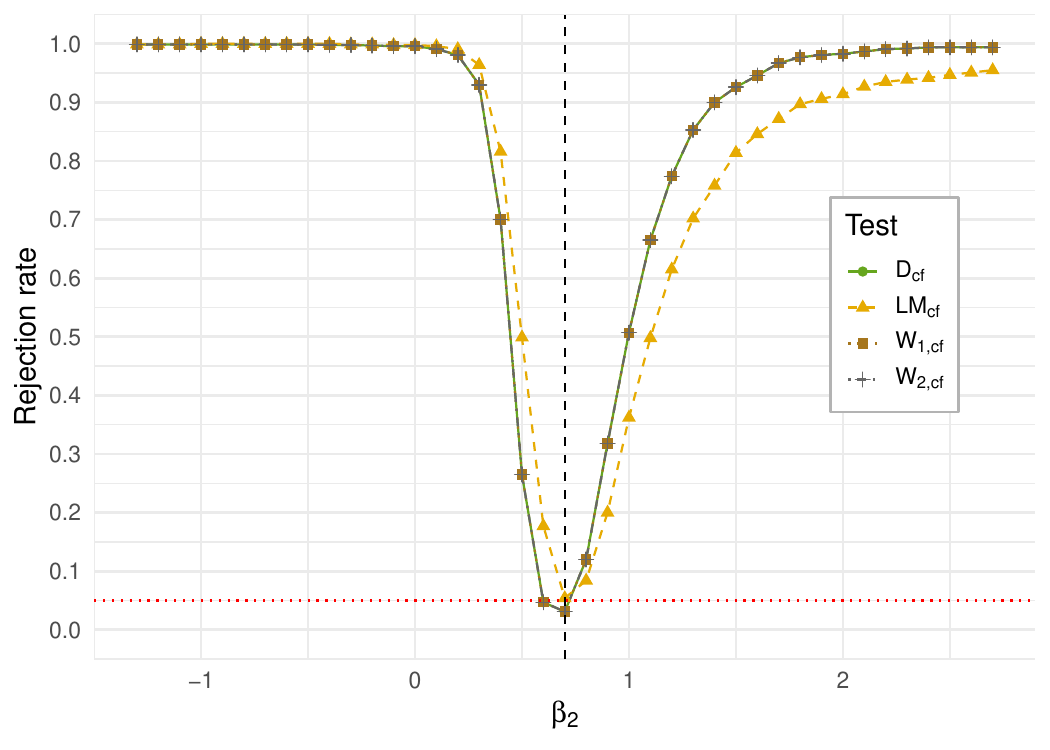}
        \caption{JIVE1}
    \end{subfigure}
    \hfill
    \begin{subfigure}[b]{0.47\textwidth}
        \includegraphics[width=\textwidth]{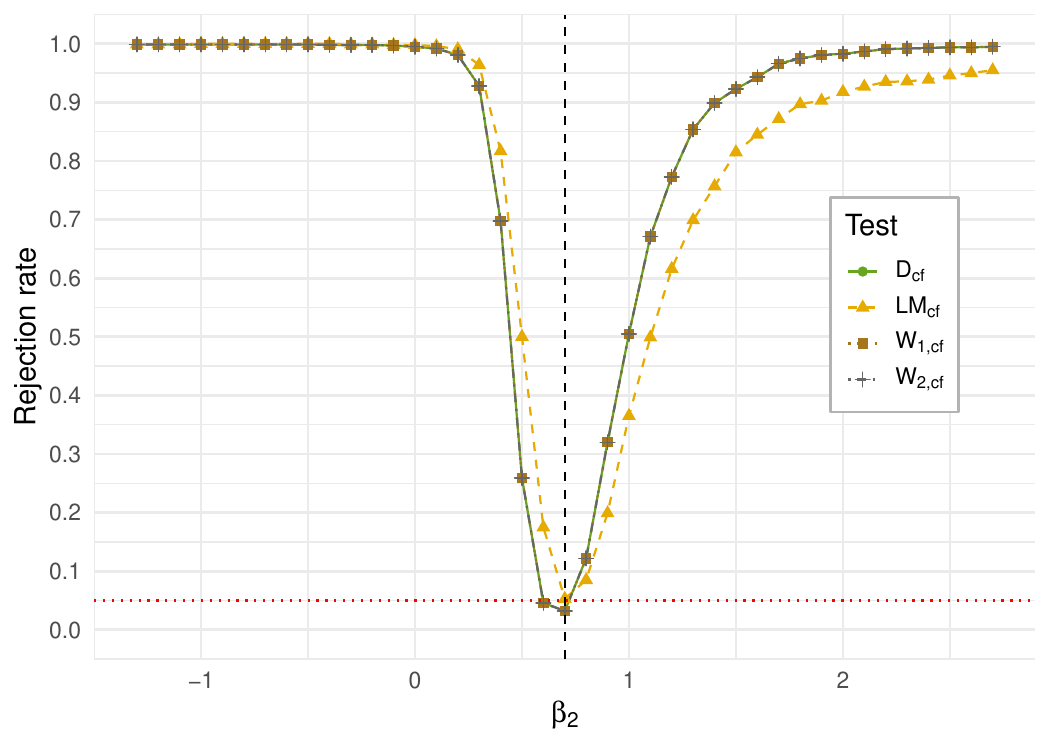}
        \caption{JIVE2}
    \end{subfigure}
    
    \caption{Power curves for DGP2 ($n=200$, $\alpha = 0.05$, $r = 0.1$). Trinity of test statistics distributed as a $\bar\chi^2$ with cross-fit variance. Results based on 1000 repetitions. The horizontal dotted red line denotes the $5\%$ nominal rejection level, while the vertical dotted black line corresponds to $\beta=1$. Panel (a) plots statistics based on the SJIVE objective function; panel (b) plots statistics based on the HLIM objective function; panel (c) plots statistics based on the JIVE1 objective function; panel (d) plots statistics based on the JIVE2 objective function.}
    \label{fig:Trinity_chibar2_cf_DGP2_200_0.05_0.1}
\end{figure}

\begin{figure}[ht]
    \centering
    % Replace 'chibar2' with 'chi2' for the second set of figures as needed
    \begin{subfigure}[b]{0.47\textwidth}
        \includegraphics[width=\textwidth]{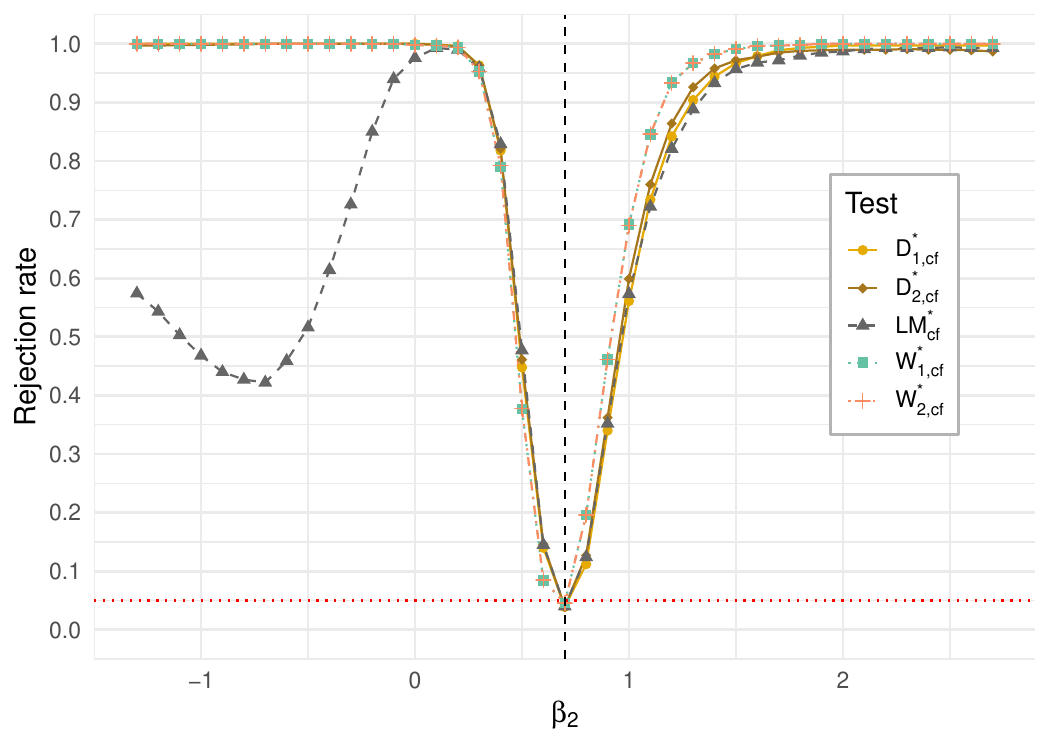}
        \caption{SJIVE}
    \end{subfigure}
    \hfill
    \begin{subfigure}[b]{0.47\textwidth}
        \includegraphics[width=\textwidth]{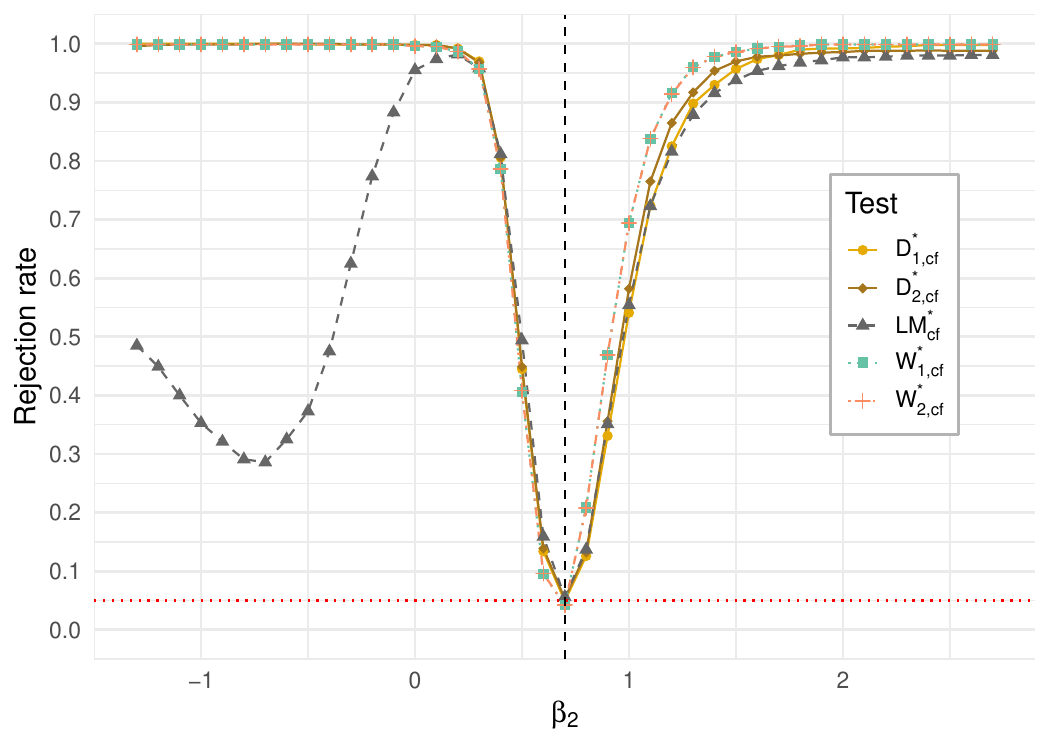}
        \caption{HLIM}
    \end{subfigure}
    
    \begin{subfigure}[b]{0.47\textwidth}
       \includegraphics[width=\textwidth]{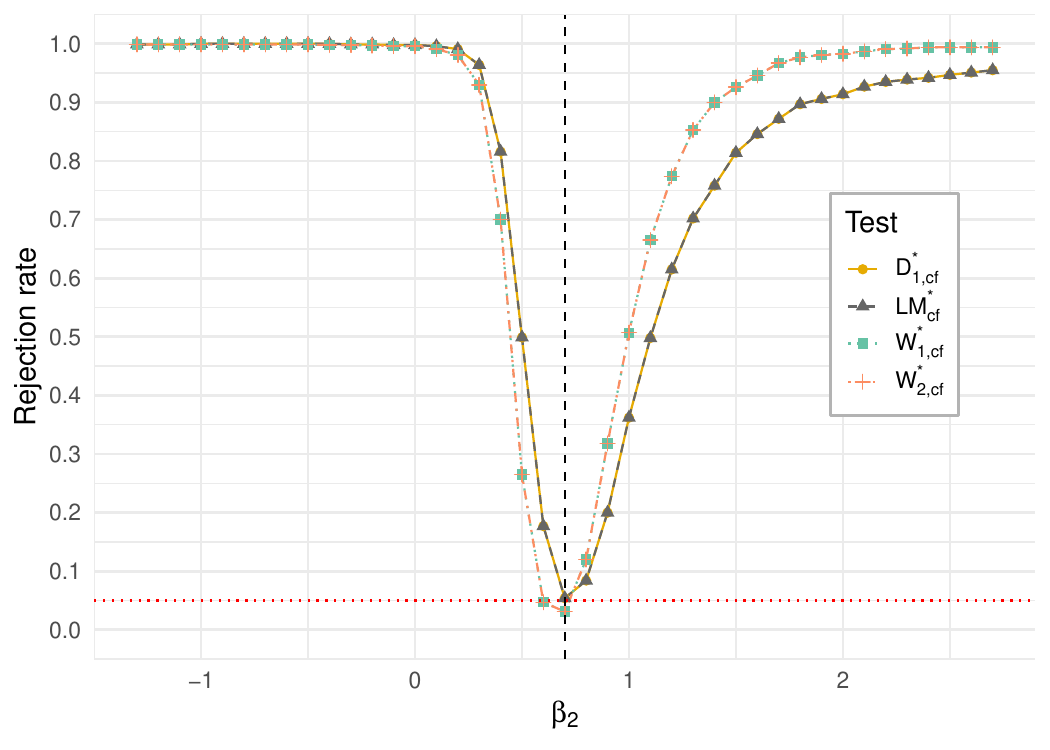}
        \caption{JIVE1}
    \end{subfigure}
    \hfill
    \begin{subfigure}[b]{0.47\textwidth}
        \includegraphics[width=\textwidth]{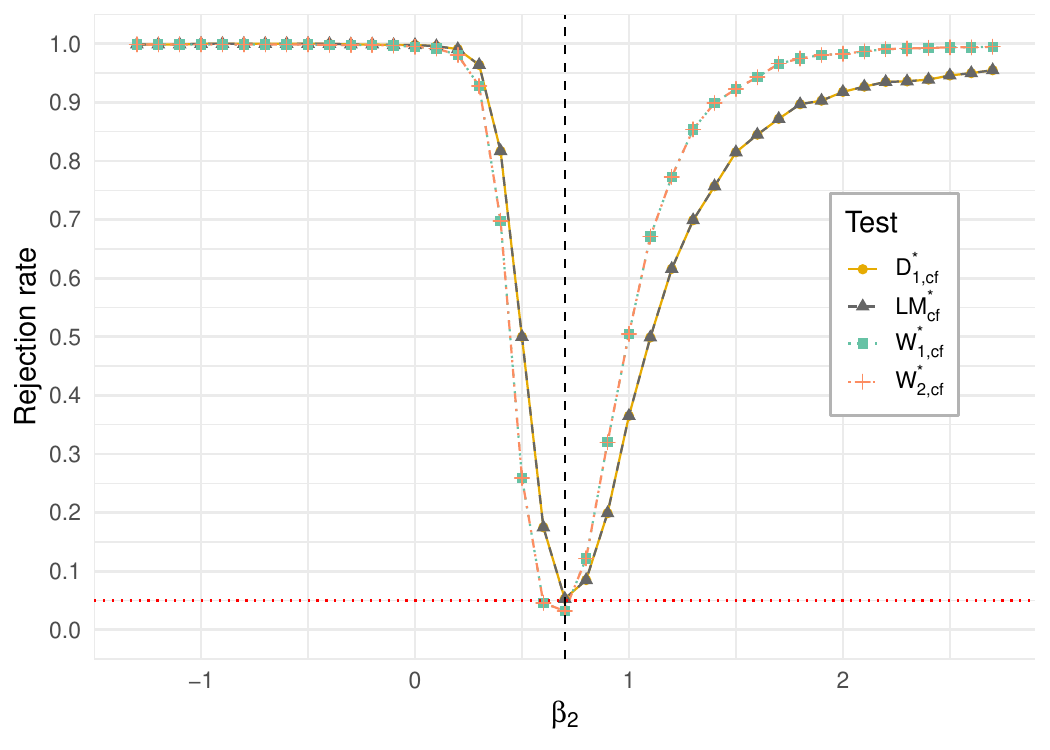}
        \caption{JIVE2}
    \end{subfigure}
    
    \caption{Power curves for DGP2 ($n=200$, $\alpha = 0.05$, $r = 0.1$). Trinity of test statistics distributed as a $\chi^2$ with cross-fit variance. Results based on 1000 repetitions. The horizontal dotted red line denotes the $5\%$ nominal rejection level, while the vertical dotted black line corresponds to $\beta=1$. Panel (a) plots statistics based on the SJIVE objective function; panel (b) plots statistics based on the HLIM objective function; panel (c) plots statistics based on the JIVE1 objective function; panel (d) plots statistics based on the JIVE2 objective function.}
    \label{fig:Trinity_chi2_cf_DGP2_200_0.05_0.1}
\end{figure}

%%%%%%%%%%%%%%%%%%%%%%%%%
% Trinity_200_0.05_0.2 

\begin{figure}[ht]
    \centering
    % Replace 'chibar2' with 'chi2' for the second set of figures as needed
    \begin{subfigure}[b]{0.47\textwidth}
        \includegraphics[width=\textwidth]{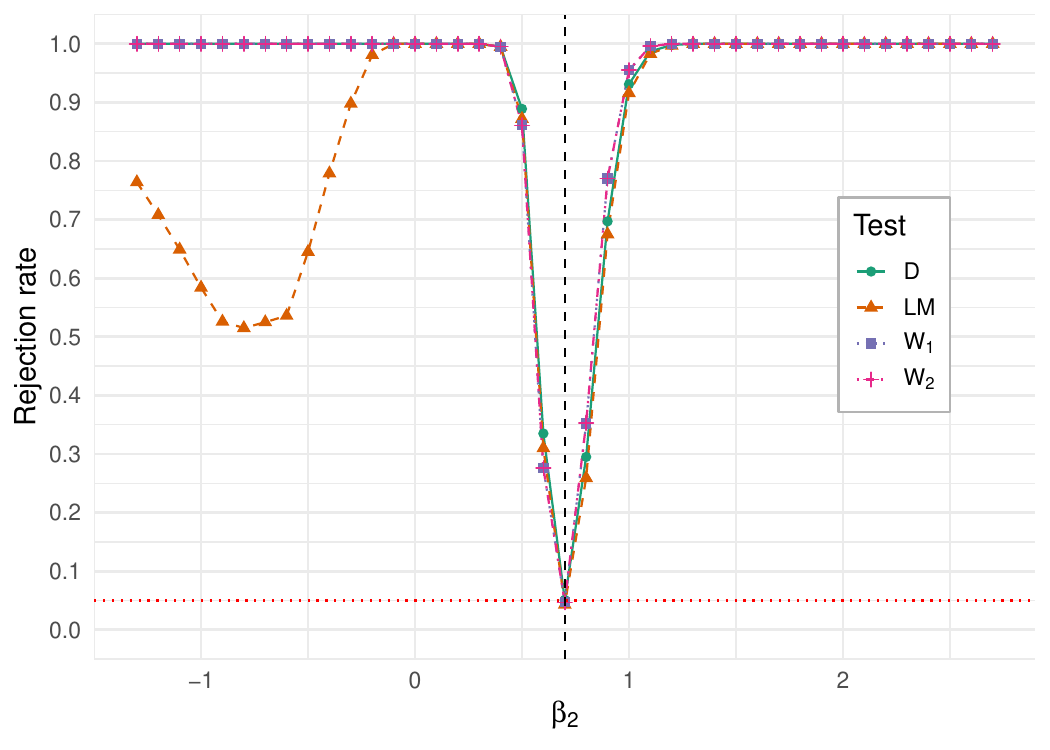}
        \caption{SJIVE}
    \end{subfigure}
    \hfill
    \begin{subfigure}[b]{0.47\textwidth}
        \includegraphics[width=\textwidth]{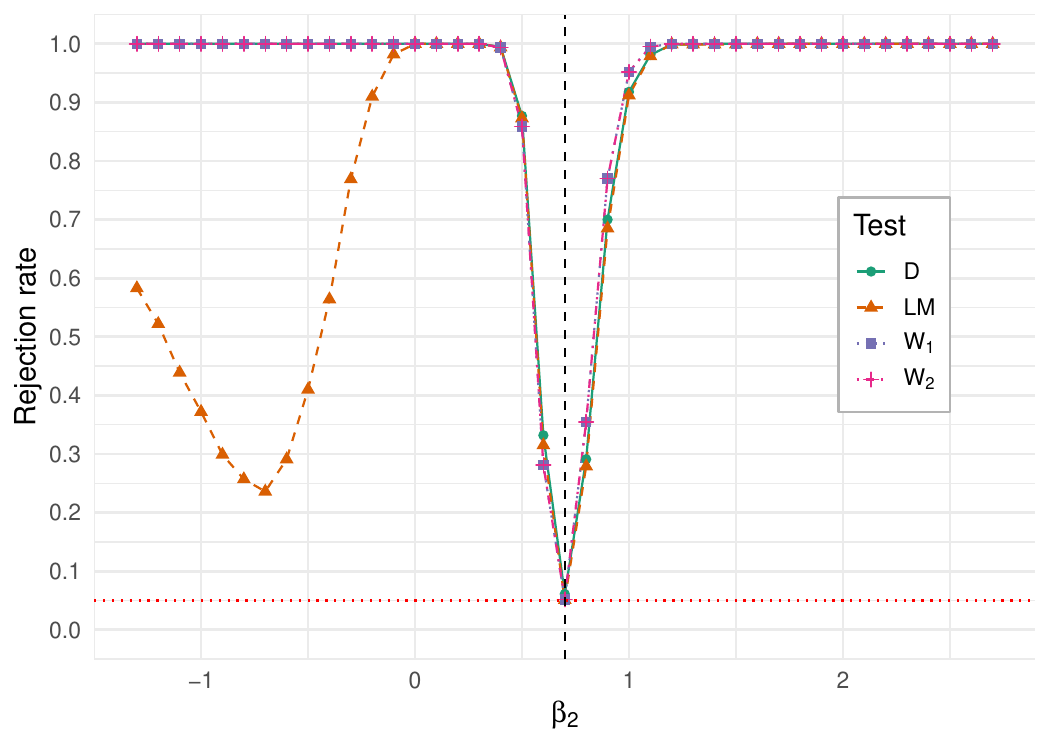}
        \caption{HLIM}
    \end{subfigure}
    
    \begin{subfigure}[b]{0.47\textwidth}
        \includegraphics[width=\textwidth]{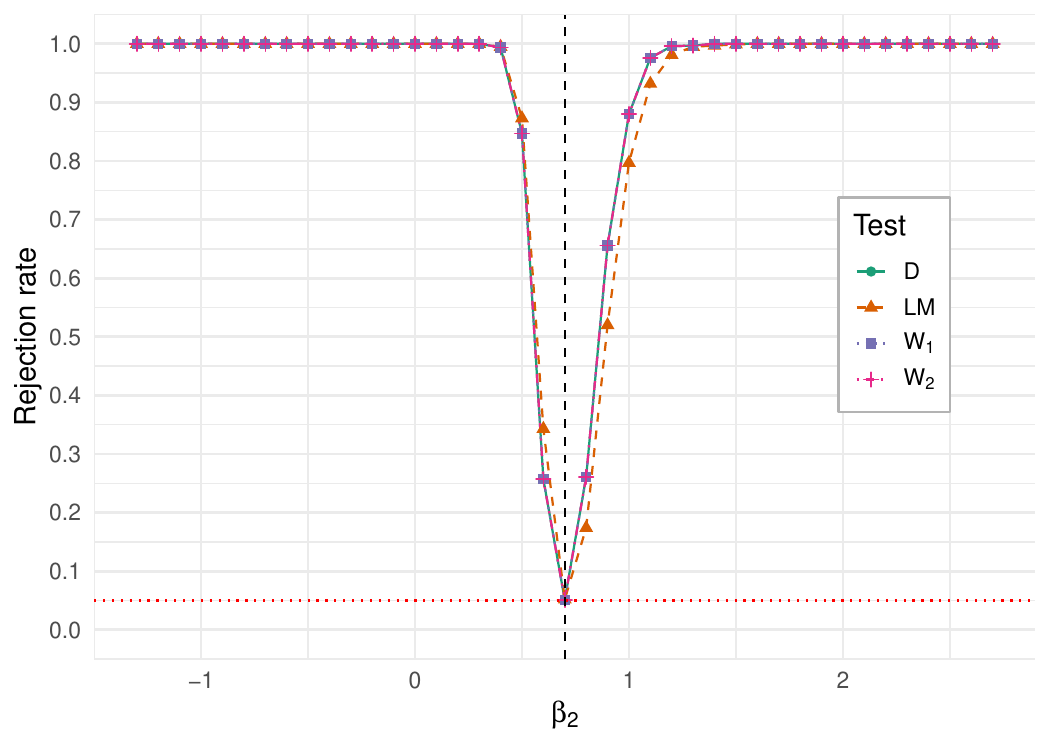}
        \caption{JIVE1}
    \end{subfigure}
    \hfill
    \begin{subfigure}[b]{0.47\textwidth}
        \includegraphics[width=\textwidth]{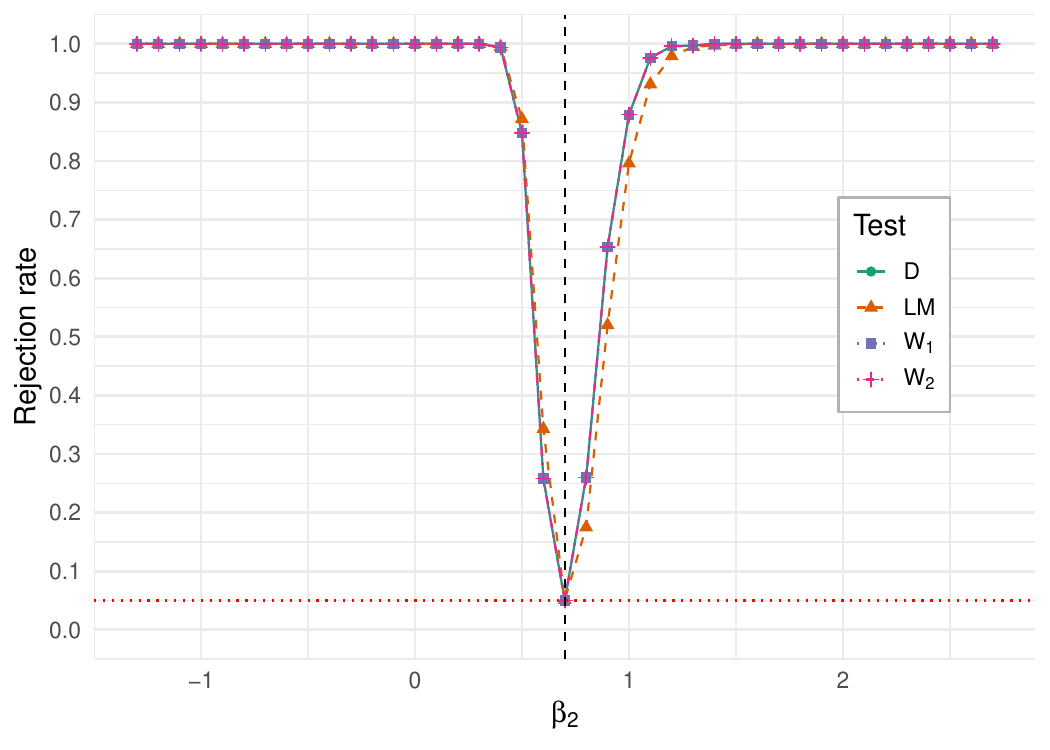}
        \caption{JIVE2}
    \end{subfigure}
    
    \caption{Power curves for DGP2 ($n=200$, $\alpha = 0.05$, $r = 0.2$). Trinity of test statistics distributed as a $\bar\chi^2$. Results based on 1000 repetitions. The horizontal dotted red line denotes the $5\%$ nominal rejection level, while the vertical dotted black line corresponds to $\beta=1$. Panel (a) plots statistics based on the SJIVE objective function; panel (b) plots statistics based on the HLIM objective function; panel (c) plots statistics based on the JIVE1 objective function; panel (d) plots statistics based on the JIVE2 objective function.}
    \label{fig:Trinity_chibar2_DGP2_200_0.05_0.2}
\end{figure}

\begin{figure}[ht]
    \centering
    % Replace 'chibar2' with 'chi2' for the second set of figures as needed
    \begin{subfigure}[b]{0.47\textwidth}
        \includegraphics[width=\textwidth]{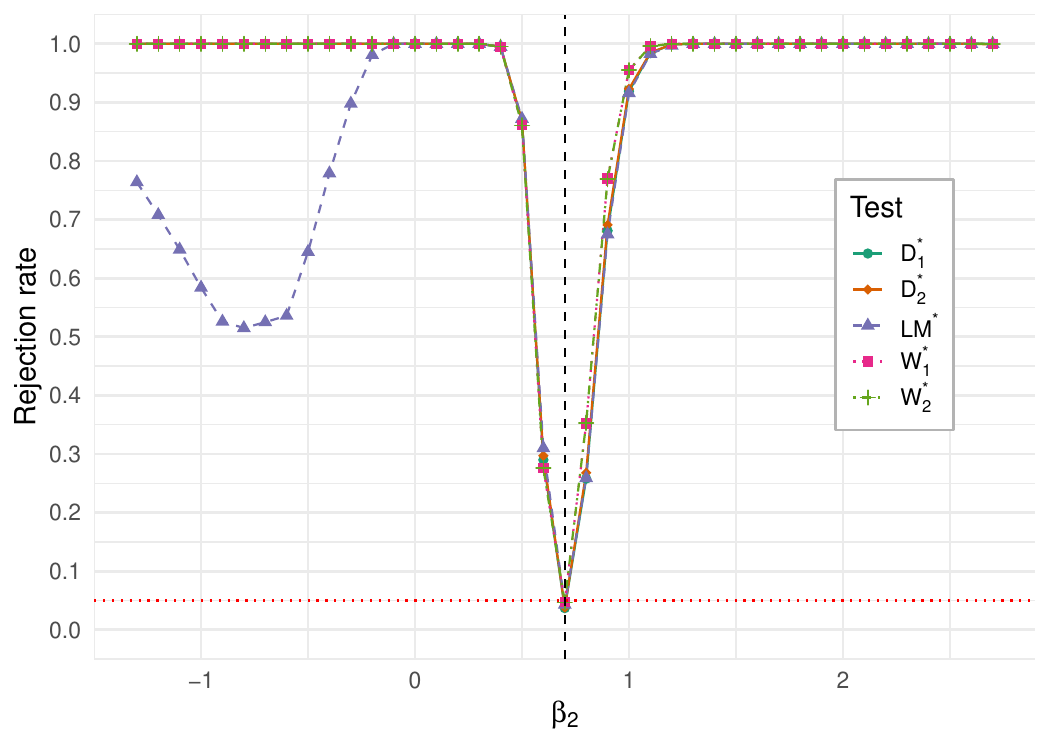}
        \caption{SJIVE}
    \end{subfigure}
    \hfill
    \begin{subfigure}[b]{0.47\textwidth}
        \includegraphics[width=\textwidth]{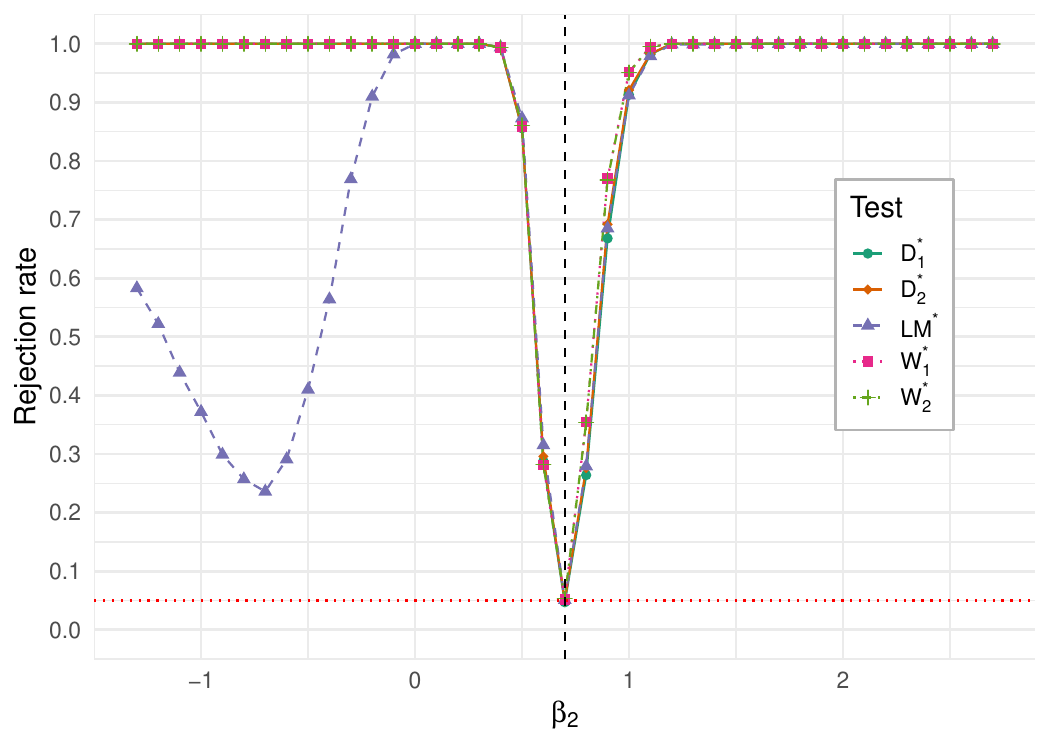}
        \caption{HLIM}
    \end{subfigure}
    
    \begin{subfigure}[b]{0.47\textwidth}
        \includegraphics[width=\textwidth]{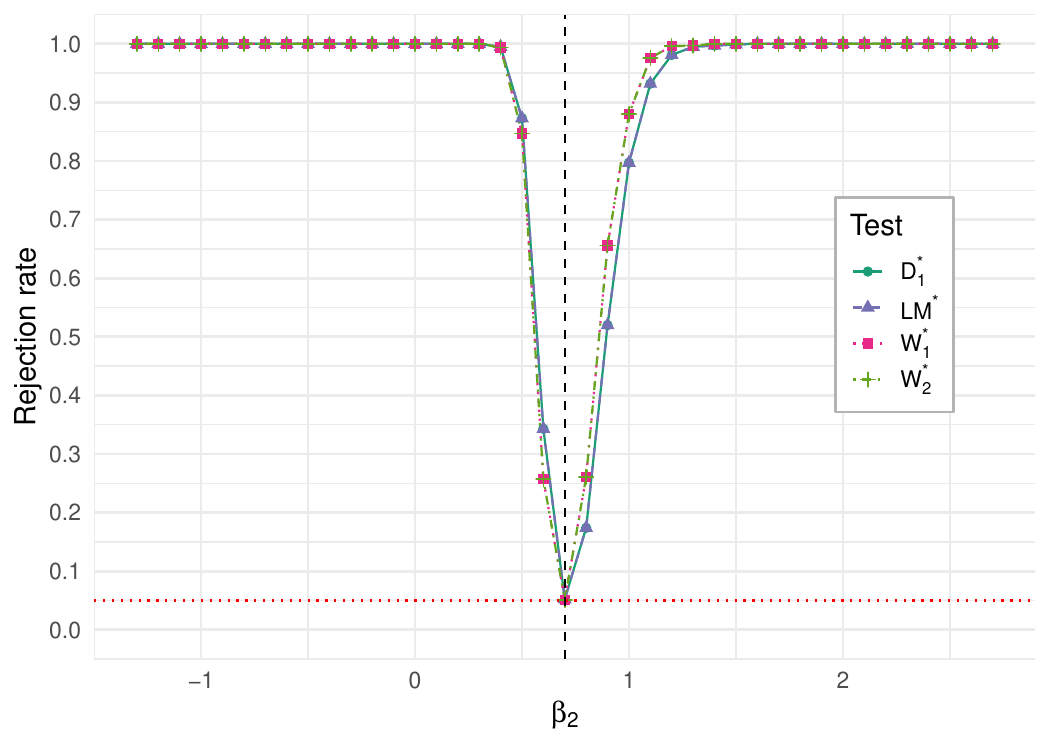}
        \caption{JIVE1}
    \end{subfigure}
    \hfill
    \begin{subfigure}[b]{0.47\textwidth}
        \includegraphics[width=\textwidth]{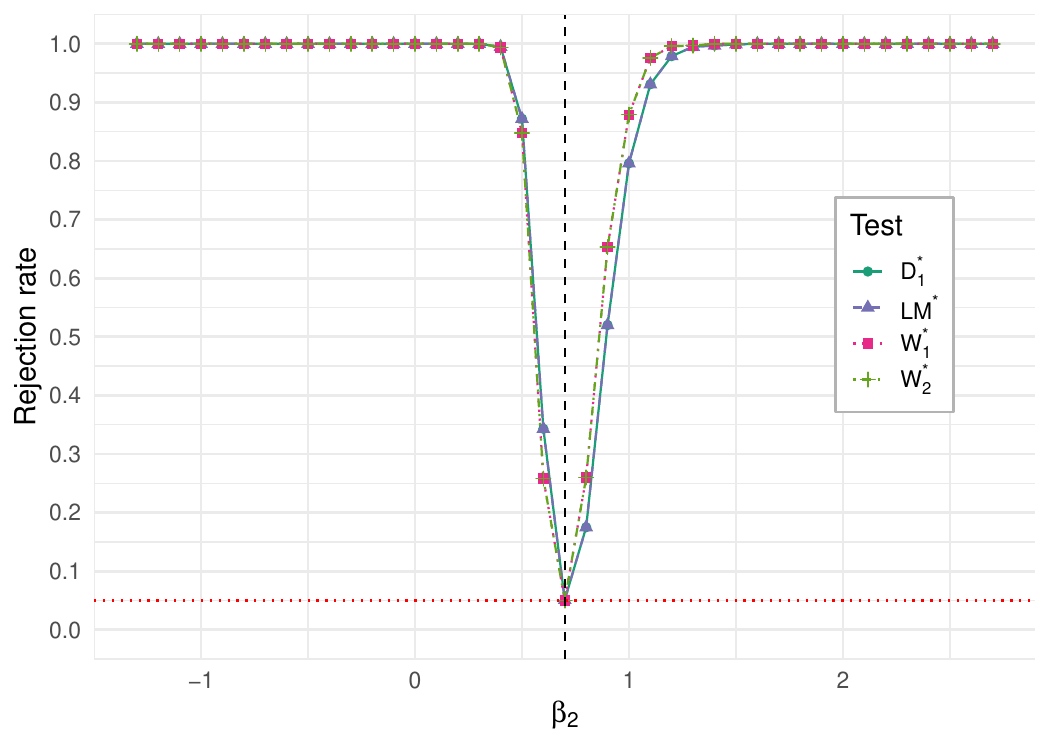}
        \caption{JIVE2}
    \end{subfigure}
    
    \caption{Power curves for DGP2 ($n=200$, $\alpha = 0.05$, $r = 0.2$). Trinity of test statistics distributed as a $\chi^2$. Results based on 1000 repetitions. The horizontal dotted red line denotes the $5\%$ nominal rejection level, while the vertical dotted black line corresponds to $\beta=1$. Panel (a) plots statistics based on the SJIVE objective function; panel (b) plots statistics based on the HLIM objective function; panel (c) plots statistics based on the JIVE1 objective function; panel (d) plots statistics based on the JIVE2 objective function.}
    \label{fig:Trinity_chi2_DGP2_200_0.05_0.2}
\end{figure}

% cf variance

\begin{figure}[ht]
    \centering
    % Replace 'chibar2' with 'chi2' for the second set of figures as needed
    \begin{subfigure}[b]{0.47\textwidth}
        \includegraphics[width=\textwidth]{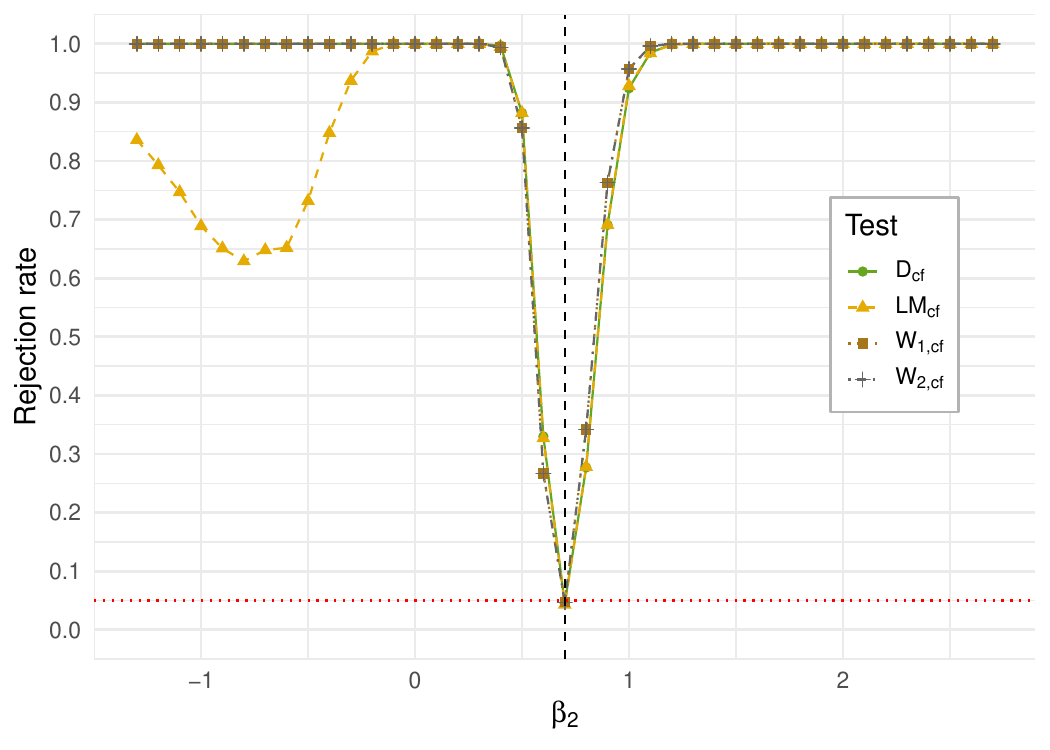}
        \caption{SJIVE}
    \end{subfigure}
    \hfill
    \begin{subfigure}[b]{0.47\textwidth}
        \includegraphics[width=\textwidth]{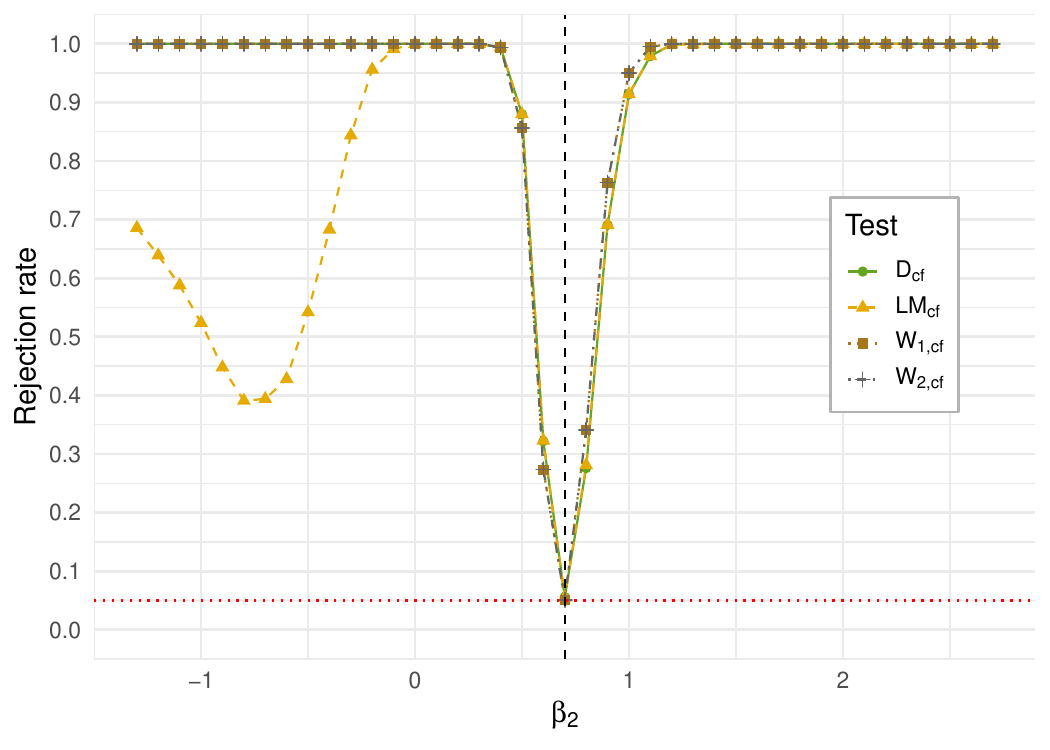}
        \caption{HLIM}
    \end{subfigure}
    
    \begin{subfigure}[b]{0.47\textwidth}
        \includegraphics[width=\textwidth]{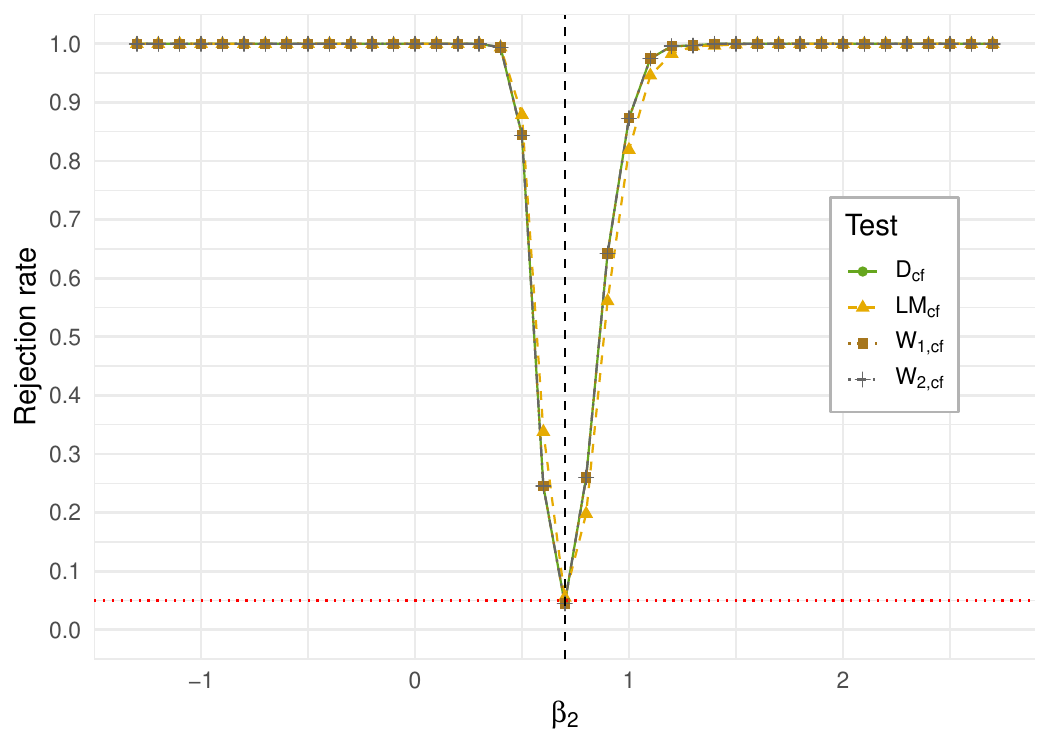}
        \caption{JIVE1}
    \end{subfigure}
    \hfill
    \begin{subfigure}[b]{0.47\textwidth}
        \includegraphics[width=\textwidth]{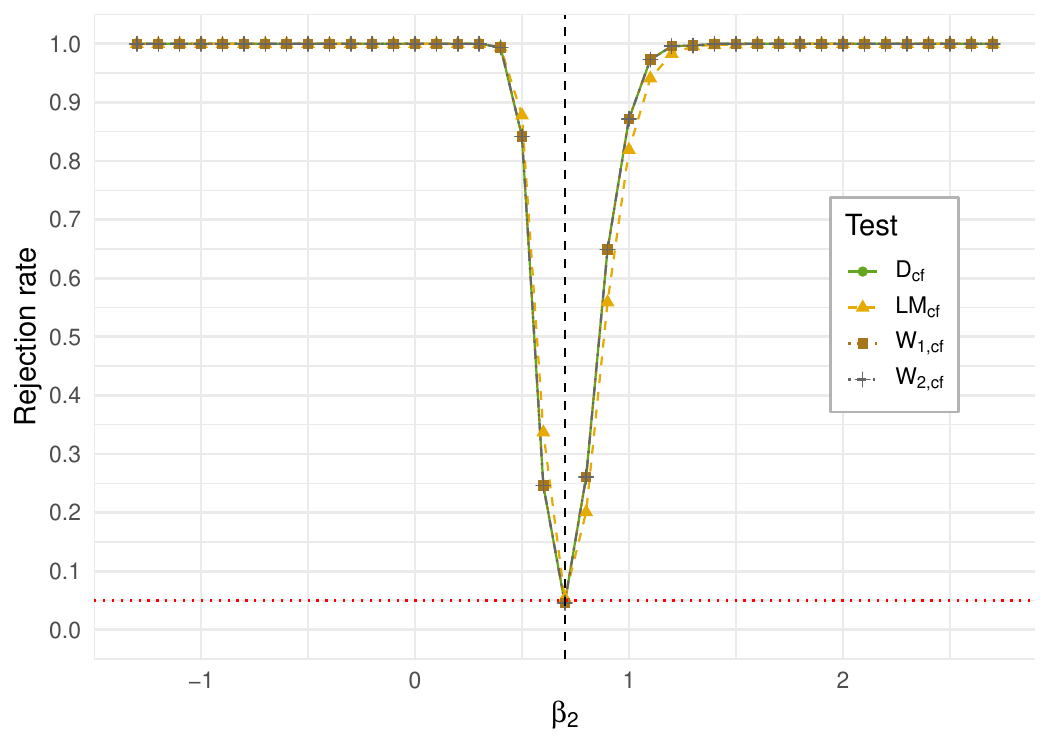}
        \caption{JIVE2}
    \end{subfigure}
    
    \caption{Power curves for DGP2 ($n=200$, $\alpha = 0.05$, $r = 0.2$). Trinity of test statistics distributed as a $\bar\chi^2$ with cross-fit variance. Results based on 1000 repetitions. The horizontal dotted red line denotes the $5\%$ nominal rejection level, while the vertical dotted black line corresponds to $\beta=1$. Panel (a) plots statistics based on the SJIVE objective function; panel (b) plots statistics based on the HLIM objective function; panel (c) plots statistics based on the JIVE1 objective function; panel (d) plots statistics based on the JIVE2 objective function.}
    \label{fig:Trinity_chibar2_cf_DGP2_200_0.05_0.2}
\end{figure}

\begin{figure}[ht]
    \centering
    % Replace 'chibar2' with 'chi2' for the second set of figures as needed
    \begin{subfigure}[b]{0.47\textwidth}
        \includegraphics[width=\textwidth]{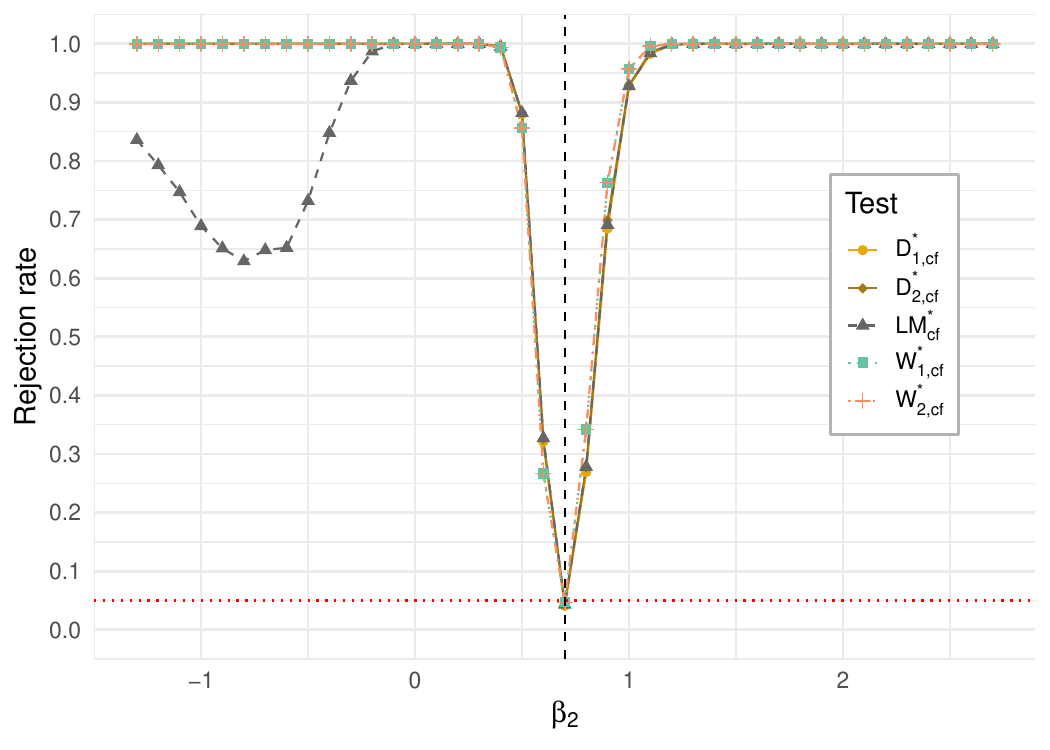}
        \caption{SJIVE}
    \end{subfigure}
    \hfill
    \begin{subfigure}[b]{0.47\textwidth}
        \includegraphics[width=\textwidth]{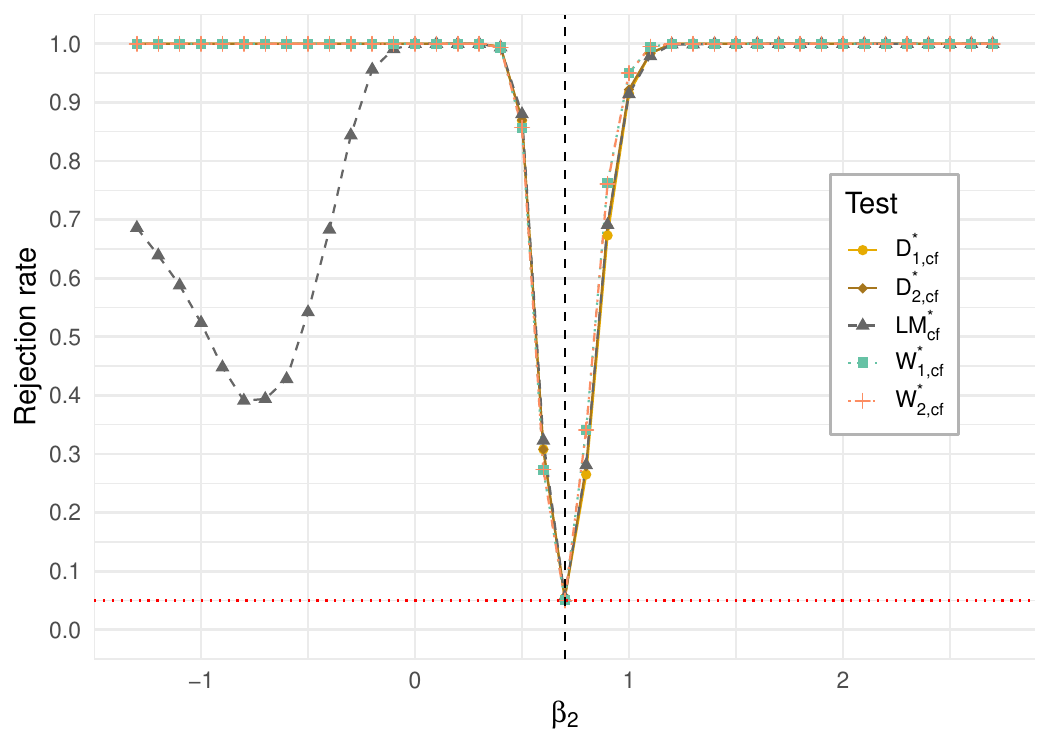}
        \caption{HLIM}
    \end{subfigure}
    
    \begin{subfigure}[b]{0.47\textwidth}
        \includegraphics[width=\textwidth]{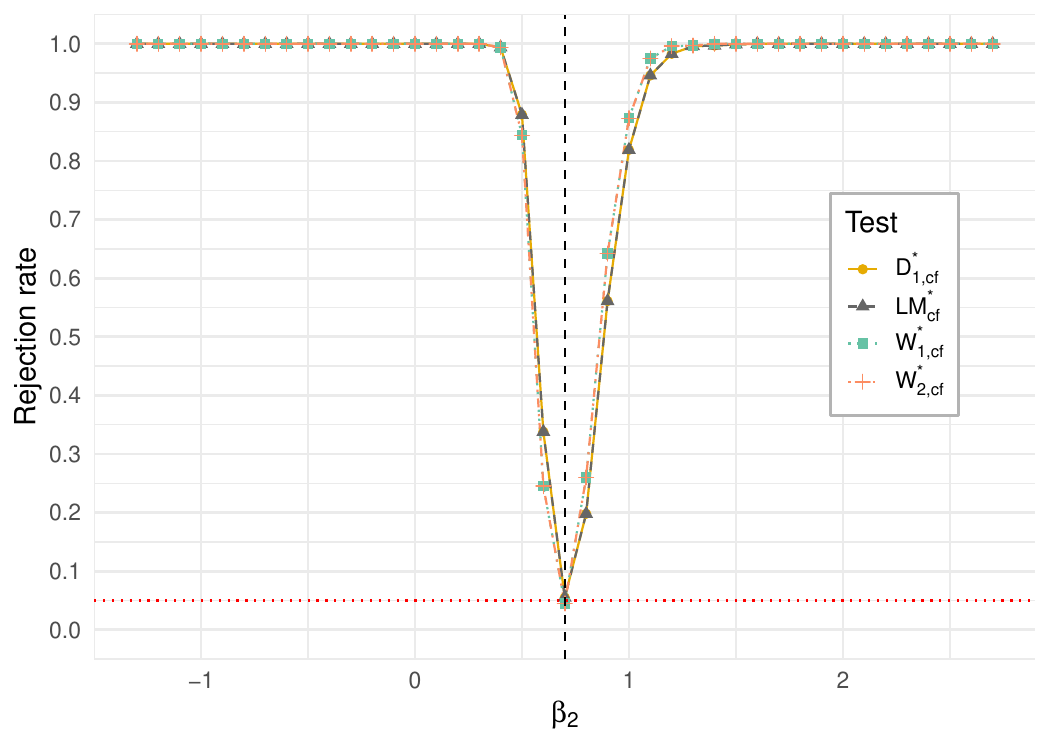}
        \caption{JIVE1}
    \end{subfigure}
    \hfill
    \begin{subfigure}[b]{0.47\textwidth}
        \includegraphics[width=\textwidth]{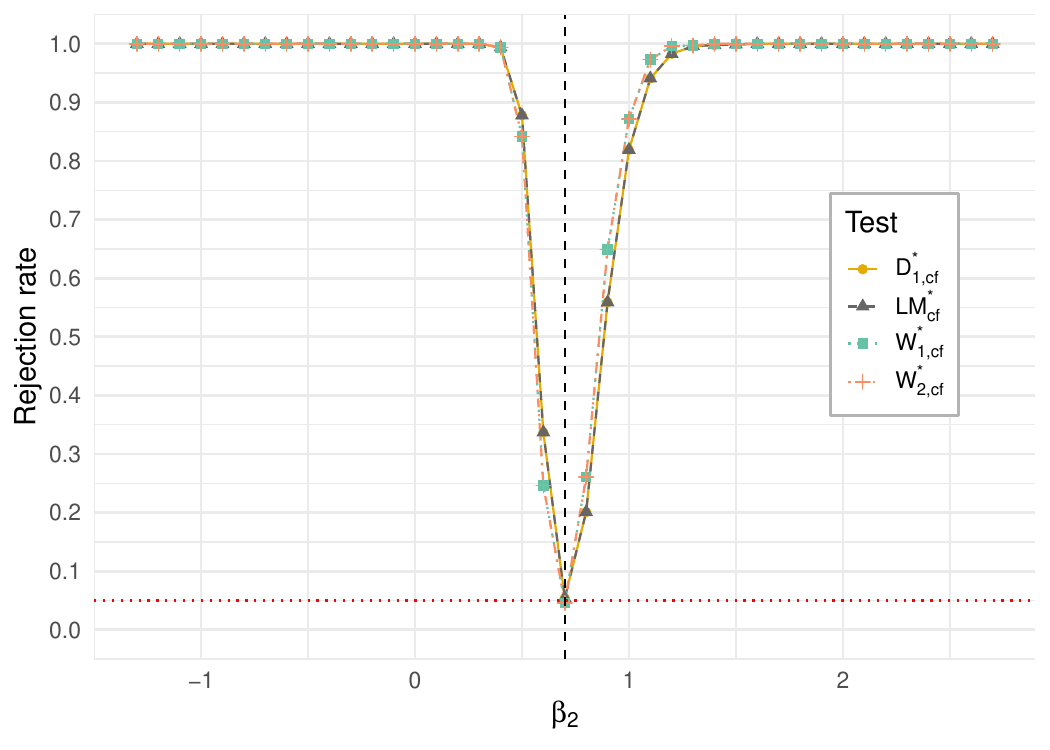}
        \caption{JIVE2}
    \end{subfigure}
    
    \caption{Power curves for DGP2 ($n=200$, $\alpha = 0.05$, $r = 0.2$). Trinity of test statistics distributed as a $\chi^2$ with cross-fit variance. Results based on 1000 repetitions. The horizontal dotted red line denotes the $5\%$ nominal rejection level, while the vertical dotted black line corresponds to $\beta=1$. Panel (a) plots statistics based on the SJIVE objective function; panel (b) plots statistics based on the HLIM objective function; panel (c) plots statistics based on the JIVE1 objective function; panel (d) plots statistics based on the JIVE2 objective function.}
    \label{fig:Trinity_chi2_cf_DGP2_200_0.05_0.2}
\end{figure}

%%%%%%%%%%%%%%%%%%%%%%%%%
% Trinity_200_0.1_0.1 

\begin{figure}[ht]
    \centering
    % Replace 'chibar2' with 'chi2' for the second set of figures as needed
    \begin{subfigure}[b]{0.47\textwidth}
        \includegraphics[width=\textwidth]{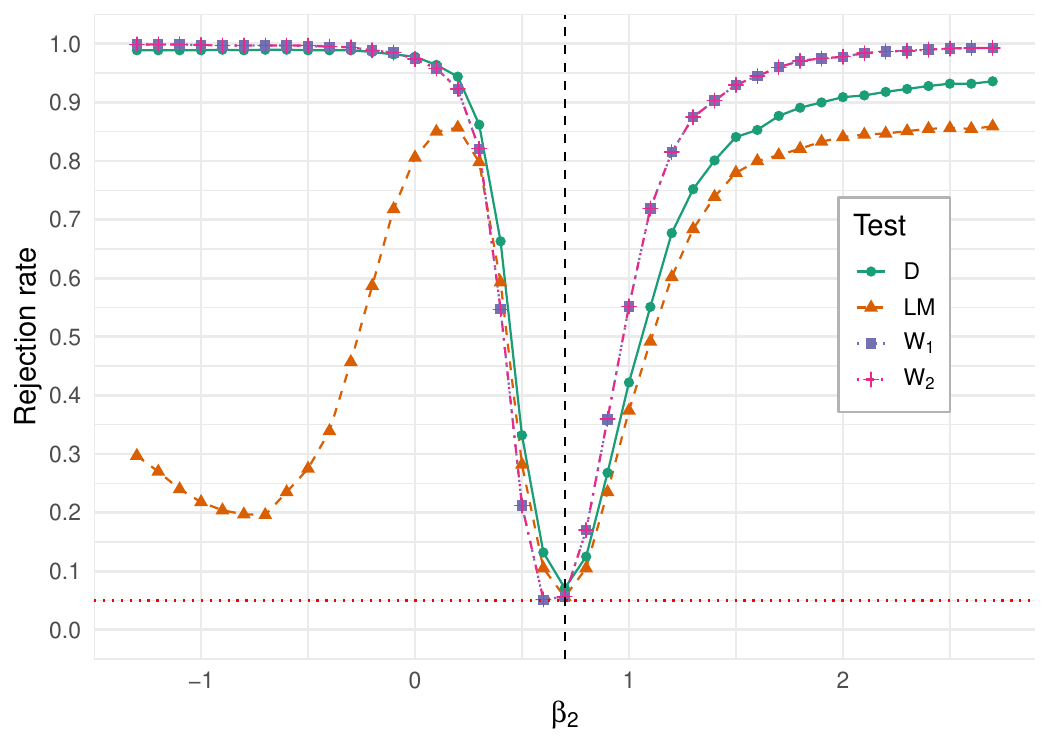}
        \caption{SJIVE}
    \end{subfigure}
    \hfill
    \begin{subfigure}[b]{0.47\textwidth}
        \includegraphics[width=\textwidth]{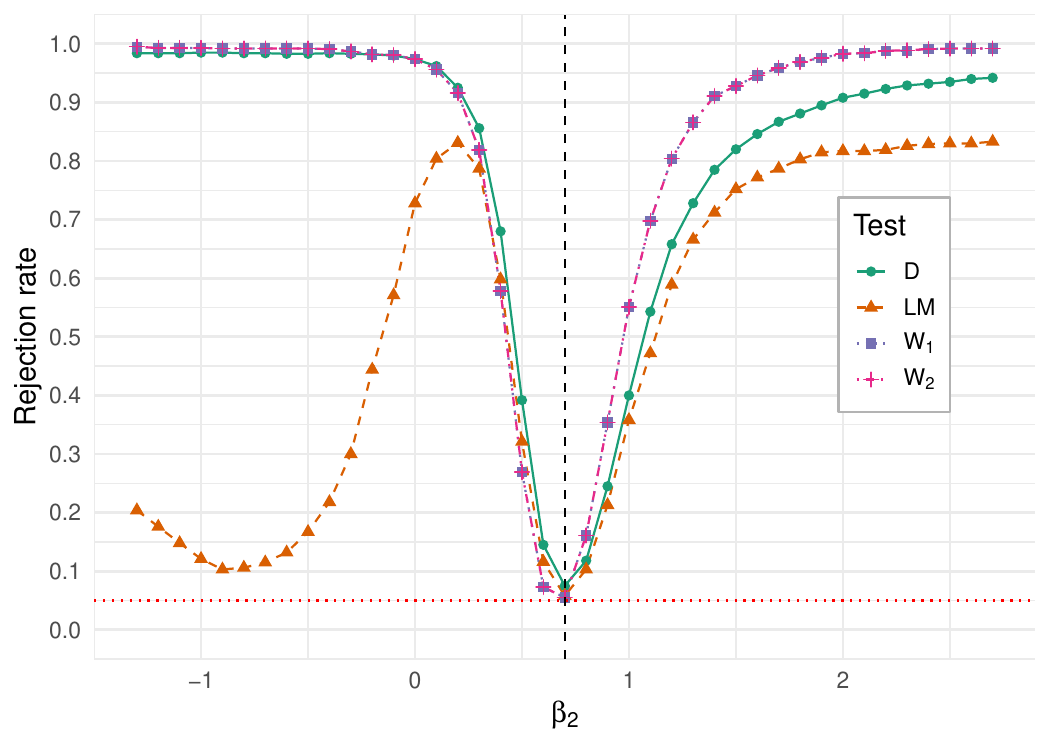}
        \caption{HLIM}
    \end{subfigure}
    
    \begin{subfigure}[b]{0.47\textwidth}
        \includegraphics[width=\textwidth]{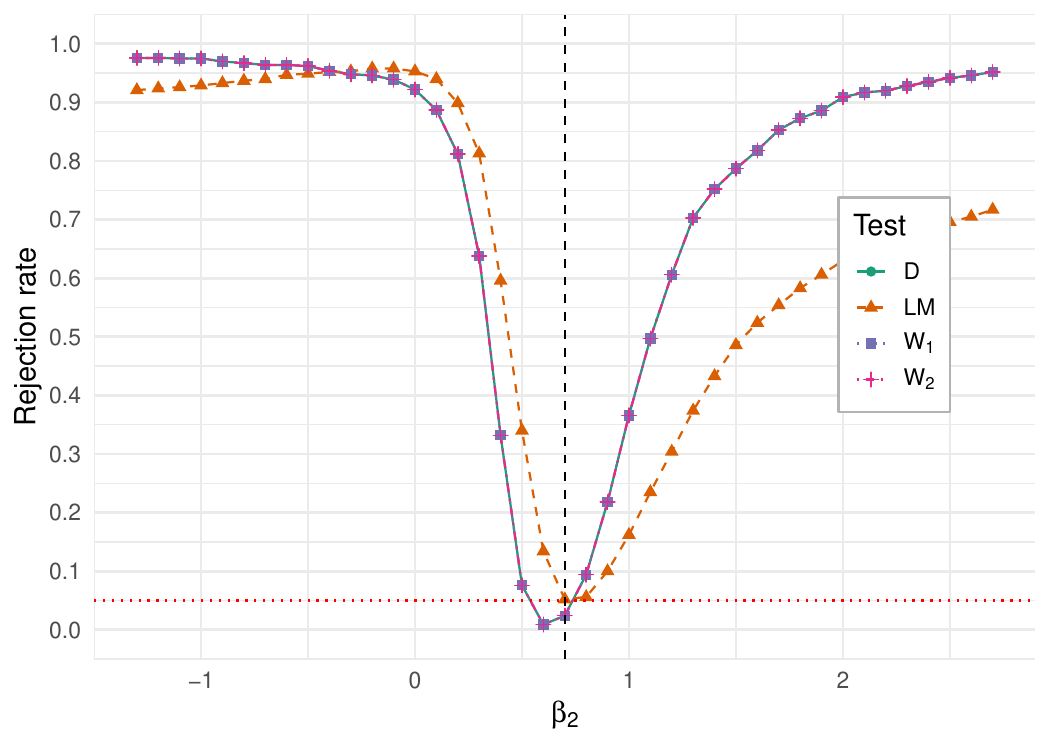}
        \caption{JIVE1}
    \end{subfigure}
    \hfill
    \begin{subfigure}[b]{0.47\textwidth}
        \includegraphics[width=\textwidth]{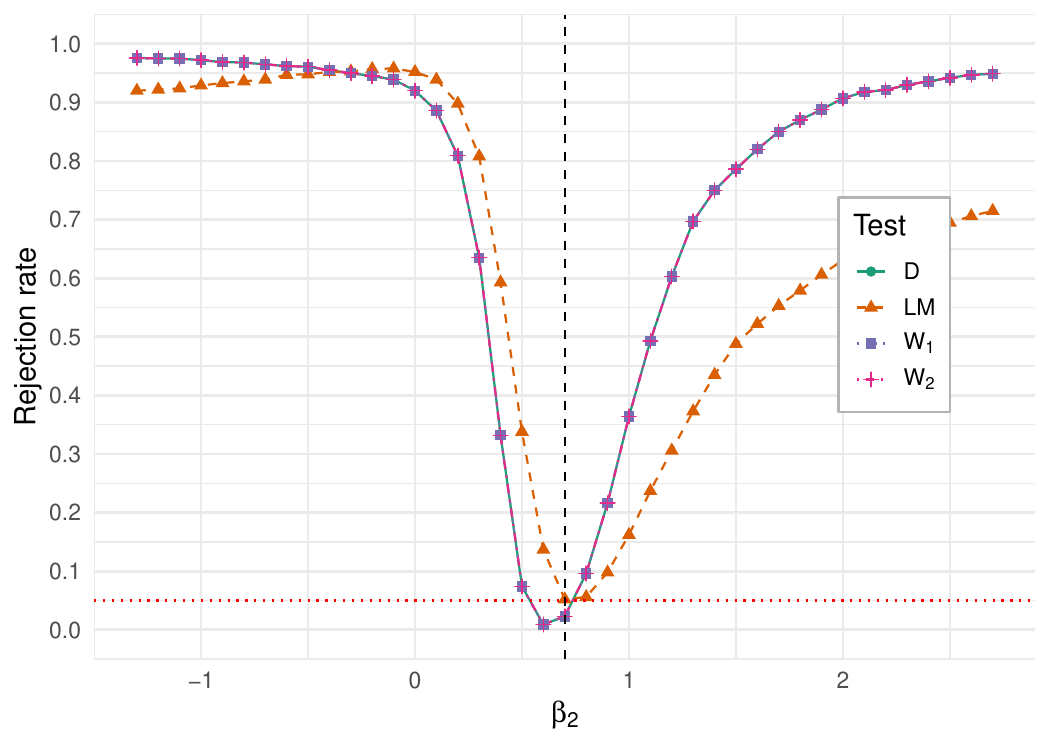}
        \caption{JIVE2}
    \end{subfigure}
    
    \caption{Power curves for DGP2 ($n=200$, $\alpha = 0.1$, $r = 0.1$). Trinity of test statistics distributed as a $\bar\chi^2$. Results based on 1000 repetitions. The horizontal dotted red line denotes the $5\%$ nominal rejection level, while the vertical dotted black line corresponds to $\beta=1$. Panel (a) plots statistics based on the SJIVE objective function; panel (b) plots statistics based on the HLIM objective function; panel (c) plots statistics based on the JIVE1 objective function; panel (d) plots statistics based on the JIVE2 objective function.}
    \label{fig:Trinity_chibar2_DGP2_200_0.1_0.1}
\end{figure}

\begin{figure}[ht]
    \centering
    % Replace 'chibar2' with 'chi2' for the second set of figures as needed
    \begin{subfigure}[b]{0.47\textwidth}
        \includegraphics[width=\textwidth]{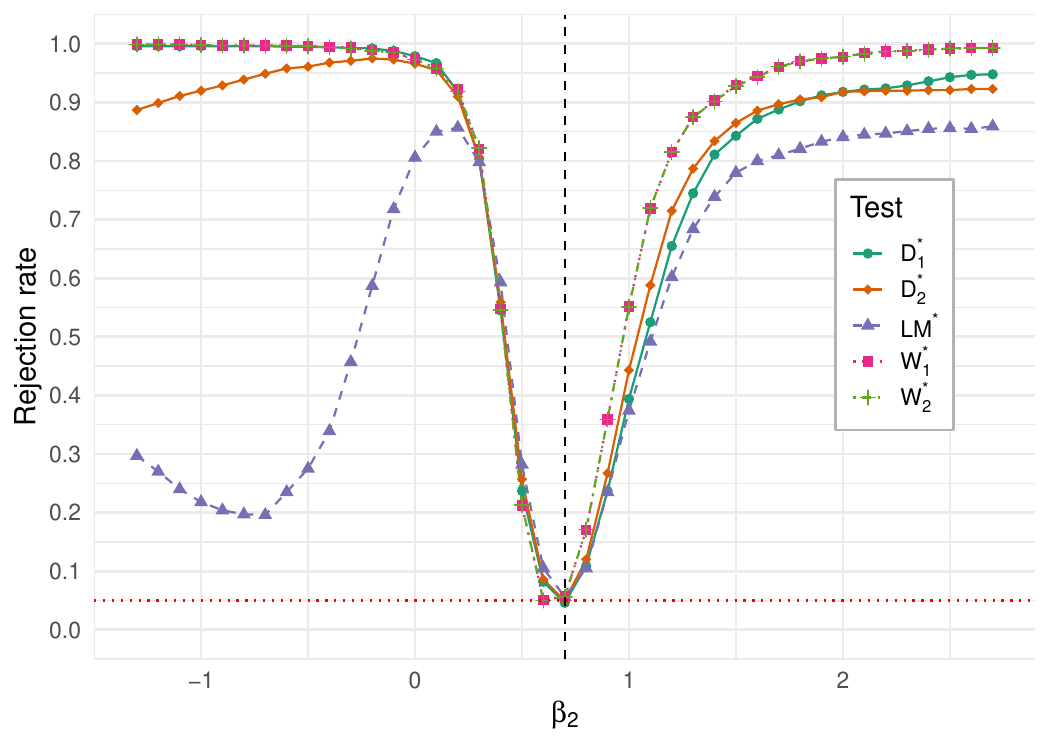}
        \caption{SJIVE}
    \end{subfigure}
    \hfill
    \begin{subfigure}[b]{0.47\textwidth}
        \includegraphics[width=\textwidth]{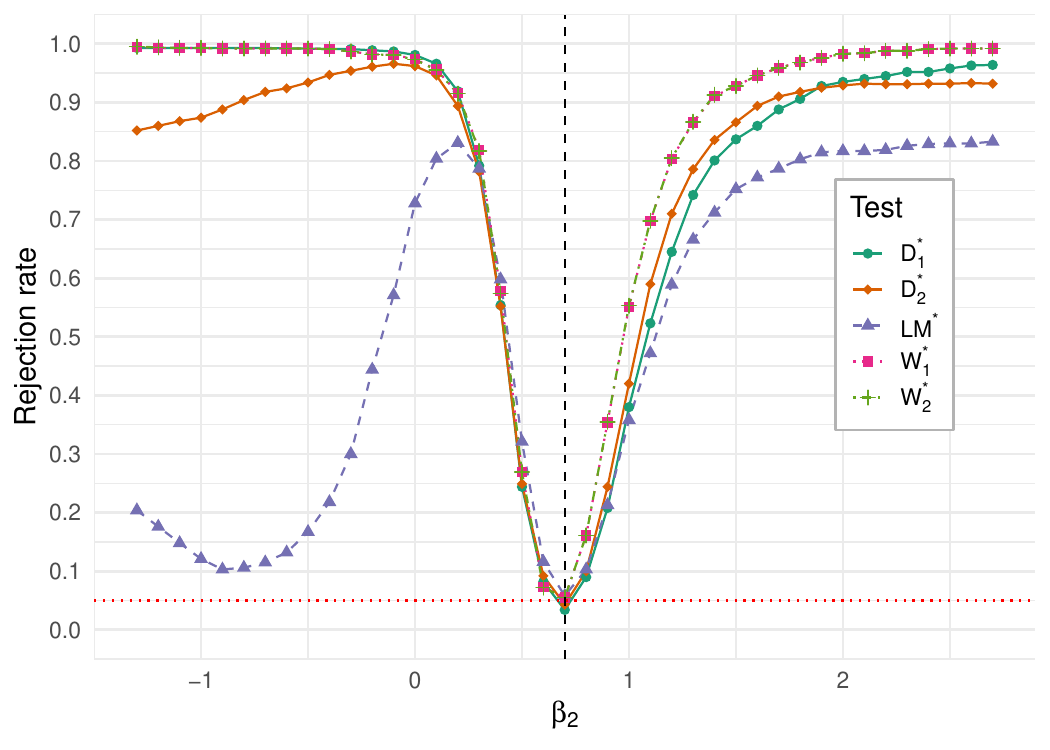}
        \caption{HLIM}
    \end{subfigure}
    
    \begin{subfigure}[b]{0.47\textwidth}
        \includegraphics[width=\textwidth]{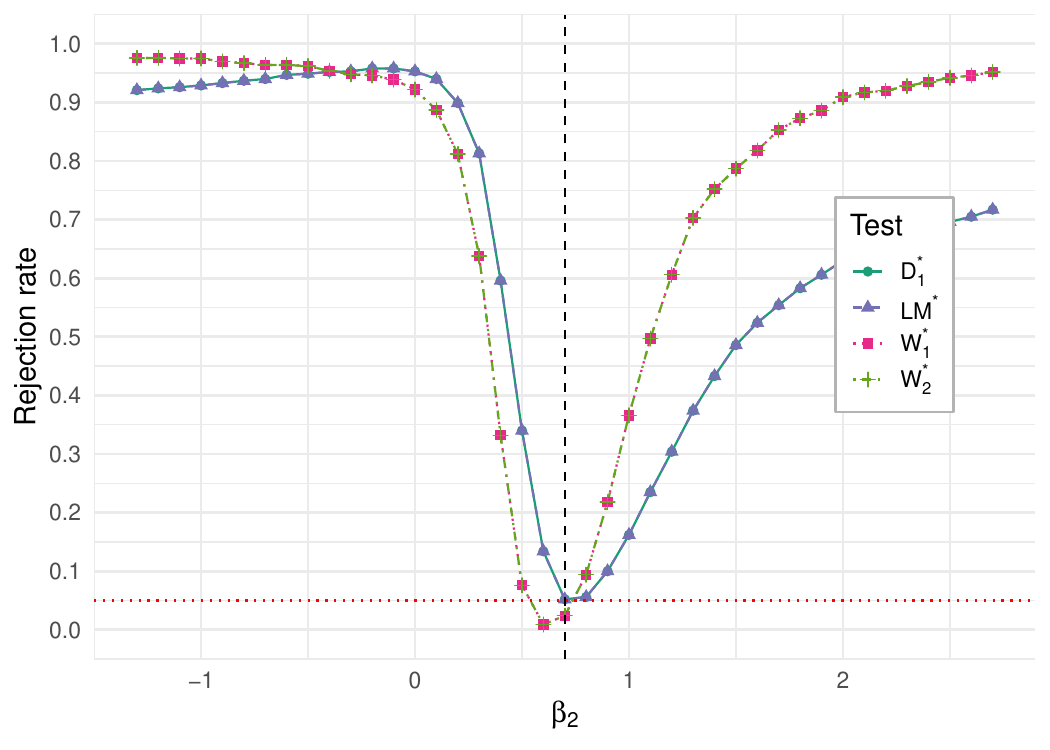}
        \caption{JIVE1}
    \end{subfigure}
    \hfill
    \begin{subfigure}[b]{0.47\textwidth}
        \includegraphics[width=\textwidth]{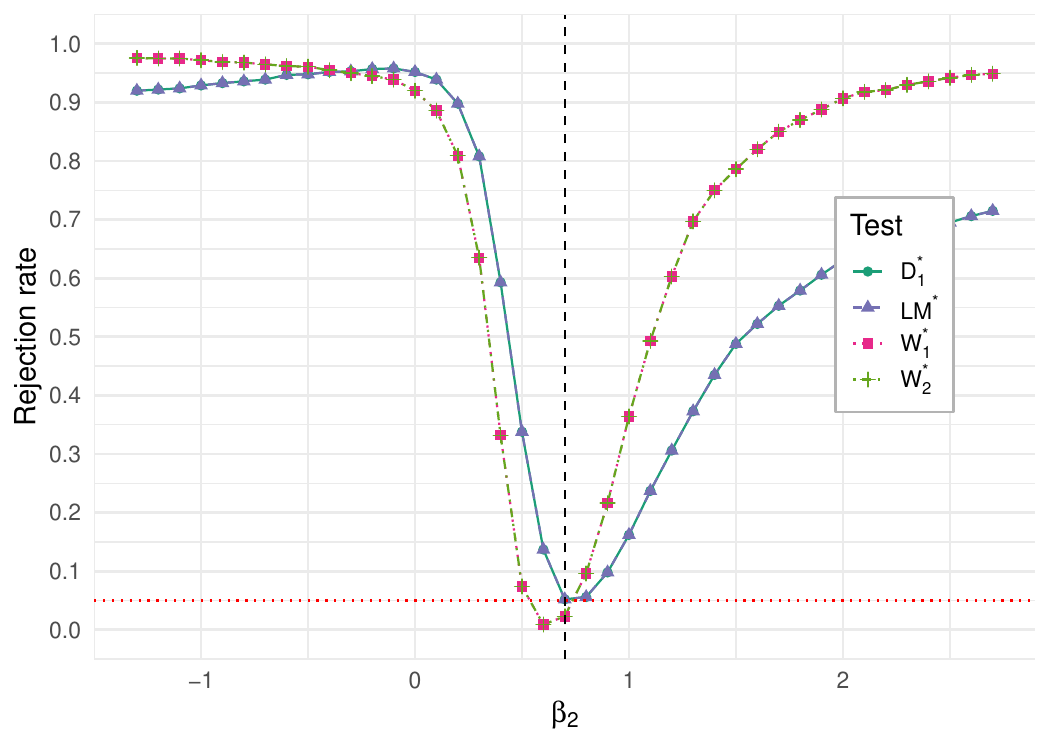}
        \caption{JIVE2}
    \end{subfigure}
    
    \caption{Power curves for DGP2 ($n=200$, $\alpha = 0.1$, $r = 0.1$). Trinity of test statistics distributed as a $\chi^2$. Results based on 1000 repetitions. The horizontal dotted red line denotes the $5\%$ nominal rejection level, while the vertical dotted black line corresponds to $\beta=1$. Panel (a) plots statistics based on the SJIVE objective function; panel (b) plots statistics based on the HLIM objective function; panel (c) plots statistics based on the JIVE1 objective function; panel (d) plots statistics based on the JIVE2 objective function.}
    \label{fig:Trinity_chi2_DGP2_200_0.1_0.1}
\end{figure}

% cf variance

\begin{figure}[ht]
    \centering
    % Replace 'chibar2' with 'chi2' for the second set of figures as needed
    \begin{subfigure}[b]{0.47\textwidth}
        \includegraphics[width=\textwidth]{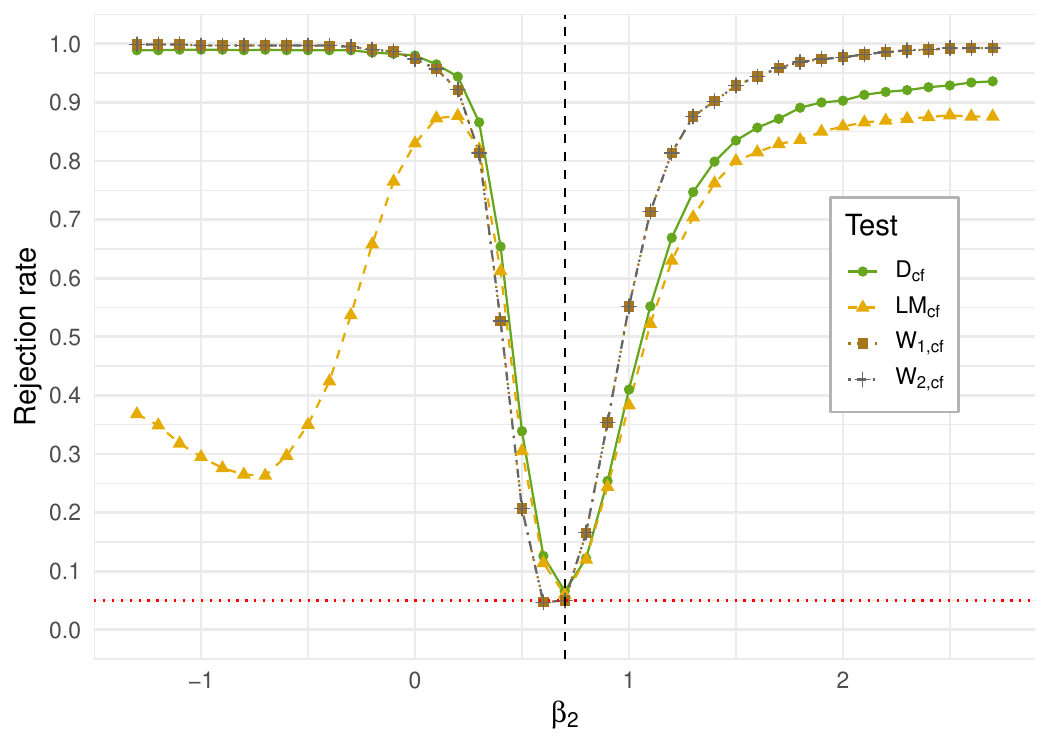}
        \caption{SJIVE}
    \end{subfigure}
    \hfill
    \begin{subfigure}[b]{0.47\textwidth}
        \includegraphics[width=\textwidth]{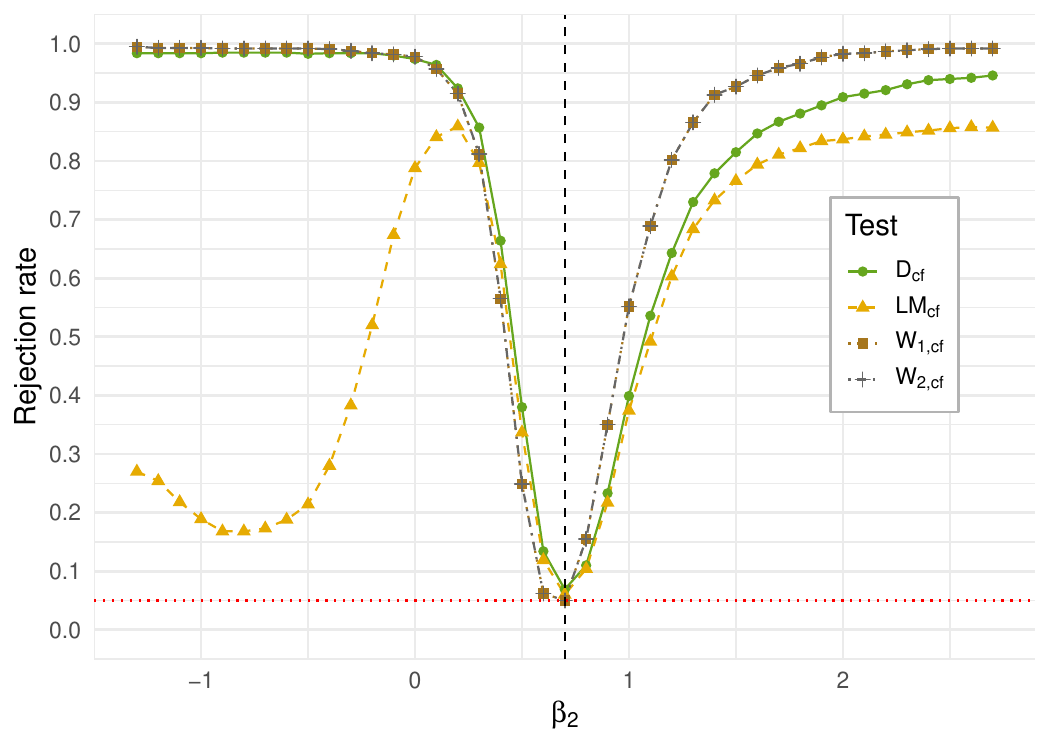}
        \caption{HLIM}
    \end{subfigure}
    
    \begin{subfigure}[b]{0.47\textwidth}
       \includegraphics[width=\textwidth]{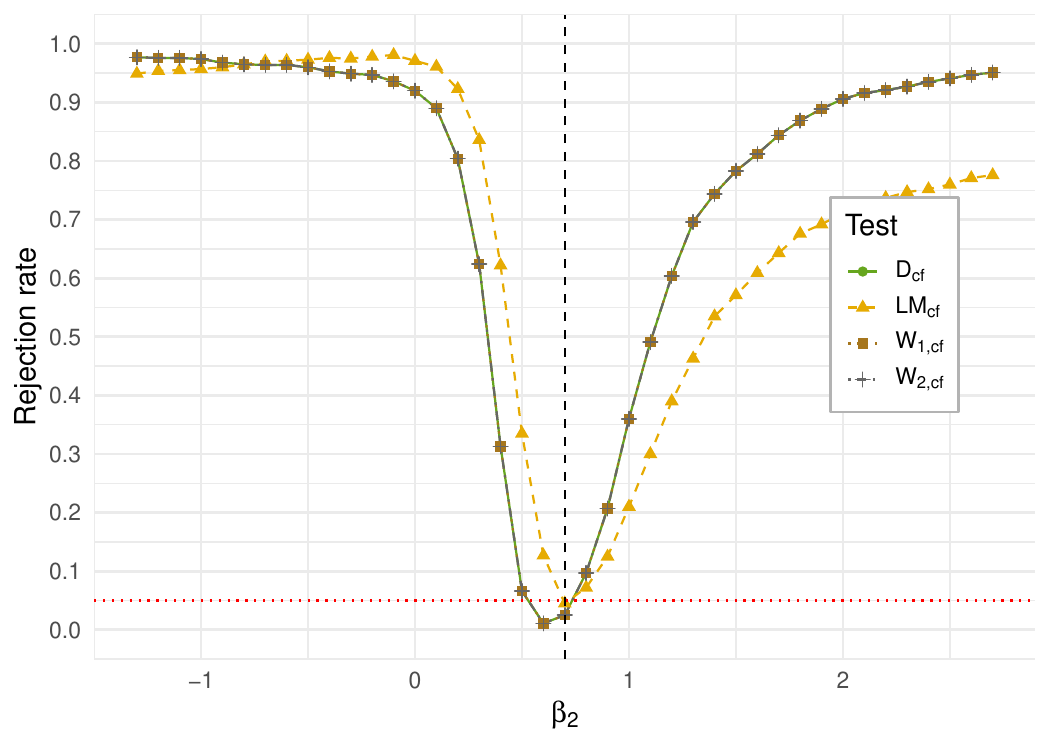}
        \caption{JIVE1}
    \end{subfigure}
    \hfill
    \begin{subfigure}[b]{0.47\textwidth}
        \includegraphics[width=\textwidth]{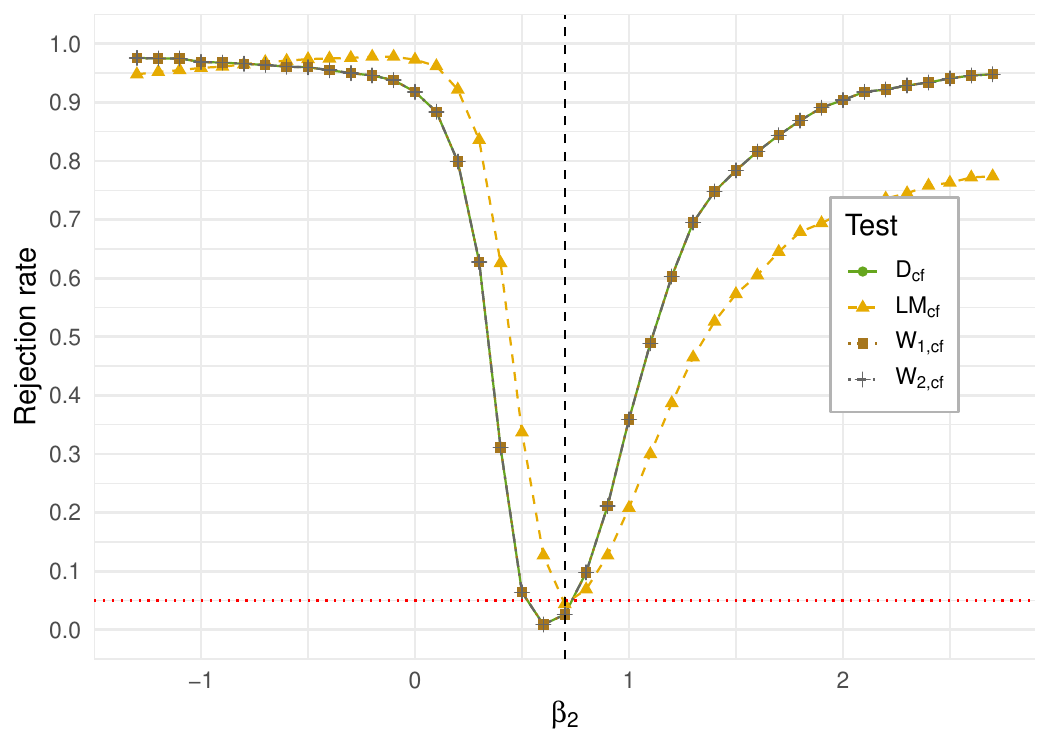}
        \caption{JIVE2}
    \end{subfigure}
    
    \caption{Power curves for DGP2 ($n=200$, $\alpha = 0.1$, $r = 0.1$). Trinity of test statistics distributed as a $\bar\chi^2$ with cross-fit variance. Results based on 1000 repetitions. The horizontal dotted red line denotes the $5\%$ nominal rejection level, while the vertical dotted black line corresponds to $\beta=1$. Panel (a) plots statistics based on the SJIVE objective function; panel (b) plots statistics based on the HLIM objective function; panel (c) plots statistics based on the JIVE1 objective function; panel (d) plots statistics based on the JIVE2 objective function.}
    \label{fig:Trinity_chibar2_cf_DGP2_200_0.1_0.1}
\end{figure}

\begin{figure}[ht]
    \centering
    % Replace 'chibar2' with 'chi2' for the second set of figures as needed
    \begin{subfigure}[b]{0.47\textwidth}
        \includegraphics[width=\textwidth]{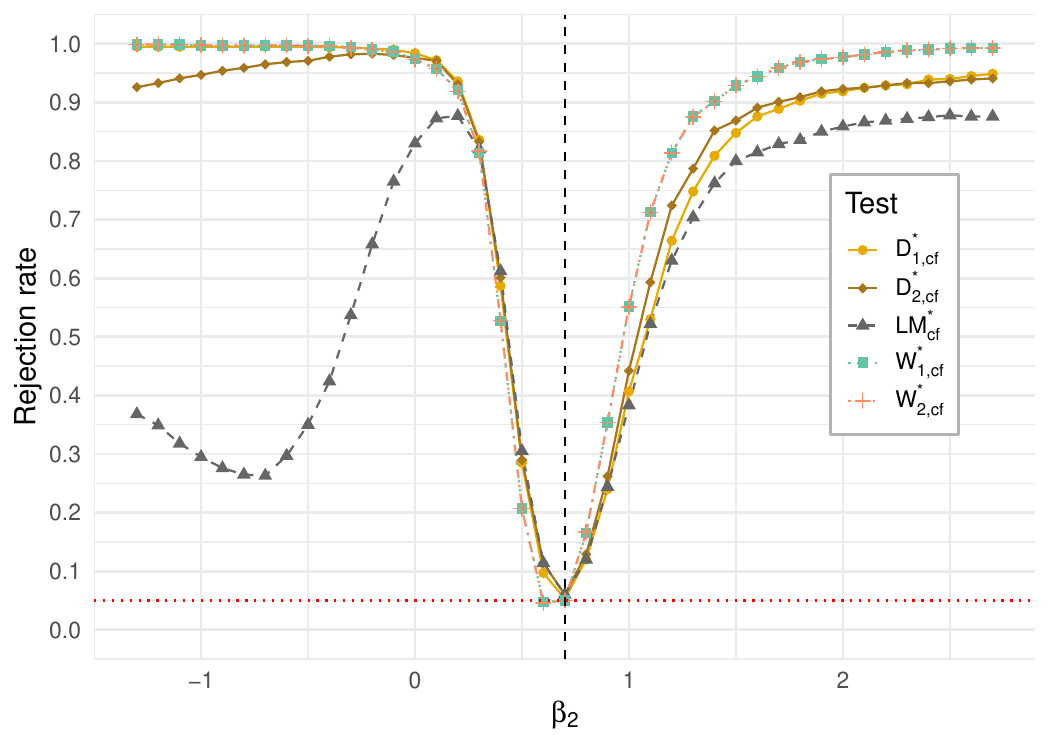}
        \caption{SJIVE}
    \end{subfigure}
    \hfill
    \begin{subfigure}[b]{0.47\textwidth}
        \includegraphics[width=\textwidth]{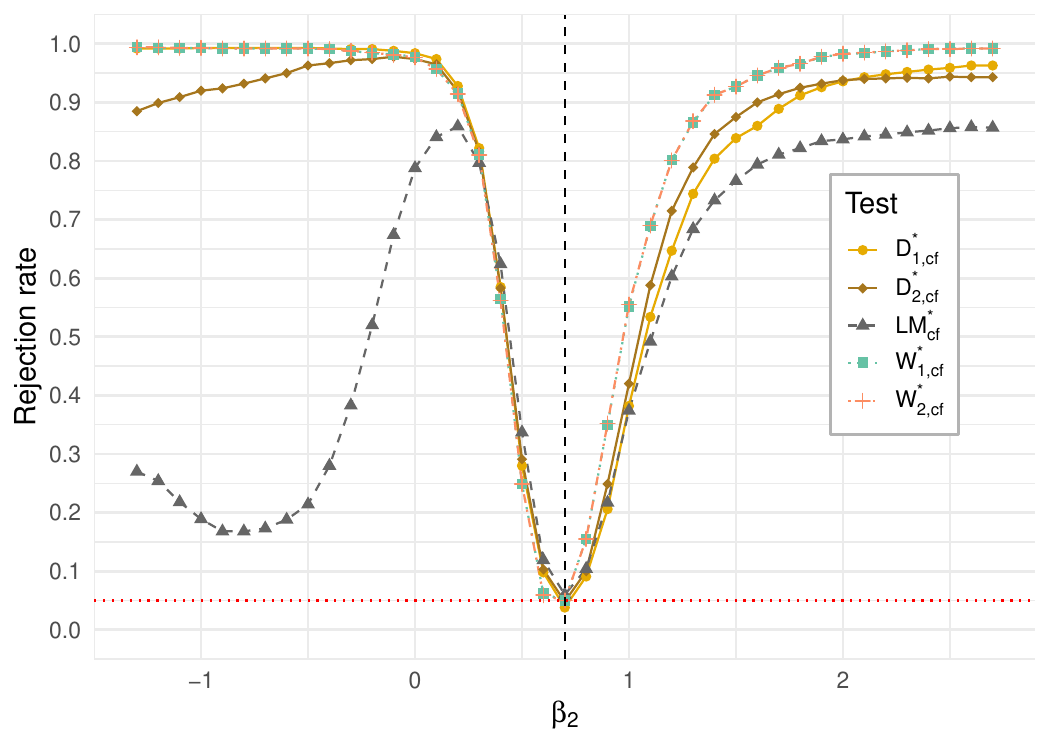}
        \caption{HLIM}
    \end{subfigure}
    
    \begin{subfigure}[b]{0.47\textwidth}
        \includegraphics[width=\textwidth]{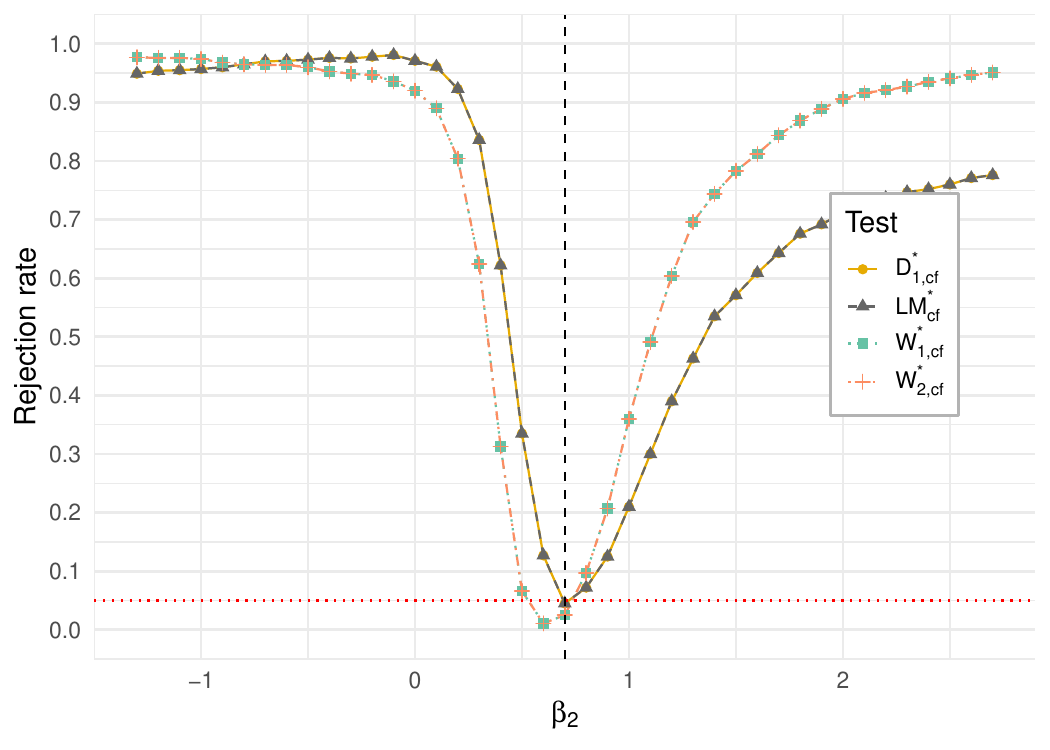}
        \caption{JIVE1}
    \end{subfigure}
    \hfill
    \begin{subfigure}[b]{0.47\textwidth}
        \includegraphics[width=\textwidth]{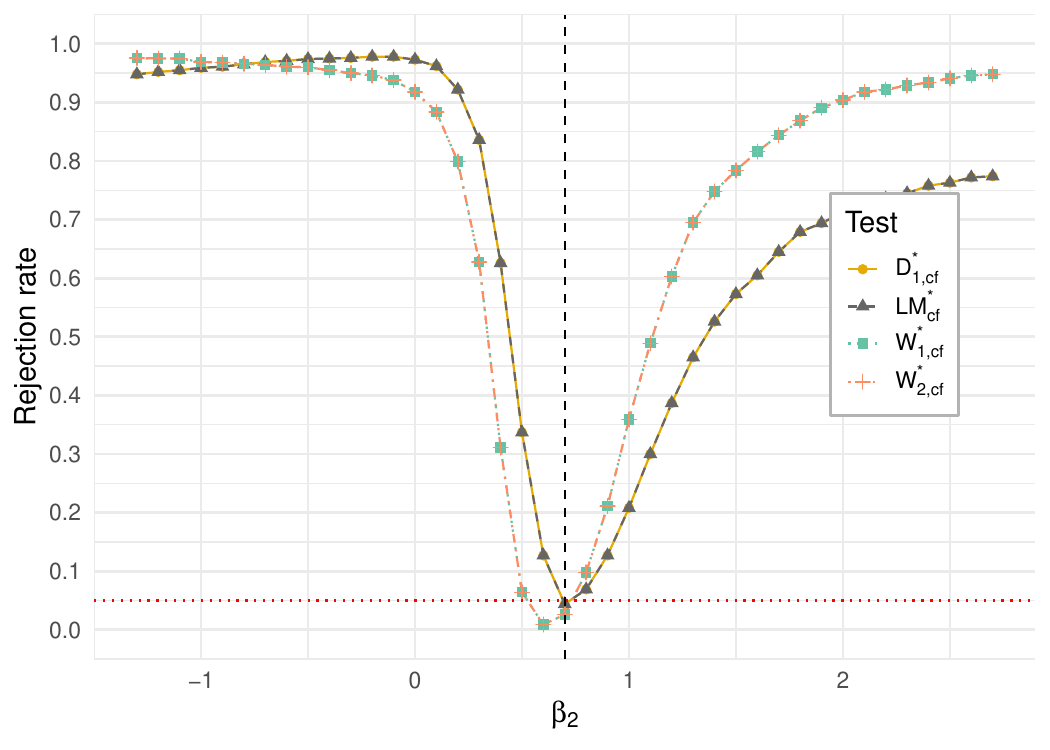}
        \caption{JIVE2}
    \end{subfigure}
    
    \caption{Power curves for DGP2 ($n=200$, $\alpha = 0.1$, $r = 0.1$). Trinity of test statistics distributed as a $\chi^2$ with cross-fit variance. Results based on 1000 repetitions. The horizontal dotted red line denotes the $5\%$ nominal rejection level, while the vertical dotted black line corresponds to $\beta=1$. Panel (a) plots statistics based on the SJIVE objective function; panel (b) plots statistics based on the HLIM objective function; panel (c) plots statistics based on the JIVE1 objective function; panel (d) plots statistics based on the JIVE2 objective function.}
    \label{fig:Trinity_chi2_cf_DGP2_200_0.1_0.1}
\end{figure}

%%%%%%%%%%%%%%%%%%%%%%%%%
% Trinity_100_0.1_0.2 

\begin{figure}[ht]
    \centering
    % Replace 'chibar2' with 'chi2' for the second set of figures as needed
    \begin{subfigure}[b]{0.47\textwidth}
        \includegraphics[width=\textwidth]{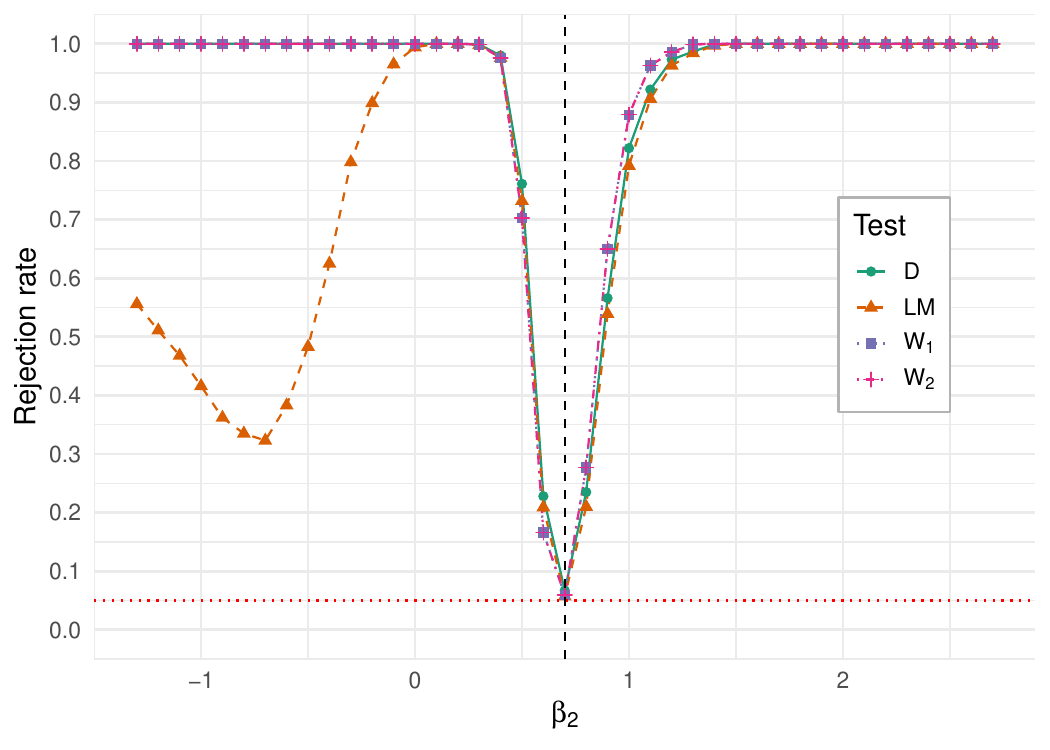}
        \caption{SJIVE}
    \end{subfigure}
    \hfill
    \begin{subfigure}[b]{0.47\textwidth}
        \includegraphics[width=\textwidth]{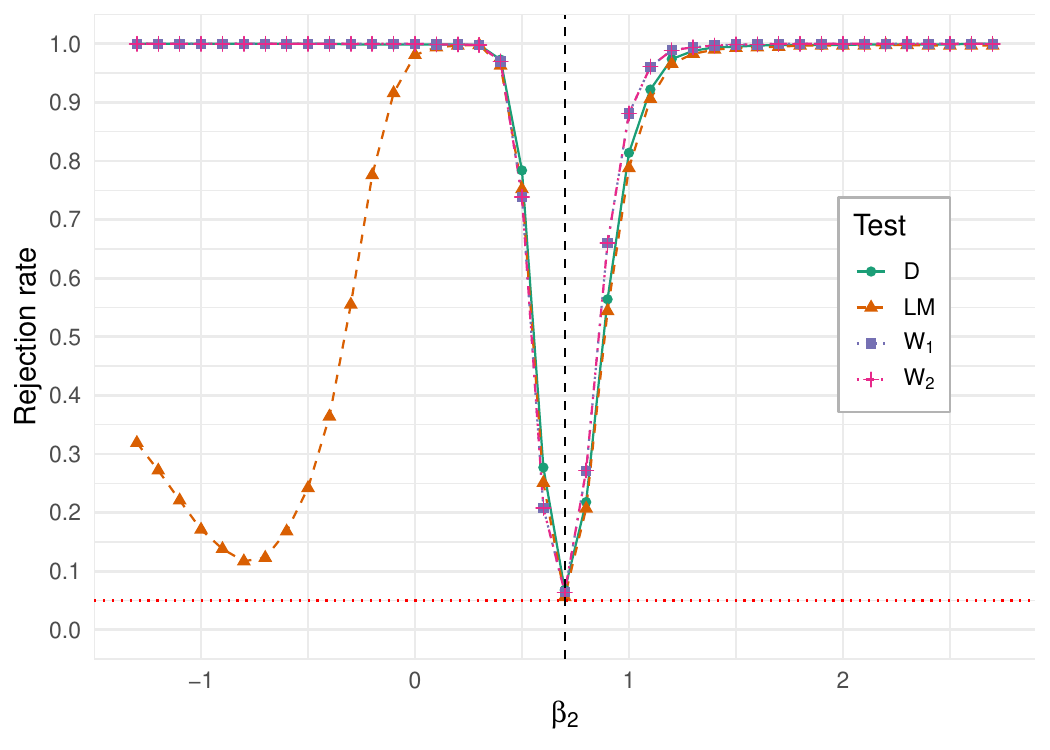}
        \caption{HLIM}
    \end{subfigure}
    
    \begin{subfigure}[b]{0.47\textwidth}
        \includegraphics[width=\textwidth]{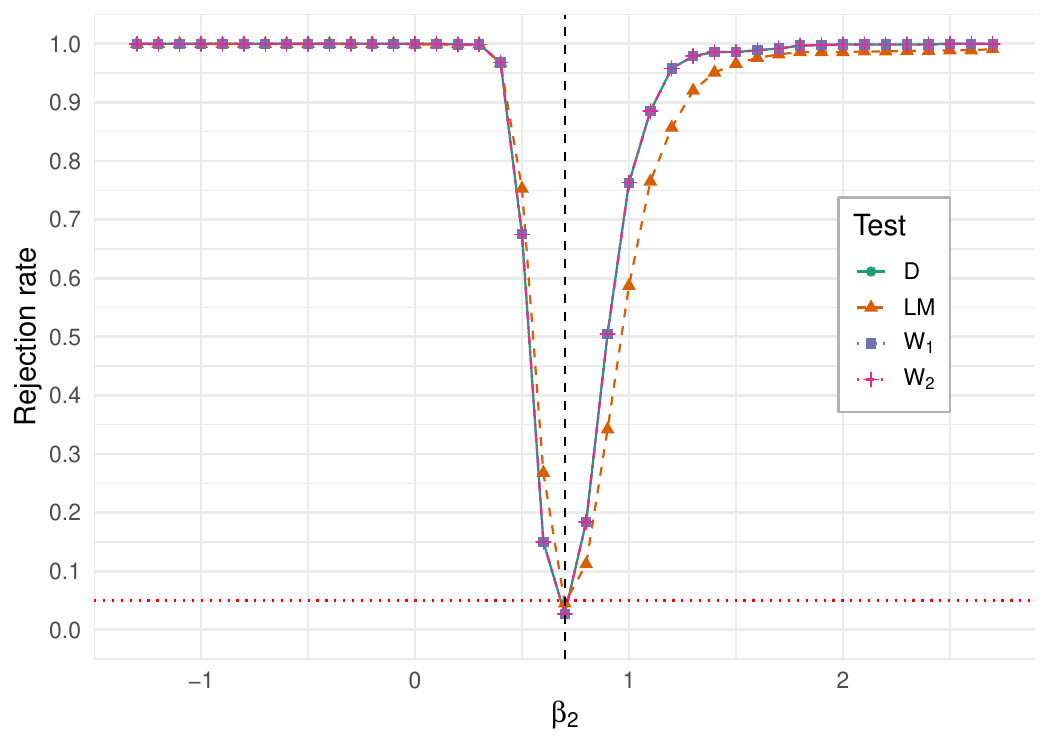}
        \caption{JIVE1}
    \end{subfigure}
    \hfill
    \begin{subfigure}[b]{0.47\textwidth}
        \includegraphics[width=\textwidth]{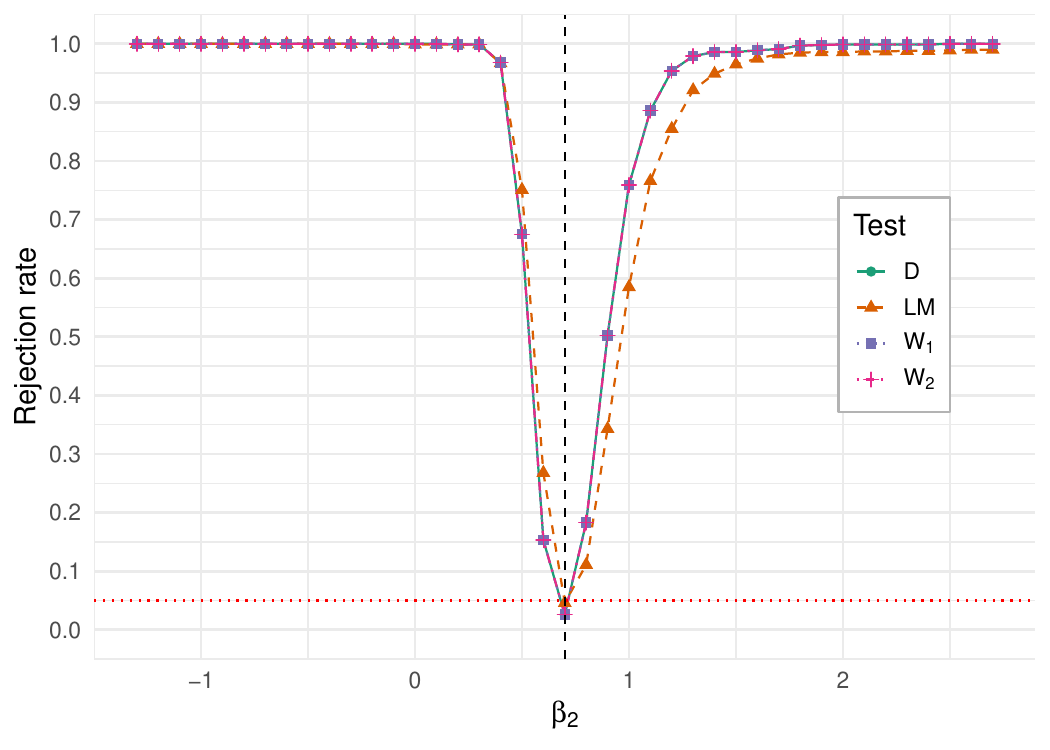}
        \caption{JIVE2}
    \end{subfigure}
    
    \caption{Power curves for DGP2 ($n=200$, $\alpha = 0.1$, $r = 0.2$). Trinity of test statistics distributed as a $\bar\chi^2$. Results based on 1000 repetitions. The horizontal dotted red line denotes the $5\%$ nominal rejection level, while the vertical dotted black line corresponds to $\beta=1$. Panel (a) plots statistics based on the SJIVE objective function; panel (b) plots statistics based on the HLIM objective function; panel (c) plots statistics based on the JIVE1 objective function; panel (d) plots statistics based on the JIVE2 objective function.}
    \label{fig:Trinity_chibar2_DGP2_200_0.1_0.2}
\end{figure}

\begin{figure}[ht]
    \centering
    % Replace 'chibar2' with 'chi2' for the second set of figures as needed
    \begin{subfigure}[b]{0.47\textwidth}
        \includegraphics[width=\textwidth]{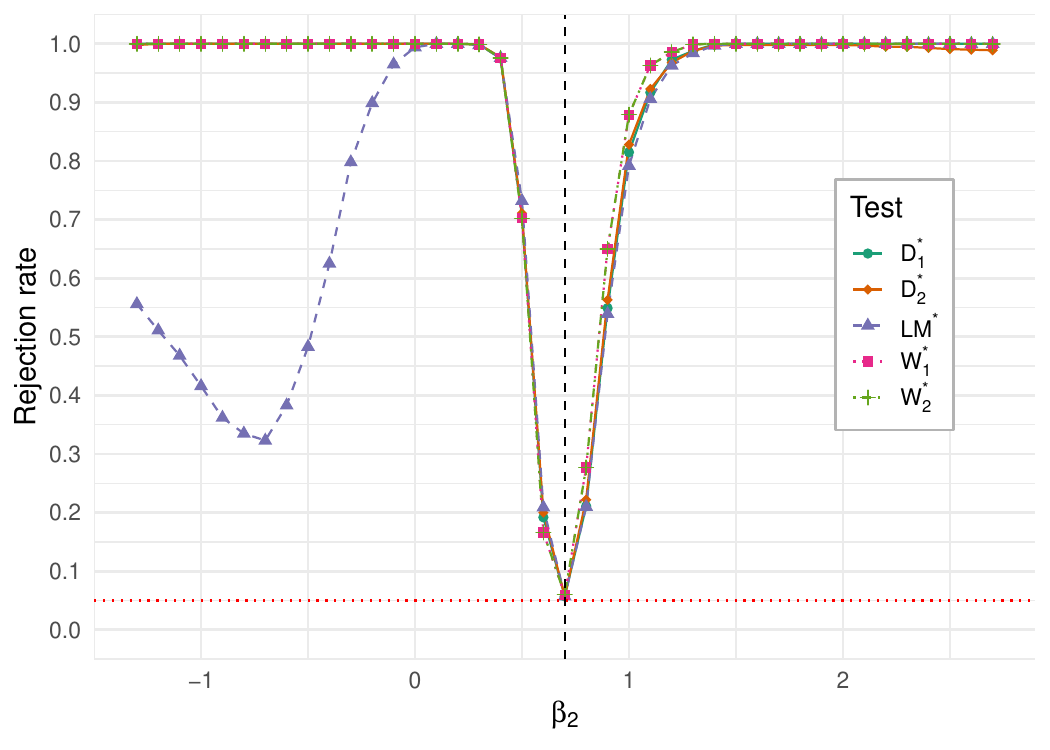}
        \caption{SJIVE}
    \end{subfigure}
    \hfill
    \begin{subfigure}[b]{0.47\textwidth}
        \includegraphics[width=\textwidth]{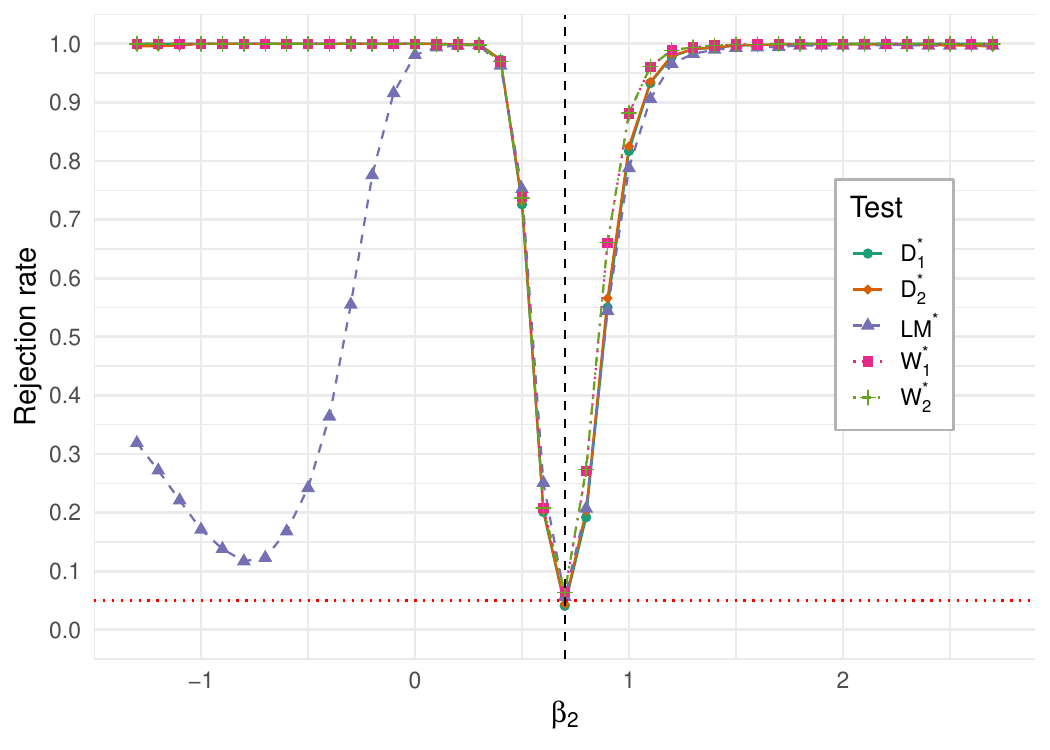}
        \caption{HLIM}
    \end{subfigure}
    
    \begin{subfigure}[b]{0.47\textwidth}
        \includegraphics[width=\textwidth]{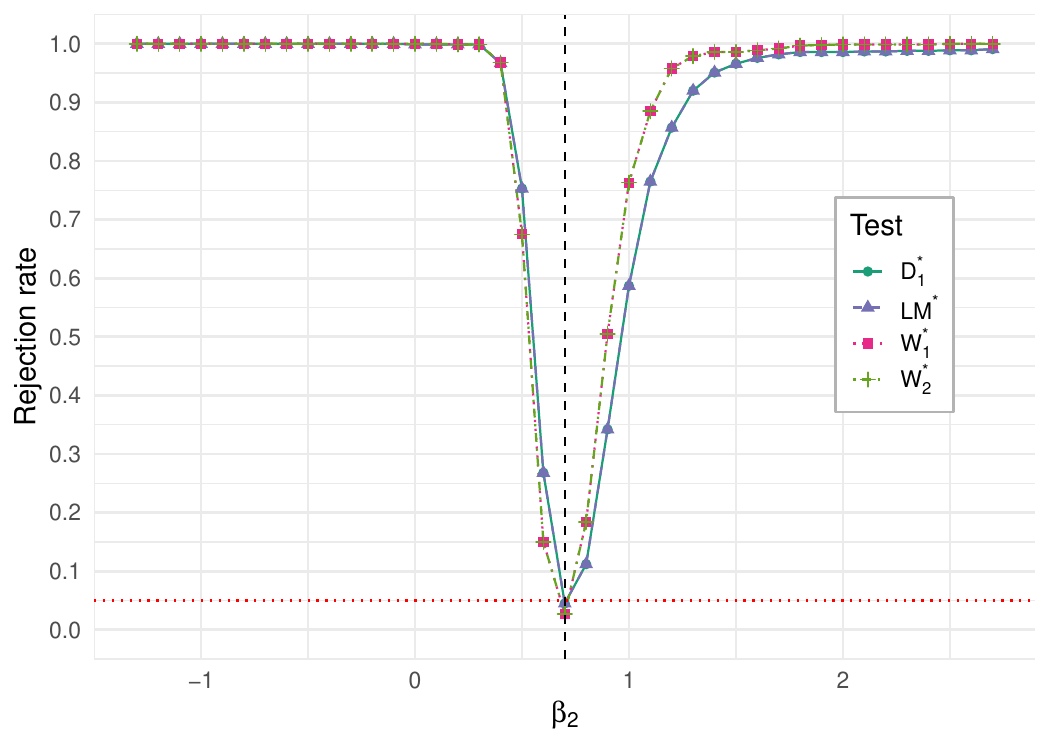}
        \caption{JIVE1}
    \end{subfigure}
    \hfill
    \begin{subfigure}[b]{0.47\textwidth}
        \includegraphics[width=\textwidth]{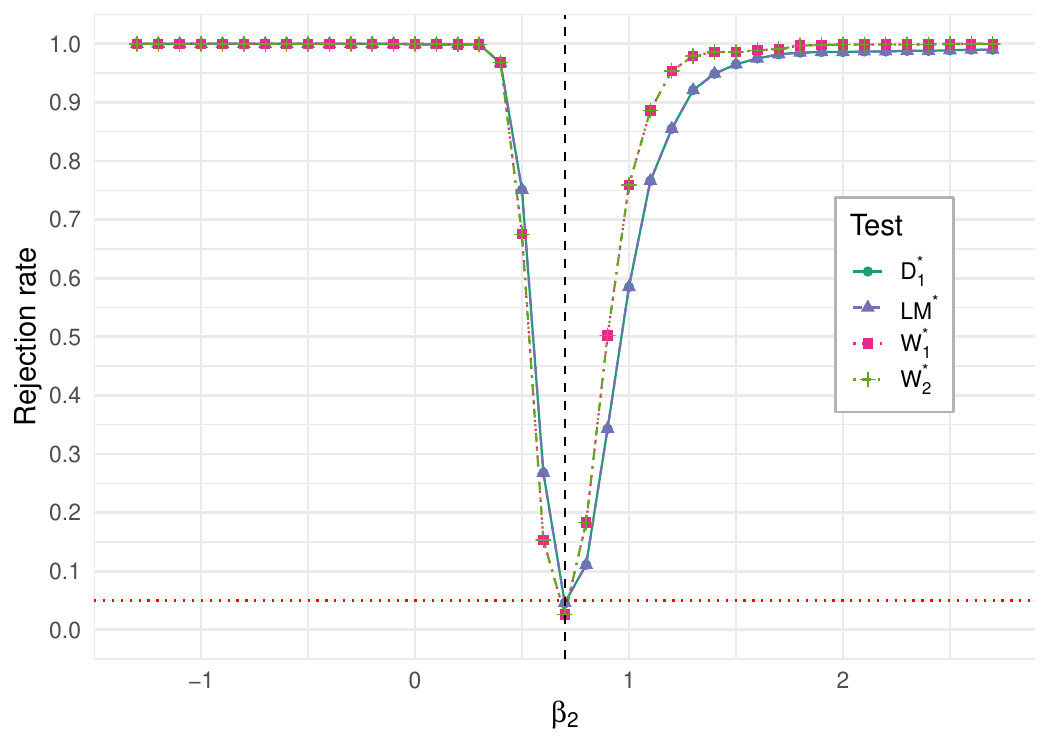}
        \caption{JIVE2}
    \end{subfigure}
    
    \caption{Power curves for DGP2 ($n=200$, $\alpha = 0.1$, $r = 0.2$). Trinity of test statistics distributed as a $\chi^2$. Results based on 1000 repetitions. The horizontal dotted red line denotes the $5\%$ nominal rejection level, while the vertical dotted black line corresponds to $\beta=1$. Panel (a) plots statistics based on the SJIVE objective function; panel (b) plots statistics based on the HLIM objective function; panel (c) plots statistics based on the JIVE1 objective function; panel (d) plots statistics based on the JIVE2 objective function.}
    \label{fig:Trinity_chi2_DGP2_200_0.1_0.2}
\end{figure}

% cf variance

\begin{figure}[ht]
    \centering
    % Replace 'chibar2' with 'chi2' for the second set of figures as needed
    \begin{subfigure}[b]{0.47\textwidth}
        \includegraphics[width=\textwidth]{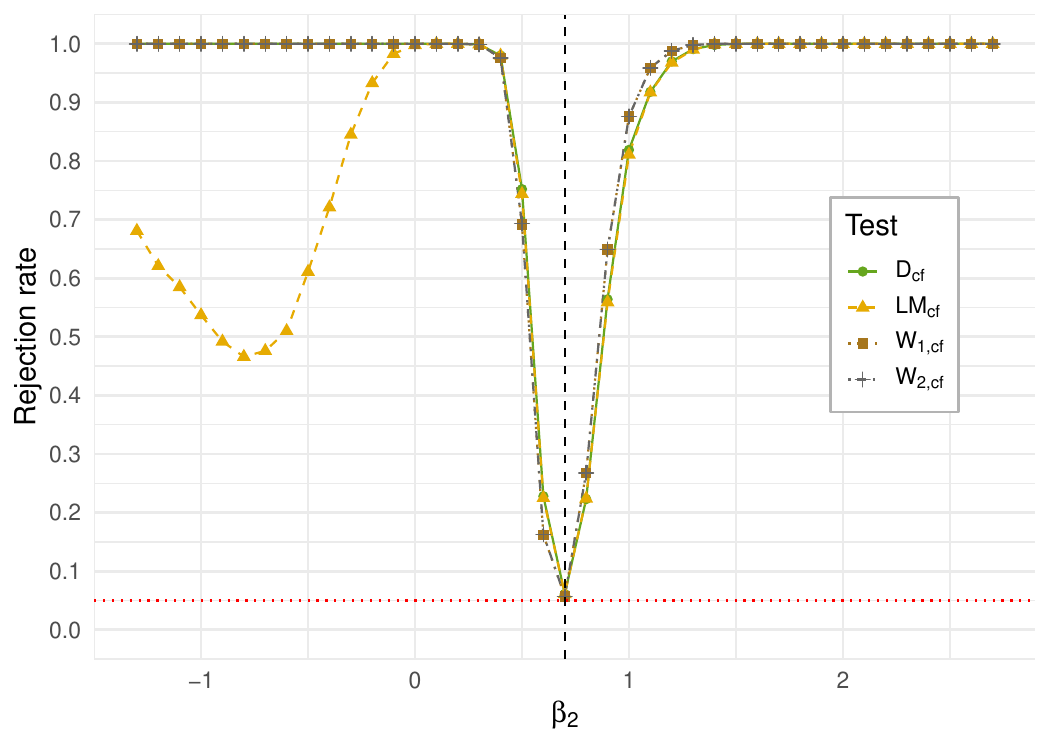}
        \caption{SJIVE}
    \end{subfigure}
    \hfill
    \begin{subfigure}[b]{0.47\textwidth}
        \includegraphics[width=\textwidth]{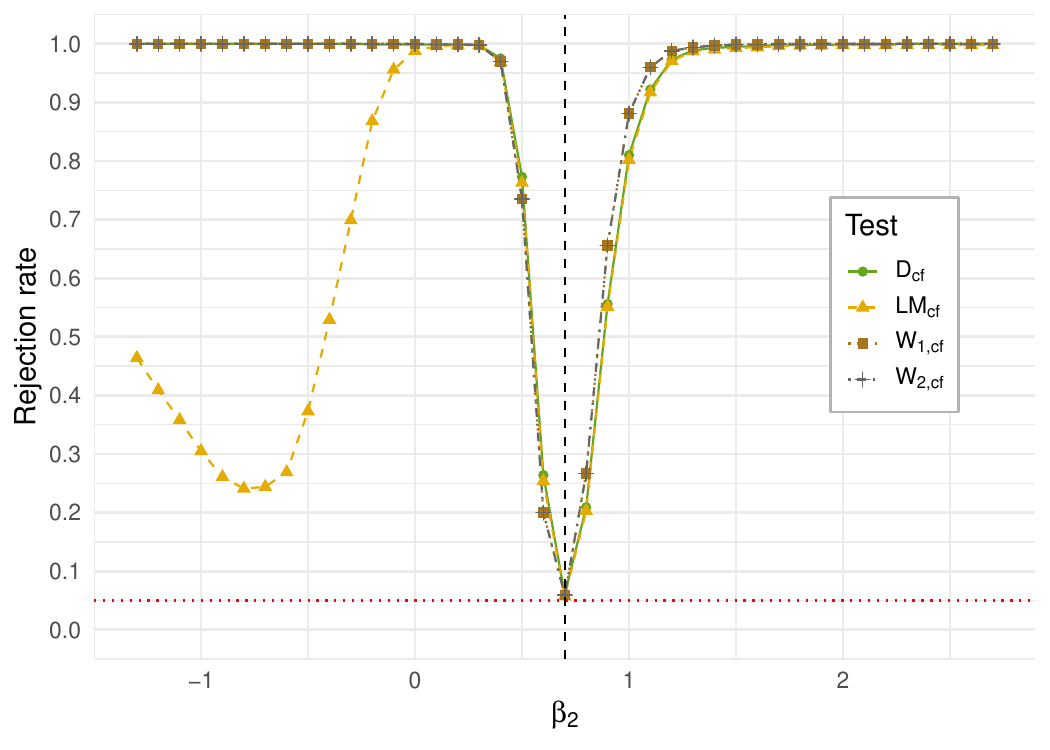}
        \caption{HLIM}
    \end{subfigure}
    
    \begin{subfigure}[b]{0.47\textwidth}
        \includegraphics[width=\textwidth]{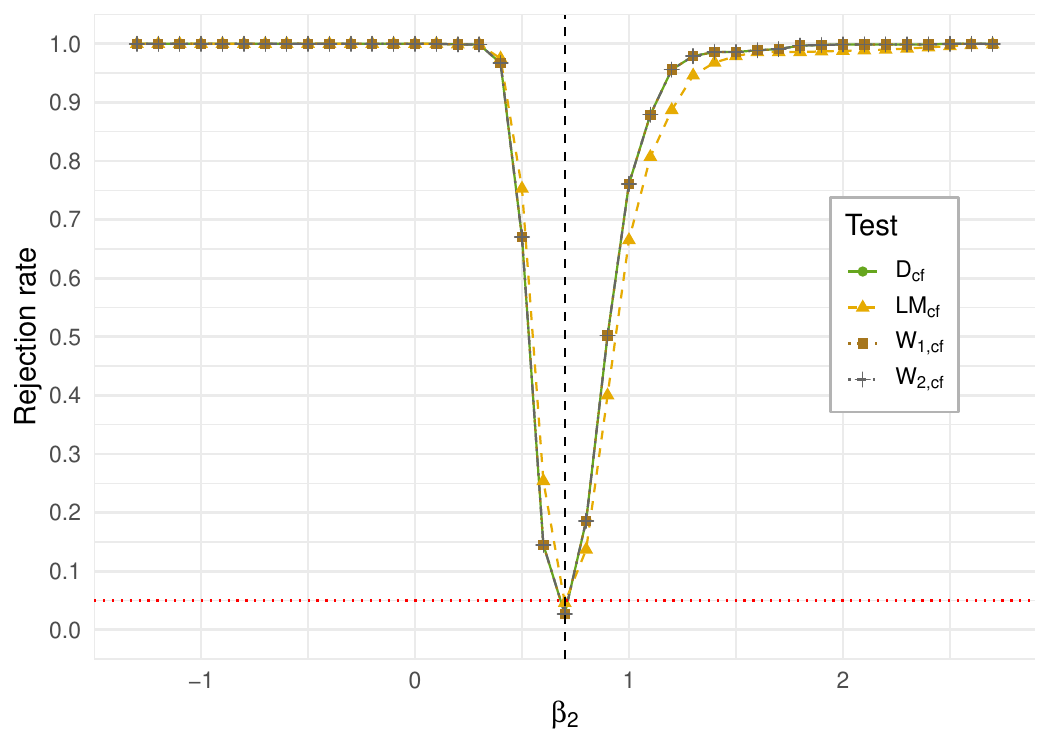}
        \caption{JIVE1}
    \end{subfigure}
    \hfill
    \begin{subfigure}[b]{0.47\textwidth}
        \includegraphics[width=\textwidth]{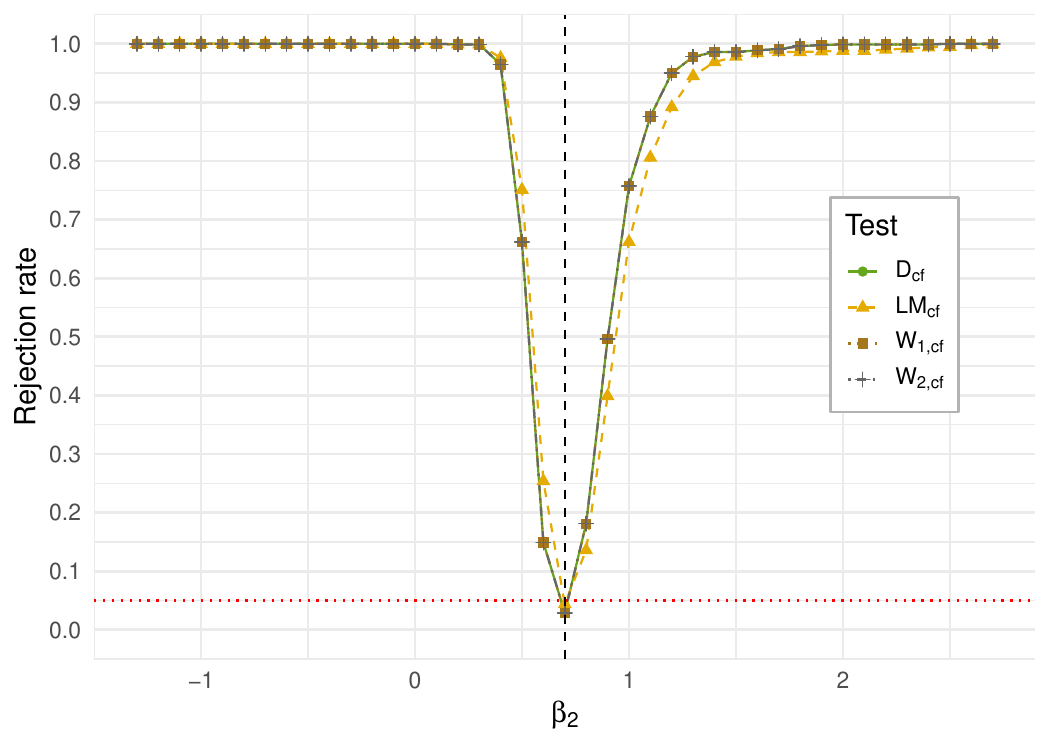}
        \caption{JIVE2}
    \end{subfigure}
    
    \caption{Power curves for DGP2 ($n=200$, $\alpha = 0.1$, $r = 0.2$). Trinity of test statistics distributed as a $\bar\chi^2$ with cross-fit variance. Results based on 1000 repetitions. The horizontal dotted red line denotes the $5\%$ nominal rejection level, while the vertical dotted black line corresponds to $\beta=1$. Panel (a) plots statistics based on the SJIVE objective function; panel (b) plots statistics based on the HLIM objective function; panel (c) plots statistics based on the JIVE1 objective function; panel (d) plots statistics based on the JIVE2 objective function.}
    \label{fig:Trinity_chibar2_cf_DGP2_200_0.1_0.2}
\end{figure}

\begin{figure}[ht]
    \centering
    % Replace 'chibar2' with 'chi2' for the second set of figures as needed
    \begin{subfigure}[b]{0.47\textwidth}
        \includegraphics[width=\textwidth]{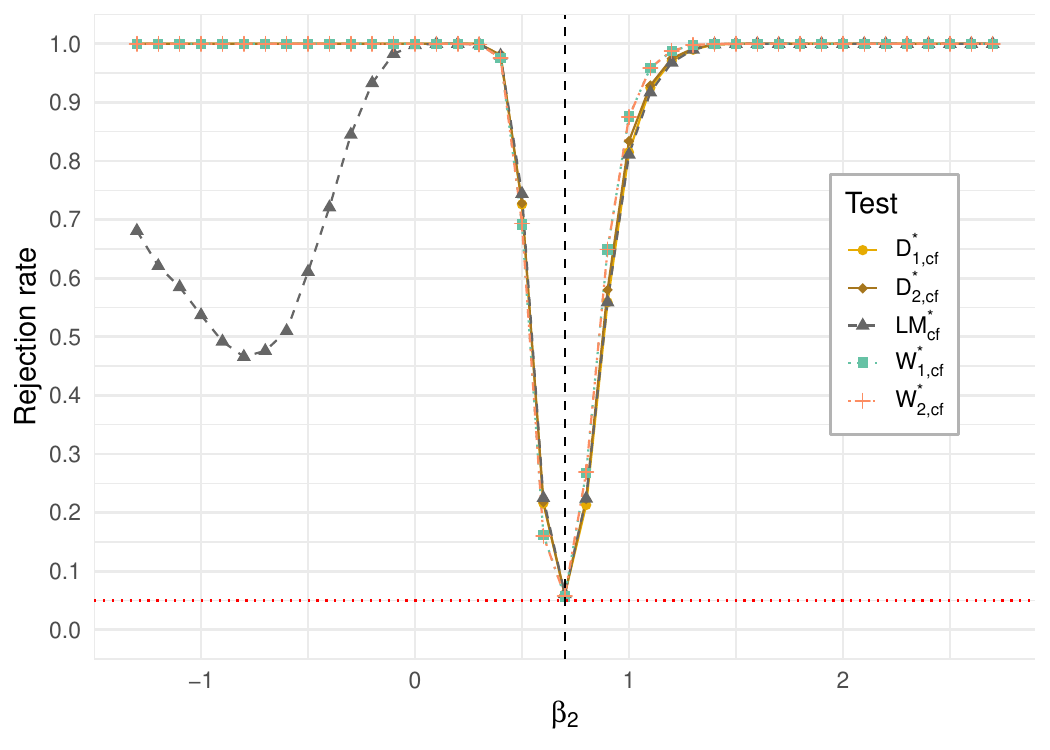}
        \caption{SJIVE}
    \end{subfigure}
    \hfill
    \begin{subfigure}[b]{0.47\textwidth}
        \includegraphics[width=\textwidth]{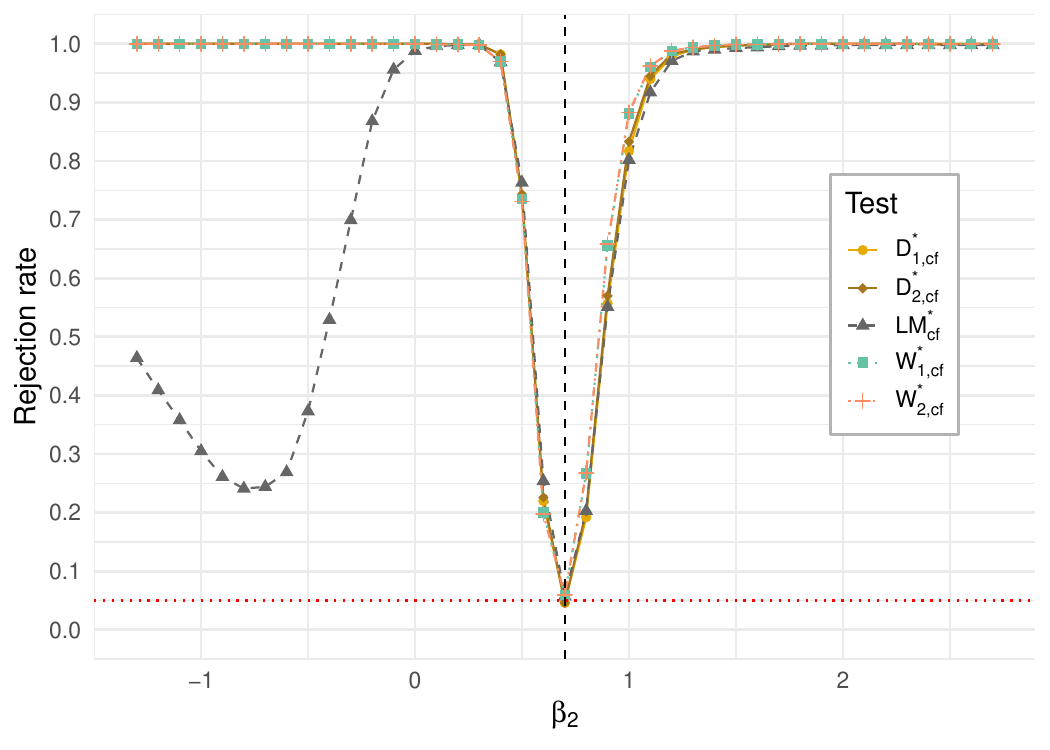}
        \caption{HLIM}
    \end{subfigure}
    
    \begin{subfigure}[b]{0.47\textwidth}
        \includegraphics[width=\textwidth]{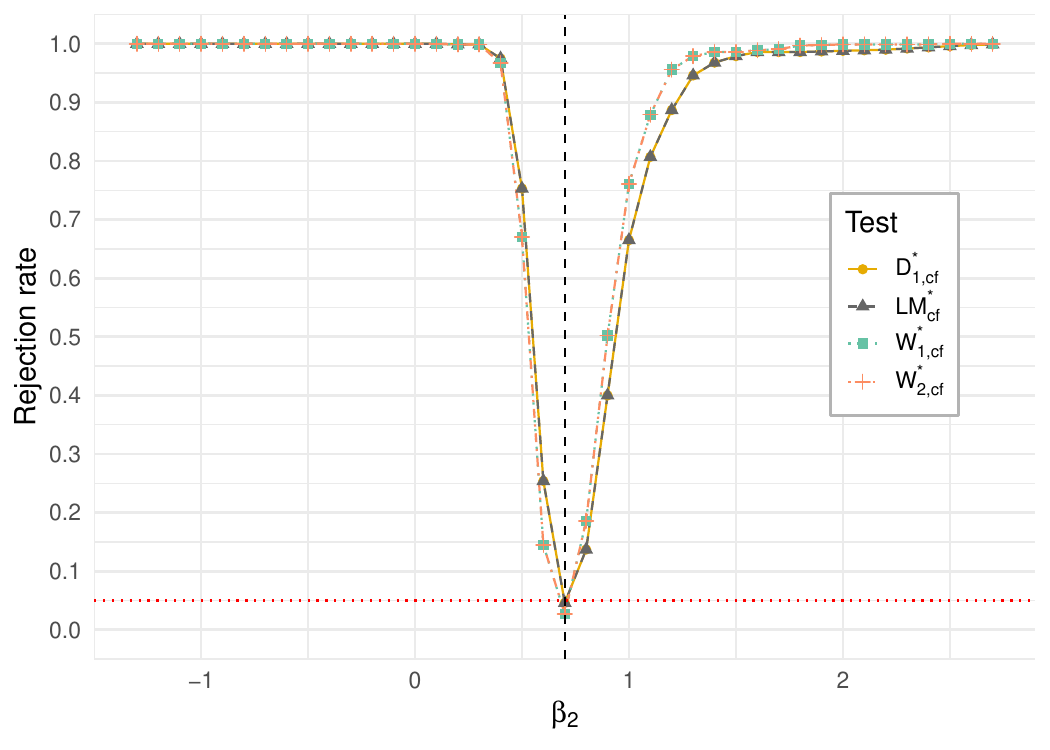}
        \caption{JIVE1}
    \end{subfigure}
    \hfill
    \begin{subfigure}[b]{0.47\textwidth}
        \includegraphics[width=\textwidth]{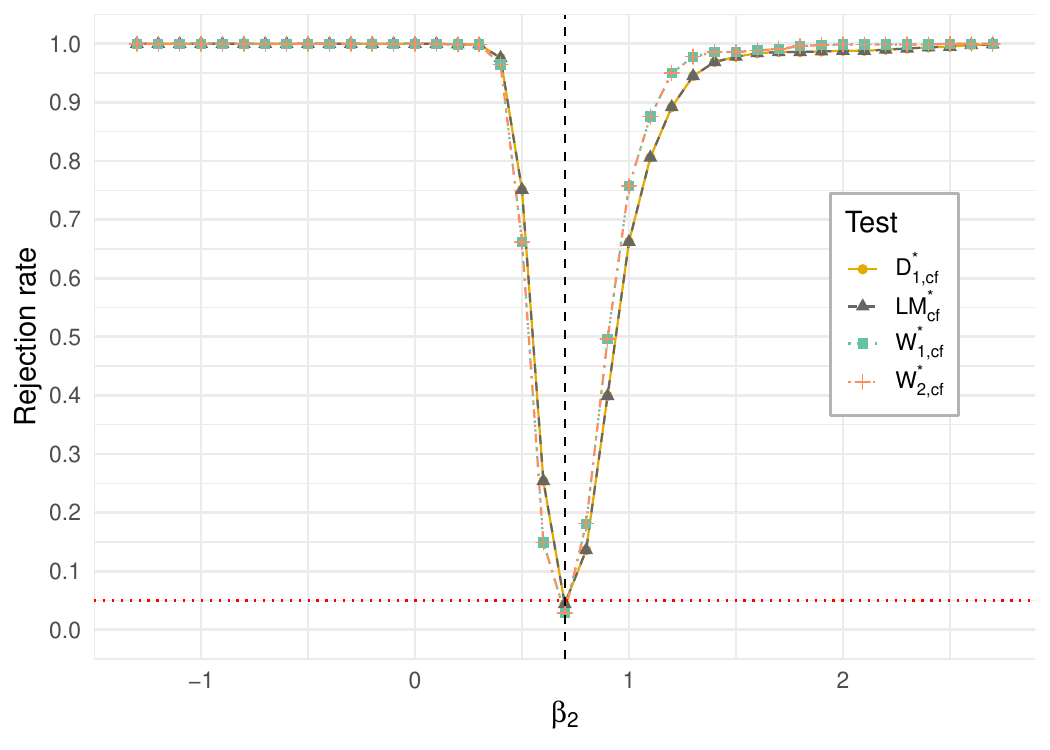}
        \caption{JIVE2}
    \end{subfigure}
    
    \caption{Power curves for DGP2 ($n=200$, $\alpha = 0.1$, $r = 0.2$). Trinity of test statistics distributed as a $\chi^2$ with cross-fit variance. Results based on 1000 repetitions. The horizontal dotted red line denotes the $5\%$ nominal rejection level, while the vertical dotted black line corresponds to $\beta=1$. Panel (a) plots statistics based on the SJIVE objective function; panel (b) plots statistics based on the HLIM objective function; panel (c) plots statistics based on the JIVE1 objective function; panel (d) plots statistics based on the JIVE2 objective function.}
    \label{fig:Trinity_chi2_cf_DGP2_200_0.1_0.2}
\end{figure}

%%%%%%%%%%%%%%%%%%%%%%%%%%%%%%%%%%%%%%%%%%%%%%%%%%%%%%%%%%%%%%%%%%%%%%%%%%%%%%%%%
%%%%%%%%%%%%%%%%%%%%%%%%%%%%%%%%%%%%%%%%%%%%%%%%%%%%%%%%%%%%%%%%%%%%%%%%%%%%%%%%%%

%%%%%%%%%%%%%%%%%%%%%%%%%
% AR_200_0.05 

\begin{figure}[ht]
    \centering
    \begin{subfigure}[b]{0.47\textwidth}
        \includegraphics[width=\textwidth]{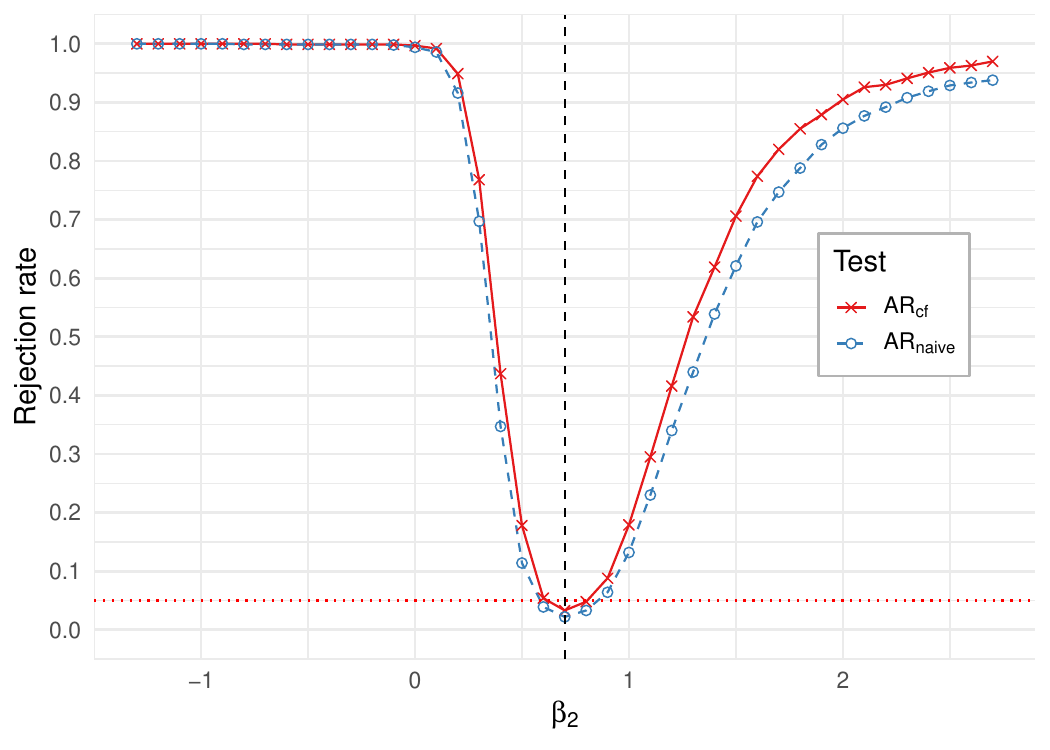}
        \caption{JIVE1, $r=0.1$}
    \end{subfigure}
    \hfill
    \begin{subfigure}[b]{0.47\textwidth}
        \includegraphics[width=\textwidth]{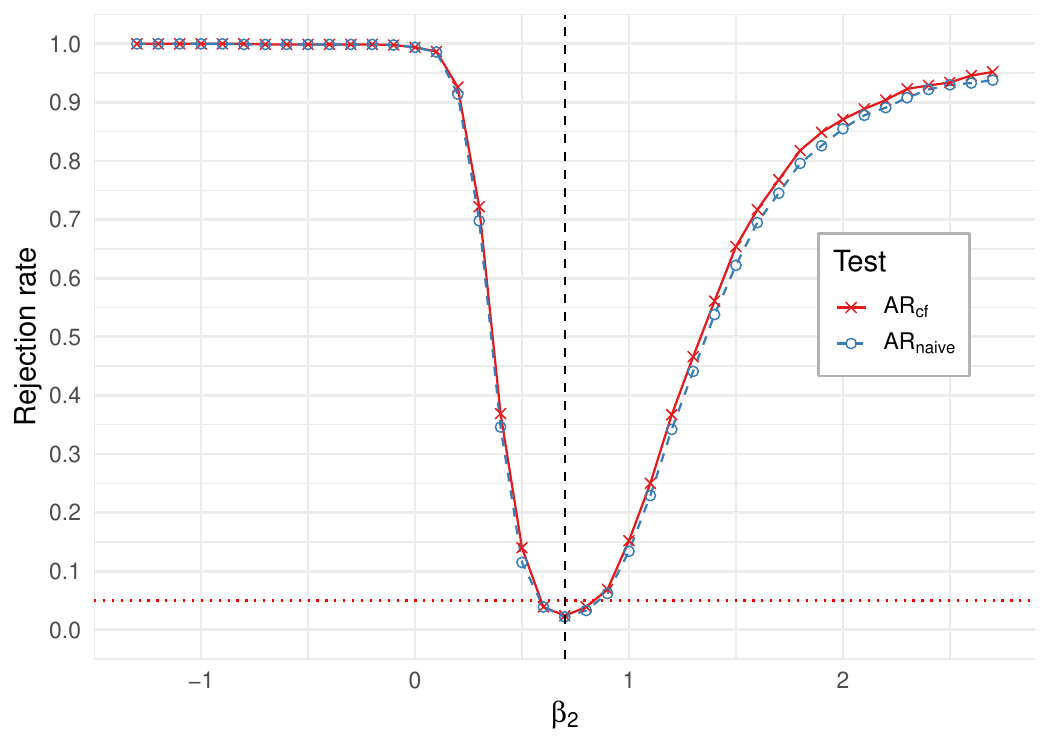}
        \caption{JIVE2, $r=0.1$}
    \end{subfigure}
    
    \begin{subfigure}[b]{0.47\textwidth}
        \includegraphics[width=\textwidth]{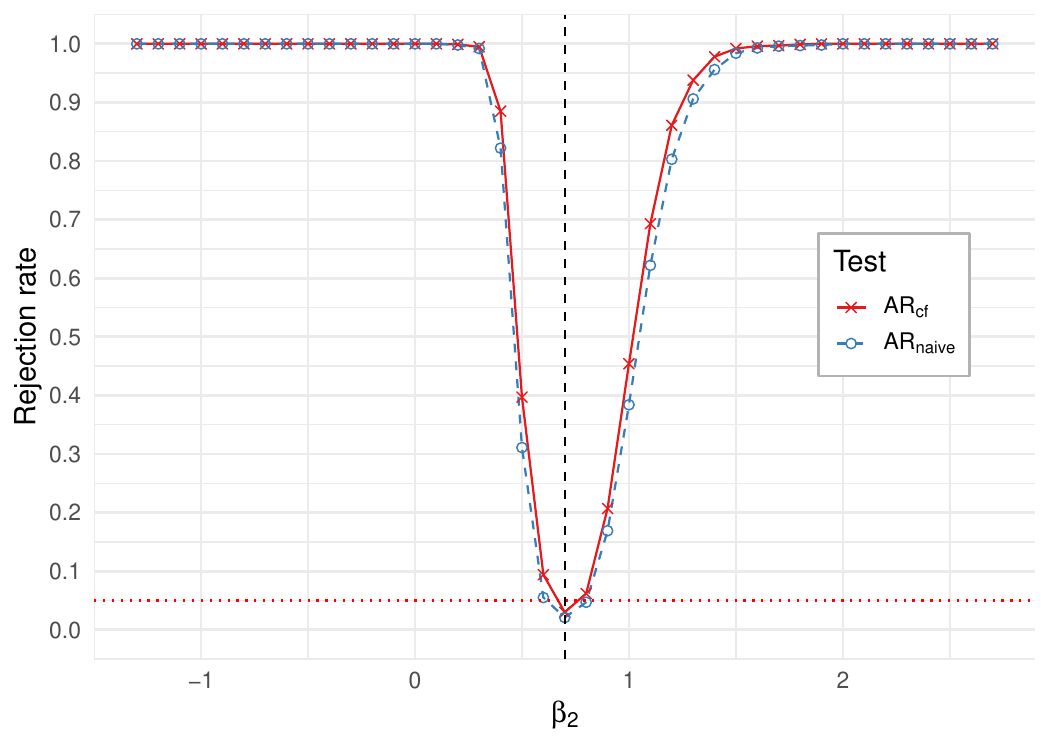}
        \caption{JIVE1, $r=0.2$}
    \end{subfigure}
    \hfill
    \begin{subfigure}[b]{0.47\textwidth}
        \includegraphics[width=\textwidth]{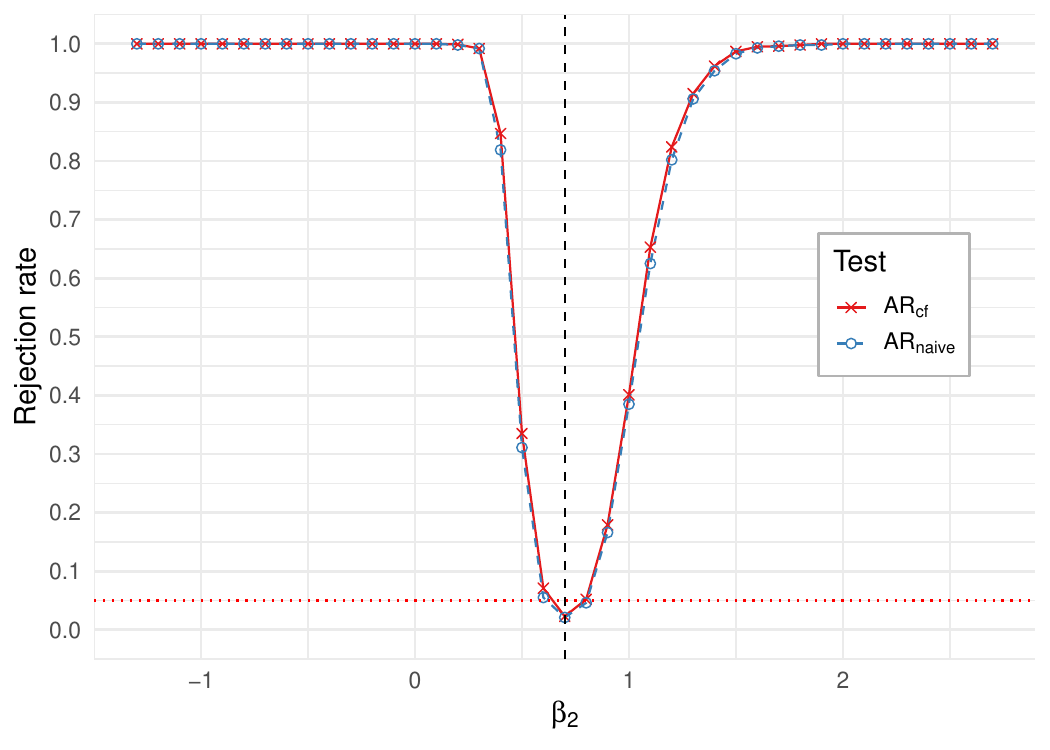}
        \caption{JIVE2, $r=0.2$}
    \end{subfigure}
    
    \caption{Power curves for DGP2 ($n=200$ and $\alpha = 0.05$). AR tests distributed as $\mathcal{N}(0,1)$. Results based on 1000 repetitions. The horizontal dotted red line denotes the $5\%$ nominal rejection level, while the vertical dotted black line corresponds to $\beta=1$. Panel (a) and panel (c) plot AR statistics based on the JIVE1 objective function with weak and strong instruments, respectively; panel (b)  and panel (d) plot AR statistics based on the JIVE2 objective function for weak and strong instruments, respectively.}
    \label{fig:AR_N_DGP2_200_0.05}
\end{figure}

%%%%%%%%%%%%%%%%%%%%%%%%%
% AR_200_0.1

\begin{figure}[ht]
    \centering
    \begin{subfigure}[b]{0.47\textwidth}
        \includegraphics[width=\textwidth]{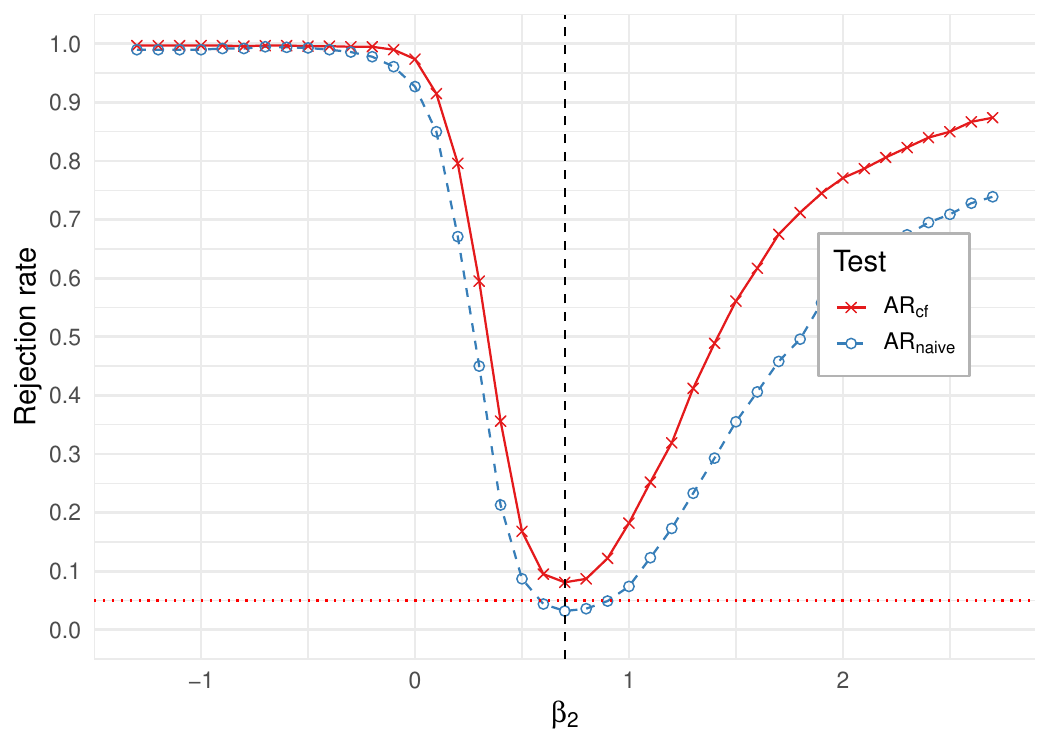}
        \caption{JIVE1, $r=0.1$}
    \end{subfigure}
    \hfill
    \begin{subfigure}[b]{0.47\textwidth}
        \includegraphics[width=\textwidth]{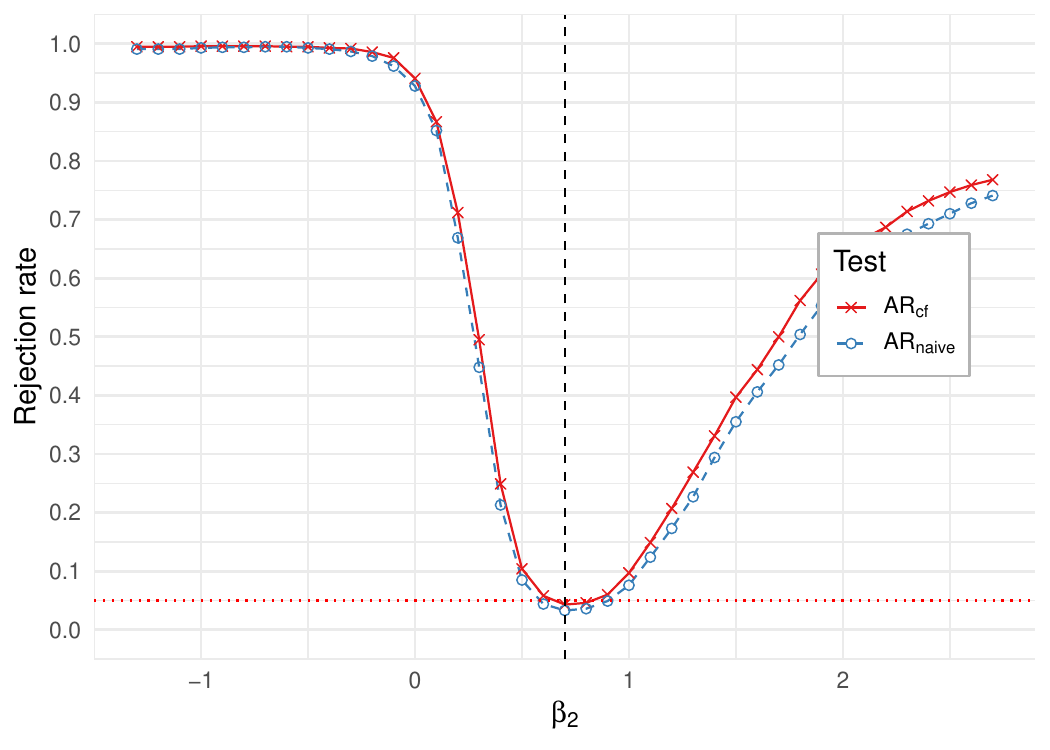}
        \caption{JIVE2, $r=0.1$}
    \end{subfigure}
    
    \begin{subfigure}[b]{0.47\textwidth}
        \includegraphics[width=\textwidth]{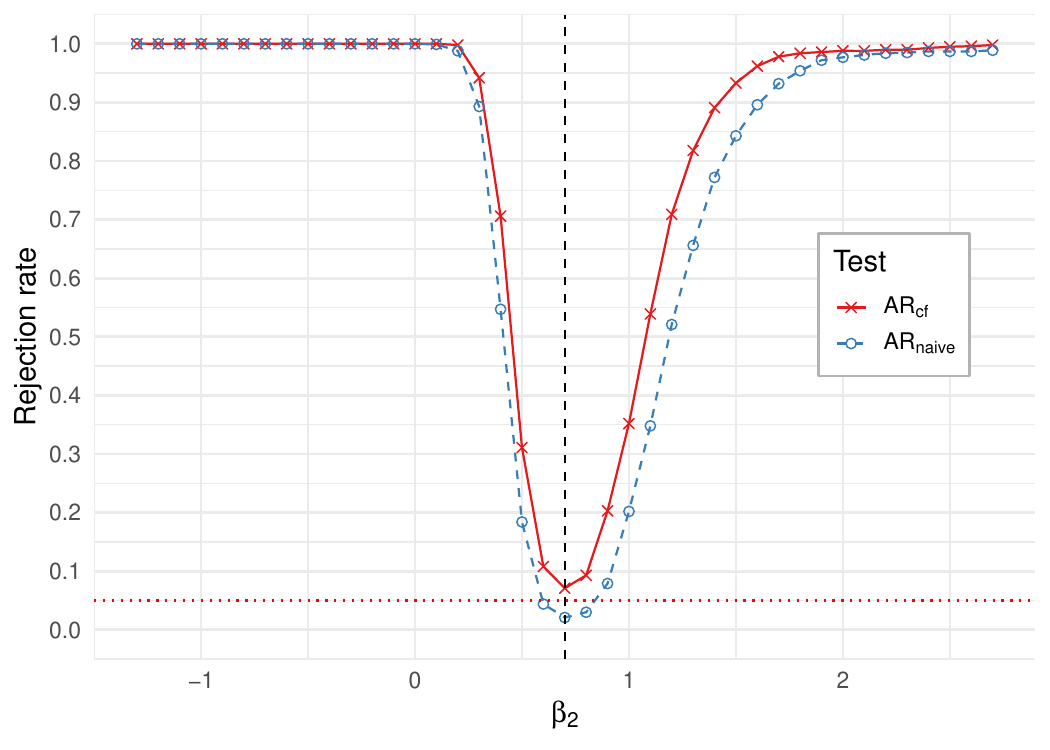}
        \caption{JIVE1, $r=0.2$}
    \end{subfigure}
    \hfill
    \begin{subfigure}[b]{0.47\textwidth}
        \includegraphics[width=\textwidth]{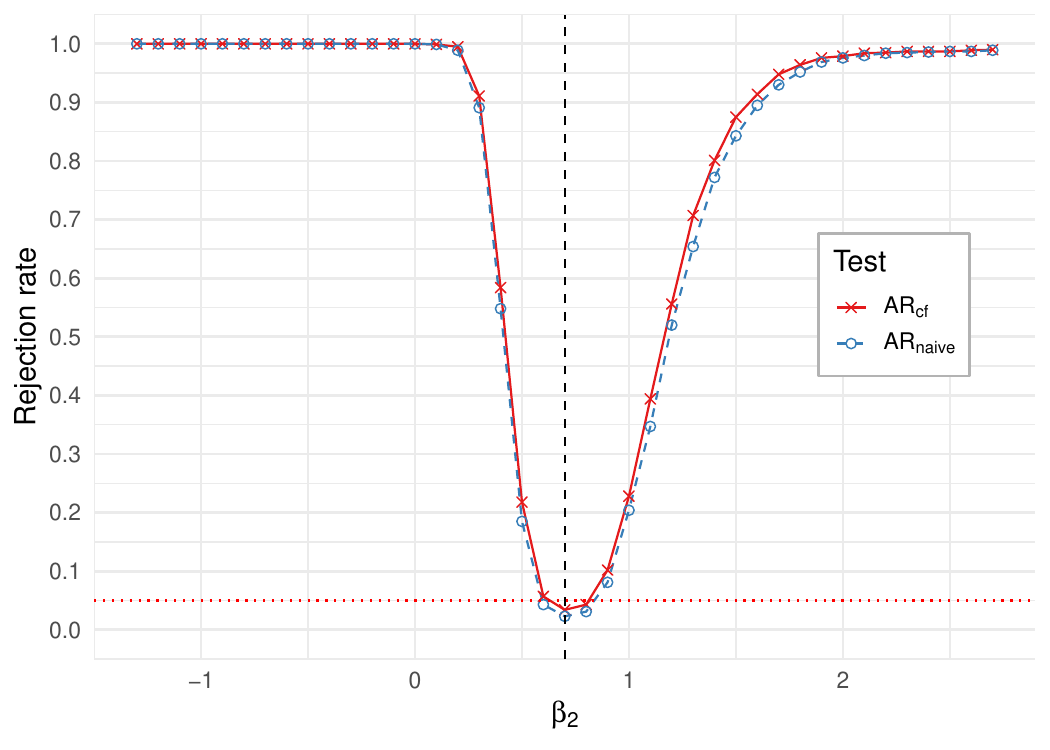}
        \caption{JIVE2, $r=0.2$}
    \end{subfigure}
    
    \caption{Power curves for DGP2 ($n=200$ and $\alpha = 0.1$). AR tests distributed as $\mathcal{N}(0,1)$. Results based on 1000 repetitions. The horizontal dotted red line denotes the $5\%$ nominal rejection level, while the vertical dotted black line corresponds to $\beta=1$. Panel (a) and panel (c) plot AR statistics based on the JIVE1 objective function with weak and strong instruments, respectively; panel (b)  and panel (d) plot AR statistics based on the JIVE2 objective function for weak and strong instruments, respectively.}
    \label{fig:AR_N_DGP2_200_0.1}
\end{figure}

%%%%%%%%%%%%%%%%%%%%%%%%%%%%%%%%%%%%%%%%%%%%%%%%%%%%%%%%%%%%%%%%%%%%%%%%%%%%%%%%%%
%%%%%%%%%%%%%%%%%%%%%%%%%%%%%%%%%%%%%%%%%%%%%%%%%%%%%%%%%%%%%%%%%%%%%%%%%%%%%%%%%%

\clearpage
%\nocite{*}
\bibliographystyle{apalike}
\bibliography{IVLiterature2}
\end{document}